\documentclass[acmsmall,screen]{acmart}

\AtBeginDocument{%
	}

\setcopyright{none}

\usepackage[T1]{fontenc}
\usepackage{hyperref}
\usepackage{fancyhdr}
\usepackage{graphicx}
\usepackage{amsmath}
\usepackage{stmaryrd,wasysym}
\usepackage{booktabs,comment} 
\usepackage{makecell}
\usepackage{units}
\usepackage{pifont}
\usepackage{algpseudocode}

\usepackage{multirow}

\usepackage{framed}

\usepackage[lined,boxed,ruled]{algorithm2e}

\usepackage{amsthm}
\usepackage[normalem]{ulem}
\usepackage{tikz,pgffor}
\usetikzlibrary{arrows}
\usetikzlibrary{shapes}
\usetikzlibrary{calc}
\usetikzlibrary{automata}
\usetikzlibrary{positioning}
\usepackage{longtable}
\usepackage{ltablex}
\usepackage{ltxtable}
\usepackage{ragged2e}
\usepackage{xltabular}

\usepackage{bbm}

\tikzstyle{proglabel}=[shape=circle,draw,inner sep=0pt,minimum size=5mm]
\tikzstyle{tran}=[draw,->,>=stealth, rounded corners]

\usepackage{listings}
\lstdefinelanguage{prog}
{
	morekeywords={if, then, else, fi, while, do, od, true, false, and, or, skip, sample, observe,return,score,normal,uniform, prob},
	sensitive = false
}

\usepackage{threeparttable}

\usepackage{subfig}
\usepackage{framed}

\usepackage{tabularx}
\usepackage{dsfont}

\usepackage{bussproofs}

\usepackage[inline]{enumitem}

\usepackage[capitalise]{cleveref}
\usepackage{thmtools}  

\makeatletter
\providecommand{\bigsqcap}{
	\mathop{
		\mathpalette\@updown\bigsqcup
	}
}
\newcommand*{\@updown}[2]{
	\rotatebox[origin=c]{180}{$\m@th#1#2$}
}
\makeatother

{\gdef\scalefactor{#1}\begin{center}\proofSkipAmount \leavevmode}%
	{\scalebox{\scalefactor}{\DisplayProof}\proofSkipAmount \end{center} }

\newif\ifdraft
\draftfalse

\newcommand{\cc}{\ensuremath{C}} 
\newcommand{\AssignSymbol}{\mathrel{\textnormal{$\mathtt{:=}$}}}
\newcommand{\ASSIGN}[2]{\ensuremath{#1 \AssignSymbol #2}}

\newcommand{\AppAssignSymbol}{\mathrel{\textnormal{\texttt{:}}\hspace{-.1em}{\approx}}}
\newcommand{\PASSIGN}[2]{\ensuremath{#1 \AppAssignSymbol\hspace{.1em} #2}}
\newcommand{\IFSYMBOL}{\ensuremath{\textnormal{\texttt{if}}}}

\newcommand{\ELSESYMBOL}{\ensuremath{\textnormal{\texttt{else}}}}

\newcommand{\WHILESYMBOL}{\ensuremath{\textnormal{\texttt{while}}}}

\newcommand{\TAG}[1]{{\tag*{[\small{{#1}}]}}}

\newtheorem{remark}{Remark}

\definecolor{codegreen}{rgb}{0,0.6,0}
\definecolor{codegray}{rgb}{0.5,0.5,0.5}
\definecolor{codepurple}{rgb}{0.58,0,0.82}
\definecolor{backcolour}{rgb}{0.95,0.95,0.92}

\lstdefinestyle{myStyle}{
    belowcaptionskip=1\baselineskip,
    breaklines=true,
    frame=none,
    basicstyle=\footnotesize\ttfamily,
    keywordstyle=\bfseries\color{green!40!black},
    commentstyle=\itshape\color{purple!40!black},
    identifierstyle=\color{blue},
    backgroundcolor=\color{gray!10!white},
    numberstyle=\tiny\color{codegray},
    stringstyle=\color{codepurple},
    breakatwhitespace=false,                          
    keepspaces=true,                 
    numbers=left,       
    numbersep=5pt,                  
    showspaces=false,                
    showstringspaces=false,
    showtabs=false,                  
    tabsize=2,
}

\newcommand\doubleplus{+\kern-1.3ex+\kern0.8ex}

\newcommand{\Uniform}{\ensuremath{\mathrm{Uniform}}}

\begin{document}

\title{Piecewise Analysis of Probabilistic Programs via $k$-Induction}

\author{Tengshun Yang}
\orcid{0000-0002-2072-0836}
\authornote{Equal Contribution}
\affiliation{%
  \institution{Institute of Software Chinese Academy of Sciences}
  \city{Beijing}
  \country{China}
}
\email{yangts@ios.ac.cn}

\author{Shenghua Feng}
\orcid{0000-0002-5352-4954}
\authornotemark[1]
\affiliation{%
  \institution{Institute of Software Chinese Academy of Sciences}
  \city{Beijing}
  \country{China}
}
\email{fengshenghua@iscas.ac.cn}

\author{Hongfei Fu}
\orcid{0000-0002-7947-3446}
\authornote{The corresponding author} 
\affiliation{%
  \institution{Shanghai Jiao Tong University}
  \city{Shanghai}
  \country{China}
}
\email{jt002845@sjtu.edu.cn}

\author{Naijun Zhan}
\orcid{0000-0003-3298-3817}
\affiliation{%
  \department{School of Computer Science \&  Key Laboratory of High Confidence Software Technology}  
  \institution{Peking University}
  \city{Beijing}
  \country{China}
}
\affiliation{%
  \institution{Zhongguancun Laboratory}
  \city{Beijing}
  \country{China}
}
\email{znj@ios.ac.cn}

\author{Jingyu Ke}
\orcid{0009-0008-6848-3105}
\affiliation{%
  \institution{Shanghai Jiao Tong University}
  \city{Shanghai}
  \country{China}
}
\email{windocotber@gmail.com}

\author{Shiyang Wu}
\orcid{0009-0002-9919-6340}
\affiliation{%
  \institution{Shanghai Jiao Tong University}
  \city{Shanghai}
  \country{China}
}
\email{bravendergclio@gmail.com}

\begin{abstract}

In probabilistic program analysis, quantitative analysis aims at deriving tight numerical bounds for probabilistic properties such as expectation and assertion probability. Most previous works consider numerical bounds over the whole program state space monolithically and do not consider piecewise bounds. Not surprisingly, monolithic bounds are either conservative, or not expressive and succinct enough in general. To derive better bounds, we propose a novel approach for synthesizing piecewise bounds over probabilistic programs. First, we show how to extract useful piecewise information from latticed $k$-induction operators, and combine the piecewise information with Optional Stopping Theorem to obtain a general approach to derive piecewise bounds over probabilistic programs. Second, we develop algorithms to synthesize piecewise polynomial bounds, and show that the synthesis can be reduced to bilinear programming in the linear case, and soundly relaxed to semidefinite programming in the polynomial case. Experimental results show that our approach generates tight piecewise bounds for a wide range of benchmarks when compared with the state of the art.

\end{abstract}

\maketitle

\section{Introduction}\label{sec:introduction}
Probabilistic programming~\cite{DBLP:journals/jcss/Kozen81,DBLP:conf/icse/GordonHNR14,DBLP:journals/corr/abs-1809-10756} is a programming paradigm that extends classical programming languages with probabilistic statements such as sampling and probabilistic branching, and provides a powerful modelling mechanism for randomized algorithms~\cite{DBLP:conf/lics/BartheGGHS16}, machine learning~\cite{DBLP:conf/pldi/BeutnerOZ22}, reliability engineering~\cite{DBLP:conf/oopsla/CarbinMR13}, etc. 
Therefore, analysis of probabilistic programs is becoming increasingly significant, and attracting more and more attention in recent years. 

In this work, we consider the quantitative analysis problem that aims at automated approaches that derive quantitative bounds for probabilistic programs. Common quantitative properties include expected runtime~\cite{DBLP:conf/esop/KaminskiKMO16,DBLP:journals/jacm/KaminskiKMO18,DBLP:conf/vmcai/FuC19,DBLP:conf/cav/AbateGR20}, expected resource consumption~\cite{DBLP:conf/pldi/Wang0GCQS19,DBLP:conf/pldi/Wang0R21,DBLP:conf/pldi/NgoC018}, sensitivity~\cite{DBLP:journals/pacmpl/0001BHKKM21},  assertion probabilities~\cite{DBLP:conf/pldi/WangS0CG21,DBLP:conf/popl/ChatterjeeNZ17,DBLP:journals/toplas/TakisakaOUH21}, and so forth. Most existing works focus on deriving numerical bounds instead of solving the semantic equations exactly, as the latter is impossible theoretically in general.  
In the literature, various approaches have been proposed to address the quantitative analysis problem, including template-based constraint solving~\cite{DBLP:conf/qest/GretzKM13,DBLP:conf/cav/ChakarovS13,DBLP:conf/cav/ChatterjeeFG16,DBLP:journals/toplas/ChatterjeeFNH18}, trace abstraction~\cite{DBLP:journals/pacmpl/SmithHA19},
sampling~\cite{DBLP:conf/pldi/SankaranarayananCG13}, etc. Most of these approaches consider to synthesize a monolithic bound over the whole state space of a probabilistic program of interest, and have the following disadvantages: First, a monolithic bound is either too conservative (e.g., only very coarse bounds exist) or not succinct enough (e.g., although tight monolithic  bounds exist, the tightness usually requires complicated polynomials with higher degree). 
Second, it may be even worse that no monolithic polynomial bounds exist. 

It is straightforward to observe that piecewise bounds are more accurate than monolithic bounds.  Moreover, a recent work~\cite{DBLP:journals/pacmpl/BatzBKW24} demonstrates that probabilistic program analysis requires piecewise feature. However, the synthesis of piecewise bounds for probabilistic programs is not  well investigated in the literature. 
To our best knowledge, a handful relevant work is  by~\cite{DBLP:conf/tacas/BatzCJKKM23}. They propose an approach for generating (piecewise) invariants to \emph{verify} user-provided linear bounds for probabilistic programs with discrete probabilistic choices, which is based on Counterexample-Guided Inductive Synthesis (CEGIS) and template refinement. Another relevant work is ~\cite{DBLP:conf/cav/BaoTPHR22} that proposes a data-driven approach that can synthesize piecewise (sub-)invariants over probabilistic programs with discrete probabilistic choices. Their approach prefers a suitable list of numerical program features (such as multiplication expressions over variables), which requires prior knowledge of the program or user's assistance. Both of these related works require a bound to be verified as an additional program input when synthesizing (super-/sub-) invariants.  

In this work, we propose a novel automated approach that synthesizes piecewise polynomial bounds for probabilistic programs with discrete probability choices without user-provided bounds or piecewise features to assist the derivation of the piecewise bound. The challenges are that (a) We need to resolve a good criterion to partition the state space of a probabilistic program into multiple parts in order to derive the form of the target piecewise bound. (b) We need to devise efficient algorithms to synthesize piecewise bounds given the criterion. Our detailed contributions to address these challenges are as follows.

To address the first challenge, we consider latticed $k$-induction operators~\cite{DBLP:conf/cav/BatzCKKMS20,DBLP:conf/apsec/LuX22}. $k$-induction is a powerful proof tactic in software and hardware verification that generalizes normal inductive reasoning~\cite{DBLP:conf/fmcad/SheeranSS00,DBLP:conf/cav/KrishnanVGG19,DBLP:conf/cav/MouraRS03,DBLP:conf/sas/DonaldsonHKR11}. While standard induction assumes the property holds after a single step, $k$-induction extends this approach by considering sequences of $k$ steps, allowing properties to be established even when a single-step induction is impossible. Latticed $k$-induction~\cite{DBLP:conf/cav/BatzCKKMS20,DBLP:conf/apsec/LuX22} further adapts $k$-induction to lattices and has been applied to probabilistic program analysis~\cite{DBLP:conf/cav/BatzCKKMS20}. 
In this work, we leverage latticed $k$-induction as the central criterion for partitioning the whole state space into multiple parts, and the Optional Stopping Theorem (see the classical Optional Stopping Theorem (OST)~\cite[Chapter 10]{DBLP:books/daglib/0073491}) provides the theoretical foundation of our methods. We develop a novel combination of operators from latticed $k$-induction and OST, which enables the synthesis of both upper and lower bounds for quantitative properties for  probabilistic programs without requiring a global bound of program values (such as non-negativity in~\cite{DBLP:conf/cav/BatzCKKMS20,DBLP:conf/apsec/LuX22,DBLP:conf/tacas/BatzCJKKM23}).
Importantly, the combination is non-trivial: the classical OST is insufficient in our setting, and we rely on an extended version of OST~\cite{10.1145/3656432} to establish our results. Additionally, as a by-product, we slightly extend existing latticed $k$-induction operators. 

To address the second challenge, we propose novel algorithms for synthesizing piecewise linear and polynomial bounds w.r.t our combination of latticed $k$-induction and OST. It is important to observe that the latticed $k$-induction involves \emph{minimum/maximum} operation, and therefore increases the difficulty to synthesize a bound algorithmically. We first introduce a key improvement in time efficiency on the unrolling of the $k$-induction operators. Then, we show that the synthesis of piecewise linear bounds can be equivalently transformed into a bilinear programming problem. A bilinear programming problem is that the variables can be decomposed into two groups so that within each group of variables the constraints are linear, and is a special non-convex programming that admits efficient constraint solving~\cite{DBLP:journals/mp/McCormick76}.
Finally, since even on the linear benchmarks we require piecewise polynomials to upper/lower bound the quantitative properties, we show that the synthesis of the more general piecewise polynomial bounds can be soundly relaxed to semidefinite programming. 
Experimental results over an extensive set of benchmarks that includes various benchmarks from the literature show that our approach is capable of generating tight or even accurate piecewise bounds and can solve benchmarks that previous approaches could not handle. 

\smallskip
\noindent\emph{Technical Contributions.} 
Approaches with latticed $k$-induction has inherent combinatorial explosion~\cite{DBLP:conf/cav/BatzCKKMS20,DBLP:conf/apsec/LuX22}. To address the difficulty, we propose two techniques. The first is a heuristic selection of a small part of the functions in the minimum operation of latticed $k$-induction. The second is the sound relaxation that over-approximates the minimum operation with convex combination.  

\section{Preliminaries}\label{sec:prelim}
In this section, we briefly review probability theory, define the $k$-induction operators, present the probabilistic loops under consideration, and finally formulate the problem 
of interest. 

\subsection{Probability Theory and Martingales}

Consider a probability space $(\Omega, \mathcal{F}, \mathbb{P})$, where $\Omega$ is the sample space, $\mathcal{F}$ is a $\sigma$-algebra on $\Omega$ and $\mathbb{P}: \mathcal{F} \rightarrow [0, 1]$ is a probability measure on the measurable space $(\Omega, \mathcal{F})$. A \textit{random variable} is an $\mathcal{F}$-measurable function $X: \Omega \rightarrow \mathbb{R} \cup \{+\infty, -\infty\}$, i.e., a function satisfying that for all $d \in \mathbb{R} \cup \{+\infty, -\infty\}$, $\{\omega \in \Omega: X(\omega) \leq d\} \in \mathcal{F}$.
The \textit{expectation} of a random variable $X$, denoted by $\mathbb{E}(X)$, is the Lebesgue integral of $X$ w.r.t. $\mathbb{P}$, i.e., $\mathbb{E}(X) = \int \! X d\mathbb{P}$. A \textit{filtration} of the probability space $(\Omega, \mathcal{F}, \mathbb{P})$ is an infinite sequence $\{\mathcal{F}_n\}_{n=0}^{\infty}$ such that for every $n$, the triple $(\Omega, \mathcal{F}_n, \mathbb{P})$ is a probability space and $\mathcal{F}_n \subseteq \mathcal{F}_{n+1} \subseteq \mathcal{F}$. A \textit{stopping time} w.r.t. $\{\mathcal{F}_n\}_{n=0}^{\infty}$ is a random variable $\tau: \Omega \rightarrow \mathbb{N} \cup \{0, \infty\}$ such that for every $n \geq 0$, the event $\{\tau \leq n\} \in \mathcal{F}_n$, i.e., $\{\omega \in \Omega: \tau(\omega) \leq n\} \in \mathcal{F}_n$. Intuitively, $\tau$ is interpreted as the time at which the stochastic process shows a desired behavior. A \textit{discrete-time stochastic process} is a sequence $\Gamma = \{X_n\}_{n=0}^{\infty} $ of random variables in $(\Omega, \mathcal{F}, \mathbb{P})$. The process $\Gamma$  is adapted to a filtration $\{\mathcal{F}_n\}_{n=0}^{\infty}$, if for all $n \geq 0$, $X_n$ is a random variable in $(\Omega, \mathcal{F}_n, \mathbb{P})$.
A discrete-time stochastic process $\Gamma = \{X_n\}_{n=0}^{\infty} $ adapted to a filtration $\{\mathcal{F}_n\}_{n=0}^{\infty}$ is a \textit{martingale} (resp. supermartingale, submartingale) if for all $n \geq 0$, $\mathbb{E}(|X_n|) < \infty$ and it holds almost surely that $\mathbb{E}(X_{n+1}|\mathcal{F}_n) = X_n$ (resp. $\mathbb{E}(X_{n+1}|\mathcal{F}_n) \leq X_n$, $\mathbb{E}(X_{n+1}|\mathcal{F}_n) \geq X_n$). See~\cite {DBLP:books/daglib/0073491} for more details about martingale theory. Applying martingales for probabilistic programs analysis is well-studied~\cite{DBLP:conf/cav/ChakarovS13,DBLP:conf/cav/ChatterjeeFG16,DBLP:conf/popl/ChatterjeeNZ17}.

\subsection{$k$-Induction Operators}\label{sec:operators}

To present $k$-induction operators, we briefly review lattice theory.
Informally, a lattice is a partially ordered set $(E,\sqsubseteq)$ (where $E$ is a set and $\sqsubseteq$ is a partial order on $E$) equipped with a \emph{meet} operation $\sqcap$ and a \emph{join} operation $\sqcup$. Given two elements $u,v\in E$, the meet $u \sqcap v$ is defined as the infimum of $\{u,v\}$ and dually the join $u\sqcup v$ is defined as the supremum of $\{u,v\}$. A partially ordered set $(E,\sqsubseteq)$ is a \emph{lattice} if for any $u,v\in E$, we have that both $u\sqcap v$ and $u\sqcup v$ exist. Given a lattice $(E,\sqsubseteq)$, we say that an operator $\varPhi: E \rightarrow E$ is \emph{monotone} if for all $u,v\in E$, $u\sqsubseteq v$ implies $\varPhi(u)\sqsubseteq \varPhi(v)$. Throughout this section, we fix a lattice $(E,\sqsubseteq)$ and a monotone operator~$\varPhi: E \rightarrow E$.

We recall the $k$-induction operator given in~\cite{DBLP:conf/cav/BatzCKKMS20} as follows, which we refer to as the \emph{upper} $k$-induction operator. 

\begin{definition}[Upper $k$-Induction Operator {\cite{DBLP:conf/cav/BatzCKKMS20}}]\label{def:k_induction_operator}
Given any element $u \in E$, the upper $k$-induction operator $\varPsi_u$ w.r.t. $u$ and the monotone operator $\varPhi$ is defined by: $\varPsi_{u}: E \rightarrow E, v \mapsto \varPhi(v) \sqcap u$~.
\end{definition}

Below we propose a dual version for the upper $k$-induction operator. The intuition is simply to replace the meet operation with join. We call this dual operator as the \emph{lower} $k$-induction operator. 

\begin{definition}[Lower $k$-Induction Operator]\label{def:dual_k_induction_operator1}
Let $u\in E$. The \emph{dual} $k$-induction operator $\varPsi'_u$ w.r.t. $u$ and the aforementioned monotone operator $\varPhi$ is defined by: $\varPsi'_u: E \rightarrow E, v \mapsto \varPhi(v) \sqcup u$~.
\end{definition}

\begin{remark}
Alterative formulation of the $k$-induction operators have also been proposed in~\cite{DBLP:conf/apsec/LuX22}. In Appendix~\ref{app:operators}, We show that these formulations are essentially equivalent to the definitions adopted in this work. Therefore, in the rest of this paper, we focus exclusively on the upper and lower $k$-induction operators defined above.
\qed
\end{remark}

\subsection{Probabilistic Loops}\label{sec:programs}

In this work, we  
use simple 
probabilistic while loops of the form (\ref{eq:pwhileloop}) for easing  the explanation of our basic idea, and will  
discuss how to
extend our approach to general probabilistic while loops like nested loops without substantial changes in ~\cref{sec:extensions}.
Below we define the class of single probabilistic loops.

\smallskip
\noindent {\em Syntax. }A probabilistic while loop takes the form 
\begin{equation}\label{eq:pwhileloop}
\textbf{while} ~ (\varphi) ~  \{C\}
\end{equation}
\noindent where $\varphi$ is the loop guard and $C$ is the loop body without loops. Formally, the loop guard $\varphi$ and loop body $C$ are generated by the grammar in Figure~\ref{fig: syntax}, 
\begin{figure}
    \centering
    \begin{align*}
        C &::= ~\textsf{skip} \mid x:=e  \mid x:\approx \mu \mid C;C \mid \{C\} ~\lbrack p \rbrack  ~\{C\} \mid 
         \textsf{if} ~(\varphi)~\{C\} ~\textsf{else} ~\{C\} \\
        \varphi &::= ~e < e \mid \neg \varphi \mid \varphi \wedge \varphi \ \qquad
        e ::= c \mid x \mid e \cdot e \mid e + e \mid e - e 
    \end{align*}
    \caption{Syntax of Loop Guard and Body in the form (\ref{eq:pwhileloop})}
    \label{fig: syntax}
\end{figure}
where $x$ is a program variable taken from a countable set \textsf{Vars} of variables, $c \in \mathbb{R}$ is a real constant, $e$ is an arithmetic expression that involves addition and multiplication, $\varphi$ is a formula over program variables that is a Boolean combination of arithmetic inequalities, and $\mu$ is a predefined probability distribution. 
In this work, we consider $\mu$ to be a finite discrete probability distribution (i.e., distributions with a finite support) such as Bernoulli distribution and discrete uniform distribution. The semantics of 
\textsf{skip}, assignment, sequential composition, conditional, and  while statement can be understood 
as their counterparts in imperative programs. 
The semantics of a probabilistic choice $\{C_1\}[p]\{C_2\}$ is   that flips a coin with bias $p\in [0, 1]$ and executes the statement $C_1$ if the coin yields head, and $C_2$ otherwise. 
The semantics of a sampling statement 
$x\,{:\approx}\,\mu$  samples a value according to the predefined distribution $\mu$ and assigns the value to the variable $x$. 

Given a probabilistic while loop, a \emph{program state} is a function that maps every program variable to a real number. 
We denote by $S$ the set of program states.
The initial state for a probabilistic while loop is denoted by $s^*$. 
The evaluation $\varphi(s)$ of a logical formula $\varphi$ 
and the evaluation $e(s)$ of an arithmetic expression $e$ over a program state $s$ are defined in the standard way.
$\varphi(s)=\mathtt{true}$ is 
denoted by $s \models \varphi$. 

\smallskip
\noindent {\em Semantics. } The semantics of a probabilistic loop of the form (\ref{eq:pwhileloop}) can be interpreted  as a discrete-time Markov chain, where the state space is the set of all program states $S$, and the transition probability function $\mathbf{P}$ is given by the loop body $C$ and determines the probability  $\mathbf{P}(s,s')$ for $s,s'\in S$, meaning the probability 
producing output state $s'$ 
from input state $s$. If the loop guard $\varphi(s)$ evaluates to false, then we treat the program state $s$ as a sink state, that is $\mathbf{P}(s,s)=1$ and $\mathbf{P}(s,s')=0$ for $s\ne s'$.

Given the Markov chain of a probabilistic while loop as described above, a \emph{path} is an infinite sequence $\pi=s_0, s_1,\dots,s_n,\dots$ of program states such that $\mathbf{P}(s_n, s_{n+1})>0$ for all $n\ge 0$. Intuitively, each $s_n$  corresponds to the state right before the $(n+1)$-th loop iteration. A program state $s$ is \emph{reachable} from an initial program state $s^*$ if there exists a path $\pi = s_0, s_1,\dots$ such that $s_0 = s^*$ and $s_n = s$ for some $n \geq 0$, and define $\mbox{\sl Reach}(s^*)$ as the set of reachable states starting from the initial state $s^*$. By the standard cylinder construction (see e.g. ~\cite[Chapter 10]{DBLP:books/daglib/0020348}), the Markov chain with a designated initial program state $s^*$ for the probabilistic loop induces a probability space over paths and reachable states. We denote the probability measure in this probability space by $\mathbb{P}_{s^*}$ and its related expectation operator by  $\mathbb{E}_{s^*}$. 

\smallskip
\noindent {\em Problem formulation. } Given  a probabilistic loop $P$ in the form (\ref{eq:pwhileloop}), assuming that $P$ terminates with probability 1, a \emph{return function} $f$ is a function $f:S\rightarrow \mathbb{R}$ that is used to specify the output of the loop $P$ in the sense that when the loop $P$ terminates at a program state $s$, then the return value is given as $f(s)$.  
A return function is \emph{piecewise polynomial} if it can be expressed as a piecewise polynomial expression in program variables. We denote by $X_f$ the random variable for the return value of the loop given a return function $f$. In this work, we consider the following problem: Given a probabilistic while loop $P$ in the form (\ref{eq:pwhileloop}) and a piecewise polynomial return function $f$, synthesize \emph{piecewise upper and lower bounds} on the expected value of $X_f$.

\section{An Overview of Our Approach}\label{sec:overview}
Our approach falls in the background of (latticed) $k$-induction~\cite{DBLP:conf/cav/BatzCKKMS20,DBLP:conf/apsec/LuX22}. $k$-induction is an induction principle that generalizes the standard  induction by considering $k$ consecutive transitions together in the inductive condition. Roughly speaking, given a predicate $P$ to be proved via induction, the $k$-induction principle considers the inductive condition as $(P(x_1)\wedge \dots \wedge P(x_k))\rightarrow P(x_{k+1})$, 
for which the premise $P(x_1)\wedge \dots\wedge P(x_k)$ means that the predicate $P$ holds for $k$ consecutive transitions, and the whole condition states that if $P$ holds for $k$ consecutive transitions, then $P$ holds after these consecutive transitions. In particular, $1$-induction coincides with the usual inductive condition.

Latticed $k$-induction~\cite{DBLP:conf/cav/BatzCKKMS20} adapts the idea of $k$-induction to lattices for deriving bounds of fixed points. It considers $k$ consecutive applications of a monotone operator over a lattice and applies the {\em meet/join} operations iteratively in the $k$ consecutive applications. The parameter $k$ here does not matter in the monotone operator (see Definitions~\ref{def:k_induction_operator} and \ref{def:dual_k_induction_operator1}), but is the number of iterative applications (see Definition~\ref{def:potential function}) when the operator is applied. In this work, we propose a novel combination of latticed $k$-induction operators and Optional Stopping Theorem (OST), and propose novel algorithms for deriving piecewise linear and polynomial bounds on probabilistic programs.

We illustrate the main idea of our approach via the following example, which is 
a discretized version of the \textsc{Growing Walk}  in~\cite{DBLP:conf/pldi/BeutnerOZ22}:
    \[
 	{\textsc{Growing Walk}}\colon
		\quad
		\WHILESYMBOL\left(\,0\leq x\,\right)\{
		\{\ASSIGN{x}{x+1}; \ASSIGN{y}{y+x}\}~[0.5]~\{\ASSIGN{x}{-1}\}
		\}
    \]
    
\noindent The example models a simple random walk where the step size $x$ is increased by $1$ with one half probability, and set to $-1$ with the other half probability. 
The program terminates when $x$ becomes negative. The objective is to analyze the expected value of the return function $f(x, y)=y$, which corresponds to the total traveled distance $y$, after the program terminates. We take the synthesis of piecewise linear upper bound as an example.

\smallskip
\noindent {\em Step 1: Establishing $k$-induction operators. }
Let  $\overline{\varPhi}_f\colon (\mathbb{R}\times \mathbb{R} \to \mathbb{R}) \to (\mathbb{R}\times \mathbb{R} \to \mathbb{R})  $ be the operator
\[
\overline{\varPhi}_f(h):=\lambda (x,y). [x<0] \cdot y + [x\ge 0] (0.5\cdot h(x+1, y+x+1) + 0.5\cdot h(-1, y))
\]
for function $h:\mathbb{R}\times \mathbb{R}\rightarrow\mathbb{R}$, and $[x\ge 0]$ denotes the Iverson-bracket of the predicate $x\ge 0$, which evaluates to 1 if $x \ge 0$ holds at state $s$ and 0 otherwise.
Intuitively, $\overline{\varPhi}_f$ outputs $y$ if the loop guard $x \ge 0$ is violated, and the expected value of $h(x, y)$ after the execution of the loop body $\{\ASSIGN{x}{x+1}; \ASSIGN{y}{y+x}\}~[0.5]~\{\ASSIGN{x}{-1}\}$ otherwise. 
We introduce the $k$-induction operator $\varPsi_h$ (c.f. ~\cite{DBLP:conf/cav/BatzCKKMS20}), defined by $\varPsi_h(g) := \min \{\overline{\varPhi}_f(g), h \}$ for any fixed function $h:\mathbb{R}\times \mathbb{R}\rightarrow\mathbb{R}$. Informally, when applied to a function $g$, the operator $\varPsi_h(g)$ pulls $\overline{\varPhi}_f(g)$ down via the pointwise minimum operation with $h$. 

\smallskip
\noindent {\em Step 2: Applying $k$-induction condition. }
Let $k=2$. We unroll the loop $P$ $(k=2)$ times and examine the $(k=2)$-induction condition to upper-bound the expected value of $X_f$.
The resultant inductive condition from our approach is as follows (here $\le$ is taken pointwise), which is obtained by applying the operator $\varPsi_h$ to a candidate bound function $h$ once (i.e., $k-1$ times):
\begin{equation}\label{eg:overview_condition}
    \overline{\varPhi}_f(\varPsi_h(h)) \le h
\end{equation}
We show that under a mild assumption and by using OST, if we have a function $h$ that fulfills this inductive condition, then $\varPsi_h(h)$ is an upper bound for the expected value of $X_f$, for which the \emph{pointwise minimum} in $\varPsi_h(h) = \min \{\overline{\varPhi}_f(h), h\}$ is the key to derive the piecewise partition of the bound apart from loop unrolling.

\smallskip
\noindent {\em Step 3: Simplifying the $k$-induction condition. }
Our approach synthesizes a function $h$ w.r.t the condition (\ref{eg:overview_condition}). To the end, we reduce the condition~(\ref{eg:overview_condition}) to the form below with four functions $h_i$ ($1\le i\le 4$) combined with a minimum operation: 
\begin{equation}\label{eg:overview_condition_miniform}
    \min \{h_1, h_2, h_3, h_4\} \le h,
\end{equation}
where $h_1 = [x < 0]\cdot y + [x \ge 0] \cdot (0.5\cdot h(x+1, x+y+1)+0.5 \cdot h(-1, y))$, $h_2 = [x < 0]\cdot y + [x \ge 0] \cdot (0.25\cdot h(-1, y+x+1) + 0.25 \cdot h(x+2,2x+y+3)+0.5 \cdot h(-1, y))$, $h_3 = [x < 0]\cdot y + [x \ge 0] \cdot (0.25\cdot h(-1, y+x+1) + 0.25\cdot h(x+2,2x+y+3) + 0.5 \cdot y)$ and $h_4 = [x < 0]\cdot y + [x \ge 0] \cdot (0.5\cdot h(x+1, x+y+1) + 0.5 \cdot y)$. Using our algorithm, we employ a loop unrolling based approach to efficiently derive the simplified constraint (\ref{eg:overview_condition_miniform}) 
and we show that each $h_i$ results from the unfolding of the loop up to depth $k=2$ and corresponds to a loop-free program from the unfolding. See {\bf Stage 2} in~\cref{sec:algorithm} for the details. 

\smallskip
\noindent {\em Step 4: Solving the simplified $(k=2)$-induction condition. }
After {\em Step 3}, we obtain the constraint in (\ref{eg:overview_condition_miniform}) and further synthesize the function $h$ in (\ref{eg:overview_condition_miniform}) by assuming a template for $h$ and solving the template w.r.t. the constraint (\ref{eg:overview_condition_miniform}). Every synthesized  function $h$ leads to a piecewise upper bound $\varPsi_h(h) = \min \{\overline{\varPhi}_f(h), h \}$ for the expected value of $X_f$.
Since this constraint includes a minimum operation, it is non-convex and non-trivial to solve.
Our approach reduces the synthesis problem with a linear template to bilinear programming, and obtains a piecewise linear upper bound $[x<0]\cdot y + [x\geq 0]\cdot (x+y+2)$, which is actually the exact expected value of $y$. Similarly, our method can also obtain a piecewise linear lower bound $[x<0]\cdot y + [x\geq 0]\cdot (x+y+\nicefrac{13}{8})$.  

\section{Piecewise Bounds via Latticed $k$-Induction}\label{sec:functions}
In this section, we propose a novel combination of OST and latticed $k$-induction operators to derive bounds for the expected value of $X_f$. We first introduce expectation functions over which we construct concrete $k$-induction operators, then define potential functions, and finally show the soundness of potential functions to derive expectation bounds via OST. Throughout this section, we fix a probabilistic while loop $P=\textbf{while} (\varphi) \{C\}$ in the form of (\ref{eq:pwhileloop}) and a return function $f$. 

\subsection{Expectation Functions}

\begin{definition}[Expectation Functions]\label{def:expfunc}
An \emph{expectation function} is a function $h: S \rightarrow \mathbb{R}$ that assigns to each program state a real value. The partial order $\preceq$ over expectation functions is defined in the pointwise fashion, i.e., $h_1 \preceq h_2  \! \iff \! \forall s \in S, h_1(s) \leq h_2(s)$. We denote the set of expectation functions by $\mathcal{E}$ and the lattice by $(\mathcal{E}, \preceq)$, for which the meet operation $\sqcap$ in the lattice is given by $h_1\sqcap h_2 := \min \{h_1, h_2\}, $where $\min$ is the pointwise minimum on functions, i.e., $\forall s \in S, \min \{h_1, h_2\}(s) = \min\{h_1(s), h_2(s) \}$, and the join operation $\sqcup$ is given by $h_1\sqcup h_2 := \max \{h_1, h_2\}$, where $\max$ is the pointwise maximum.
\end{definition}

Informally, an expectation function $h$ is that for each program state $s\in S$, the value $h(s)$ bounds the expected value of  $X_f$ after the execution of the while loop $P$ when the loop starts with the program state $s$. Although one observes that the partially ordered set $(\mathcal{E}, \preceq)$ with the meet and join operations defined above is a lattice, we do not use lattice properties in our approach.

To instantiate the $k$-induction operators for expectation functions, we construct the monotone operator for the lattice $(\mathcal{E}, \preceq)$. To this end, we first define the notion of pre-expectation as follows, 
wherein $[\varphi]$ denotes the Iverson-bracket of $\varphi$. Notice that the random assignment command $x\,{:\approx}\,\mu$ (where $\mu$ is a discrete distribution  of finite support) can be written in an iterative style of $\{C_1\} ~\lbrack p \rbrack  ~\{C_2\}$, so that we define pre-expectation without random assignment commands.

\begin{definition}[Pre-expectation~\cite{DBLP:conf/cav/ChakarovS13,DBLP:conf/pldi/Wang0GCQS19}]\label{def:pre-exp}
    Given an expectation function $h: S\rightarrow \mathbb{R}$. We define its \textit{pre-expectation} over a loop-free program $Q$, $pre_Q(h): S \rightarrow \mathbb{R}$, recursively on the structure of $Q$:
    \begin{itemize}
        \item $pre_Q(h) := h$, if $Q \equiv \textsf{skip}$. 
        \item $pre_Q(h) := h[x/e]$, if $Q \equiv x := e$, where $h[x/e]$ denotes $h[x/e](s) = h(s[x/e])$ with $s[x/e](x) =e(s)$ and $s[x/e](y) = s(y)$ for all $y \in \textsf{Vars} \backslash \{x\}$.
        \item $pre_Q(h) := pre_{Q_1}(pre_{Q_2}(h))$, if $Q \equiv Q_1; Q_2$. 
        \item $pre_Q(h) := p \cdot pre_{Q_1}(h) + (1-p) \cdot pre_{Q_2}(h) $, if $Q \equiv \{Q_1\} ~\lbrack p \rbrack  ~\{Q_2\}$. 
        \item $pre_Q(h) := [ \phi ] \cdot pre_{Q_1}(h) + [\neg \phi] \cdot pre_{Q_2}(h) $, if $Q \equiv \textsf{if} ~(\phi)~\{Q_1\} ~\textsf{else} ~\{Q_2\}$. 
    \end{itemize}
\end{definition}

The intuition of pre-expectation is that given an expectation function $h$, the pre-expectation $pre_Q$ computes the expected value $pre_Q(h)$ of $h$ after the execution of the command $Q$. 
With pre-expectation, we then define the monotone operator to be the characteristic function $\overline{\varPhi}_f$ of the probabilistic loop $P$ with respect to the return function $f$ as follows. 

For the rest of this section, we fix an initial state $s^*$ and override the set $S$ of program states with $\mbox{\sl Reach}(s^*)$ in Definition~\ref{def:expfunc} so that we consider expectation functions restricted to $\mbox{\sl Reach}(s^*)$.

\begin{definition}[Characteristic Function~\cite{DBLP:conf/esop/KaminskiKMO16,DBLP:conf/cav/ChakarovS13}]\label{def:characteristic function}
The \emph{characteristic function} $\overline{\varPhi}_f: \mathcal{E} \rightarrow \mathcal{E}$ is defined by 
$\overline{\varPhi}_f(h) := [\neg \varphi] \cdot f + [\varphi] \cdot pre_C(h)$.
The monotone operator for the lattice $(\mathcal{E}, \preceq)$ is defined as $\overline{\varPhi}_f$.
\end{definition}
Informally, the characteristic function $\overline{\varPhi}_f$ outputs $f$ if the loop guard $\varphi$ is violated and the loop terminates in the next step, and the pre-expectation of $h$ w.r.t. the loop body $C$ otherwise. 
It is straightforward to verify the monotonicity of $\overline{\varPhi}_f$. 
In the following, we omit the subscript $f$ in $\overline{\varPhi}_f$ if it is clear from the context.
Given the monotone operator, we establish the concrete $k$-induction operators as follows. 

\begin{definition}[$k$-Induction Operators for $(\mathcal{E}, \preceq)$]\label{def:k_induc_func}

Given an expectation function $h$, the \emph{upper (resp. lower) $k$-induction operator} $\overline{\varPsi}_h: \mathcal{E} \rightarrow \mathcal{E}$ (resp. $\overline{\varPsi}'_h: \mathcal{E} \rightarrow \mathcal{E}$) is defined by 
$\overline{\varPsi}_h(g) = \min\{\overline{\varPhi}_f(g), h\}$ (resp. $\overline{\varPsi}'_h(g)= \max\{\overline{\varPhi}_f(g), h\}$)~ for arbitrary expectation function $g \in \mathcal{E}$.

\end{definition}

Note that $k$ does not explicitly appear within the operators; rather, it denotes the number of times these operators are iteratively applied.

\subsection{Potential Functions}\label{sec:potential function}

We define potential functions as expectation functions satisfying the $k$-induction conditions. These potential functions serve as candidate bounds to be synthesized.

\begin{definition}[Potential Functions]~\label{def:potential function}
Let $k$ be a positive integer. A \emph{$k$-upper} (resp. \emph{$k$-lower}) potential function is an expectation function $h$ that satisfies the \emph{upper} (resp. \emph{lower}) $k$-induction condition $\overline{\varPhi}_f(\overline{\varPsi}_h^{k-1}(h)) \preceq h $ (resp. $\overline{\varPhi}_f((\overline{\varPsi}'_h)^{k-1}(h)) \succeq h$), respectively.
\end{definition}

\begin{remark}
Note that both the \emph{base case} and the \emph{induction case} are encapsulated in the operator $\overline{\Phi}_f$ encoding the $k$-induction condition. Specifically, by the definition of $\overline{\Phi}_f$, the $k$-induction condition $\overline{\Phi}_f(\overline{\Psi}_h^{k-1}(h)) \preceq h$ is equivalent to $[\neg \varphi] \cdot f + [\varphi] \cdot \operatorname{pre}_C(\overline{\Psi}_h^{k-1}(h)) \preceq h$. This, in turn, requires that:
\begin{enumerate}
\item $f(s) \preceq h(s)$ when the program state $s$ satisfies $s \models \neg\varphi$, corresponding to the base case; and
\item $\operatorname{pre}_C(\overline{\Psi}_h^{k-1}(h))(s) \preceq h(s)$ when $s \models \varphi$, corresponding to the induction case. \qed
\end{enumerate}
\end{remark}

We apply Optional Stopping Theorem (OST) to address our soundness results. We find that the classical OST~\cite{DBLP:books/daglib/0073491,doi:10.1080/00029890.1971.11992788} cannot handle our problem due to the requirement of bounded step-wise difference (see~\cref{app:classical_OST}), while the OST variant proposed in~\cite{10.1145/3656432} can handle our problem.

\begin{theorem}[Extended OST~\cite{10.1145/3656432}]~\label{thm:ost-variant}
Let $\{X_n\}_{n = 0}^{\infty}$ be a supermartingale adapted to a filtration 
$\mathcal{F}=\{\mathcal{F}_n\}_{n=0}^{\infty}$ and $\tau$ be a stopping time w.r.t the filtration $\mathcal{F}$.
Suppose there exist positive real numbers $b_1,b_2,c_1,c_2,c_3$ such that $c_2>c_3$ and
\begin{itemize}
\item[(a)]  For all sufficiently large natural numbers $n$, it holds that $\mathbb{P}(\tau > n) \leq c_1 \cdot e^{-c_2 \cdot n}$. 
\item [(b)] For every natural number $n\ge 0$, it holds almost-surely that $\left\vert X_{n+1}-X_n \right\vert \le b_1\cdot n^{b_2}\cdot e^{c_3\cdot n}$.
\end{itemize} 
Then we have that $\mathbb{E}(\lvert X_\tau\rvert)<\infty$ and $\mathbb{E}(X_\tau) \leq \mathbb{E}(X_0)$. 
\end{theorem}

Under certain side conditions that guarantee the validity of the extended OST, the potential functions provide upper and lower bounds on the expected value of $X_f$. Before presenting this result, we introduce some concepts that capture the magnitude of updates to program variables between two consecutive steps.

\begin{definition}[Affine Programs]
    A probabilistic program is affine if all conditions and assignments within the program are \textbf{affine} functions of the program variables. 
\end{definition}

\begin{definition}[Termination Time]\label{def:termination_time}
    The {\em termination time} $T$ of the loop $P$ is the random variable that for any path of the loop, measures the number of total loop iterations in the path.
\end{definition}
\vspace{0.5ex}
\begin{definition}[Uniform Amplifier]~\label{def:uniform_amplifier}
    Suppose that the loop $P$ is affine, and for each program variable $x$, let $x_n$ denote the random variable representing the value of $x$ at the $n$-th iteration of the loop. A \emph{uniform amplifier} $c$ is  a constant $c > 0$ such that, for all $n \geq 0$, $|x_{n+1}| \leq c \cdot |x_n| + a$ holds for some fixed constant $a$.
\end{definition}

\begin{definition}[Bounded Update]
    The loop $P$ has the \emph{bounded-update} property if there exists a real constant $a > 0$ such that for each program variable $x$, $|x_{n+1} - x_n |\le a$ for every $n \ge 0$ (see Definition~\ref{def:uniform_amplifier} for the meaning of $x_n$).  
\end{definition}

\begin{remark}
    Note that any program satisfying the bounded update property also admits a uniform  amplifier with $c = 0$. \qed
\end{remark}

We now present the soundness theorem of $k$-upper (resp. lower) potential functions. We distinguish between \emph{affine programs} and \emph{polynomial programs}, as each requires different side conditions for potential functions to serve as upper or lower bounds. Notably, the side conditions for affine programs are weaker than those for polynomial programs. 

\begin{theorem}\label{thm:soundness}
Suppose the loop $P$ is affine. Let $k$ be a positive integer and $h$ be a polynomial potential polynomial in the program variables with degree $d$. If there exist real numbers $c_1 >0$ and $c_2 > c_3 > 0$ such that
\begin{itemize}
    \item [(P1)] there exists a uniform  amplifier $c$ satisfying $c \leq e^{c_3/d}$, and 
    \item [(P2)] the termination time $T$ of $P$ has the \emph{concentration property}, i.e., $\mathbb{P}(T > n) \leq c_1 \cdot e^{-c_2 \cdot n}$ holds for sufficiently large $n \in \mathbb{N}$.
\end{itemize}
hold, then for any initial program state $s^*$, we have:
\begin{itemize}
\item $\mathbb{E}_{s^*}(X_f)\le \overline{\varPsi}_h^{k-1}(h)(s^*) \le h(s^*)$ holds for any $k$-upper potential polynomial $h$.

\item $\mathbb{E}_{s^*}(X_f)\ge (\overline{\varPsi}'_h)^{k-1}(h)(s^*) \ge h(s^*)$ holds for any $k$-lower potential polynomial $h$. 
\end{itemize}

\end{theorem}

\begin{proof}[Proof Sketch]
(See~\cref{app:soundness} for the full proof)
Let $s_n$ be the random variable of the program state at the $n$-th iteration with $s_0=s^*$, and let $H = \overline{\varPsi}_h^{k-1}(h)$. 
A {\em key point} is that since $H$ is piecewise polynomial (by the definition of $\overline{\varPsi}_h$) and condition (P1) holds, condition (b) in Theorem~\ref{thm:ost-variant} holds for process $ \{X_n\}_{n\in \mathbb{N}} = \{H(s_n)\}_{n\in \mathbb{N}}$. Combining with the fact that $h$ is a $k$-upper potential function, one can further deduce $ \{X_n\}_{n\in \mathbb{N}} = \{H(s_n)\}_{n\in \mathbb{N}}$ is a supermartingale. 
Note that $X_0$ is the initial value of the process $X_0$. and $X_T$ represents the value of the process $X_n$ at loop termination. 
By applying Theorem~\ref{thm:ost-variant}, we have $\mathbb{E}_{s^*}(X_T) \leq\mathbb{E}_{s^*}(X_0)$ ($T$ is a stopping time), thus $\mathbb{E}_{s^*}(X_f)\le \mathbb{E}_{s^*}(X_0) = H(s^*)$. The lower case is derived similarly.
\end{proof}

The side condition (P1) for affine programs requires that the loop $P$ possesses a uniform  amplifier constant. In contrast, for polynomial programs, a stronger property is needed: the program must satisfy the bounded update property, which imposes stricter constraints than (P1).

\begin{theorem}\label{thm:soundness_poly}
Let $k$ be a positive integer. Suppose there exist real numbers $c_1 >0$ and $c_2 > 0$ such that condition
(P1') loop $P$ has the bounded update property; and condition (P2) in~\cref{thm:soundness} holds, then for any initial program state $s^*$, we have
\begin{itemize}
\item $\mathbb{E}_{s^*}(X_f)\le \overline{\varPsi}_h^{k-1}(h)(s^*) \le h(s^*)$ holds for any $k$-upper potential polynomial $h$.

\item $\mathbb{E}_{s^*}(X_f)\ge (\overline{\varPsi}'_h)^{k-1}(h)(s^*) \ge h(s^*)$ holds for any $k$-lower potential polynomial $h$. 
\end{itemize}

\end{theorem}

\begin{remark}
    See~\cref{app:soundness_poly} for the proof of ~\cref{thm:soundness_poly}. The concentration condition (P2), which ensures exponentially decreasing nontermination probabilities as stated in~\cref{thm:soundness,thm:soundness_poly}, guarantees that loop $P$ terminates almost surely. This condition has been extensively studied in the literature (see, e.g., \cite{DBLP:conf/popl/ChatterjeeFNH16,DBLP:conf/cav/ChatterjeeFG16,DBLP:journals/pacmpl/FengCSKKZ23}).  \qed
\end{remark}

According to~\cref{thm:soundness,thm:soundness_poly}, synthesizing upper and lower bounds reduces to finding a potential function $h$ that satisfies the conditions outlined in these theorems. However, solving the $k$-upper and $k$-lower potential function conditions is challenging due to the intricate combination of minimum and indicator functions involved. In the following sections, we introduce algorithmic approaches to systematically synthesize these upper and lower bounds.

\section{Algorithms for Bound Synthesis}\label{sec:algorithm}
In this section, we first present algorithms for synthesizing upper and lower bounds for single-loop programs. We then demonstrate how our approach naturally extends to handle programs containing nested or sequential loops.

\subsection{Algorithms for Probabilistic  Single Loops}
In this subsection, we present algorithms for synthesizing 
$k$-upper and $k$-lower potential functions that satisfy the conditions specified in~\cref{thm:soundness} and~\cref{thm:soundness_poly}, leading to piecewise bounds on the expected value of $X_f$.
Below, we consider a fixed probabilistic loop $P$ of the form~(\ref{eq:pwhileloop}) along with a return function $f$. Due to the space limit, we only illustrate the synthesis procedure for upper bounds. The case for lower bounds is nearly  analogous, obtained by replacing minimum with maximum and substituting $\preceq$ by $\succeq$. The pseudocode for our algorithm is presented in~\cref{alg}. Our approach consists of the following major steps: 

\smallskip
\noindent\textbf{Stage 1: Prerequisites Checking and External Inputs.}
Our algorithm first verifies the side conditions (P1) and (P2) (respectively, (P1') and (P2')) for affine (respectively, polynomial) programs, as specified by~\cref{thm:soundness,thm:soundness_poly}. The algorithm also accepts the hyperparameter $k$ and a program invariant as input parameters.

\smallskip
\noindent {\em Prerequisites checking.}  
When $P$ is affine, condition (P1) is verified by syntactically inspecting the loop body to identify a positive constant $c_3$, ensuring that each program variable is amplified by at most $e^{c_3/d}$, up to an additive constant, within a single loop iteration, where $d$ denotes the degree of the polynomial template potential function $h$ (c.f. Stage 2). Condition (P2) is guaranteed either by synthesizing a difference-bounded ranking supermartingale (dbRSM) that demonstrates the exponentially decreasing concentration property~\cite{DBLP:conf/popl/ChatterjeeFNH16,DBLP:conf/cav/ChatterjeeFG16}, or by syntactically analyzing probabilistic branching within the loop to extract a suitable constant $c_2$ satisfying $c_2 > c_3 > 0$. For polynomial programs, condition (P1')—the bounded update property—is checked via an SMT solver (e.g., Z3~\cite{Z3paper}), while condition (P2) is ensured analogously to the affine case.

\smallskip
\noindent {\em External inputs. } 
Our algorithm requires the following hyperparameters as input:
\begin{enumerate*}
    \item \emph{Induction parameter $k$:} We specify a positive real number $k$ as the parameter for $k$-induction, along with the initial program state $s^*$.
    \item \emph{Program invariant:} We assume the existence of an invariant $I$ at the entry point of the loop, which over-approximates the set of reachable program states $\mbox{\sl Reach}(s^*)$. That is, for every $s \in \mbox{\sl Reach}(s^*)$, we have $s \models I$. The state space is thus restricted to program states satisfying $I$, and the relation $\preceq$ is interpreted over $I$, i.e., $h_1 \preceq h_2 \iff \forall s \models I,\, h_1(s) \leq h_2(s)$. The rational  of this restriction follows from the over-approximation property of $I$. Invariants can be obtained using external invariant generators, such as~\cite{DBLP:conf/sas/SankaranarayananSM04}.
\end{enumerate*}

\begin{example}\label{eg:runningexample}
We still take the example in Section~\ref{sec:overview} as a running example, which is a discretized version of the \textsc{Growing Walk}  in~\cite{DBLP:conf/pldi/BeutnerOZ22}: 
\[ \WHILESYMBOL\left(\,0\leq x\,\right)\{
		\{\ASSIGN{x}{x+1}; \ASSIGN{y}{y+x}\}~[0.5]~\{\ASSIGN{x}{-1}\} \}
\]
In this example, our goal is to analyze the expected value of $y$ upon program termination. 
We check the prerequisites and specify the external inputs as follows:
\begin{enumerate*}
    \item \emph{Prerequisite Verification}: We find that $c = 1$ serves as a uniform amplifier, satisfying $c \leq e^{c_3/d}$ with $c_3 = \ln 1.5$ and $d = 1$. The concentration condition (P2) is also met with $c_2 = \ln 2$.
    \item \emph{External Inputs}: We set $k=2$, and choose the invariant $I = \{x\mid -1 \leq x\}$ with initial state $s^* = (x, y) = (1, 1)$.
\end{enumerate*}
\qed 
\end{example}

\smallskip
\noindent\textbf{Stage 2: Templates and Constraints.} After verifying the prerequisites and identifying the external inputs as described in \textbf{Stage 1}, our algorithm predefines a $d$-degree polynomial template $h$ as the candidate $k$-upper potential function for the loop $P$. This template consists of a linear combination of all monomials in the program variables of degree at most $d$, where each monomial is multiplied by an unknown coefficient.

Next, we apply the $k$-induction conditions in Definition~\ref{def:potential function}, resulting in the constraint ${\overline{\varPhi}_f(\overline{\varPsi}_h^{k-1}(h))} \allowbreak \preceq h$. The presence of $\min$ and indicator operators within this constraint complicates direct simplification. To address this, we reformulate the constraint into the form $\min\{h_1, h_2, \dots, h_m\} \preceq h$, where each $h_i$ is free of the minimum operator. Although a brute-force arithmetic expansion can achieve this transformation (see~\cref{app:general_approach} for details), our algorithm employs a more efficient unfolding strategy, which we outline below.

\smallskip
\noindent {\em The unfolding process for constraint simplification: }
We symbolically unroll the probabilistic loop from the initial state up to $k$ iterations, exploring all possible unfolding strategies. Here, "symbolic" means that program variables in each program state retain their original variable names and represent undetermined values.
An \emph{unfolding strategy} operates at each symbolic program state encountered during the unfolding process (excluding the initial state),  and iteratively selects among actions (i), (ii), and (iii): Action (i) corresponds to unfolding the loop once more. Action (ii) represents the scenario where we actively choose to stop unrolling before reaching the total of $k$ unfoldings (the $k$-induction parameter). In contrast, action (iii) occurs when unrolling is forcibly stopped because the number of unfoldings has reached $k$; beyond this point, no further unrolling is allowed by our method.
Each unfolding strategy, determined by the choices made at each unfolding step, yields a loop-free program. Let $C_1, \dots, C_m$ denote all loop-free programs generated by applying the above decision process across all possible unfolding strategies. 
For each loop-free program $C_d$, we compute the \emph{pre-expectation} $pre_{C_d}(h)$ of $h$ with respect to $C_d$ (see~\cref{def:pre-exp}), allowing us to equivalently rewrite the constraint $\overline{\varPhi}_f(\overline{\varPsi}_h^{k-1}(h)) \preceq h$ as:
\begin{align}\label{eq:minform}
\min \{h_1, h_2,\dots, h_m \} \preceq h,
\end{align}
where each $h_i$ is given by $pre_{C_d}(h)$ for some $C_d$. According to the computation of pre-expectation (\cref{def:pre-exp}), each $h_i$ can be represented as $h_i = \sum_r [B_{ir}] \cdot e_{ir}$, where $B_{ir}$ is a predicate independent of the template's unknown coefficients, and $e_{ir}$ is a monolithic polynomial in the program variables, potentially containing unknown coefficients.
Moreover, the $B_{ir}$'s are pairwise logically disjoint. 

The following proposition formally establishes the relationship between the unfolding process and the $k$-induction condition. The proof is provided in~\cref{proof:relation}.
\begin{proposition}\label{prop:relation}
The upper $k$-induction condition  $\overline{\varPhi}_f(\overline{\varPsi}_h^{k-1}(h)) \preceq h$ is equivalent to  constraint $\min\{h_1, h_2,\dots, h_m\} \preceq h$, where each $h_i$ equals $pre_{C_{d}}(h)$ for some unique $C_d \in \{C_1,\dots, C_m\}$ from the unfolding process above. 
\end{proposition}

By Proposition~\ref{prop:relation}, the $k$-induction constraint can be simplified by computing the pre-expectations of all programs $\{C_1, \dots, C_m\}$ generated by all possible unfolding strategy within $k$ loop iterations. Since these programs are structurally similar, we can efficiently compute $pre_{C_d}(h)$ for all $C_d \in \{C_1, \dots, C_m\}$ simultaneously by traversing the $k$-unfolding of the program loop once. This approach reduces runtime by eliminating excessive and repeated computations.

\begin{remark}
The unfolding process also provides an explanation for how the partitioning of the input domain is determined in the final piecewise bound $\overline{\varPsi}_h^{k-1}(h)$. Specifically, the partition of the input domain is primarily governed by the number of iterations required for the state to potentially violate the loop guard, mirroring the $k$-fold unrolling of the loop in $k$-induction. This partition is further refined by the application of the $\min$ operator in the $k$-induction construction, resulting in a tighter and more precise bound. \qed
\end{remark}

\smallskip
\noindent {\em Illustrative Example of the Unfolding Process. }
We demonstrate our unfolding process via a simple but illustrative example as follows: \begin{equation}\label{eg:simple_loop}
P:= \WHILESYMBOL\left(\,\varphi(x)\,\right)\{\{\ASSIGN{x}{a_1 x + b_1}\}~[p]~\{\ASSIGN{x}{a_2 x + b_2}\} \}
\end{equation}
where $x$ is a real-valued program variable, $a_i, b_i (i = 1, 2)$ are real constants, 
$p\in [0, 1]$ and $\varphi(x)$ is a guard condition. 
Let $f$ be the return function, and let $\overline{\varPhi}_f$ be the operator defined as
\[
\overline{\varPhi}_f(h)(x):=[\neg \varphi(x)] \cdot f(x) + [\varphi(x)] (p\cdot h(a_1 x + b_1) + (1-p)\cdot h(a_2 x + b_2))
\]
for any function $h : \mathbb{R} \rightarrow \mathbb{R}$ (with $S = \mathbb{R}$), where $[\varphi]$ denotes the Iverson bracket for the predicate $\varphi$.
In this example, we consider the $2$-induction operator $\overline{\varPsi}_h$ for a fixed function $h: \mathbb{R} \rightarrow \mathbb{R}$, as defined in~\cite{DBLP:conf/cav/BatzCKKMS20}. Specifically, $\overline{\varPsi}_h(g)$ is given by
$
\overline{\varPsi}_h(g) := \min \{\overline{\varPhi}_f(g), h \},
$
and the corresponding $2$-upper induction condition is :
\begin{equation}\label{eg:condition}
    \overline{\varPhi}_f(\overline{\varPsi}_h(h)) \preceq h \,.
\end{equation}
According to Proposition~\ref{prop:relation}, we simplify this constraint by transforming~(\ref{eg:condition}) into the following form, which expresses the minimum over four functions $h_i$ ($1 \leq i \leq 4$): 
\begin{equation*}\label{eg:condition_miniform}
    \min \{h_1, h_2, h_3, h_4\} \preceq h\,,
\end{equation*}
where each $h_i$ corresponds to a loop-free program $C_i$ generated during the unfolding process up to depth $k=2$. All such unfolded programs are summarized in~\cref{fig:k=2induction}.

\begin{figure}[htbp]
  \centering
  \subfloat[\textbf{Case 1:} Program $C_1$]
  {\includegraphics[width=0.45\textwidth]{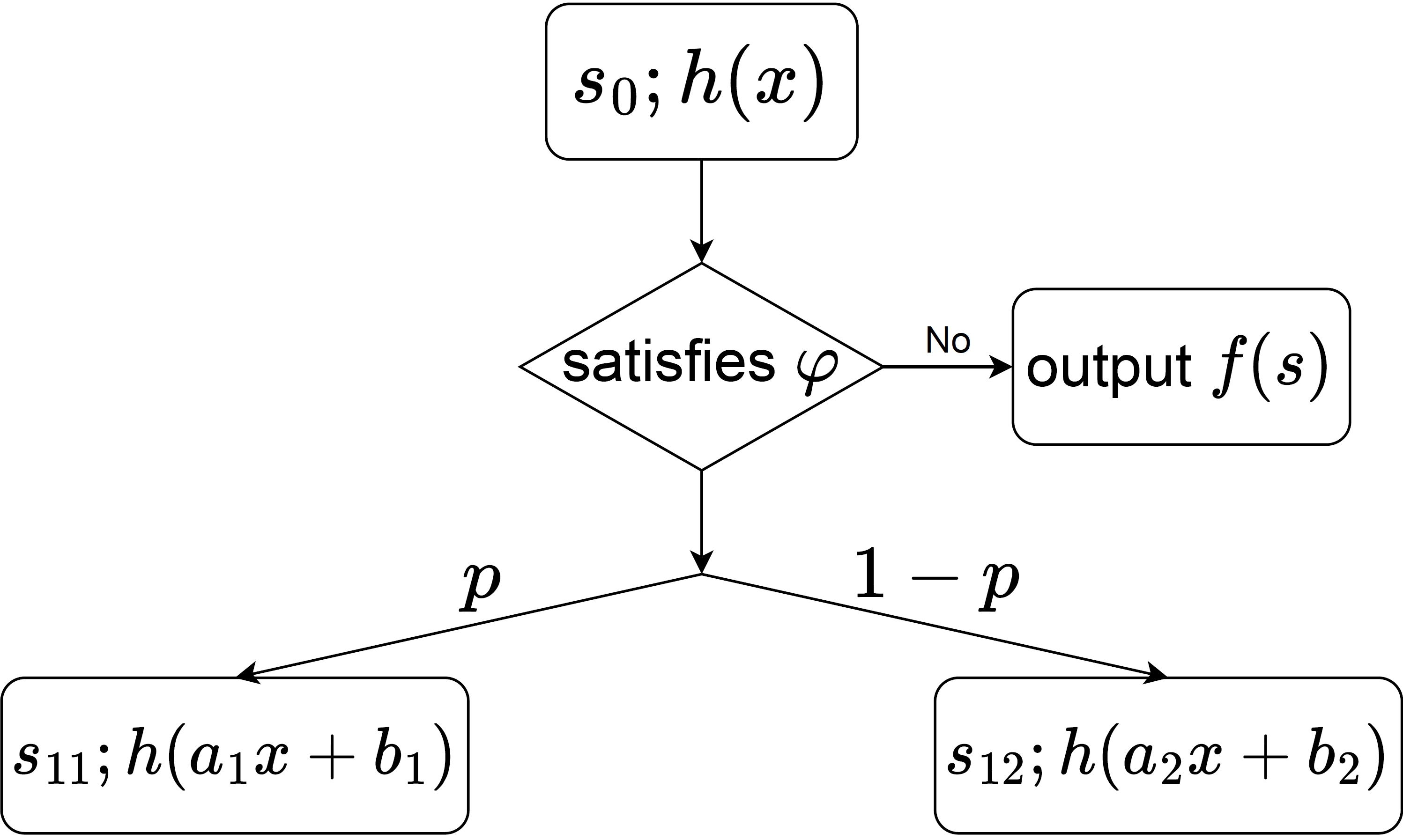}\label{fig:case1}} \quad   
  \subfloat[\textbf{Case 2:} Program $C_2$]
  {\includegraphics[width=0.45\textwidth]{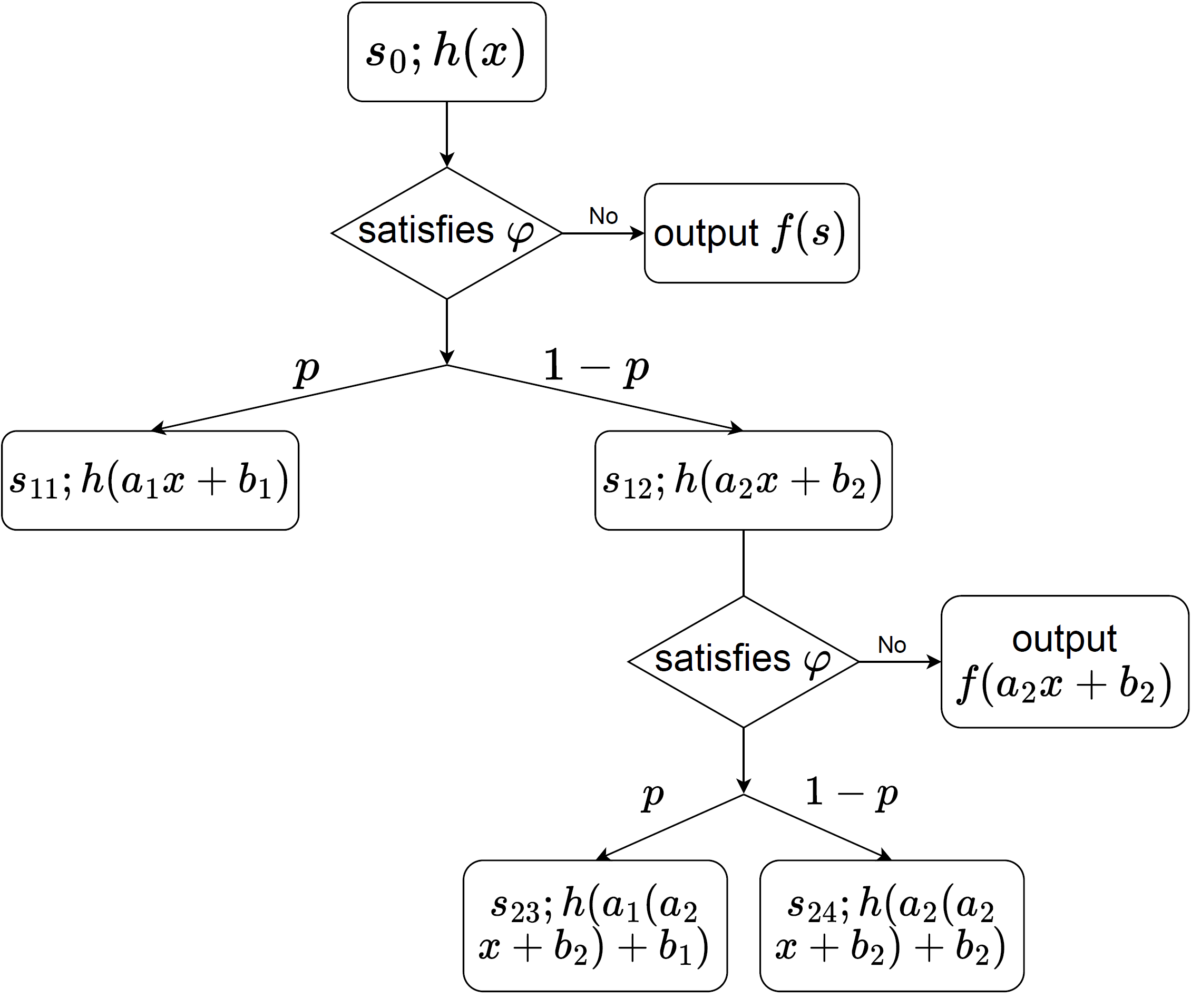}\label{fig:case2}} \quad 
  \subfloat[\textbf{Case 3:} Program $C_3$]
  {\includegraphics[width=0.45\textwidth]{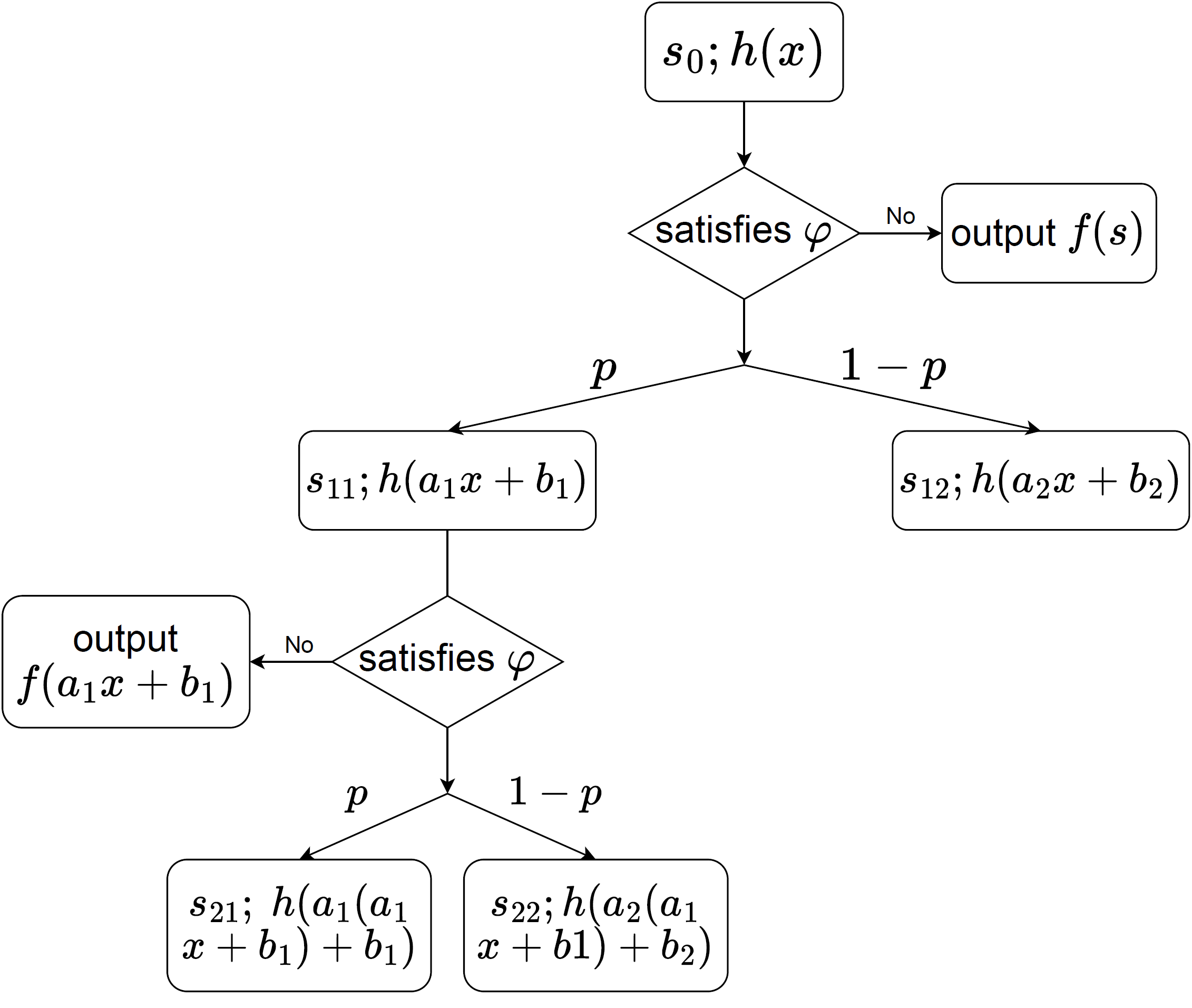}\label{fig:case3}} \quad 
  \subfloat[\textbf{Case 4:} Program $C_4$]
  {\includegraphics[width=0.45\textwidth]{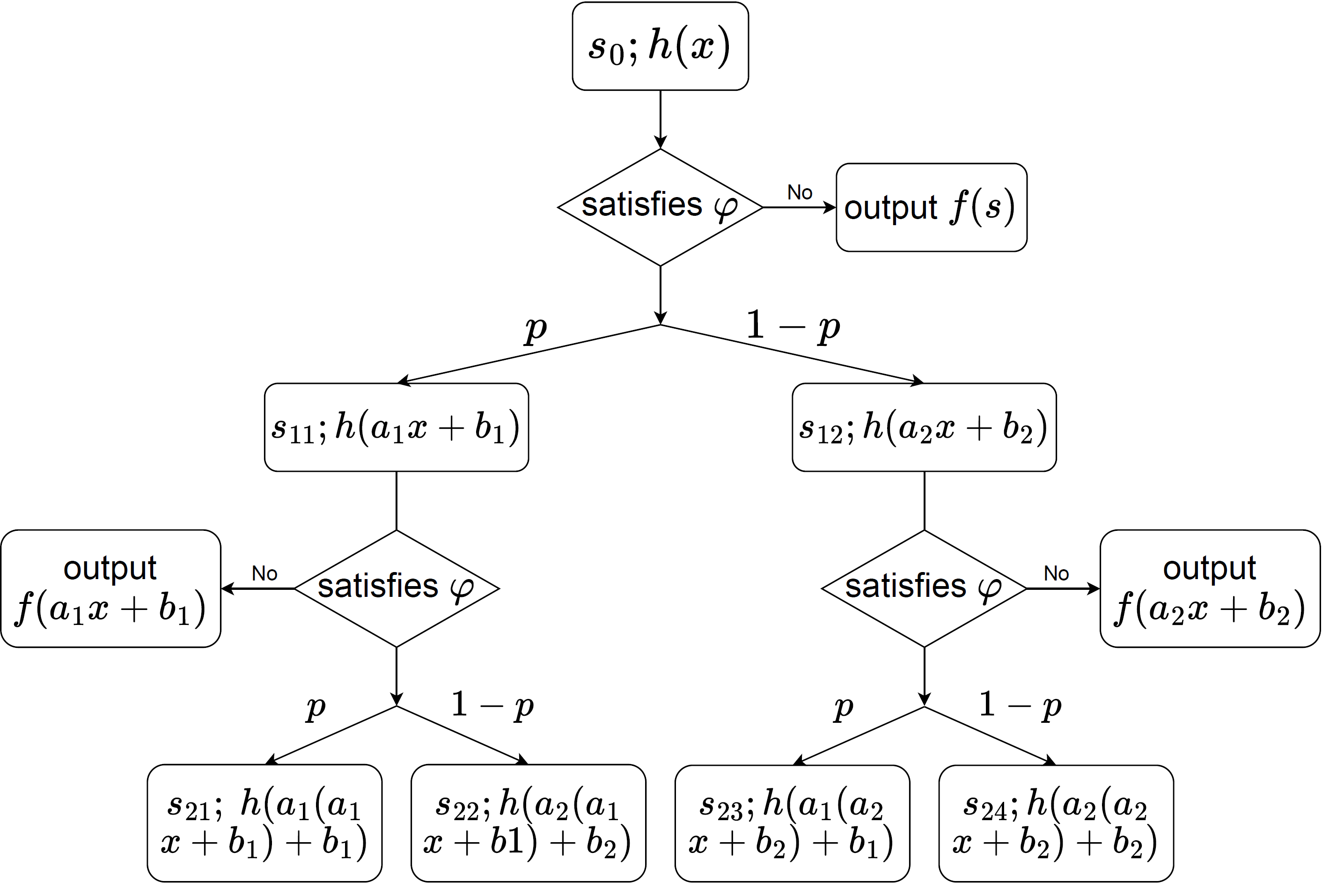}\label{fig:case4}} \quad 

  \caption{Loop-free programs generated by $(k=2)$-induction} \label{fig:k=2induction}
\end{figure}

We illustrate the unfolding process as follows. Starting from an initial value $x$, if $\varphi(x)$ is not satisfied, the loop terminates immediately and outputs $f(x)$. If $\varphi(x)$ holds, we proceed to unfold the loop, resulting in four distinct cases. Due to space constraints, we describe only the first case in detail here; the remaining three cases are depicted in~\cref{fig:k=2induction}, with further explanations provided in~\cref{app:explanation}.
In \text{Case 1}, the loop executes once and transitions to two possible states, $a_1 x + b_1$ and $a_2 x + b_2$, then it terminates. This corresponds to a single unrolling of the loop and terminating the unfolding at both resulting symbolic states, yielding the loop-free program $C_1$ as shown in~\cref{fig:case1}. The associated expression is $h_1 = [\neg \varphi(x)] \cdot f(x) + [\varphi(x)] (p\cdot h(a_1 x + b_1) + (1-p)\cdot h(a_2 x + b_2))$, which represents the expected value of $h(x)$ after executing program $C_1$. Cases 2, 3, and 4 are derived analogously by unrolling the loop up to two iterations.

\begin{example}\label{eg:extract}
Returning to the running example in Example~\ref{eg:runningexample}, 
we establish a 1-degree, i.e., linear template $h = a\cdot x+b\cdot y+c$, where $a,b,c$ are unknown coefficients.
We apply $2$-induction condition to synthesize a piecewise linear upper bound. Starting from a symbolic initial program state $s^* = (x, y)$, we unroll the loop once and arrive at two new symbolic program states $(x+1, x+y+1)$ and $(-1, y)$. Over each new state, we take the decision separately and the unfolding strategy produces four loop-free programs. The $pre_{C_{d}}(h)$ w.r.t. these four programs are as follows:
\begin{equation}\label{eq:unfolding term}
  {\small   \begin{aligned}
        h_1 &= [x < 0]\cdot y + [x \geq 0] \cdot (0.5\cdot h(x+1, x+y+1)+0.5 \cdot h(-1, y))  \\
        h_2 &= [x < 0]\cdot y + [x \geq 0] \cdot (0.25\cdot h(-1, y+x) + 0.25\cdot h(x+2,2x+y+3)+0.5\cdot h(-1, y))  \\
        h_3 &= [x < 0]\cdot y + [x \geq 0] \cdot (0.25\cdot h(-1, y+x) + 0.25\cdot h(x+2,2x+y+3) + 0.5 \cdot y)  \\
        h_4 &= [x < 0]\cdot y + [x \geq 0] \cdot (0.5\cdot h(x+1, x+y+1) + 0.5 \cdot y )
    \end{aligned} }
    \end{equation} 
   Thus,  we have the simplified constraint $\forall (x, y) \models I, \min \{h_1, h_2,h_3, h_4 \} \preceq h$.   \qed
\end{example}

\smallskip
\noindent {\em Branch reduction.}\label{page:branchreduction} 
During the unfolding process used to simplify the latticed $k$-induction condition $\overline{\varPhi}_f(\overline{\varPsi}_h^{k-1}(h)) \preceq h$, the number of resulting functions $h_i$ in (\ref{eq:minform}) grows rapidly with the number of probabilistic choices in the loop body. This combinatorial growth occurs because, when computing the pre-expectation for probabilistic branches, the sum of two minimum expressions results in a new minimum taken over the Cartesian product of the original function sets. To address this issue, we introduce a heuristic that selects only a small subset of "representative" functions from the complete set of $h_i$ in (\ref{eq:minform}). Importantly, this approach does not compromise soundness (see~\cref{thm:soundness,thm:soundness_poly}), as the minimum over any subset is always at least as the minimum over the full set.

Taking the case of $k=2$ as an example, by definition of operator $\overline{\varPsi}_h$, we have
\begin{equation*}
    \begin{aligned}
        \overline{\varPhi}_f(\overline{\varPsi}_h(h)) & = \overline{\varPhi}_f(\min\{\{\overline{\varPhi}_f(h), h \}\}  \\
        & = [\neg \varphi]\cdot f + [\varphi]\cdot \sum_{i=1}^n p_i\cdot \min\{\{\overline{\varPhi}_f(h(u_i(s))), h(u_i(s))\}
    \end{aligned}
\end{equation*}
where each $p_i$ denotes a probabilistic choice in the characteristic function $\overline{\varPhi}_f$, and $u_i$ represents the corresponding state update function under that choice.
Instead of enumerating all possible $2^n$ combinations in choosing either $\overline{\varPhi}_f(h(u_i(s)))$ or $h(u_i(s))$ for each $p_i$ (to expand into the minimum form (\ref{eq:minform})), one could consider combinations that have at most one $\overline{\varPhi}_f(h(u_i(s)))$ and at most one $h(u_i(s))$, so that only a linear number of combinations are considered while retaining soundness. For the case of $k>2$, a possible way for relaxation is to recursively consider combinations that have at most one $\overline{\varPhi}_f(\overline{\varPsi}_h^{k-2}(h(u_i(s)))$ and at  most one $h(u_i(s))$.

\smallskip
\noindent\textbf{Stage 3: Transforming to Canonical Form.} At this stage, our algorithm transforms the constraint of the form~(\ref{eq:minform}) from \textbf{Stage 2} into the following canonical form:
\begin{equation}\label{eq:canon_form}  
\textstyle
[B_{1}] \implies \min \{e_{11}, \dots, e_{m 1}\} \leq h, \,\dots,\, [B_{l}] \implies \min \{e_{1l}, \dots, e_{m l}\} \leq  h 
\end{equation}
where $h$ is the predefined polynomial template. Each $B_{j} (j\in \{1,...,l\})$ is a conjunction of predicates over the program variables that does not involve the template's unknown coefficients, and each $e_{ij}$ is a polynomial expression in these unknown coefficients.
The transformation begins by rewriting the inequality~(\ref{eq:minform}) as
\begin{align}\label{eq:minformpushed}
\textstyle
 \min \left\{\sum_r [B_{1r}]\cdot e_{1r},\,\dots,\,\sum_r [B_{mr}]\cdot e_{mr} \right\} \preceq h
\end{align}
where, as described previously, each $h_i$ is expressed as $h_i=\sum_r [B_{ir}]\cdot e_{ir}$. Next, for each conjunction $B=\bigwedge_{i=1}^m B_{ir_i}$ -- with each $B_{ir_i}$ taken from the summation $\sum_r [B_{ir}] \cdot e_{ir}$ -- we obtain the constraint $\Psi_B=[B] \implies \min_{i=1}^m e_{ir_i} \leq h$. The transformed system of inequalities~(\ref{eq:canon_form}) is thus precisely the set of all such $\Psi_B$ constraints. Infeasible constraints (i.e., those with unsatisfiable $B$) are removed, whenever possible, using an SMT solver such as Z3~\cite{Z3paper}.

\begin{example}\label{eg:gw_canonical}
Continuing from Example~\ref{eg:extract}, we convert (\ref{eq:unfolding term}) into its canonical form by partitioning the state space $S$ into two regions: $[x < 0]$ and $[x \geq 0]$, as indicated in~(\ref{eq:unfolding term}). Applying \textbf{Stage 3} and eliminating unsatisfiable predicates yields the following canonical form:

{\small
    \begin{equation*}\label{eq:x<0}
        [x<0]\implies \min\{y\} \leq h 
    \end{equation*}
    
     \begin{equation}\label{eq:x>=0}
        [x \geq 0] \implies \min \left\{
        \begin{aligned}
            &0.5\cdot h(x+1, x+y+1)+0.5 \cdot h(-1, y)  \\
            &0.25\cdot h(-1, y+x+1) + 0.25 \cdot h(x+2,2x+y+3)+0.5 \cdot h(-1, y)  \\
            &0.25\cdot h(-1, y+x+1) + 0.25\cdot h(x+2,2x+y+3) + 0.5 \cdot y  \\
            &0.5\cdot h(x+1, x+y+1) + 0.5 \cdot y 
        \end{aligned}
        \right\} \leq h  
    \end{equation} \qed}    
\end{example}

\smallskip
\noindent \textbf{Stage 4: Solving Constraints. }
Below, we describe our approach for solving the canonical constraints given in (\ref{eq:canon_form}). It is important to note that the presence of the \emph{minimum} operator in this canonical form makes the constraint \emph{non-convex}. To address this, we develop distinct algorithms for the linear and polynomial cases. In the linear case, where the program is affine (i.e., all conditions and assignments are linear), we employ a linear template for the $k$-upper potential function $h$. In the polynomial case, where the program may be non-affine, we utilize a polynomial template.

\smallskip 
\noindent {\em Solving constraints (linear case).} 
In our algorithm for the linear case, we require that the return function be piecewise linear and that the invariant be affine in the program variables. We first eliminate the minimum operator in~(\ref{eq:canon_form}) by considering its negation. This allows us to transform the constraint into a set of bilinear constraints using Motzkin's Transposition Theorem, which can then be solved with off-the-shelf bilinear programming solvers such as \emph{Gurobi}.

Below, we present a variant of Motzkin’s Transposition Theorem~\cite{motzkin1936beitrage}, which will be utilized in the subsequent analysis. The proof is provided in~\cref{app:constraint}.

\begin{theorem}[Motzkin's Transposition Theorem~\cite{motzkin1936beitrage}]
\label{corollary_motzkin}
Let  $S = \left( A_1\cdot \mathbf{x} + \mathbf{b_1}\leq 0 \right)$ and $T = \left( A_2 \cdot \mathbf{x} + \mathbf{b_2} < 0\right)$ be systems of linear inequalities, where $A_1 = (\alpha_{i,j} )\in \mathbb{R}^{m\times n}$ and $A_2 = (\alpha_{m+i,j} )\in \mathbb{R}^{k\times n}$ are real coefficient matrices, $\mathbf{b_1} = (\beta_1,...\beta_m)^\top $ and $\mathbf{b_2} = (\beta_{m+1},...\beta_{m+k})^\top $ are real vectors, and $\mathbf{x} = (x_1,...x_n)^\top $.
If $S$ is satisfiable, then $S \wedge T$ is unsatisfiable if and only if there exist non-negative real numbers $\lambda_0, \lambda_1, ..., \lambda_{m+k}$, with at least one $\lambda_i$ for $i\in \lbrace m+1, ..., m+k \rbrace$ being nonzero, such that:

\begin{equation}\label{motzkin bilinear}
\small
\sum_{i=1}^{m+k} \lambda_i \alpha_{(i,1)} =0,\; \dots, \sum_{i=1}^{m+k} \lambda_i \alpha_{(i,n)} =0 ,\;   (\sum_{i=1}^{m+k} \lambda_i \beta_i) - \lambda_0= 0. 
\end{equation}
\end{theorem}
\begin{remark}
    Note that, since $\lambda_i \geq 0$ for $0 \leq i \leq m+k$, the requirement that at least one $\lambda_i$ for $i \in \{m+1, \ldots, m+k\}$ is nonzero can be equivalently encoded as the linear constraint $\sum_{i = m+1}^{m+k} \lambda_i > 0$.
\end{remark}

In what follows, we demonstrate how to apply Theorem~\ref{corollary_motzkin} to solve the canonical constraints~(\ref{eq:canon_form}). We begin by conjunct the affine invariant $I$ with the antecedent predicates in~(\ref{eq:canon_form}) and eliminating any constraints with unsatisfiable antecedents, resulting in
\begin{equation}\label{eq:indcano}
[I \wedge B_j] \implies \min \{e_{1j}, e_{2j},..., e_{mj}\}\leq  h \quad \text{for $j\in \{1,2,\dots, l\}$}\,,
\end{equation}
where we assume that each $I \wedge B_j$ is satisfiable. For each $j$, we have
\begin{align*}
    & \left([I \wedge B_j] \implies \min \{e_{1j}, e_{2j},..., e_{mj}\}\leq  h \right) \text{ holds} \\
    \iff &\, \left([I \wedge B_j] \wedge \left(\wedge_{i=1}^m (e_{ij} > h) \right)\right) \text{ is not satisfiable}  \TAG{Apply Thm~\ref{corollary_motzkin}} \\
    \iff &\, 
    \begin{aligned}[t]
        & \text{exists nonnegative real vector $\boldsymbol{\lambda_j} = (\lambda_{0,j},\dots,\lambda_{m_j+k_j, j})$}, \\
    & \text{s.t. $(\lambda_{m_j+1, j}, \dots, \lambda_{m_j+k_j, j} )\neq \mathbf{0}$, and eq.~\eqref{motzkin bilinear} holds.} 
    \end{aligned}
\end{align*}
The second equivalence follows from the Motzkin’s Transposition Theorem by setting $S = I \wedge B_j$ and $T = \left(\wedge_{i=1}^m (e_{ij} > h) \right)$ for each $j \in \{1,2,\dots, l\}$.
Note that (\ref{motzkin bilinear}) constitutes a bilinear constraint problem, as its nonlinearity arises solely from the products of unknown template coefficients and the variables $\boldsymbol{\lambda_j}$. Our approach aggregates all such bilinear constraints and utilizes off-the-shelf bilinear solvers to obtain concrete solutions for the template $h$.

\begin{example}\label{eg:gw_bilinear}
Continuing from Example~\ref{eg:gw_canonical}, recall that we choose $x \geq -1$ as the invariant. For the constraint (\ref{eq:x>=0}), substituting $h(x, y)$ with the template $a x + b y + c$ and considering its negation as previously illustrated, we obtain the following inequalities:
    {\small
        \begin{gather*}
            -x \leq 0 \qquad 0.5(a-b) x-0.5 b < 0 \qquad
            0.75 (a-b)  x - (b-0.25 a) < 0 \\
            0.75 (a-b)  x+ 0.5 (b-1)  y + (0.5c-0.25a-b) < 0 \qquad
            0.5(a-b)  x + 0.5(b-1)  y + 0.5(c-a-b) < 0.
        \end{gather*}}
Then by~\cref{corollary_motzkin}, the constraint (\ref{eq:x>=0}) is equivalent to solving the following set of bilinear constraints involving the unknown coefficients $a$, $b$, and $c$.
    {\small   \begin{align*} \label{eq:compostion-1}
            &\exists \lambda_{0} \ge 0, \lambda_{1} \ge 0, \cdots,  
             \lambda_{5}  \ge 0  \quad  \textbf{s.t.}  \quad (\lambda_2 \neq 0 \vee \lambda_3 \neq 0 \vee \lambda_4 \neq 0 \vee \lambda_5 \neq 0)  \wedge\\
            0 &= (-1)\cdot \lambda_1 + 0.5(a-b) \cdot \lambda_2 + 0.75 (a-b) \cdot \lambda_3 + 0.75 (a-b)\cdot \lambda_4 + 0.5(a-b) \cdot \lambda_5 \ \wedge\\
            0 &= 0.5 (b-1)\cdot \lambda_4 + 0.5(b-1) \cdot \lambda_5 \ \wedge \\
            0 &= -0.5b \cdot \lambda_2 - (b-0.25 a) \cdot \lambda_3 + (0.5c-0.25a-b) \cdot \lambda_4 + 0.5(c-a-b) \cdot \lambda_5 - \lambda_0.  \qed
    \end{align*} } 
\end{example}

Our algorithm utilizes bilinear solvers to address the derived bilinear constraints. Since these constraints define only a feasible region, we heuristically select an objective function to guide the solver toward solutions that yield tighter upper bounds. 
Specifically, we minimize $h(s^*)$, where $s^*$ is a designated initial program state of interest. Once the template coefficients for $h$ are determined (yielding a candidate $h^*$), we reconstruct the piecewise linear upper bound by applying the upper $k$-induction operator $\overline{\varPsi}_{h^*}$ iteratively $k-1$ times, resulting in $\overline{\varPsi}_{h^*}^{k-1}(h^*)$. We claim that our linear bound algorithm is complete in the sense that the reduction to bilinear programming preserves the original $k$-induction condition.

\label{opt1}

\begin{example}
Continuing with Example~\ref{eg:gw_bilinear}, we use the objective function $h = a x + b y + c$ with the initial state $s^* = (x, y) = (1, 1)$. Solving the optimization yields the candidate $h^*(x,y) = x + y + 2$. We then reconstruct the piecewise upper bound by applying $\overline{\varPsi}_{h^*}$ once, resulting in the upper bound $[x < 0] \cdot y + [x \geq 0] \cdot (x + y + 2)$.
\end{example}

\smallskip
\noindent {\em Solving constraints (polynomial case). } In our algorithm for the polynomial case, we assume that the return function is piecewise polynomial and that the invariant is a polynomial predicate over the program variables. We design a sound approach that relaxes the $k$-induction constraint and reduces the relaxed formulation to a semidefinite programming (SDP) problem using Putinar's Positivstellensatz~\cite{putinar}. This relaxation guarantees that the synthesized upper bound $h$ satisfies the original $k$-induction condition (see Definition~\ref{def:potential function}). The algorithm is described as follows.
 
First, for each constraint in the canonical form~(\ref{eq:canon_form}), namely $[B_{j}] \implies \min \{e_{1j}, \dots, e_{m j}\} \leq h$ for $j \in \{1, \dots, l\}$, we relax the constraint by replacing the minimum operator with a convex combination of the terms $\{e_{ij}\}_{i=1}^m$. This results in the following relaxed form:
\begin{equation}\label{eq:approximation_each_form}
\textstyle
   [B_{j}] \implies \sum_{i=1}^m w_i \cdot e_{i j} \le h, \, ,j \in \{1, \dots, l\} 
\end{equation}
where each weight $w_i \geq 0$ and the set of weights satisfies $\sum_{i=1}^m w_i = 1$. Various forms of weight combinations $\{w_i \}_{i=1}^m$ can be employed, such as uniform weights (where each $w_i = 1/m$) or randomly generated weights normalized to sum to one. This relaxation is sound: any function $h$ and set $\{e{ij}\}_{i=1}^m$ that satisfy the relaxed constraint~(\ref{eq:approximation_each_form}) will also satisfy the original canonical form~(\ref{eq:canon_form}). This follows from the fact that $\sum_{i=1}^m w_i\cdot e_{ij} \le h \implies \min\limits_{i \in \{1, \dots, m\}} \{e_{i j}\} \le h$.

Next, we conjunct the invariant $I$ with each constraint in~(\ref{eq:approximation_each_form}), resulting in the following form:
\begin{equation}\label{eq:approximation_form}
\small
   \bigwedge_{j \in \{1, \dots, l\}}  [I \wedge B_{j}] \implies \sum_{i=1}^m w_i \cdot e_{i j} \le h,
\end{equation}
We then apply Putinar's Positivstellensatz~\cite{putinar}, following previous work~\cite{DBLP:conf/cav/ChatterjeeFG16,10.1145/3656432}, to generate constraints on the unknown coefficients, which are solved using off-the-shelf SDP solvers (see~\cref{app:putinar} for details). As these constraints define only a feasible region, we employ a heuristic objective function to guide the solver towards tighter upper bounds. Specifically, we minimize $\sum_{i}h(s^*_i)$, where $s^*_i$ are selected initial program states of interest. \label{opt2}
After obtaining the optimal solution $h^*$ from the SDP solver, we reconstruct the piecewise polynomial upper bound $\overline{\varPsi}_{h^*}^{k-1}(h^*)$ by iteratively applying the upper $k$-induction operator $\overline{\varPsi}_{h^*}$ to $h^*$ $k-1$ times.

{\small
\begin{algorithm}
    \SetKwData{Left}{left}\SetKwData{This}{this}\SetKwData{Up}{up}
    \SetKwFunction{Union}{Union}\SetKwFunction{FindCompress}{FindCompress}
    \SetKwInOut{Input}{Input}\SetKwInOut{Output}{Output}

    \Input{Probabilistic loop $P$ in the form of~(\ref{eq:pwhileloop}) and a piecewise return function $f$}
    \Output{Piecewise bounds for the expected value of $X_f$ upon termination of $P$}
    \BlankLine
    \textbf{Prerequisites Checking and External Inputs:}
    
    (a) Prerequisites Checking: Verify the prerequisites in~\cref{thm:soundness} (\cref{thm:soundness_poly}). \\
    (b) External Inputs: Generate an invariant $I$,   
    select parameter $k$ and specify initial program state $s^*$. \\
    \textbf{Templates and Constraints: }\\
    (a) Predefining a (monolithic) polynomial template $h$. \\(b) Unfolding the loop within $k$ times and calculate $pre_{C_d}(h)$ for all $C_d \in \{C_1,\dots, C_m\}$ (generated by our unfolding process) to obtain the constraint $\min \{h_1, h_2,\dots, h_m \} \preceq h$. \\
    \textbf{Transforming to Canonical Form: } \\
    Transform the constraints (\ref{eq:minform}) into the form of (\ref{eq:canon_form}) through an iterative approach and obtain $l$ canonical constraints\;
    \textbf{Constraints Solving:} \\
    \uIf {\emph{the loop $P$ is linear and the template $h$ is linear}} { 
    $Cons \leftarrow \emptyset$; \Comment{Linear Case} \\
    \For {$j \leftarrow 1~ \KwTo~ l$}{
        Extract the coefficients of the variables from canonical-formed constraints\;
        Construct bilinear constraints $K_j$ with auxiliary variables $\boldsymbol{\lambda}_j$\;
        $Cons \leftarrow Cons \cup K_j$\;
     }
     Call bilinear solver to solve $Cons$ and obtain the piecewise bound with the solution $h^*$}
    \Else{  
    (a) Soundly relax the original canonical constraints (\ref{eq:canon_form}) into (\ref{eq:approximation_form}). \Comment{Polynomial Case}\\
    (b) Call SDP solver to solve and obtain the piecewise bound with $h^*$.
}
    \caption{Synthesizing Bounds}
    \label{alg}
\end{algorithm}
}

\smallskip
\noindent {\bf Correctness. }
Our algorithms are guaranteed to produce correct bounds by~\cref{thm:soundness,thm:soundness_poly}. The \emph{Prerequisites Checking} stage ensures that all prerequisites in~\cref{thm:soundness} and~\cref{thm:soundness_poly} are met, and the function $h$ is determined according to the $k$-induction conditions (see Definition~\ref{def:potential function}). Additionally, the invariants we use over-approximate the set of reachable program states, thereby preserving the soundness of our approach. Specifically, our linear bound algorithm is both sound and complete in the sense that the reduction to bilinear programming exactly preserves the original $k$-induction condition. In the polynomial case, our algorithm employs a sound relaxation, which likewise guarantees the correctness of the synthesized bounds.

\subsection{Extensions: Handling Probabilistic Programs with Multiple Loops}\label{sec:extensions}

Below, we describe the extension of our approach to probabilistic programs with multiple loops, including both sequential compositions of probabilistic loops and nested loops. For brevity, we focus on the synthesis of upper bounds; the synthesis of lower bounds is entirely analogous.

\smallskip
\noindent {\em Sequential Composition. }
For a sequential composition $P = P_1;\dots;P_n$ of loops $P_1,\dots,P_n$ with return function $f$, our method analyzes each loop in reverse order. To illustrate the approach, we focus on the case $P = P_1; P_2$. Given a $k$-induction parameter $k$, the procedure for synthesizing upper bounds proceeds as follows:
(i) Begin by computing a piecewise upper bound $h_2$ for the expected value of $f$ after the execution of $P_2$.
(ii) Then, treat $h_2$ as the return function for $P_1$ and compute its piecewise upper bound, resulting in a final bound $h_1$ for the entire composition.
This backward compositional reasoning can be systematically extended to compositions with more than two loops.

\smallskip
\noindent {\em Nested Loops. }
To address nested loops, we incorporate our approach with the methods proposed in~\cite{DBLP:conf/atva/FengZJZX17,DBLP:journals/pacmpl/FengCSKKZ23}, applying $k$-induction exclusively to the innermost loop and $1$-induction to the outer loops. Since the innermost loop can be unfolded independently of the outer loops, we are able to derive tight piecewise bounds for the inner loop via $k$-induction and subsequently propagate these bounds to the outer loops. For clarity, we focus on the case where the program $R$ contains a single inner loop and has the following structure:
\[
R=\WHILESYMBOL (\, \psi \, ) \{P\}\mbox{ with }P = \WHILESYMBOL (\, \varphi \, ) \{Q\}\mbox{ and }Q\mbox{ loop-free}.
\]
Our objective is to analyze the expected value of $X_f$ upon termination of the loop. Let $\varPhi^{out}_f$ denote the characteristic function (see Definition~\ref{def:characteristic function}) with respect to the outer loop and return function $f$, and let $\varPhi^{in}_g$ denote the characteristic function for the inner loop $P$ and return function $g$. While $\varPhi^{in}_g$ can be computed explicitly, $\varPhi^{out}_f$ typically cannot. We therefore apply $1$-induction to the outer loop and $k$-induction to the inner loop, as summarized below:
\begin{itemize}
\item Define templates $h_{out}$ and $h_{in}$ at the entry of the outer and inner loop respectively.
\item For the outer loop, the $1$-induction rule yields the constraint $\varPhi^{out}_f(h_{out}) \preceq h_{out}$. Since $\varPhi^{out}_f(h_{out})$ cannot generally be computed explicitly, we upper-approximate the expected value of $h_{out}$ after executing the inner loop $P$ by $h_{in}$, i.e., $\varPhi^{out}_f(h_{out})  \preceq [\neg \psi] \cdot f + [\psi] \cdot h_{in} $,  and the original constraint  $\varPhi^{out}_f(h_{out}) \preceq h_{out}$ can be strengthened into $[\neg \psi] \cdot f + [\psi] \cdot h_{in} \preceq h_{out} $.
\item For the inner loop, we apply the $k$-induction condition (see Definition~\ref{def:k_induc_func}) to ensure that $h_{in}$ upper-approximates the expected value of $h_{out}$ after executing the inner loop. This leads to the constraint $\varPhi^{in}_{h_{out}}((\varPsi^{in}_{h_{in}})^{k-1}(h_{in})) \preceq h_{in}$, where $\varPsi^{in}_{h_{in}}(g) = \min \{\varPhi^{in}_{h_{out}}(g), h_{in}\}$ is the upper $k$-induction operator for the inner loop $P$  (see Definition~\ref{def:k_induction_operator}). 
\item Collect the resulting constraints and apply our synthesis algorithm as described in~\cref{sec:algorithm}.
\end{itemize}

Through this process, we obtain $h_{out}$ as a piecewise upper bound for the expected value of $X_f$ with respect to the return function $f$ upon termination of the entire while loop $R$.

\section{Experimental Results}\label{sec:experiment}
We implement our algorithms in Python 3.9.12 and Julia 1.9.4. We use Gurobi in Python for bilinear programming and Mosek in Julia for semi-definite programming. All experiments are conducted on a Windows 10 (64-bit) machine equipped with an Intel(R) Core(TM) i7-9750H CPU at 2.60GHz and 16GB of RAM. We evaluate our algorithms for synthesizing piecewise linear and polynomial upper bounds, as detailed in~\cref{sec:linear_res} and~\cref{sec:poly_res}. Results for lower bound synthesis, which exhibit similar performance and comparative advantages, are provided in Appendix~\ref{app:linear_lower} and Appendix~\ref{app:poly_lower} due to space limitations.

\smallskip
\noindent {\bf Evaluation Goals. } Our experiments are designed to address the following research questions:
\begin{itemize}
\item[\textbf{RQ1.}] How effective is our approach in generating piecewise bounds?
\item[\textbf{RQ2.}] How does our approach compare to the most closely related methods?
\item[\textbf{RQ3.}] How do our piecewise bounds compare to monolithic polynomial bounds?
\end{itemize}

\smallskip
\noindent {\bf Experimental Settings. } 
We address the evaluation goals for our piecewise linear and polynomial algorithms separately. The experiments are conducted under the following settings:

\smallskip
\noindent{\em Invariants.} We employ invariants to over-approximate the set of reachable states, which is standard in various existing results~\cite{DBLP:conf/cav/ChakarovS13,DBLP:conf/vmcai/FuC19,DBLP:conf/cav/ChatterjeeFG16}. Note that invariants do not provide information about the piecewise partitioning of the bounds to be computed. In our experiment, we minimize their impact by deliberately choosing trivial interval-bound invariants that can be directly derived as the union of loop guard and its post image under the increment/decrement operations within the loop body.

\smallskip
\noindent{\em Prerequisites Checking.} Our experiments cover both linear and polynomial probabilistic programs (see~\cref{app:programs} for details). For linear programs with monolithic linear return functions, we use a linear template and apply our linear algorithm. For more general cases involving polynomial programs with piecewise polynomial return functions, we employ a higher-degree polynomial template and apply our polynomial algorithm.
In our piecewise linear experiments, we ensure that the prerequisites (P1) and (P2) in~\cref{thm:soundness} are satisfied as follows. For (P1), we verify syntactically that the uniform amplifier $c$ can typically be set to $1$ across most benchmarks, ensuring that (P1) holds for any positive $c_2$. For the remaining benchmarks, we take the maximum coefficient of the program variables in the loop body as $c$. For example, in the \textsc{St-Petersburg} benchmark, we set the uniform amplifier $c$ to $2$, choosing $c_3 = \ln 2$ (since $e^{c_3} = 2$) and $c_2 = \ln 4$ to meet the required conditions.
For (P2), in benchmarks where each loop iteration terminates with probability $p$ and continues with probability $1-p$, we can syntactically extract $p$ and verify that the concentration property holds, exhibiting exponential decay at a rate of $e^{\ln(1-p)}$. For the remaining benchmarks, we construct difference-bounded ranking supermartingales (dbRSMs) to ensure the concentration property. Such dbRSMs can be synthesized automatically using methods described in~\cite{DBLP:conf/popl/ChatterjeeFNH16,DBLP:conf/cav/ChatterjeeFG16}.
In our piecewise polynomial experiments, we ensure that the prerequisites (P1') and (P2) in~\cref{thm:soundness_poly} are satisfied as follows. For (P1'), we verify the bounded-update property on each polynomial benchmark using an SMT solver~\cite{Z3paper}. For (P2), we apply the same approach as in the linear case to establish the concentration property for polynomial programs.

\smallskip
\noindent{\em Bound Optimization. } 
Recall that in our algorithms described on pages~\pageref{opt1} and~\pageref{opt2}, we optimize the synthesized upper bounds by minimizing their values over the initial states of interest, which serve as the objective function. In the piecewise linear experiments, we typically set the default initial state $s^*$ by assigning the value 1 to all program variables across most benchmarks. For specific cases, such as \textsc{Fair Coin}, we assign initial values $x = 0$ and $y = 0$ --- since $(x, y) = (0, 0)$ is the only state from which the loop can be entered --- and set the variable $i$ to its default value of 1.
In the piecewise polynomial experiments, for path probability estimation benchmarks selected from~\cite{Beutner2022b,DBLP:conf/pldi/SankaranarayananCG13,10.1145/3656432,DBLP:conf/cav/GehrMV16}, we adopt the default initial state $s^*$ used in previous work to ensure consistency. For the remaining benchmarks, we first define an interval-bound region, with real-valued variables ranging over $[0, 10]$ and Boolean variables over $[0, 1]$. We then select 10 initial states comprising the boundary points of the region, the midpoints of each boundary, the center point, and uniformly distributed integer points within the region.

\smallskip
\noindent{\em Weights Selection.} For the polynomial experiments, recall that our algorithm requires a predefined set of weight combinations (see~\cref{eq:approximation_form}). We employ uniformly distributed weights (i.e., each weight is $\frac{1}{m}$) and additionally generate 10 sets of randomly selected weights, each normalized to sum to one. Independent computations are performed for each of these 11 weight combinations. From the resulting solutions, we select the function with the minimum objective value as the synthesized upper bound $h^*$. The total execution time is reported as the cumulative runtime for the 11 independent runs with different weight settings. 

\smallskip
\noindent{\em Numerical Repair.} To address the inherent numerical issues associated with numerical solvers, we apply a post-processing step to repair the computed results. In the linear experiments, we approximate the output floating-point coefficients with rational numbers using continued fractions (see~\cref{app:continued fraction}, ~\cite{Jones1984ContinuedFA}), and validate these approximations by checking the constraints in (\ref{eq:canon_form}). This numerical repair is applied to all benchmarks except \textsc{Expected Time}. For this benchmark, since suitable rational approximations could not be found, we truncate the floating-point results to a precision of $10^{-4}$ and verify their validity against the same constraints. In the polynomial experiments, we similarly truncate all floating-point results to $10^{-4}$ precision, then substitute the results into the constraints in~(\ref{eq:approximation_form}) to check feasibility. Of the 20 benchmarks evaluated, the results for 16 passed our validation procedure, while the remainder remain unknown.

\subsection{Piecewise Linear Bound Synthesis}\label{sec:linear_res}
\noindent {\bf Benchmark Selection. }
We choose upper-bound benchmarks from existing works ~\cite{DBLP:conf/cav/BaoTPHR22,DBLP:conf/cav/ChenHWZ15,DBLP:conf/atva/FengZJZX17,DBLP:conf/tacas/BatzCJKKM23,DBLP:conf/cav/BatzCKKMS20,DBLP:conf/cav/GehrMV16,DBLP:conf/pldi/BeutnerOZ22,DBLP:journals/pacmpl/FengCSKKZ23} that fall into our scope and have the following adaptions. First, for those that do not have linear return functions, we add simple linear return functions. Second, for those whose upper bound that can be handled directly by 1-induction (except for several classical examples: \textsc{k-Geo, RevBin, Fair Coin}), we adapt them by reasonable perturbations (such as changing the assignment statement, changing the probability parameters, reducing the continuous distribution to discrete distribution, etc) so that they require $(k>1)$-induction. Third, for those whose upper bound that cannot be handled by $k$-induction with small $k=1,2,3$, we adapt them by reasonable perturbations as above so that they can be handled by $(k>1)$-induction, while still cannot be handled by 1-induction. 

In detail, we consider 7 original examples and 6 adapted examples from the literature. The examples \textsc{Geo}, \textsc{k-Geo} and \textsc{Equal-Prob-Grid} are taken from~\cite{DBLP:conf/cav/BatzCKKMS20,DBLP:conf/tacas/BatzCJKKM23}, for which we replace the assertion probability with a linear return function $goal$ in \textsc{Equal-Prob-Grid}. We consider the benchmark \textsc{Zero-Conf-Variant} adapted from~\cite{DBLP:conf/tacas/BatzCJKKM23,DBLP:journals/pacmpl/FengCSKKZ23}. We revise the assignments and probabilistic parameters in the original program, and add a linear return function $curprobe$. The benchmark \textsc{St-Petersburg variant} is taken from~\cite{DBLP:journals/pacmpl/FengCSKKZ23} where we replace the probability parameter $\frac{1}{2}$ with $\frac{3}{4}$ since the original program does not satisfy the prerequisites in~\cref{thm:soundness}. From~\cite{DBLP:conf/cav/BaoTPHR22,DBLP:conf/cav/ChenHWZ15,DBLP:conf/atva/FengZJZX17}, we consider the benchmarks \textsc{Coin}, \textsc{Mart}, \textsc{RevBin} and \textsc{Fair Coin}, and revise the assignments, guards on the original benchmarks \textsc{Bin} series so that we obtain a more complex version \textsc{Bin-Ran}. The remaining three examples, \textsc{Expected Time, Growing Walk} and its variant, are all adapted  from~\cite{DBLP:conf/cav/GehrMV16,DBLP:conf/pldi/BeutnerOZ22} by reducing the continuous distributions to discrete distributions. For affine programs, as $k$ increases, the bottleneck emerges in the unfolding process of Stage 2. In our experiments with piecewise linear bounds, the parameter $k$ is chosen as the largest value such that the runtime of Stage 2 in Algorithm~\ref{alg} remains below 600 seconds.

\smallskip 
\noindent{\bf Answering RQ1.} We present the experimental results on these 13 benchmarks in Table~\ref{table:1}. As bilinear solving is an iterative search for optimal solutions, we set the maximum searching time for Gurobi to $100$s. On most benchmarks, we find that a monolithic linear bound with 1-induction does not exist but obtain a piecewise linear upper bound via ($k>1$)-induction in a few minutes. Our approach derives the exact bound, i.e., the tightest upper bound, on the benchmarks~\textsc{Geo, Coin, k-Geo, Mart, Growing Walk, Equal-Prob-Grid, RevBin, Fair Coin, St-Petersburg variant}. The exactness of these bounds is established by comparison with the exact invariants synthesized in~\cite{DBLP:conf/cav/BaoTPHR22} (see \textbf{RQ2}) and with the piecewise lower bounds presented in~\cref{app:linear_lower}.
We also show that on a significant number of benchmarks (e.g., \textsc{k-Geo, Bin-Ran, Growing Walk-variant, Expected Time}, etc), the piecewise bounds we synthesize are non-trivial (i.e., the program state space $S$ is partitioned into more than $[\varphi]$ and $[\neg \varphi]$). 

\begin{table*}
    \renewcommand{\arraystretch}{1.8}
	\caption{Experimental Results for {\bf RQ1} and {\bf RQ2}, Linear Case (Upper Bounds). "$f$" stands for the return function considered in the benchmark, "T(s)" (of our approach) stands for the execution time of our approach (in seconds), including the parsing from the program input,  transforming the $k$-induction constraint into the bilinear problems, bilinear solving time and verification time. "Conventional Approach ($k=1$)" stands for the monolithic linear upper bound synthesized via 1-induction, "$k$" stands for the $k$-induction we apply, "Solution" stands for the linear candidate solved by Gurobi, and "Piecewise Linear Upper Bound" stands for our piecewise results. "Result" stands for the synthesized results by other tools and "T(s)" (of their approaches) stands for the execution time of their tools.}
	\label{table:1}
	\resizebox{\textwidth}{!}{
	\begin{threeparttable}
	   \begin{tabular}{|c|c|c|c|c|c|c|c|c|c|c|}
				\hline
				\multicolumn{1}{|c|}{\multirow{2}{*}{\textbf{Benchmark}}}  &
				\multicolumn{1}{c|}{\multirow{2}{*}{\textbf{$f$}}}      &
                \multicolumn{1}{c|}{\multirow{2}{*}{\textbf{\makecell{Conventional \\ Approach \\ ($k=1$)}}}}  &
				\multicolumn{4}{c|}{\multirow{1}{*}{\textbf{Our Approach}}}  &
                \multicolumn{2}{c|}{\multirow{1}{*}{\textsc{cegispro2}}}  &
				\multicolumn{2}{c|}{\multirow{1}{*}{\textsc{exist}}} 
				\\ \cline{4-11}
				\multicolumn{1}{|c|}{}  & \multicolumn{1}{c|}{} &  \multicolumn{1}{c|}{} &
                \multicolumn{1}{c|}{\textbf{$k$}}  & 
				\multicolumn{1}{c|}{\textbf{ Solution }} &
				\multicolumn{1}{c|}{\textbf{ Piecewise Linear Upper Bound }}   &    
                \multicolumn{1}{c|}{\textbf{T(s)}} &
                \multicolumn{1}{c|}{\textbf{ Result }} &
                \multicolumn{1}{c|}{\textbf{ T(s) }} &
                \multicolumn{1}{c|}{\textbf{ Result }} &
                \multicolumn{1}{c|}{\textbf{ T(s) }} 
                \\ \hline \hline
                {\textsc{Geo}}& $x$ &  \ding{55} &  3 &   $x+1$   & \makecell{$[c>0]\cdot x +$ $[c\leq 0]\cdot (x+1)$} & $1.92$  &  \makecell{$[c>0]\cdot x +$ \\ $[c\leq 0]\cdot (x+1)$} & 0.05 &  $x+[c=0]$ & 17.29 \\
                \hline
                {\textsc{k-Geo}}& $y$ &  $\makecell{-k + N + \\  x + y +1}$ &  3  &    $\makecell{-k + N  \\+  x + y +1}$  & \makecell{$[k>N]\cdot y + [k\leq N-1] \cdot $ \\ $(-k + N +x + y +1)+$ \\$ [N-1 < k \leq N] \cdot $\\ $(-0.5k + 0.5N + x + y +1)$} & 132.76   & \makecell{$ [k>N] \cdot y + [ k \leq N]$\\$ \cdot (-k+N+x+y+1)$} & 0.38  & \makecell{$y + [k\le n] \cdot $\\$(x - k + n+1)$} & 76.74 \\
                \hline
                {\textsc{Bin-ran}}& $y$  &  \ding{55}  & 2  &  $\makecell{0.9x-21i\\+y+233}$   & \makecell{$[i>10]\cdot y + $\\$[\frac{90}{11} < i\leq 10]\cdot (0.9x-21i+y+233)$ \\ $+[i\leq \frac{90}{11}] \cdot (0.9x-18.8i+y+215)$} &  106.29  &\makecell{inconsistent \\ results} & -  & \makecell{inner  error } & - \\
                \hline
                {\textsc{Coin}}& $i$ & \ding{55}  &  2   &    $i+\frac{8}{3}$ & \makecell{$[x\neq y]\cdot i + $ \\ $[x=y]\cdot (i+\frac{8}{3})$} &104.13  & not terminate &  - & fail & -\\
                \hline
                {\textsc{Mart}}& $i$ & \ding{55}  &  3   &  $i+2$ & \makecell{$[x\le 0]\cdot i + [x > 0]\cdot (i+2)$} & $19.29$  & \makecell{violation of \\ non-negativity } & -  & $i + [x > 0] *  2 $ & 37.23 \\
                \hline
                {\makecell{\textsc{Growing} \\\textsc{Walk}  }}& $y$ & \ding{55}  & 3  &  $x+y+2$ & \makecell{$[x<0]\cdot y + [x\geq 0]\cdot (x+y+2)$} & $4.03$  & \makecell{violation of \\ non-negativity } & -  & \makecell{$y + [x \ge 0] \cdot $\\$ (x+2)$} & 21.98 \\
                \hline
                {\makecell{\textsc{Growing} \\ \textsc{Walk}\\ \textsc{-variant}}}& $y$ &  $x+y+1$  & 3  &  $x+y+1$ & \makecell{$[x<0]\cdot y +$ \\ $ [0\leq x <1]\cdot (0.5x+y+0.25) $ \\ $+[x\geq 1] \cdot (x+y)$} & $125.19$  & \makecell{violation of \\ non-negativity }& - & not terminate &- \\
                \hline
                {\makecell{\textsc{Expected} \\ \textsc{Time}}}& $t$  &  \ding{55}  &  3  &  \makecell{$4.4280x + t$\\$ + 6.2461$} & \makecell{$[x<0]\cdot t + [0\leq x<1]\cdot (t+1) +$ \\ $ [1\leq x <3.258] \cdot (3.9852x+t+7.39)$\\ $+[3.258 \leq x < 3.3772]\cdot $\\$(4.4280 x + t + 6.2461)+$ \\ $[3.3772 \leq x] \cdot (3.5867 x + t + 9.0874)$} &  $109.35$  & \makecell{violation of \\ non-negativity }& - & not terminate &-\\
                \hline
                {\makecell{\textsc{Zero-Conf} \\ \textsc{-variant}}}& $\textit{cur}$  & \ding{55} & 3   & $\textit{cur} + 140$ &  \makecell{$[\textit{est} >0]\cdot \textit{cur}+$ \\ $[\textit{start}=0 \wedge \textit{est} \leq 0]\cdot (\textit{cur} + 140)$ \\ $+[\textit{start}\geq 1 \wedge \textit{est} \leq 0]\cdot(\textit{cur} + 42)$} &  180.42  & \makecell{violation of \\ non-negativity } &-  & \makecell{$cur + [\textit{est}=0]$\\$\cdot (-49\cdot \textit{start}^2 -$\\$ 49\cdot \textit{start}+141 )$} & 392.19 \\
				\hline
                {\makecell{\textsc{Equal-} \\ \textsc{Prob-Grid}}}& $\textit{goal}$ & \ding{55} & 2  &  $\textit{goal} + 1.5$   &  \makecell{$[a>10\vee b>10 \vee \textit{goal} \neq 0]\cdot \textit{goal}$ \\ $+[a\leq 10 \wedge b\leq 10 \wedge \textit{goal} = 0]\cdot 1.5$} &$142.68$ & \makecell{$[a>10\vee b>10 \vee $\\$\textit{goal} \neq 0]\cdot \textit{goal} +$ \\ $[a\leq 10 \wedge b\leq 10 $\\$ \wedge \textit{goal} = 0] \cdot 1.5$}  &0.11  & fail &- \\
				\hline
                {\textsc{RevBin}}& $z$  & $2x+z$   &  3    & $2x+z$ &  \makecell{$[x<1] \cdot z +$  $[1 \leq x <2] \cdot (z+x+1) $ \\ $+[x \geq 2] \cdot (z+2x)$} &  $70.30$  & \makecell{$[x<1] \cdot z +$\\  $[x \geq 1] \cdot (z+2x) $} & 0.22  & $z + [x > 0]\cdot 2x$ & 151.26 \\
				\hline

                {\textsc{Fair Coin}}& $i$   &  $i-2y+2$ &  3  &  $i+\frac{4}{3}$ & \makecell{$[x > 0 \vee y>0]\cdot i + $ \\ $[x\leq 0 \wedge y \leq 0]\cdot (i+\frac{4}{3})$} &  $129.34$ & \makecell{$[x > 0 \vee y>0]\cdot i + $ \\ $[x\leq 0 \wedge y \leq 0]\cdot $\\$(i+\frac{4}{3})$} &  0.06  & \makecell{$[ x+y = 0] \cdot $\\$\frac{4}{3} + i$} & 17.95 \\
                \hline
                {\makecell{\textsc{St-Petersburg} \\ \textsc{variant}}} & $y$  & \ding{55} & 3  & $\frac{3}{2} y$ & $[x > 0]\cdot y+ [x \leq 0]\cdot \frac{3}{2} y$ &  1.53  & \makecell{$[x > 0]\cdot y+ $\\$[x \leq 0]\cdot \frac{3}{2} y$} & 0.04 & \makecell{$y+[x=0]\cdot$\\$ 0.5y$} & 13.39  \\
                \hline
			         \end{tabular}
	\end{threeparttable}}
\end{table*}

\smallskip 
\noindent {\bf Answering RQ2.} 
We answer {\bf RQ2} by comparing our approach with the most related approaches~\cite{DBLP:conf/tacas/BatzCJKKM23,DBLP:conf/cav/BaoTPHR22}. We present our comparison results in~\cref{table:1}. The main difference between \textsc{cegispro2}~\cite{DBLP:conf/tacas/BatzCJKKM23} and our approach is that \textsc{cegispro2} requires an upper bound to be verified as an additional program input and it will only return a super-invariant (i.e., a possibly piecewise upper-bound) that is sufficient to \emph{verify} (i.e., smaller than) the input upper bound, while we intend to synthesize a tight piecewise upper bound directly. 
The benchmarks \textsc{Geo, k-Geo} are the common benchmarks in these two works and the direct comparisons are as follows: For the benchmark \textsc{Geo}, the piecewise upper bounds of the two methods are the same. For \textsc{k-Geo}, their piecewise result is consistent with our result over $\mathbb{Z}_{\ge 0}$. While in the scope of real numbers, our piecewise upper bound is tighter than theirs. To have a richer comparison with \textsc{cegispro2}, we give \textsc{cegispro2} an advantage by feeding our benchmarks (including the above two benchmarks) in~\cref{table:1} to \textsc{cegispro2} paired with the piecewise upper bounds synthesized by our approach. We find \textsc{cegispro2} cannot adequately handle piecewise inputs. Additionally, it reports violation of non-negativity on 5 of our benchmarks (see~\cref{table:1}). By feeding one segment from the piecewise bounds synthesized via our approaches for the remaining 8 benchmarks, we find on 6 benchmarks, \textsc{cegispro2} produce the consistent results with our inputs on $\mathbb{Z}_{\ge 0}$, while some of them (e.g., \textsc{k-Geo}) are incorrect over $\mathbb{R}$. On \textsc{Bin-Ran}, the results they produce are impossible to compare since it produces sophisticated and different results when we feed different segments from our piecewise upper bound. On \textsc{Coin}, the execution using their tool does not terminate, which prevents the output of a result. 

The work~\cite{DBLP:conf/cav/BaoTPHR22} considers the probabilistic invariant synthesis via data-driven approach. Note that the synthesis of upper bounds (i.e., super-invariants) is not considered in their work, and the only relevant work in~\cite{DBLP:conf/cav/BaoTPHR22} with our upper bound synthesis is the exact invariant synthesis. For a further comparison, We apply their tool \textsc{exist} on our benchmarks to try to generate exact invariants. On the benchmarks~\textsc{Geo, k-Geo, Mart, Growing Walk, RevBin, Fair Coin, St-Petersburg variant}, \textsc{exist} can generate an exact invariant for each benchmark and we show that on these benchmarks, the piecewise upper bounds we synthesize are equal to their exact invariants so that the upper bounds we synthesize are actually the exact expected value of  $X_f$. On the benchmark~\textsc{Zero-Conf-Variant}, they spend about 400s while we obtain a respectable piecewise linear bound in around 180s. For the remaining benchmarks, their tool fails or the computation seems to be stuck. 

In conclusion, our approach can handle many benchmarks that these two works~\cite{DBLP:conf/tacas/BatzCJKKM23,DBLP:conf/cav/BaoTPHR22} cannot handle. When feeding our benchmarks with the bounds synthesized through our approach to \textsc{cegispro2} and \textsc{exist}, they fail on about $40 \%$ of our benchmarks. Over most of the benchmarks that their and our approaches can handle, our bounds are comparable with theirs. 

\smallskip
\noindent{\bf Answering RQ3.} In addition {\bf RQ2}, we compare our piecewise linear upper bounds with monolithic polynomial bounds via 1-induction in~\cref{table:comparison_upper}. Following~\cite{DBLP:conf/cav/ChatterjeeFG16,10.1145/3656432}, we implement the polynomial synthesis with Putinar's Positivstellensatz~\cite{putinar} (see~\cref{app:putinar}). For a fair comparison, we generate the polynomial bounds with the same invariant and optimal objective function for each benchmark. All the numerical results in the polynomial bounds are cut to $10^{-4}$ precision. 
We compare two results by uniformly taking the grid points in the invariant and evaluate two results, and we compute the percentage of the points that our piecewise upper bound are larger (i.e., no better) than monolithic polynomial, which is shown in the last column "PCT" in~\cref{table:comparison_upper}. 
For each benchmark, we provide difference plots that classify all grid points into three disjoint regions according to the magnitude of difference: red points correspond to cases where our piecewise linear upper bounds are notably smaller (diff $ > 10^{-3}$), blue points correspond to cases where the monolithic upper bounds are notably smaller, and gray points represent regions where the two bounds are nearly identical (diff $ \le 10^{-3}$). We display part of the comparison in \cref{fig:diff_plots_upper1}, see~\cref{app:linear_upper_fig} for other figures.
We observe that on our benchmarks except \textsc{Expected Time}, our piecewise linear bounds are significantly tighter than monolithic polynomial bounds. 
In addition, we quantify the difference between the two upper bounds by subtracting the piecewise upper bound from the monolithic one and taking the unbiased average of the resulting values, which is reported in the last column “Diff” of~\cref{table:comparison_upper}. We observe that, except for \textsc{Expected Time}, our piecewise bounds are generally tighter than the monolithic ones, with especially notable improvements on benchmarks such as \textsc{Mart} and \textsc{Zero-Conf-Variant}.
We conjecture that the relatively poor performance of our piecewise linear algorithm in  \textsc{Expected Time} is due to the fact that the true expected value is closer to a piecewise polynomial function.

\begin{table*}
    \renewcommand{\arraystretch}{1.8}
        \caption{Experimental Results for {\bf RQ3}, Linear Case (Upper Bounds). "$f$" stands for the return function considered in the benchmark, "$k$" stands for the $k$-induction condition we apply in this comparison, "Monolithic Polynomial via 1-Induction" stands for the monolithic polynomial bounds synthesized via 1-induction, and "d" stands for the degree of polynomial template we use. "PCT" stands for the percentage of the points that our piecewise upper bound are larger (i.e., no better) than monolithic polynomial, and "Diff" stands for the unbiased average difference between our piecewise polynomial bound and the monolithic polynomial. A positive value indicates how much tighter our piecewise bounds are on average. }
        \label{table:comparison_upper}
        \resizebox{\textwidth}{!}{
		\begin{threeparttable}
            \begin{tabular}{|c|c|c|c|c|c|c|c|}
				\hline
				\multicolumn{1}{|c|}{\multirow{2}{*}{\textbf{Benchmark}}}  &
				\multicolumn{1}{c|}{\multirow{2}{*}{\textbf{$f$}}}      & 
				\multicolumn{2}{c|}{\multirow{1}{*}{\textbf{Our Approach}}}  &
				\multicolumn{2}{c|}{\multirow{1}{*}{\textbf{\makecell{Monolithic  Polynomial  via 1-Induction }}}} &
                \multicolumn{1}{c|}{\multirow{2}{*}{\textbf{PCT}}} &
                \multicolumn{1}{c|}{\multirow{2}{*}{\textbf{Diff}}} \\ 
                \cline{3-6}
                \multicolumn{1}{|c|}{} & \multicolumn{1}{c|}{} &  
                \multicolumn{1}{c|}{$\; k\;$}  &   
				\multicolumn{1}{c|}{\textbf{Piecewise Linear Upper Bound} } & \multicolumn{1}{c|}{\textbf{d}}  &   
				\multicolumn{1}{c|}{\textbf{Monolithic Polynomial Upper Bound} }
                 & \multicolumn{1}{c|}{}  & \multicolumn{1}{c|}{}    \\ \hline \hline

                {\textsc{Geo}}& $x$  & 3  & \makecell{$[c>0]\cdot x +[c\leq 0]\cdot (x+1)$} & 3 &\makecell{$1.0000 - 1.9996*c + 1.0000*x + $\\$0.9996*c^2 - 0.0002*x*c + 0.0002*x*c^2$}& $0.0\%$ & 0.3249 \\
                \hline
                {\textsc{k-Geo}}& $y$  & 3 &     \makecell{$[k>N]\cdot y + $ \\ $[k\leq N-1] \cdot (-k + N +x + y +1)+$ \\$ [N-1 < k \leq N] \cdot (-0.5k + 0.5N + x + y +1)$} & 2 & \makecell{$269.8049 - 52.761*N - 1.0000*k +$\\$ 1.0000*y + 1.0000*x + 2.688*N^2$} & 10.0\% & 77.813 \\
                \hline
                {\textsc{Bin-ran}}& $y$   &  2   & \makecell{$[i>10]\cdot y + $\\$[\frac{90}{11} < i\leq 10]\cdot (0.9x-21i+y+233)$ \\ $+[i\leq \frac{90}{11}] \cdot (0.9x-18.8i+y+215)$} & 3 & \makecell{$54.2875 + 26.0308*i - 38.3708*y -$\\$ 20.4776*x - 1.683*i^2 - 0.2*y*i - $\\$ 0.035*y^2 - 0.1881*x*i - 0.4591*x*y - $\\$ 1.8397*x^2 - 0.0129*i^3 + 0.3751*y*i^2 + $\\$ 0.0141*y^2*i - 0.0062*y^3 + 0.5049*x*i^2 - $\\$0.0123x*y*i - 0.0126*x*y^2 + 0.9318x^2*i $\\ - $0.0562*x^2*y + 1.0057*x^3$} &  40.29\% & 233.49 \\
                \hline
                {\textsc{Coin}}& $i$    & 3  & \makecell{$[x\neq y]\cdot i + [x=y]\cdot (i+\frac{8}{3})$} & 2 & \makecell{$2.6667 + 1.0000*i - 0.6381*y + 4.2840*x$\\$ - 2.0286*y^2 - 2.0067*x*y + 0.3893*x^2$} & 0.0\% & 3.2915 \\
                \hline
                {\textsc{Mart}}& $i$    & 3   & \makecell{$[x\le 0]\cdot i + [x > 0]\cdot (i+2)$} & 2 & \makecell{$630.7295 + 0.9941*i + $\\$ 199099.682*x + 44.2303*x^2$} &  0.0\% & $9.9*10^5$ \\
                \hline
                {\makecell{\textsc{Growing} \\ \textsc{Walk}}}& $y$   & 3  &   \makecell{$[x<0]\cdot y + [x\geq 0]\cdot (x+y+2)$} & 3 & \makecell{$2.5000+1.0000*y+1.900*x $ \\ $-0.5000*x^2+0.1000*x^3$} & 0.0\% & 8.5964 \\
                \hline
                {\makecell{\textsc{Growing} \\ \textsc{Walk} \\ \textsc{variant}}}& $y$   & 3    &  \makecell{$[x<0]\cdot y +$ \\ $ [0\leq x <1]\cdot (0.5x+y+0.25) $ \\ $+[x\geq 1] \cdot (x+y)$} & 3 &\makecell{$1.0000*y-0.2380*x+0.1041*y^2-$ \\ $0.0686*x*y+0.0951*x^2+0.03558*x*y^2$ \\ $+0.0686*x^2*y+0.1430*x^3$} & 5.52\% & 52.315 \\
                \hline
                {\makecell{\textsc{Expected} \\ \textsc{Time}}}& $t$   &  3  & \makecell{$[x<0]\cdot t + [0\leq x<1]\cdot (t+1)+ $ \\ $ [1\leq x <3.258] \cdot (3.9852x+t+7.39)+$\\ $[3.258 \leq x < 3.3772]\cdot (4.4280 x + t + 6.2461)$ \\ $+[3.3772 \leq x] \cdot (3.5867 x + t + 9.0874)$} & 3 & \makecell{$3.1203 + 0.9622*t + 2.8278*x + $\\$0.0015*t^2 - 0.01558*x*t - 0.1397*x^2 - $\\$0.0003*x*t^2 - 0.0002*x^2*t + 0.0025*x^3$} & 82.0\% & -11.669 \\
                \hline
                {\makecell{\textsc{Zero-Conf} \\ \textsc{-variant}}}& $\textit{cur}$   &3 &  \makecell{$[\textit{est} >0]\cdot \textit{cur}+$ \\ $[\textit{start}==0 \wedge \textit{est} \leq 0]\cdot (\textit{cur} + 140)$ \\ $+[\textit{start}\geq 1 \wedge \textit{est} \leq 0]\cdot(\textit{cur} + 42)$} 
                & 2 & \makecell{$109.8660 - 0.1357*cur + 293795.0410*start + $\\$209178.7117est + 0.0019cur^2 + 0.7202start*cur$\\$-293865.0570*start^2 + 1.0313*est*cur + $\\$274251.8886*est*start - 209283.0750*est^2$} & 0.5 \% & $2.0*10^5$ \\
				\hline
                {\makecell{\textsc{Equal-} \\ \textsc{Prob-Grid}}}& $\textit{goal}$  & 2 &    \makecell{$[a>10\vee b>10 \vee \textit{goal} \neq 0]\cdot \textit{goal}$ \\ $+ [a\leq 10 \wedge b\leq 10 \wedge \textit{goal} = 0]\cdot 1.5$} & 2 &\makecell{$1.6661 + 5.7396*goal - 9.4857*10^{-5}*b +$\\$ 1.5707*10^{-5}*a + 0.6003goal^2 - 0.6740b*goal $\\$ + 1.5975^10{-5}*b^2 + 2.2074*10^{-5}*a*goal$} & 0.0\% & 3.757 \\
				\hline
                {\textsc{RevBin}}& $z$  & 3   &  \makecell{$[x<1] \cdot z +[1 \leq x <2] \cdot (z+x+1) $ \\ $+[x \geq 2] \cdot (z+2x)$} & 2 & $z + 0.7098*x + 1.2902*x^2$ & 0.0\% & 36.909 \\
				\hline
                {\textsc{Fair Coin}}& $i$   &  3  & \makecell{$[x > 0 \vee y>0]\cdot i + $ \\ $[x\leq 0 \wedge y \leq 0]\cdot (i+\frac{4}{3})$} & 2 & \makecell{$1.3333 + 1.0000*i - 0.4141*y - 0.4141*x$\\$ + 1.1743*i^2 - 2.3486*y*i + 0.2551*y^2 $\\$- 2.3486*x*i + 3.6820*x*y + 0.2551*x^2$} & 0.0\% & 23.677 \\
                \hline
                {\makecell{\textsc{St-Petersburg} \\ \textsc{variant}}} & $y$ & 3 & $[x > 0]\cdot y+ [x \leq 0]\cdot \frac{3}{2} y$ & 3 &\makecell{$0.0018 + 1.5006*y + 808.8832*x $\\$- 0.0112*y^2 + 7.7505*x*y - 191.2143*x^2$\\$ + 0.0113*x*y^2 - 8.251*x^2*y - 617.6696*x^3$} & 0.0\% & 192.25 \\
                \hline
                
            \end{tabular}
        \end{threeparttable}}
\end{table*}

\begin{figure}[htbp]
  \centering
  \subfloat[\textsc{Growing Walk}]
  {\includegraphics[width=0.2\textwidth]{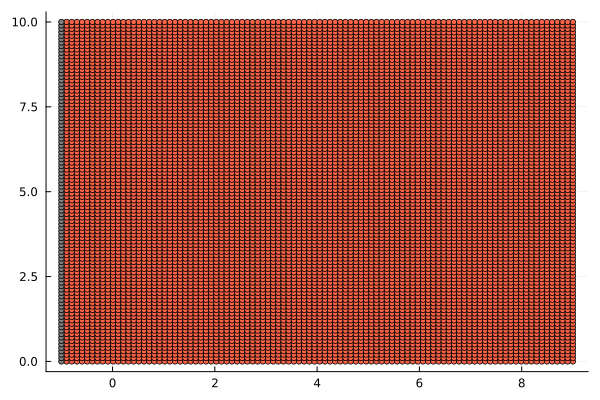}} \hfill   
  \subfloat[\textsc{Geo}]
  {\includegraphics[width=0.2\textwidth]{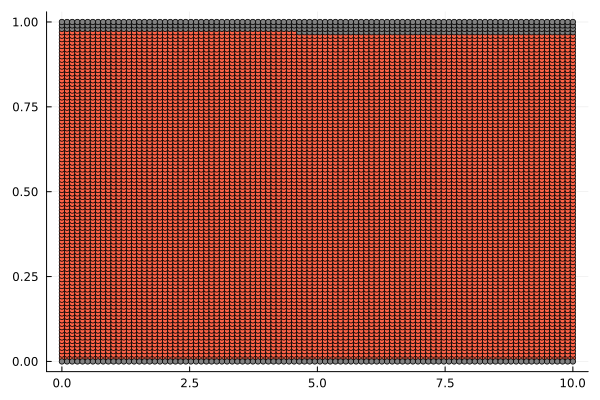}}  \hspace{0.1cm}
  \subfloat[\textsc{Zero-Conf-variant}]
  {\includegraphics[width=0.31\textwidth]{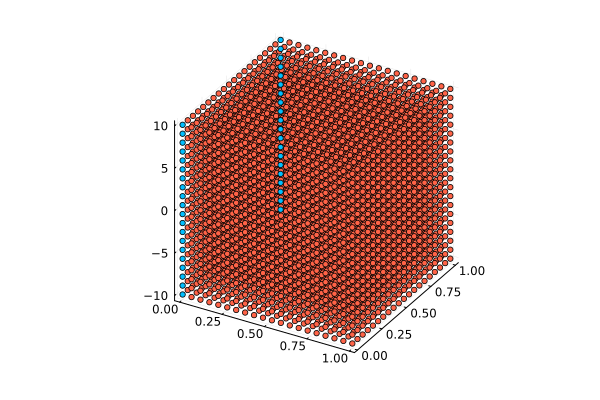}} \hspace{-0.8cm}
  \subfloat[\textsc{Bin-Ran}]
  {\includegraphics[width=0.31\textwidth]{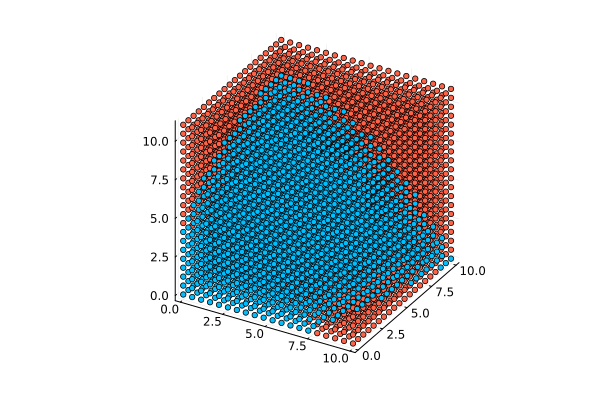}}\hspace{-0.5cm}  

  \caption{Difference Plots of the Comparison in Piecewise Linear Case}
  \label{fig:diff_plots_upper1}
  
  \begin{tablenotes}
      \item \textcolor{red}{Red} points indicate where our piecewise upper bounds are smaller than the monolithic ones \\ by more than $10^{-3}$; 
      \textcolor{blue}{blue} points indicate where the monolithic bounds are smaller by more \\ than $10^{-3}$;
      \textcolor{gray}{gray} points denote cases with negligible differences ($\le 10^{-3}$).
  \end{tablenotes}
  \vspace{-2ex}
\end{figure}

\subsection{Piecewise Polynomial Bound Synthesis}\label{sec:poly_res}

\noindent {\bf Benchmark Selection. } We select all remaining benchmarks from~\cite{DBLP:conf/cav/BaoTPHR22,DBLP:conf/tacas/BatzCJKKM23} that are not used in the previous linear experiments, as well as  path probability estimation benchmarks from~\cite{Beutner2022b,DBLP:conf/pldi/SankaranarayananCG13,10.1145/3656432,DBLP:conf/cav/GehrMV16}, including all unbounded loop benchmarks from~\cite{DBLP:conf/pldi/SankaranarayananCG13} in particular. 
For the former 7 benchmarks from~\cite{DBLP:conf/cav/BaoTPHR22}, we instantiate the probability parameters with commonly used values (such as $0.5$). Note that among them, the benchmarks \textsc{GeoAr, Bin0, Bin2, Sum0, Duel} cannot be handled by our piecewise linear algorithm with $k$-induction when $k=1,2,3$, even though both the program and the return function are linear. 
For the benchmarks from~\cite{DBLP:conf/tacas/BatzCJKKM23}, the benchmarks \textsc{chain, brp} exhibit numerical pathologies due to extremely large constants, which can cause numerical instability and render our algorithms ineffective. To address this issue, we scale down these pathological values to more moderate magnitudes—for instance, replacing 1000000000000 with 100 in the \textsc{chain} benchmark and 8000000000 with 800 in the \textsc{brp} benchmark—so that our numerical algorithm can operate reliably.
For the benchmarks from~\cite{Beutner2022b,DBLP:conf/pldi/SankaranarayananCG13,10.1145/3656432,DBLP:conf/cav/GehrMV16}, since 5 of 9 benchmarks contain continuous distributions originally, we make simple adaptions on these benchmarks by replacing each continuous distribution (e.g. uniform distribution over $[0,1]$) with a uniform discrete choice of the same range (e.g. $0$ with probability $0.5$ and $1$ also with $0.5$), resulting in 5 adapted benchmarks. 
The benchmark \textsc{inv-Pend} in~\cite{DBLP:conf/pldi/SankaranarayananCG13} does not pass our checking of prerequisite (P2). Therefore we make minor modifications to the coefficients in this benchmark so that we can synthesize a dbRSM to satisfy~(P2), thereby obtaining the benchmark~\textsc{inv-Pend variant}.
For piecewise polynomial bound synthesis, larger values of $k$ significantly increase the computational cost. Therefore, we apply $2$-induction on these 24 benchmarks.

\smallskip 
\noindent{\bf Answering RQ1.} 
Our algorithm successfully handles all of the aforementioned benchmarks except for four. The failures in these cases are attributed to excessive branching introduced by our algorithm based on Proposition~\ref{prop:relation} (see \textbf{Stage 2} in~\cref{sec:algorithm}), and branch reduction techniques (see Page~\pageref{page:branchreduction}) have not yet been incorporated into our implementation. Nevertheless, our current implementation is capable of addressing a wide range of complex benchmarks.
For example, the benchmark \textsc{cav5} comprises 35 lines of code (see~\cref{app:programs}), the benchmark \textsc{inv-Pend variant} benchmark features 4 variables with complex polynomial updates, posing significant challenges for analysis.
We leave further optimization for future work. 
We present the experimental results for the synthesis of piecewise polynomial upper bounds on the remaining 20 benchmarks in Table~\ref{table:poly1}. Our approach successfully derives piecewise polynomial upper bounds for 16 out of 20 benchmarks within seconds. Of the remaining four, two benchmarks (\textsc{fig-6} and \textsc{fig-7}) are solved within tens of seconds, while only \textsc{inv-Pend variant} and \textsc{cav-5} require more than five minutes to compute a result.
Our algorithms obtain the exact bound (i.e., the tightest upper bound) on the benchmarks~\textsc{Bin0, Bin2, DepRV, Prinsys, Sum0}. The exactness of these results is verified by comparison with the exact invariants synthesized in~\cite{DBLP:conf/cav/BaoTPHR22} (see \textbf{RQ2}) and with our corresponding lower bounds in~\cref{app:poly_lower}.

\begin{table*}
    \renewcommand{\arraystretch}{1.5}
	\caption{Experimental Results for {\bf RQ1} and {\bf RQ2}, Polynomial Case (Upper Bounds). "$f$" stands for the return function considered in the benchmark, "T(s)" stands for the execution time of our approach (in seconds), including the parsing procedure from the program input, relaxing the $k$-induction constraint into the SDP problems, the SDP solving time and verification time. "d" stands for the degree of polynomial template we use and "Solution $h^*$" is the candidate polynomial solved directly by the solver. "Piecewise Polynomial upper Bound" stands for the piecewise bound we synthesize. "Exact" stands for the exact expected result from~\textsc{exist}.}
	\label{table:poly1}
	\resizebox{\textwidth}{!}{
		\begin{threeparttable}
			\begin{tabular}{|c|c|c|c|c|c|c|c|}
				\hline
				\multicolumn{1}{|c|}{\multirow{2}{*}{\textbf{Benchmark}}}  &
				\multicolumn{1}{c|}{\multirow{2}{*}{\textbf{$f$}}}      &
				\multicolumn{4}{c|}{\multirow{1}{*}{\textbf{Our Approach}}}  &
                    \multicolumn{2}{c|}{\multirow{1}{*}{\textbf{\textsc{exist}}}}  
				\\ \cline{3-8}
				\multicolumn{1}{|c|}{} & \multicolumn{1}{c|}{} &  
                    \multicolumn{1}{c|}{\textbf{d}} &
                    \multicolumn{1}{c|}{\textbf{Solution $h^*$}}  & 
				\multicolumn{1}{c|}{\textbf{ T(s) }} &
				\multicolumn{1}{c|}{\textbf{ Piecewise Polynomial Upper Bound}}   &    
				\multicolumn{1}{c|}{\textbf{ Exact }} &  
                \multicolumn{1}{c|}{\textbf{T(s)}} 
                \\ \hline \hline
                {\textsc{GeoAr}}& $x$ &  2  &  \makecell{$0.0001*x^2 - 0.0004*x*y + 0.0005*x*z $\\$ +0.0011*y^2 + 0.0079*z^2 + 0.9998*x $\\$+ 1.5398*y - 0.0085*z + 5.0078$}  & 3.92 & \makecell{$\min\{[z>0] \cdot (0.0001x^2 - 0.0003*x*y + $\\$0.0003*x*z + 0.0010y^2 + 0.0003*y*z $\\$+ 0.0040z^2 + 0.9995x + 2.0416y - $\\$0.004z + 7.0485) + [z\le 0]\cdot x, h^*\}$} &  \makecell{inner error} & - \\
                \hline
                {\textsc{Bin0}}& $x$   &  2 &  $x + 0.5*y*n$  &  6.05 & $x + [n>0]\cdot 0.5*y*n$  & \makecell{$x + [n>0]\cdot $\\$0.5*y*n$} & 79.04\\
                \hline
                {\textsc{Bin2}}& $x$ &  2  &  $0.25*n + x + 0.25*n^2 + 0.5*y*n$  &  5.81   & \makecell{$x + [n>0]\cdot(0.25*n + x $\\$+ 0.25*n^2 + 0.5*y*n)$} &  \makecell{$x + [n>0]\cdot$\\$(0.25*n$\\$ +  x + 0.25*n^2$\\$ + 0.5*y*n)$} & 250.60\\
                \hline
                {\textsc{DepRV}}& $x*y$  &  2  &  \makecell{$-0.25*n + 0.25*n^2 + 0.5*y*n$\\$ + 0.5*x*n + x*y$}  &  5.91  & \makecell{$[n >0] \cdot (-0.25*n + 0.25*n^2 +$ \\$ 0.5*y*n + 0.5*x*n + x*y) $\\$+ [n\leq 0] \cdot x*y$} &  \makecell{inner error} & -\\
                \hline
                {\textsc{Prinsys}}& $[x==1]$ &   2  &  $0.5 + 0.5*x$  &  2.35   & \makecell{$[x==1]*1 + [x==0]*0.5$} &  \makecell{$[x==1]*1 + $\\$[x==0]*0.5$} & 3.02 \\
                \hline
                {\textsc{Sum0}}& $x$ & 2 & $0.25*i^2+0.25*i+x$ &  2.33 & $[i>0]*(0.25*i^2+0.25*i)+x$ & \makecell{$x+[i>0]*$\\$(0.25i+0.25i^2)$} & 105.01\\
                \hline
                {\textsc{Duel}}& $t$ & 2 & \makecell{$-20.267*x^2 - 0.4198*x*t - 2.5502*t^2$\\$ + 20.6657*x + 3.5505*t + 0.0013$} & 6.9 & \makecell{$\min\{[t>0\wedge x\ge 1]\cdot (-10.1335x^2 - 2.5502t^2$\\$+0.2099*x*t  + 10.1230*x  + 2.5502*t $\\$ + 0.5015)+[t\le 0\wedge x \ge 1]\cdot (-5.0668*x^2 $\\$+ 0.1050*x*t - 2.5502*t^2 + 5.0615*x $\\$+ 3.0504*t + 0.2514)+[x < 1]\cdot t, h^*\}$} & fail & -\\
                \hline
                {\textsc{brp}}& \makecell{$[failed$\\$=10]$}  & 2 & \makecell{$38912.3699*failed^2 + 0.7329*sent^2+$\\$3.2173*failed*sent + 1486.258*failed $\\$ -573.6644*sent - 2459.9909$} & 10.12 & \makecell{$\min\{[failed<10 \wedge sent < 800]\cdot (0.7329sent^2 $\\$+0.0322*failed*sent + 389.1237*failed^2 $\\$+793.1100*failed - 572.1811*sent$\\$ - 2623.2068)+[failed=10] , h^*\}$} & \makecell{not \\ terminate} & -\\
                \hline
                {\textsc{chain}}& $[y=1]$ &  2 & \makecell{$-0.006*x*y + 0.4841*y^2 - $\\$0.0021*x + 0.4477*y + 0.1007$} & 4.79 & \makecell{$\min\{[y= 0 \wedge x < 100] \cdot (-0.0059*x*y $\\$ + 0.4793*y^2 -0.0022*x + 0.4373*y$\\$ + 0.1079) + [y=1], h^* \}$} & fail & -\\
                \hline
                {\textsc{grid small}}& \makecell{$[a<10 \wedge $\\$ b\ge 10]$} &  3 & \makecell{$ 0.0018*a*b^2 -0.0003*a^3 - 0.0008*a^2*b  $\\$  -0.0011*b^3 + 0.0117*a^2 - 0.0154*a*b + $\\$ 0.0136*b^2 - 0.097*a + 0.0239*b + 0.5355$} 
                & 6.71 & \makecell{$\min\{[a<10 \wedge b < 10] \cdot (-0.0003*a^3 - $\\$0.0011*b^3 -0.0008*a^2*b+ 0.0018*a*b^2  $\\$+ 0.0109*a^2 \cdots + 0.0277*b$\\$  + 0.5109)+[a<10 \wedge b\ge 10], h^*\}$} & Not support & - \\
                \hline
                {\textsc{grid big}}& \makecell{$[a<1000 \wedge$\\$ b \ge 1000]$} &  2 &  \makecell{$0.0159*a^2 - 0.0319*a*b + 0.0159*b^2$\\$ + 0.2715*a - 0.3086*b - 0.437$}& 7.74 & \makecell{$\min\{[a < 1000 \wedge b < 1000]\cdot (0.0159*a^2 $\\$ - 0.0319*a*b + 0.0159*b^2 + 0.2714*a $\\$- 0.3087*b - 0.4397)+$\\$  [a < 1000 \wedge b \ge 1000], h^* \}$} & Not support & -\\
                \hline
                {\textsc{cav-2}}& $[h>1+t]$ &  3  & 0.0 & 3.78 & $[h>t+1]$ & fail & - \\
                \hline
                {\textsc{cav-4}}& $[x \le 10]$  & 2 & 1.0 & 2.75 & 1.0 & \makecell{inner error} & - \\
                \hline
                {\textsc{fig-6}}& $[y\le 5]$ & 4 & \makecell{$ 0.0011*x^3*y -0.0001*x^4  - 0.0001*y^4  $\\$+ 0.0008*x*y^3  - 0.001*x^2*y^2 + \cdots $\\$+ 0.5712*x - 0.281*y + 0.6009$} & 109.03 & \makecell{$\min\{[x \le 4 ]\cdot (-0.0001*x^4 + 0.0011*x^3*y $\\$- 0.001*x^2*y^2 + 0.0008*x*y^3 - 0.0001*y^4 $\\$+ 0.0023*x^3 \cdots - 0.0094*y^2+ 0.5530*x - $\\$ 0.2782*y + 0.6027) + [x > 4 \wedge y \le 5], h^*\}$} & inner error & - \\
                \hline
                {\textsc{fig-7}}& $[x \le 1000]$  & 2 & \makecell{$0.0005*y^2 - 0.0008*y*i $\\$+ 0.0002*i^2 - 0.0001*x +$\\$ 0.001*y - 0.0005*i + 1.0003$} & 24.32 & \makecell{$\min\{[y\le 0]\cdot (0.0002*i^2 - 0.0002*x - $\\$0.0005*i + 1.0004)+[y > 0\wedge x \le 1000], h^* \}$} & inner error & - \\
                \hline
                {\textsc{\makecell{inv-Pend \\ variant}}}& $[pA\le 1]$ & 3 & \makecell{$0.0058*pAD^2*pA + 0.0023*pAD^2*cV- $\\$0.1313*pAD^2*cP - 0.6278*pAD*pA^2-$\\$ 0.2352pAD*pA*cV\cdots +5.9637*cV*cP $\\$+ 60.4194*cP^2 + 7.1495*cV + 1.0$} & 412.20 & \makecell{$\min \{[cp>0.5 \vee cp < -0.5 \vee pA > 0.1 \vee $\\$  pA < -0.1] \cdot (0.0058pAD^2pA - 0.0011pAD^2cV $\\$- 0.1313pAD^2*cP+ \cdots + 0.0689*cP + 0.3238)$\\$ + [-0.5 \le cp \le 0.5 \wedge -0.1 \le pA \le 0.1], h^* \}$} &  inner error & - \\

                \hline
                {\textsc{CAV-7}}& $[x\leq 30]$   &  3  &  \makecell{$-0.0001*i^3 + 0.0002*i^2*x $\\$+  0.0011*i^2-0.0012*i*x - $\\$ 0.0009*i - 0.0001*x + 0.9993$}  &  5.26   & \makecell{$\min\{[i<5]\cdot (0.0001*i^2*x + 0.0005*i^2 - $\\$0.0006*i*x + 0.0004*i - 0.0011*x $\\$+ 0.9983) + [i<5 \wedge x \le 30], h^* \}$} &  \makecell{inner error} & -\\
                \hline
                {\textsc{cav-5}}& $[i\le 10]$ &   3  & \makecell{$ -0.0001*i*money^2 - 0.0004*i^2 - $\\$ 0.0006*i*money + 0.1029*money^2 $\\$+ 0.0037*i + 1.0$} &  892.6  & \makecell{$\min\{[money \ge 10]\cdot (-0.0001*i*money^2 -$\\$ 0.0004*i^2 - 0.0004*i*money + 0.0015*i $\\$+0.1028*money^2  - 0.2118*money + $\\$3.1283)+[money < 10 \wedge i \le 10], h^* \}$} &  \makecell{inner error} & - \\
                \hline
                {\textsc{Add}}& $[x>5]$ &   3  &  \makecell{$0.9402 - 0.1887*y + 0.0833*x - 0.0751*y^2 $\\$+ 0.0556*x*y - 0.0107*x^2 - 0.0281*y^3 $\\$+ 0.023*x*y^2 - 0.0052*x^2*y + 0.0005*x^3$} &  3.63  & \makecell{$\min\{[y \le 1]\cdot (0.9403 - 0.1953*y + 0.0831*x$\\$ - 0.0736*y^2 + 0.0565*x*y - 0.0105*x^2 $\\$- 0.0281*y^3 + 0.023*x*y^2 - 0.0052*x^2*y$\\$ + 0.0005*x^3) + [y > 1 \wedge x > 5], h^* \}$} &  \makecell{inner error} & - \\ 
                \hline
                {\textsc{\makecell{Growing\\Walk\\variant2}}} & $y$ & 2 & \makecell{$0.0622*x^2 - 1.2722*x*r +  6.5027*r^2 $\\$+0.6396*x + y - 6.5379*r + 1.6433$} & 5.33 & \makecell{$\min\{[r\le 0]\cdot(0.0622*x^2 + 0.6279*x$\\$ + y + 1.6914)+[r > 0]\cdot y, h^*\}$}& inner error & - \\
                \hline
			\end{tabular}
	\end{threeparttable}}
\end{table*}

\smallskip
\noindent{\bf Answering RQ2.} We answer {\bf RQ2} by comparing our approach with the relevant work \textsc{exist}~\cite{DBLP:conf/cav/BaoTPHR22} in~\cref{table:poly1}, whose illustration is the same to~\cref{table:1}. 
It is worth noting that \textsc{cegispro2} only supports linear bounds and does not accept nonlinear expressions as additional program input. Therefore, we exclude it from our comparison.
Note that the only relevant aspect of~\cite{DBLP:conf/cav/BaoTPHR22} with respect to upper bound synthesis (i.e., super-invariants) is their method for exact invariant synthesis. For comparison, we apply their tool \textsc{exist} to our benchmarks in an attempt to generate exact invariants.
On the benchmarks \textsc{Bin0}, \textsc{Bin2}, \textsc{Prinsys}, and \textsc{Sum0}, we show that the piecewise polynomial upper bounds we synthesize are actually the exact expected value of $X_f$, i.e., the tightest upper bounds, by comparing them with the exact invariants synthesized by \textsc{exist}. Among these benchmarks, \textsc{exist} spends about 80s on \textsc{Bin0}, about 250s on \textsc{Bin2}, and about 100s on \textsc{Sum0}, while we spend only several seconds to obtain the specific results. Thus, our algorithm is much more efficient.
For the benchmarks \textsc{Duel}, \textsc{chain}, and \textsc{cav2}, the tool \textsc{exist} is able to identify candidates for exact invariants but fails to verify them, and thus does not produce exact invariants. Additionally, \textsc{exist} does not support the benchmarks \textsc{grid small} and \textsc{grid big}. For the remaining benchmarks, \textsc{exist} fails to generate results due to internal errors.
Moreover, for the benchmark \textsc{DepRV}, we demonstrate that the piecewise polynomial upper bound synthesized by our approach is exact, as it coincides with the corresponding lower bound (see~\cref{app:poly_lower}). Thus, our method yields the tightest upper bound for 5 out of the 20 benchmarks in this table.
In summary, our approach successfully handles more benchmarks than~\cite{DBLP:conf/cav/BaoTPHR22}, and for those benchmarks that both methods can process, our approach is more efficient and produces comparable bounds.

\smallskip
\noindent{\bf Answering RQ3.} In addition to the comparisons in \textbf{RQ2}, we further evaluate our piecewise polynomial upper bounds (obtained via $k$-induction) against monolithic polynomial bounds of higher degree synthesized using simple induction (i.e., $1$-induction). The synthesis of these monolithic polynomial bounds is implemented using Putinar's Positivstellensatz~\cite{putinar} (see~\cref{app:putinar} for details). For a fair comparison, we use the same invariant and optimal objective function for each benchmark. We also verify the validity of the monolithic polynomial bounds (see \emph{Numerical Repair}). In our experimental evaluation, we observe that for most benchmarks, when the degree of the polynomial template exceeds 5, numerical performance deteriorates and the synthesized monolithic bounds fail our validation process. Therefore, in this experiment, we restrict the degree of monolithic polynomial bounds to at most 5. 

We present the comparison results in~\cref{table:comparison_upper_poly}, whose illustration is the same to~\cref{table:comparison_upper}. To compare the two synthesized bounds, we uniformly sample grid points from a region of interest (typically a subset of the invariant) and evaluate both results at these points. We compute the percentage of points where our piecewise polynomial upper bound is larger (i.e., no better) than the (higher degree) monolithic polynomial, which is shown in the column "PCT" in~\cref{table:comparison_upper_poly}. 
For each benchmark, we provide difference plots that classify all grid points into three disjoint regions according to the magnitude of difference: red points correspond to cases where our piecewise linear upper bounds are notably smaller (diff $ > 10^{-3}$), blue points correspond to cases where the monolithic upper bounds are notably smaller, and gray points represent regions where the two bounds are nearly identical (diff $ \le 10^{-3}$). We display part of the comparison in \cref{fig:diff_plots_upper2}, see~\cref{app:poly_upper_fig} for other figures.
We show that on all the benchmarks except~\textsc{grid small, grid big, fig-6}, our piecewise polynomial bounds are {\em significantly} tighter and simpler than monolithic polynomial bounds. 
In addition, we quantify the difference between the two upper bounds by subtracting the piecewise upper bound from the monolithic one and taking the unbiased average of the resulting values, which is reported in the column “Diff” of~\cref{table:comparison_upper_poly}. We observe that, on the benchmarks~\textsc{Bin0, DepRV, Prinsys, Sum0}, our piecewise polynomial result is consistent with the monolithic one, which makes both the "PCT" and "Diff" values nearly zero. On the benchmarks \textsc{Bin2, grid small, cav-2, fig-6, inv-Pend variant, Add}, our piecewise upper bounds are close to the monolithic upper bounds on average. On the remaining benchmarks, our piecewise bounds are generally tighter than the monolithic ones, with especially notable improvements on benchmarks such as \textsc{chain} and \textsc{grid big}.
Although our running time is a bit longer than that of monolithic polynomial experiments, our approach allows to synthesize lower-degree polynomials while achieving better precision against higher-degree polynomials. This advantage is critical as the synthesis of higher-degree polynomials suffers from a large amount of numerical errors as stated previously. 

\begin{table*}
    \renewcommand{\arraystretch}{1.8}
        \caption{Experimental Results for {\bf RQ3}, Polynomial Case (Upper Bounds). "$f$" stands for the return function considered in the benchmark. "Piecewise Polynomial Upper Bound" stands for the results synthesized by our algorithm. "Monolithic Polynomial via 1-Induction" stands for the monolithic polynomial bounds synthesized via 1-induction, and "T(s)" stands for the total execution time. "PCT" stands for the percentage of the points that our piecewise polynomial upper bound are larger (i.e., no better) than (higher degree) monolithic polynomial, and "Diff" stands for the unbiased average difference between our piecewise polynomial bound and the monolithic polynomial. A positive value indicates how much tighter our piecewise bounds are on average. }
        \label{table:comparison_upper_poly}
        \resizebox{\textwidth}{!}{
		\begin{threeparttable}
            \begin{tabular}{|c|c|c|c|c|c|c|c|c|c|}
				\hline
				\multicolumn{1}{|c|}{\multirow{2}{*}{\textbf{Benchmark}}}  &
				\multicolumn{1}{c|}{\multirow{2}{*}{\textbf{$f$}}}      &
				\multicolumn{3}{c|}{\multirow{1}{*}{\textbf{Our Approach}}}  &
				\multicolumn{3}{c|}{\multirow{1}{*}{\makecell{\textbf{Monolithic Polynomial  via 1-induction}}}} &
                \multicolumn{1}{c|}{\multirow{2}{*}{\textbf{PCT}}} &
                \multicolumn{1}{c|}{\multirow{2}{*}{\textbf{Diff}}} \\ 
                \cline{3-8}
                \multicolumn{1}{|c|}{}  & \multicolumn{1}{c|}{} &  
                \multicolumn{1}{c|}{\textbf{d}}  &  \multicolumn{1}{c|}{\textbf{T(s)}} &  
				\multicolumn{1}{c|}{\textbf{Piecewise Polynomial Upper Bound}} &\multicolumn{1}{c|}{\textbf{d}} &  \multicolumn{1}{c|}{\textbf{T(s)}} &   
				\multicolumn{1}{c|}{\textbf{Monolithic Polynomial Upper Bound}} 
                & \multicolumn{1}{c|}{} & \multicolumn{1}{c|}{}  \\ \hline \hline

                {\textsc{GeoAr}}& $x$ & 2 &  3.92 &  \makecell{$\min\{[z>0] \cdot (0.0001x^2 - 0.0003*x*y + $\\$0.0003*x*z + 0.0010y^2 + 0.0003*y*z $\\$+ 0.0040z^2 + 0.9995x + 2.0416y - $\\$0.004z + 7.0485) + [z\le 0]\cdot x, h^*\}$} &  3 & 0.84  &  \makecell{$-0.0001x^3 + 0.0001*x^2*y - 0.0011*x^2*z-$\\$ 0.0004*x*y^2 - 0.0112*x*y*z + 0.164*x*z^2$\\$ + 0.0012*y^3+ \cdots - 0.0137*y^2 + 2.7194*y*z $\\$+ 0.9993*x + 0.0417*y + 89867.2768*z + 0.078$}   &  5.0\% & 4976.4\\
                \hline

                {\textsc{Bin0}}& $x$ & 2    &  6.05 &  \makecell{$x + [n>0]\cdot 0.5*y*n$}  & 3 & 0.65 & \makecell{$0.5*y*n + x$}  &  0.0\%  & 0.019 \\
                \hline
                {\textsc{Bin2}}& $x$    &  2  &   5.81  & \makecell{$x + [n>0]\cdot(0.25*n + $\\$x + 0.25*n^2 + 0.5*y*n)$} &  3 & 0.53 & \makecell{$-0.0001*x*y^2 - 0.0002*x*y*n -$\\$ 0.0001*x*n^2 + \dots + 0.2496*n^2 $\\$+ 0.9986*x +  0.033*y + 0.2641*n + 0.051$} & 9.2\%  & 0.014 \\
                \hline
                {\textsc{DepRV}}& $x*y$ &  2 &  5.91  &  \makecell{$[n >0] \cdot (-0.25*n + 0.25*n^2 + 0.5*y*n $ \\$+ 0.5*x*n + x*y) + [n\leq 0] \cdot x*y$}  &  3 &  0.74  & \makecell{$x*y + 0.5*x*n + 0.5*y*n + $\\$ 0.25*n^2 - 0.2499*n + 0.0001$} &  0.0\%  & 0.0005 \\
                \hline
                {\textsc{Prinsys}}& $[x==1]$ &   2  & 2.35 & \makecell{$[x==1]*1 + [x==0]*0.5$}     & 3 & 0.75 & $0.2973*x^3 + 0.2027*x + 0.5$ & 0.0\%  & 0.0 \\
                \hline
                {\textsc{Sum0}}& $x$ &   2  & 2.33 & $[i>0]*(0.25*i^2+0.25*i)+x$     & 4 & 0.7 & $0.25*i^2 + 0.25*i + x$ & 0.0\%  & 0.0 \\
                \hline
                {\textsc{Duel}}& $t$ & 2  & 6.90 &  \makecell{$\min\{[t>0\wedge x\ge 1]\cdot (-10.1335x^2 - 2.5502t^2$\\$+0.2099*x*t  + 10.1230*x  + 2.5502*t $\\$ + 0.5015)+[t\le 0\wedge x \ge 1]\cdot (-5.0668*x^2 $\\$+ 0.1050*x*t - 2.5502*t^2 + 5.0615*x $\\$+ 3.0504*t + 0.2514)+[x < 1]\cdot t, h^*\}$}  & 4 & 0.92 & \makecell{$-175.0474x^4 - 33.1201x^3t - 256.8154*x^2*t^2$\\$ + 74.5673*x*t^3 + 81.1314*t^4 - 115.4608*x^3+$\\$ 153.7459*x^2*t - 125.7204*x*t^2 - 104.9856t^3 $\\$+ 78.3171*x^2 + 186.7714*x*t - 135.7646*t^2 $\\$+ 212.334*x + 160.6187*t$} &  0.02\%  & 119.7 \\
                \hline
                {\textsc{brp}}& \makecell{$[failed$\\$=10]$} &   2  & 10.12 & \makecell{$\min\{[failed<10 \wedge sent < 800]\cdot (0.7329sent^2 $\\$+0.0322*failed*sent + 389.1237*failed^2 $\\$+793.1100*failed - 572.1811*sent$\\$ - 2623.2068)+[failed=10] , h^*\}$}
                & 4 & 1.27 & \makecell{$5.6049failed^4 + 4.902failed^3sent + 3.2666failed^3$\\$ - 0.0035failed^2sent^2   - 7.0269*failed^2*sent $\\$+ 0.0019*failed*sent^2 + 2.9608*failed*sent $\\$- 0.0001sent^3 + 5.1816*failed^2 -0.0288*sent^2$\\$ + 2.4293*failed - 7.3179*sent - 0.9176$} & 28.8\%  & 2964.2 \\
                \hline
                {\textsc{chain}}& $[y=1]$ &   2  & 4.79 & \makecell{$\min\{[y= 0 \wedge x < 100] \cdot (-0.0059*x*y $\\$ + 0.4793*y^2 -0.0022*x + 0.4373*y$\\$ + 0.1079) + [y=1], h^* \}$} & 3 & 1.15 & \makecell{$-0.0449*x^3 - 0.5045*x^2*y + 5.611*x*y^2- $\\$155242.5616*y^3 + 5.5921*x^2 + 43.0661*x*y $\\$- 668140.0947*y^2 - 117.5705*x + $\\$823721.7882*y + 160.1718$}  & 0.85\%  & $1.2*10^5$ \\
                \hline
                {\textsc{grid small}}& \makecell{$[a<10 \wedge $\\$ b \ge 10]$} &   3  & 6.71 & \makecell{$\min\{[a<10 \wedge b < 10] \cdot (-0.0003*a^3 - $\\$0.0011*b^3 -0.0008*a^2*b+ 0.0018*a*b^2  $\\$+ 0.0109*a^2 \cdots + 0.0277*b$\\$  + 0.5109)+[a<10 \wedge b\ge 10], h^*\}$} 
                & 4 & 1.16 & \makecell{$0.0001*a^3 - 0.0003*a^2*b + 0.0002*a*b^2 $\\$- 0.0001*b^3 - 0.0002*a^2 + 0.0001*a*b $\\$+ 0.0003*b^2 - 0.0326*a + 0.0322*b + 0.4628$} & 43.54\%  & 0.0014 \\
                \hline
                {\textsc{grid big}}& \makecell{$[a<1000 \wedge$\\$ b \ge 1000]$} &   2  & 7.74 & \makecell{$\min\{[a < 1000 \wedge b < 1000]\cdot (0.0159*a^2 $\\$ - 0.0319*a*b + 0.0159*b^2 + 0.2714*a $\\$- 0.3087*b - 0.4397)+$\\$  [a < 1000 \wedge b \ge 1000], h^* \}$}     & 3 & 0.83 & \makecell{$-809.0361 - 1408.2916*b + 1397.5134*a$\\$ + 2.9331*b^2 - 5.8658*a*b + 2.9431*a^2 $\\$- 0.0004*b^3 + 0.0006*a*b^2 + $\\$0.0011*a^2*b - 0.0012*a^3$} & 34.49\%  & $3.8*10^5$ \\
                \hline
                {\textsc{cav-2}}& $[h>t+1]$ &   3  & 3.78 & $[h>t+1]$  & 4 & 0.75 & \makecell{$0.0008*h^2 - 0.001*h*t + 0.001*t^2 $\\$- 0.0066*h - 0.0073*t + 0.0885$} &  0.0\%  & 0.0489\\
                \hline
                {\textsc{cav-4}}& $[x\le 10]$ &   2  & 2.75 &  1.0   & 3 & 0.62 & \makecell{$0.0007*x*y^2 - 20.236*y^3 - 0.0007*x*y$\\$ + 13.2821*y^2 + 6.9539*y + 1.0$} & 0.0\%  & 364.32\\
                \hline
                {\textsc{fig-6}}& $[y\le 5]$ &   4  & 109.03 & \makecell{$\min\{[x \le 4 ]\cdot (-0.0001*x^4 + 0.0011*x^3*y $\\$- 0.001*x^2*y^2 + 0.0008*x*y^3 - 0.0001*y^4 $\\$+ 0.0023*x^3 \cdots - 0.0094*y^2+ 0.5530*x - $\\$ 0.2782*y + 0.6027) + [x > 4 \wedge y \le 5], h^*\}$}    & 5 & 1.12 & \makecell{$-0.0001*x^5 - 0.0002*x^4*y - 0.0003*x^2*y^3 $\\$+ 0.0001*x*y^4 - 0.0002*y^5 + 0.0011*x^4+$\\$  0.0037*x^3*y \cdots + 0.1432*x*y + 0.0064*y^2 $\\$+ 0.9708*x - 0.6526*y + 0.575$} & 42.73\%  &  0.2507\\
                \hline
                {\textsc{fig-7}}& $[x \le 1000]$ &   2  & 24.32 & \makecell{$\min\{[y\le 0]\cdot (0.0002*i^2 - 0.0002*x - $\\$0.0005*i + 1.0004)+[y > 0\wedge x \le 1000], h^* \}$}     & 3 & 2.65 & \makecell{$0.0003*x^2*i-0.083*x^2*y  + 48.5638*x*y^2+$\\$  0.5267*x*y*i - 0.018*x*i^2 + 2600.9691*y^3 $\\$- 36.705*y^2*i - 2.646*y*i^2 \cdots- 3.3923*x $\\$+ 56310.8279*y - 0.0114*i + 7.2868$} & 2.58\%  & 11570.1\\
                \hline
                {\textsc{\makecell{inv-Pend \\ variant}}}& $[pA\le1]$ &   3  & 412.20 &  \makecell{$\min \{[cp>0.5 \vee cp < -0.5 \vee pA > 0.1 \vee $\\$  pA < -0.1] \cdot (0.0058pAD^2pA - 0.0011pAD^2cV $\\$- 0.1313pAD^2*cP+ \cdots + 0.0689*cP + 0.3238)$\\$ + [-0.5 \le cp \le 0.5 \wedge -0.1 \le pA \le 0.1], h^* \}$}  & 4 & 7.42 &\makecell{$0.2264*pAD^4 + 1.1448*pAD^3*pA $\\$- 0.1026*pAD^3*cV - 0.1107*pAD^3*cP +$\\$ 5.2869*pAD^2*pA^2 + \cdots + 10.6625*cP^2 $\\$- 0.0001*pA + 53.8573*cV + 1.0$} & 4.04\%  & 0.0\\
                \hline
                
                {\textsc{CAV-7}}& $[x\leq 30]$ &  3 &  5.26   & \makecell{$\min\{[i<5]\cdot (0.0001*i^2*x + 0.0005*i^2 - $\\$0.0006*i*x + 0.0004*i - 0.0011*x $\\$+ 0.9983) + [i<5 \wedge x \le 30], h^* \}$} &  4 & 1.17  & \makecell{$0.0007*i^4 - 0.0011*i^3*x + 0.0005*i^2*x^2$\\$ - 0.0001*i*x^3 - 0.0045*i^3 + 0.0052*i^2*x$\\$ - 0.0012*i*x^2 +0.0134*i^2 - 0.012*i*x + $\\$0.002*x^2 - 0.0135*i + 0.0046*x + 1.0034$} &  37.37\%  & 0.0049\\
                \hline
                {\textsc{cav-5}}& $[i \leq 10 ]$ &  3 &  892.6  & \makecell{$\min\{[money \ge 10]\cdot (-0.0001*i*money^2 -$\\$ 0.0004*i^2 - 0.0004*i*money + 0.0015*i $\\$+0.1028*money^2  - 0.2118*money + $\\$3.1283)+[money < 10 \wedge i \le 10], h^* \}$} &  4 & 1.27 & \makecell{$0.0001*i^2*money^2 + 0.0002*i*money^3 + $\\$ 0.0001*money^4  + 0.0184*i^2*money $\\$-0.0396*i*money^2 - 0.0168*money^3  $\\$+ 0.0009*i^3 -0.0291*i^2 + 2.8701*i $\\$+ 0.2414*i*money + 4.264*money^2  + 1.0$} &  0.0\%  & 507.20\\
                \hline

                {\textsc{Add}}& $[x > 5]$ &  3 &   3.63  & \makecell{$\min\{[y \le 1]\cdot (0.9403 - 0.1953*y + 0.0831*x$\\$ - 0.0736*y^2 + 0.0565*x*y - 0.0105*x^2 $\\$- 0.0281*y^3 + 0.023*x*y^2 - 0.0052*x^2*y$\\$ + 0.0005*x^3) + [y > 1 \wedge x > 5], h^* \}$} &   4 & 0.81  & \makecell{$0.9205 + 0.6262*y + 0.0819*x $\\$- 1.6833*y^2 + 0.0397*x*y - 0.0119*x^2 + $\\$0.9726*y^3 + 0.0598*x*y^2 - 0.0059*x^2*y $\\$+ 0.0006*x^3 - 0.2007*y^4 - 0.0104*x*y^3 $\\$- 0.0009*x^2*y^2 + 0.0001*x^3*y$} &  3.84 \%  & 0.253\\
                \hline
                {\textsc{\makecell{Growing\\Walk\\variant2}}} & $y$ & 2 & 5.33 & \makecell{$\min\{[r\le 0]\cdot(0.0622*x^2 + 0.6279*x$\\$ + y + 1.6914)+[r > 0]\cdot y, h^*\}$} & 3 &  1.22 &\makecell{$0.999*x*r^2 + 0.0008*y*r^2 + 700.3292*r^3 $\\$- 1.999*x*r - 0.0008*y*r - 1399.6591*r^2 $\\$+x + y + 698.3298*r + 1.0001$}  & 5.0 \%  & 211.91\\ 
                \hline
            \end{tabular}
        \end{threeparttable}}
\end{table*}

{\begin{figure}[tbp]
  \centering
  \subfloat[\textsc{Add}]
  {\includegraphics[width=0.2\textwidth]{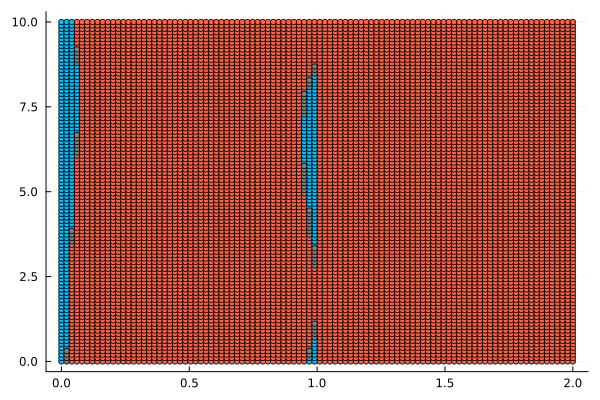}} \hfill   
  \subfloat[\textsc{cav-5}]
  {\includegraphics[width=0.2\textwidth]{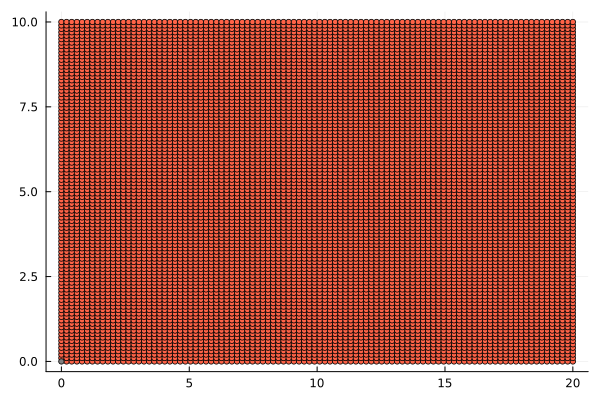}} \hspace{0.1cm} 
  \subfloat[\textsc{fig-7}]
  {\includegraphics[width=0.31\textwidth]{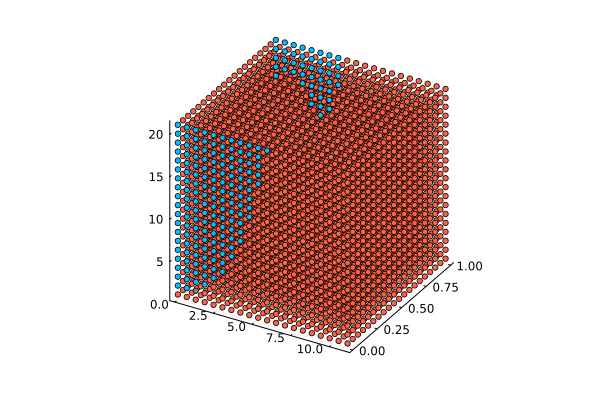}} \hspace{-0.8cm} 
  \subfloat[\textsc{Growing Walk variant2}]
  {\includegraphics[width=0.31\textwidth]{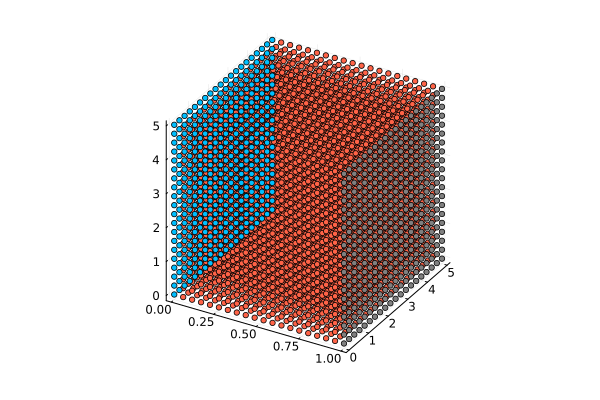}} \hspace{-0.5cm} 

  \caption{Difference Plots of the Comparison in Piecewise Polynomial Case}
  \label{fig:diff_plots_upper2}
  
  \begin{tablenotes}
      \item \textcolor{red}{Red} points indicate where our piecewise upper bounds are smaller than the monolithic ones \\ by more than $10^{-3}$; 
      \textcolor{blue}{blue} points indicate where the monolithic bounds are smaller by more \\ than $10^{-3}$;
      \textcolor{gray}{gray} points denote cases with negligible differences ($\le 10^{-3}$).
  \end{tablenotes}
  \vspace{-2ex}
\end{figure}

}

\section{Related Works \& Conclusion}
In this work, we propose a novel approach to synthesize piecewise probabilistic bounds for probabilistic programs. Further improvements include optimization on the branch reduction and the constraint solving of latticed $k$-induction constraints with minimum. Below we compare our approach with most related approaches.

Compared with previous approaches (e.g.~\cite{DBLP:conf/cav/ChakarovS13,DBLP:conf/cav/ChatterjeeFG16,DBLP:journals/toplas/ChatterjeeFNH18}) that mostly focus on synthesizing monolithic bounds over probabilistic programs, our approach targets piecewise bounds, and hence is orthogonal. 
The work~\cite{DBLP:conf/cav/BatzCKKMS20} proposes latticed $k$-induction. We claim that their work differs significantly from ours. They do not synthesize bounds and only verify whether a given bound is an upper bound or not.
The work~\cite{DBLP:conf/tacas/BatzCJKKM23} synthesize piecewise linear bounds to verify the input upper bound via counterexample-guided inductive synthesis (CEGIS), while we do not need this additional input bound and we solve the bounds by bilinear and semidefinite programming rather than CEGIS. 
For the verification of lower bounds, their work applies a proof rule in~\cite{DBLP:journals/pacmpl/HarkKGK20,DBLP:conf/vmcai/FuC19} derived from the original OST, while our approach applies extended OST. The work \cite{DBLP:conf/cav/BaoTPHR22} synthesizes (piecewise) exact invariants and sub-invariants (to verify the input lower bound) via data-driven learning. Their work additionally requires a list of features composed of numerical expressions, while our approach captures the piecewise feature via $k$-induction automatically and without such additional inputs. 
The works~\cite{Beutner2022b,10.1145/3656432,DBLP:journals/pacmpl/ZaiserMO25} focus on deriving bounds for the posterior distribution in Bayesian probabilistic programs, whereas our work aims at deriving piecewise bounds for the expected output of the probabilistic programs.

Other approaches~\cite{DBLP:conf/atva/BartocciKS19,DBLP:conf/tacas/BartocciKS20,DBLP:conf/sas/AmrollahiBKKMS22,DBLP:journals/corr/abs-2306-07072} focus on moment-based invariants generation and high-order moments derivation for probabilistic programs. These works can even handle the probabilistic program with non-polynomial expressions and continuous distributions, but they only consider the probabilistic while loop in a rather restricted form: $\textbf{while} ~ \texttt{true} ~ \{C\}$. The work~\cite{DBLP:journals/pacmpl/MoosbruggerSBK22} enlarges the theoretical foundation through the assumption
that all variables appearing in if-conditions (loop guards) are finitely valued, and \cite{DBLP:journals/pacmpl/MullnerMK24} further provides an algorithm about computing the strongest polynomial moment invariants for this kind of loops, but their works still cannot handle most of our benchmarks. Our approach can handle all the polynomial forms of loop guards and if-conditions. In a similar vein, the works~\cite{DBLP:conf/tacas/KuraUH19,DBLP:conf/pldi/Wang0R21} bound higher central moments for running time and other monotonically increasing quantities, but are limited to programs with constant size increments.

\clearpage

\section*{Data-Availability Statement}

The artifact of this paper is available at the link~\cite{Yang2025}. 

\begin{acks}
We thank the anonymous reviewers of POPL 2026 for their valuable comments and helpful suggestions. This work has been partially funded by the National Key R\&D Program of China under grant
No. 2022YFA1005100 and 2022YFA1005101, and the National Natural Science Foundation of China under grant No. 62192732, W2511064, 62172271, 62502475.
\end{acks}

\bibliographystyle{ACM-Reference-Format}
\bibliography{references}

\clearpage
\appendix
\crefalias{section}{appendix}
\crefalias{subsection}{appendix}
\section{Supplementary Material for Section~\ref{sec:operators}}~\label{app:operators}

In this section, we supplement the introduction of the variant of $k$-induction operators proposed in~\cite{DBLP:conf/apsec/LuX22}, some important properties of these two $k$-induction operators, the equivalence between them and all their proofs.

Recall that in~\cref{sec:operators}, we fix a lattice $(E,\sqsubseteq)$ and a monotone operator~$\varPhi: E \rightarrow E$.

\subsection{Property of the Upper $k$-Induction Operator in~\cite{DBLP:conf/cav/BatzCKKMS20}}\label{app:park}

We attach an important property of the upper $k$-induction operator $\varPsi_u$ in~\cite{DBLP:conf/cav/BatzCKKMS20} here.

\begin{theorem}[Park Induction from $k$-Induction~{\cite{DBLP:conf/cav/BatzCKKMS20}}]\label{lemma:park}
For any $u\in E$ and $k\in\mathbb{N}$, we have that  
$\varPhi(\varPsi_u^{k}(u)) \sqsubseteq u \iff \varPhi(\varPsi_u^{k}(u)) \sqsubseteq \varPsi_u^{k}(u)$~.
\end{theorem}

The proof is given in~\cite[Lemma 2]{DBLP:conf/cav/BatzCKKMS20}. 

\subsection{Upper $k$-Induction Operator in~\cite{DBLP:conf/apsec/LuX22}}~\label{app:k_operator}
First we recall the definition of the upper $k$-induction operator proposed in~\cite{DBLP:conf/apsec/LuX22}.

\begin{definition}[The $k$-Induction Operator in~\cite{DBLP:conf/apsec/LuX22}]\label{def:k-induction_operator2}
The upper $k$-induction operator $\Psi$ is defined by:
$\Psi: E \rightarrow E, v \mapsto \varPhi(v) \sqcap v$.
\end{definition}

Intuitively, it can be seen as a natural tightening of the operator $\varPsi_u$, which considers the meet with the input element $v$ itself. Below we introduce some important properties of the operator $\Psi$.

\begin{lemma}\label{property_itself}
    Let $\Psi$ be the $k$-induction operator in~\cite{DBLP:conf/apsec/LuX22} w.r.t. $\varPhi$. Then
    \begin{enumerate}
        \item
        \label{lemma:itself_monotonic} 
            $\Psi$ is monotonic, i.e., $\forall v_1, v_2 \in E, v_1 \sqsubseteq v_2$ implies $\Psi(v_1) \sqsubseteq \Psi(v_2)$.
        
        \item
        \label{lemma:itself_descending_chain}
            Iterations of $\Psi$ starting from $u$ are descending, i.e., 
        $$ \ldots \sqsubseteq \Psi^{k}(u) \sqsubseteq \Psi^{k-1}(u) \sqsubseteq \ldots \sqsubseteq \Psi(u) \sqsubseteq u$$
        And thus we have for all $m < n \in \mathbb{N}, \Psi^n(u) \sqsubseteq \Psi^m(u)$.
    \end{enumerate}
\end{lemma}

\begin{proof}
    For item (\ref{lemma:itself_monotonic}), observe that if we have $w_1 \sqsubseteq w_2$ and $v_1 \sqsubseteq v_2$, then we have $w_1 \sqcap v_1 \sqsubseteq w_2 \sqcap v_2 $.
    \begin{align*}
        \Psi(v_1) &= \varPhi(v_1) \sqcap v_1 \tag{by definition of $\Psi$}\\
                  &\sqsubseteq \varPhi(v_2) \sqcap v_2 \tag{by monotonicity of $\varPhi$ and above property} \\
                  &= \Psi(v_2) \tag{by definition of $\Psi$}
    \end{align*}
    
    For item (\ref{lemma:itself_descending_chain}), we can immediately derived from the definition of $\Psi$ as 
    \begin{align*}
        \Psi^k(u) &= \Psi(\Psi^{k-1}(u)) \tag{by definition of $\Psi^k(u)$} \\
                  &= \varPhi(\Psi^{k-1}(u)) \sqcap \Psi^{k-1}(u) \tag{by definition of $\Psi$} \\
                  &\sqsubseteq \Psi^{k-1}(u) \tag{by definition of $\sqcap$}
    \end{align*}  
\end{proof}

\begin{proposition}\label{prop:Psi}
For any $u\in E$, $\varPhi(\Psi^k(u)) \sqsubseteq u \iff \varPhi(\Psi^k(u)) \sqsubseteq \Psi^k(u)$.
\end{proposition} 

\begin{proof}
    The if-direction is trivial as $\Psi^k(u) \sqsubseteq u$ (by Lemma \ref{property_itself}(\ref{lemma:itself_descending_chain})). For the only-if direction:
    
    \begin{align*}
        \Psi^k(u) &\sqsupseteq \Psi^{k+1}(u) \tag{by Lemma \ref{property_itself}(\ref{lemma:itself_descending_chain})} \\ 
                  &= \varPhi(\Psi^k(u)) \sqcap \Psi^k(u) \tag{by definition of $\Psi$}\\
                  &= \varPhi(\Psi^k(u)) \sqcap \Psi(\Psi^{k-1}(u)) \tag{by definition of $\Psi^{k}(u)$} \\
                  &= \varPhi(\Psi^k(u)) \sqcap (\varPhi(\Psi^{k-1}(u)) \sqcap \Psi^{k-1}(u)) \tag{by definition of $\Psi$} \\
                  &= (\varPhi(\Psi^k(u)) \sqcap \varPhi(\Psi^{k-1}(u))) \sqcap \Psi^{k-1}(u)) \tag{by associative law} \\
                  &= \varPhi(\Psi^k(u)) \sqcap \Psi^{k-1}(u) \tag{by monotonicity of $\varPhi$ and Lemma \ref{property_itself}(\ref{lemma:itself_descending_chain})} \\
                  & \quad \vdots \\
                  &= (\varPhi(\Psi^k(u)) \sqcap \varPhi(u)) \sqcap u \tag{by unfolding $\Psi^k$ until $k = 1$}\\
                  &= \varPhi(\Psi^k(u)) \sqcap u \tag{by monotonicity of $\varPhi$ and Lemma \ref{property_itself}(\ref{lemma:itself_descending_chain})} \\
                  &= \varPhi(\Psi^k(u)) \tag{by the premise}
    \end{align*}  
\end{proof}

\subsection{Equivalence between $\varPsi_u$ and $\Psi$}

\begin{theorem}[Equivalence between $\varPsi_u$ and $\Psi$]~\label{thm:equ}
    For any element $u\in E$, the sequence $\{\varPsi_u^{k}(u)\}_{k\ge 0}$ of elements in $E$ coincides with the sequence $\{\Psi^k(u)\}_{k\ge 0}$. In other words, for any natural number $k\ge 0$, we have that $\varPsi_u^{k}(u)=\Psi^k(u)$. 
\end{theorem}

\begin{proof}
    Proof by mathematical induction. We denote $X_k = \varPsi_u^{k}(u)$ and $Y_k = \Psi^k(u)$. when $k = 0$, $X_0 = u = Y_0$. When $k = 1$, $X_1 = \varPhi(u)\sqcap u = Y_1$, by definition of two operators, respectively.

    Now we suppose that $X_k = Y_k$, i.e., $\varPsi_u^{k}(u) = \Psi ^ k(u)$, and we aim to prove that $\varPsi_u^{k+1}(u) = \Psi ^ {k+1}(u)$.
    \begin{align*}
        X_{k+1} &= \varPsi_u(\varPsi_u^k(u))  \tag{by definition of $\varPsi_u^{k+1}(u)$}\\
                &= \varPhi(\varPsi_u^k(u)\sqcap u  \tag{by definition of $\varPsi_u$}        
    \end{align*}
    \begin{align*}
        Y_{k+1} &= \Psi(\Psi^k(u)) \tag{by definition of $\Psi^{k+1}(u)$}\\
                &= \varPhi(\Psi^k(u)) \sqcap \Psi^k(u) \tag{by definition of $\Psi$}\\
                &= \varPhi(\Psi^k(u)) \sqcap \Psi(\Psi^{k-1}(u)) \tag{by definition of $\Psi^{k}(u)$} \\
                &= \varPhi(\Psi^k(u)) \sqcap (\varPhi(\Psi^{k-1}(u)) \sqcap \Psi^{k-1}(u)) \tag{by definition of $\Psi$} \\
                &= (\varPhi(\Psi^k(u)) \sqcap \varPhi(\Psi^{k-1}(u))) \sqcap \Psi^{k-1}(u)) \tag{by associative law} \\
                &= \varPhi(\Psi^k(u)) \sqcap \Psi^{k-1}(u) \tag{by monotonicity of $\varPhi$ and Lemma \ref{property_itself}(\ref{lemma:itself_descending_chain})} \\
                & \quad \vdots \\
                &= (\varPhi(\Psi^k(u)) \sqcap \varPhi(u)) \sqcap u \tag{by unfolding $\Psi^k$ until $k = 1$}\\
                &= \varPhi(\Psi^k(u)) \sqcap u \tag{by monotonicity of $\varPhi$ and Lemma \ref{property_itself}(\ref{lemma:itself_descending_chain})}
    \end{align*}
    Since we suppose that $\varPsi_u^{k}(u) = \Psi^k(u)$, we obtain that $\varPhi(\varPsi_u^k(u)\sqcap u = \varPhi(\Psi^k(u)) \sqcap u$, thus we have $\varPsi_u^{k+1}(u) = \Psi ^ {k+1}(u)$, i.e., $X_{k+1} = Y_{k+1}$. 
\end{proof}

\subsection{Supplementary Materials for the Dual $k$-Induction Operators $\varPsi'_u$ and $\Psi'$}~\label{app:proof_of_dual}

We first give the definition of the Dual $k$-Induction Operators $\Psi'$, which has been examined in~\cite{DBLP:conf/apsec/LuX22}.

\begin{definition}[Dual $k$-Induction Operator in~{\cite{DBLP:conf/apsec/LuX22}}]
The lower $k$-induction operator $\Psi'$ is given by: $\Psi': E \rightarrow E, v \mapsto \varPhi(v) \sqcup v$.
\end{definition}

\begin{lemma}~\label{pro:dual_operator}
   Fix a  lattice $(E, \sqsubseteq)$ and a monotone operator $\varPhi$. For any element $u \in E$, both of these two dual $k$-induction operators $\varPsi'_u$ and $\Psi'$ have the following properties:
    \begin{enumerate}
        \item $\varPsi'_u$(resp. $\Psi'$) is monotone. \label{pro:dual}
        \item Iterations of $\varPsi'_u$ (resp. $\Psi'$) starting from $u$ are ascending, i.e., \label{pro:dual2}
        $$ u \sqsubseteq \varPsi'_u(u) \sqsubseteq \ldots (\varPsi'_u)^{k-1}(u) \sqsubseteq (\varPsi'_u)^{k}(u) \ldots$$
        $$ u \sqsubseteq \Psi'(u) \sqsubseteq \ldots (\Psi')^{k-1}(u) \sqsubseteq (\Psi')^{k}(u) \ldots$$
    \end{enumerate}
    Thus we have for all $m < n \in \mathbb{N}$, $(\varPsi'_u)^{m}(u) \sqsubseteq (\varPsi'_u)^{n}(u)$ and $(\Psi')^{m}(u) \sqsubseteq (\Psi')^{n}(u)$.
 
\end{lemma}

\begin{proof}
    We only prove the case of dual $k$-induction operator $\varPsi'_u$, since the proof of the properties of the dual $k$-induction operator $\Psi'$ is similar with that of $\varPsi'_u$. 

    For item~(\ref{pro:dual}), observe that if we have $w_1 \sqsubseteq w_2$, then we have $w_1 \sqcup u \sqsubseteq w_2 \sqcup u$. Assume that $v_1 \sqsubseteq v_2$
    \begin{align*}
        \varPsi'_u(v_1) &= \varPhi(v_1) \sqcup u \tag{by definition of $\varPsi'_h $}\\
                  &\sqsubseteq \varPhi(v_2) \sqcup u \tag{by monotonicity of $\varPhi$ and above property} \\
                  &= \varPsi'_u(v_2) \tag{by definition of $\varPsi'_h$}
    \end{align*}

    For item~(\ref{pro:dual2}), we prove it by mathematical induction. We have $u \sqsubseteq \varPsi'_u(u)$ as $\varPsi'_u(u) = \varPhi(u) \sqcup u$. We then assume that $(\varPsi'_u)^{k}(u) \sqsupseteq (\varPsi'_h)^{k-1}(u)$, and we prove that 
    \begin{align*}
        (\varPsi'_u)^{k+1}(u) &= \varPsi'_u((\varPsi'_u)^{k}(u)) \tag{by definition of $(\varPsi'_u)^{k+1}(u)$} \\
                  &\sqsupseteq \varPsi'_u((\varPsi'_u)^{k-1}(u)) \tag{by monotonicity of $\varPsi'_u$ and assumption} \\
                  &=(\varPsi'_u)^{k}(u) \tag{by definition of $(\varPsi'_u)^{k}(u)$}
    \end{align*}
    Thus the value sequence is an ascending chain. 
\end{proof}

\begin{proposition}~\label{prop_dual}
For any element $u \in E$, the lower $k$-induction operators $\varPsi'_u$ and $\Psi'$ have the following properties:
{\small
    \begin{equation*}
        \begin{aligned}
            \varPhi((\varPsi'_u)^{k}(u)) \sqsupseteq u &\iff \varPhi((\varPsi'_u)^{k}(u)) \sqsupseteq (\varPsi'_u)^{k}(u) \\
            \varPhi((\Psi')^k(u)) \sqsupseteq u &\iff \varPhi((\Psi')^{k}(u)) \sqsupseteq (\Psi')^{k}(u)
        \end{aligned}
    \end{equation*}
}

\end{proposition}

\begin{proof}
    For the case of the dual $k$-induction operator $\varPsi'_u$:

    The if-direction is trivial as $(\varPsi'_u)^k(u) \sqsupseteq u$ (by Lemma \ref{pro:dual_operator}(\ref{pro:dual2})). For the only-if direction:
    \begin{align*}
        (\varPsi'_u)^k(u) &\sqsubseteq (\varPsi'_u)^{k+1}(u) \tag{by Lemma~\ref{pro:dual_operator}(\ref{pro:dual2}))} \\
        &= \varPsi'_u((\varPsi'_u)^k(u)) \tag{by the definition of $(\varPsi'_u)^{k+1}(u)$}\\ 
        &= \varPhi((\varPsi'_u)^k(u)) \sqcup u \tag{by the definition of $\varPsi'_u$}\\
        &= \varPsi((\varPsi'_u)^k(u)) \tag{by the premise} 
    \end{align*}

    For the case of the dual $k$-induction operator $\Psi'$:

    The if-direction is trivial as $(\Psi')^k(u) \sqsupseteq u$ (by Lemma \ref{pro:dual_operator}(\ref{pro:dual2})). For the only-if direction:
    \begin{align*}
        (\Psi')^k(u) &\sqsubseteq (\Psi')^{k+1}(u) \tag{by Lemma~\ref{pro:dual_operator}(\ref{pro:dual2}))} \\
        &= \Psi'((\Psi')^k(u)) \tag{by the definition of $(\Psi')^{k+1}(u)$} \\
        &= \varPhi((\Psi')^k(u)) \sqcup (\Psi')^k(u) \tag{by the definition of $\Psi'$} \\
        &= \varPhi((\Psi')^k(u)) \sqcup \Psi'((\Psi')^{k-1}(u)) \tag{by the definition of $(\Psi')^k(u)$} \\
        &= \varPhi((\Psi')^k(u)) \sqcup \varPhi((\Psi')^{k-1}(u)) \sqcup (\Psi')^{k-1}(u) \tag{by the definition of $\Psi'$} \\
        &= (\varPhi((\Psi')^k(u)) \sqcup \varPhi((\Psi')^{k-1}(u))) \sqcup (\Psi')^{k-1}(u) \tag{by associate law} \\
        &= \varPhi((\Psi')^k(u)) \sqcup (\Psi')^{k-1}(u) \tag{by monotonicity of $\varPhi$ and Lemma~\ref{pro:dual_operator}(\ref{pro:dual2}))} \\
        & \quad \vdots \\
        &= \varPhi((\Psi')^k(u)) \sqcup \Psi'(u) \tag{by unfolding $(\Psi')^k(u)$ until $k=1$} \\
        &= \varPhi((\Psi')^k(u)) \sqcup \varPhi(u) \sqcup u \tag{by definition of $\Psi'$} \\
        &= \varPhi((\Psi')^k(u)) \sqcup u \tag{by monotonicity of $\varPhi$ and Lemma~\ref{pro:dual_operator}(\ref{pro:dual2}))} \\
        &= \varPhi((\Psi')^k(u)) \tag{by the premise}
    \end{align*} 
    
\end{proof}

\subsection{Equivalence between $\varPsi'_u$ and $\Psi'$}

\begin{theorem}[Equivalence between $\varPsi'_u$ and $\Psi'$]~\label{thm:equ_dual}
    For any element $u\in E$, we have that the sequence $\{(\varPsi'_u)^{k}(u)\}_{k\ge 0}$ of elements in $E$ coincides with the sequence $\{(\Psi')^k(u)\}_{k\ge 0}$. In other words, for any natural number $k\ge 0$, we have that $(\varPsi'_u)^{k}(u)=(\Psi')^k(u)$. 
\end{theorem}

\begin{proof}
    Analogously, we proof it by mathematical induction. $X_k = (\varPsi'_u)^{k}(u)$ and $Y_k = (\Psi')^k(u)$. when $k = 0$, $X_0 = u = Y_0$. When $k = 1$, $X_1 = \varPhi(u)\sqcup u = Y_1$, by definition of two dual operators, respectively.

    Now we suppose that $X_k = Y_k$, i.e., $(\varPsi'_u)^{k}(u) = (\Psi') ^ k(u)$, and we aim to prove that $(\varPsi'_u)^{k+1}(u) = (\Psi') ^ {k+1}(u)$.
    \begin{align*}
        X_{k+1} &= \varPsi'_u((\varPsi'_u)^k(u))  \tag{by definition of $(\varPsi'_u)^{k+1}(u)$}\\
                &= \varPhi((\varPsi'_u)^k(u)) \sqcup u  \tag{by definition of $\varPsi'_u$}        
    \end{align*}
    \begin{align*}
        Y_{k+1} &= \Psi'((\Psi')^k(u)) \tag{by definition of $(\Psi')^{k+1}(u)$}\\
                &= \varPhi((\Psi')^k(u)) \sqcup (\Psi')^k(u) \tag{by definition of $\Psi'$} \\
                &= \varPhi((\Psi')^k(u)) \sqcup \Psi'((\Psi')^{k-1}(u)) \tag{by definition of $(\Psi')^{k}(u)$} \\
                &= \varPhi((\Psi')^k(u)) \sqcup (\varPhi((\Psi')^{k-1}(u)) \sqcup \Psi'^{k-1}(u)) \tag{by definition of $\Psi'$} \\
                &= (\varPhi((\Psi')^k(u)) \sqcup \varPhi((\Psi')^{k-1}(u))) \sqcup \Psi'^{k-1}(u)) \tag{by associative law} \\
                &= \varPhi((\Psi')^k(u)) \sqcup (\Psi')^{k-1}(u) \tag{by monotonicity of $\varPhi$ and Lemma~\ref{pro:dual_operator}(\ref{pro:dual2}))} \\
                & \quad \vdots \\
                &= (\varPhi((\Psi')^k(u)) \sqcup \varPhi(u)) \sqcup u \tag{by unfolding $(\Psi')^k$ until $k = 1$}\\
                &= \varPhi((\Psi')^k(u)) \sqcup u \tag{by monotonicity of $\varPhi$ and Lemma~\ref{pro:dual_operator}(\ref{pro:dual2})}
    \end{align*}
    Since we suppose that $(\varPsi'_u)^{k}(u) = (\Psi')^k(u)$, we obtain that $\varPhi((\varPsi'_u)^{k}(u)\sqcup u = \varPhi((\Psi')^k(u)) \sqcup u$, thus we have $(\varPsi'_u)^{k+1}(u) = (\Psi') ^ {k+1}(u)$, i.e., $X_{k+1} = Y_{k+1}$. 
\end{proof}

\section{Supplementary Material for Section~\ref{sec:functions}}~\label{app:functions}

\subsection{Classical OST}~\label{app:classical_OST}
Optional Stopping Theorem (OST) is a classical theorem in martingale theory that characterizes the relationship between the expected values initially and at a stopping time in a supermartingale. Below we present the classical form of OST.
\begin{theorem}[Optional Stopping Theorem (OST)~{\cite[Chapter 10]{DBLP:books/daglib/0073491}}]
Let $\{X_n\}_{n = 0}^{\infty}$ be a martingale (resp. supermartingale) adapted to a filtration $\mathcal{F}=\{\mathcal{F}_n\}_{n=0}^{\infty}$ and $\tau$ be a stopping time w.r.t the filtration $\mathcal{F}$. If we have that:
\begin{itemize}
    \item $\mathbb{E}(\tau)<\infty$;
    \item exists an $M \in [0, \infty)$ such that $\lvert X_{n+1} - X_n \rvert \leq M$ holds almost surely for every $n \geq 0$, 
\end{itemize}
then it follows that  $\mathbb(\lvert X_\tau \rvert ) < \infty$ and $\mathbb{E}(X_\tau) =  \mathbb{E}(X_0)$(resp. $\mathbb{E}(X_\tau) \leq \mathbb{E}(X_0)$). 
\end{theorem}

Since the classical Optional Stopping Theorem~\cite{DBLP:books/daglib/0073491,doi:10.1080/00029890.1971.11992788} requires bounded step-wise difference $\lvert X_{n+1} - X_n \rvert$ in a stochastic process $\{X_n\}_{n\ge 0}$, which cannot handle our problem due to the assignment commands in the loop body. To address this difficulty, We have sought several extended versions of OST, as proposed in~\cite{DBLP:conf/pldi/Wang0GCQS19,DBLP:journals/corr/abs-2103-16105,10.1145/3656432}, etc. Among which we find the OST variant proposed in~\cite{10.1145/3656432} can handle our problem. 

\subsection{Proof of~\cref{thm:soundness}}\label{app:soundness}
\noindent\textbf{Theorem}~\ref{thm:soundness}. 
Suppose the loop $P$ is affine. Let $k$ be a positive integer and $h$ be a polynomial potential polynomial in the program variables with degree $d$. If there exist real numbers $c_1 >0$ and $c_2 > c_3 > 0$ such that
\begin{itemize}
    \item [(P1)] there exists a uniform  amplifier $c$ satisfying $c \leq e^{c_3/d}$, and 
    \item [(P2)] the termination time $T$ of $P$ has the \emph{concentration property}, i.e., $\mathbb{P}(T > n) \leq c_1 \cdot e^{-c_2 \cdot n}$ holds for sufficiently large $n \in \mathbb{N}$.
\end{itemize}
hold, then for any initial program state $s^*$, we have:
\begin{itemize}
\item $\mathbb{E}_{s^*}(X_f)\le \overline{\varPsi}_h^{k-1}(h)(s^*) \le h(s^*)$ holds for any $k$-upper potential polynomial $h$.

\item $\mathbb{E}_{s^*}(X_f)\ge (\overline{\varPsi}'_h)^{k-1}(h)(s^*) \ge h(s^*)$ holds for any $k$-lower potential polynomial $h$. 
\end{itemize}

\begin{proof}
    We first proof the soundness of upper potential functions. Let $s_n$ be the random vector (random variable) of the program state at the $n$-th iteration of the probabilistic while loop $P$, where $s_0 = s^*$, and let $\{\mathcal{F}_n\}_{n\ge 0}$ be the filtration such that each $\mathcal{F}_n$ is the $\sigma$-algebra that describes the first $n$ iterations of the loop, i.e., the smallest $\sigma$-algebra that makes the random values during the first $n$ executions measurable. This choice of $\mathcal{F}_n$ is standard in previous martingale-based results~\cite{DBLP:conf/popl/ChatterjeeFNH16,DBLP:journals/toplas/ChatterjeeFNH18,DBLP:conf/popl/ChatterjeeNZ17,DBLP:conf/pldi/Wang0GCQS19}.
    
    We also define $H = \overline{\varPsi}_h^{k-1}(h)$. Note that $H$ is piecewise linear or polynomial (by the definition of $\overline{\varPsi}_h$ in~\cref{def:k_induc_func}) . By~\cref{def:potential function} and the property that $\overline{\varPhi}(\overline{\varPsi}_h^{k-1}(h)) \preceq h \iff \overline{\varPhi}(\overline{\varPsi}_h^{k-1}(h)) \preceq \overline{\varPsi}_h^{k-1}(h)$ (\cref{lemma:park}), we obtain that $\forall s \in \mbox{\sl Reach}(s^*)$, $\overline{\varPhi}(H)(s) \le H(s)$. We define the stochastic process $\{X_n\}_{n=0}^{\infty}$ by 
    \[
        X_n := H(s_n). 
    \]
    We first prove that the stochastic process $\{X_n\}$ is a supermartingale. We discuss this in the following two scenarios:
    \begin{itemize}
        \item if $s_n \not\models \varphi$, by the semantics of probabilistic while loop (see~\cref{sec:programs}), $s_{n+1} = s_n$, and thus $X_{n+1} = X_n$, which satisfies the conditions of supermartingale; 
        \item if $s_n \models \varphi$, 
        we have 
            \begin{align*}
                \mathbb{E}_{s^*}[X_{n+1}|\mathcal{F}_n] &= \mathbb{E}_{s^*} [H(s_{n+1})|\mathcal{F}_n] \\
                &= \mathbb{E}_{s_n}  [H(s_{n+1})|\mathcal{F}_n] \tag{by definition of conditional expectation}\\
                &= pre_{C}(H)(s_n) \tag{by definition of pre-expectation}\\
                &= \overline{\varPhi}(H)(s_n) \tag{by definition of characteristic function}\\
                &\le H(s_n) \tag{by property of $H$} \\
                & = X_n
            \end{align*}
    where the property of conditional expectation is the “take out what is known” property of conditional expectation (see~\cite{DBLP:books/daglib/0073491}).
    From (P1) and the definition of uniform amplifier (see~\cref{def:uniform_amplifier}), for each program variable $x$, the value of $x_n$ is bounded by $|X_n| \leq c^n \cdot |x_0| + a\cdot (c^0 + \cdots c^{n-1}) \leq K_n \cdot c^n \leq K_n \cdot e^{c_3*n/d} $ for some positive constant $K_n$, where $d$ is the degree of the polynomial $H$ (i.e., the degree of $h$). From that $H$ is piecewise linear (resp. polynomial with degree $d$), i.e., $H$ is linear (resp. polynomial with degree $d$) on each segment, we can obtain $\mathbb{E}_{s^*}[X_n] = \mathbb{E}_{s^*}[H(s_n)]  = \mathbb{E}_{s^*}[ M_n \cdot c^n ]< \infty$ for some positive constant $M_n >0$ by the definition of $X_n$. Thus $\{X_n\}$ is a supermartingale. 
    \end{itemize}
    
    The condition (a) in~\cref{thm:ost-variant} follows from the assumption that (P2) $P$ has the concentration property.

    Then we prove the condition (b) in~\cref{thm:ost-variant}. 
    From (P1), we have that for each program variable $x$, the value of $x_n$ at $n$-th iteration, i.e., at the program state $s_n$, is bounded by $ K_n \cdot c^n $.
    When $H$ is piecewise linear, i.e., $d=1$, we have that $H(s_n) \leq M_n \cdot c^n $ for $M_n > 0$.
        \begin{align*}
            \lvert X_{n+1} - X_n \rvert &= \lvert H(s_{n+1}) - H(s_n) \rvert \\
            &\le \lvert H(s_{n+1}) \rvert + \lvert H(s_n) \rvert \\
            &\le M_n \cdot |c|^n + M_{n+1} \cdot |c|^{n+1} \\
            &\le (M_n + |c| \cdot M_{n+1})\cdot |c|^n \\
            &\le b_1 \cdot e^{c_3 n} 
        \end{align*}
    When $H$ is piecewise polynomial with degree $d$, we have that $H(s_n) \leq M_n \cdot c^{nd} $ for $M_n > 0$.
        \begin{align*}
            \lvert X_{n+1} - X_n \rvert &= \lvert H(s_{n+1}) - H(s_n) \rvert \\
            &\le \lvert H(s_{n+1}) \rvert + \lvert H(s_n) \rvert \\
            &\le M_n \cdot |c|^{nd} + M_{n+1} \cdot |c|^{(n+1)d} \\
            &\le (M_n + |c^d| \cdot M_{n+1})\cdot |c|^{nd} \\
            &\le b_1 \cdot (e^{c_3/d})^{nd} \\
            &\le b_1 \cdot e^{c_3 n}
        \end{align*}
    Especially, if the uniform amplifier $c$ is chosen as $1$, then $c_3$ can be chosen arbitrarily small, the prerequisites of this theorem always holds regardless of the values taken by $c_2$ and $d$.
    
    By applying Theorem~\ref{thm:ost-variant}, we have that $\mathbb{E}_{s^*}(X_T) \leq \mathbb{E}_{s^*}(X_0)$. Since the termination time $T$ is a stopping time w.r.t. the filtration $\{\mathcal{F}_n\}_{n\ge 0}$, and there will be $s_{T} \not\models \varphi $, thus $X_{T} = f(s_T) = X_f$. We have $\mathbb{E}_{s^*}(X_f)\le \mathbb{E}_{s^*}(X_0) = H(s^*)$. The second inequality, i.e., $\overline{\varPsi}_h^{k-1}(h)(s^*) \le h(s^*) (\forall s^*)$ can be derived directly from the property that $\overline{\varPsi}_h^{k-1}(h) \preceq h$ holds (see~\cref{app:k_operator} and~\cite{DBLP:conf/cav/BatzCKKMS20}).
    The case of lower potential functions is completely dual to the case of upper potential functions since we can consider the stochastic process $\{-X_n\}$, that is, define the stochastic process by $Y_n := -H(s_n)$.
    The remaining proof is essentially the same.

\end{proof}

\subsection{Proof of~\cref{thm:soundness_poly}}\label{app:soundness_poly}
\noindent\textbf{Theorem}~\ref{thm:soundness_poly}. Let $k$ be a positive integer. Suppose there exist real numbers $c_1 >0$ and $c_2 > 0$ such that condition
(P1') loop $P$ has the bounded update property; and condition (P2) in~\cref{thm:soundness} holds, then for any initial program state $s^*$, we have
\begin{itemize}
\item $\mathbb{E}_{s^*}(X_f)\le \overline{\varPsi}_h^{k-1}(h)(s^*) \le h(s^*)$ holds for any $k$-upper potential polynomial $h$.

\item $\mathbb{E}_{s^*}(X_f)\ge (\overline{\varPsi}'_h)^{k-1}(h)(s^*) \ge h(s^*)$ holds for any $k$-lower potential polynomial $h$. 
\end{itemize}

\begin{proof}
    We first proof the soundness of upper potential functions. Let $s_n$ be the random vector (random variable) of the program state at the $n$-th iteration of the probabilistic while loop $P$, where $s_0 = s^*$, and let $\{\mathcal{F}_n\}_{n\ge 0}$ be the filtration such that each $\mathcal{F}_n$ is the $\sigma$-algebra that describes the first $n$ iterations of the loop, i.e., the smallest $\sigma$-algebra that makes the random values during the first $n$ executions measurable. This choice of $\mathcal{F}_n$ is standard in previous martingale-based results~\cite{DBLP:conf/popl/ChatterjeeFNH16,DBLP:journals/toplas/ChatterjeeFNH18,DBLP:conf/popl/ChatterjeeNZ17,DBLP:conf/pldi/Wang0GCQS19}.
    
    We suppose that $h$ is a $k$-upper potential polynomial with degree $d$ and define $H = \overline{\varPsi}_h^{k-1}(h)$. Note that $H$ is piecewise polynomial with degree $d$ (by the definition of $\overline{\varPsi}_h$ in~\cref{def:k_induc_func}) . By~\cref{def:potential function} and the property that $\overline{\varPhi}(\overline{\varPsi}_h^{k-1}(h)) \preceq h \iff \overline{\varPhi}(\overline{\varPsi}_h^{k-1}(h)) \preceq \overline{\varPsi}_h^{k-1}(h)$ (\cref{lemma:park}), we obtain that $\forall s \in \mbox{\sl Reach}(s^*)$, $\overline{\varPhi}(H)(s) \le H(s)$. We define the stochastic process $\{X_n\}_{n=0}^{\infty}$ by 
    \[
        X_n := H(s_n). 
    \]
    We first prove that the stochastic process $\{X_n\}$ is a supermartingale. We discuss this in the following two scenarios:
    \begin{itemize}
        \item if $s_n \not\models \varphi$, by the semantics of probabilistic while loop (see~\cref{sec:programs}), $s_{n+1} = s_n$, and thus $X_{n+1} = X_n$, which satisfies the conditions of supermartingale; 
        \item if $s_n \models \varphi$, we have 
            \begin{align*}
                \mathbb{E}_{s^*}[X_{n+1}|\mathcal{F}_n] &= \mathbb{E}_{s^*} [H(s_{n+1})|\mathcal{F}_n] \\
                &= \mathbb{E}_{s_n}  [H(s_{n+1})|\mathcal{F}_n] \tag{by definition of conditional expectation}\\
                &= pre_{C}(H)(s_n) \tag{by definition of pre-expectation}\\
                &= \overline{\varPhi}(H)(s_n) \tag{by definition of characteristic function}\\
                &\le H(s_n) \tag{by property of $H$} \\
                & = X_n
            \end{align*}

    \end{itemize}
    where the property of conditional expectation is the “take out what is known” property of conditional expectation (see~\cite{DBLP:books/daglib/0073491}).
    From (P1') that $P$ has the bounded update property and $H$ is a piecewise polynomial with degree $d$, i.e., $H$ is a polynomial with degree $d$ on each segment, we can obtain $\mathbb{E}_{s^*}[X_{n}] = \mathbb{E}_{s^*}[H(s_n)] \le \zeta \cdot n^d$ for a positive constant $\zeta>0$, thus $\{ X_n \}$ is a supermartingale.
    
    The condition (a) in~\cref{thm:ost-variant} follows from the assumption that (P2) $P$ has the concentration property.

    Then we prove the condition (b) in~\cref{thm:ost-variant}. From that $P$ has the bounded update property and $H$ is a piecewise polynomial with degree $d$, we also have that $|X_{n}| \le \zeta \cdot n^d$ for a positive constant $\zeta>0$, thus we have
    \begin{align*}
            \lvert X_{n+1} - X_n \rvert 
            &\le |X_{n+1}| + |X_n| \\
            &\le \zeta \cdot n^d + \zeta \cdot (n+1)^d \\
            &\le b_1 \cdot n^d
        \end{align*}
    Note that in this theorem, $c_3$ in ~\cref{thm:ost-variant}(b) is chosen arbitrarily small, therefore the prerequisites of~\cref{thm:ost-variant} always holds regardless of the values taken by $c_2$.

    By applying Theorem~\ref{thm:ost-variant}, we have that $\mathbb{E}_{s^*}(X_T) \leq \mathbb{E}_{s^*}(X_0)$. Since the termination time $T$ is a stopping time w.r.t. the filtration $\{\mathcal{F}_n\}_{n\ge 0}$, and there will be $s_{T} \not\models \varphi $, thus $X_{T} = f(s_T) = X_f$. We have $\mathbb{E}_{s^*}(X_f)\le \mathbb{E}_{s^*}(X_0) = H(s^*)$. The second inequality, i.e., $\overline{\varPsi}_h^{k-1}(h)(s^*) \le h(s^*) (\forall s^*)$, can be derived directly from the property that $\overline{\varPsi}_h^{k-1}(h) \preceq h$ holds (see~\cref{app:k_operator} and~\cite{DBLP:conf/cav/BatzCKKMS20}).
    The case of lower potential functions is completely dual to the case of upper potential functions since we can consider the stochastic process $\{-X_n\}$, that is, define the stochastic process by $Y_n := -H(s_n)$.
    The remaining proof is essentially the same.

\end{proof}

\section{Supplementary Material for Section~\ref{sec:algorithm}}~\label{app:algorithm}

\subsection{Supplementary Material for Brute-Force Arithmetic Expansion in {\bf Stage 2}}\label{app:general_approach}
In this section, we supplement the brute-force arithmetic expansion that can simplify the $k$-induction constraint. To transform the $k$-induction constraint $\overline{\varPhi}_f(\overline{\varPsi}_h^{k-1}(h)) \preceq h$ into a simpler form, our algorithm further unrolls this $k$-induction conditions so that the minimum operations appear at the outermost of the left-hand-side of the inequality. In detail, from the definition of the operator $\overline{\varPsi}_h$ (Definition~\ref{def:k_induc_func}), the unrolling is reduced to the recursive computation of {\em pre-expectation} and the pointwise minimum operation. Following the definition of pre-expectation (\cref{def:pre-exp}), the unrolling can be done by  the following reduction rules for functions $f_1,\dots, f_m$, $g_1,\dots, g_n$:
\begin{itemize}
    \item[(R1)] $\min\{f_1,\dots, f_m\} + \min\{g_1,\dots,g_n\}=\min_{1\le i\le m,1\le j\le n}\{f_i+g_j\}$; 
    \item[(R2)] $c\cdot \min\{f_1,\dots, f_m\}=\min\{c\cdot f_1,\dots, f_m\}$ for constant $c\ge 0$; 
    \item[(R3)] $[B]\cdot \min\{f_1,\dots, f_m\}=\min\{[B]\cdot f_1,\dots, [B]\cdot f_m\}$ for predicate $B$. 
\end{itemize}

By iterative applications of the reduction rules, the constraint $\overline{\varPhi}_f(\overline{\varPsi}_h^{k-1}(h)) \preceq h$ can be transformed into a succinct form with only one minimum operation:
\begin{align*}
    \min \{h_1, h_2,\dots, h_m \} \preceq h
\end{align*}
where $h$ is the predefined polynomial template and each $h_i~(i = 1, \dots, m)$ is a piecewise expression derived from the unrolling that does not contain the minimum operation. 

\subsection{Proof of~\cref{prop:relation}} \label{proof:relation}
We give a proof for~\cref{prop:relation} in this section.

\noindent \textbf{Proposition}~\ref{prop:relation}.
The upper $k$-induction condition  $\overline{\varPhi}_f(\overline{\varPsi}_h^{k-1}(h)) \preceq h$ is equivalent to  constraint $\min\{h_1, h_2,\dots, h_m\} \preceq h$, where each $h_i$ equals $pre_{C_{d}}(h)$ for some unique $C_d \in \{C_1,\dots, C_m\}$ from the unfolding process above. 

\begin{proof}
    We concentrate on the left side of the constraint: $\overline{\varPhi}_f(\overline{\varPsi}_h^{k-1}(h)) \preceq h$.
    
    We first proof the case of $k=2$, i.e., $\overline{\varPhi}_f(\overline{\varPsi}_h^{1}(h)) \preceq h$. Since our syntax of the probabilistic programs is defined in a compositional style (see \cref{fig: syntax} in \cref{sec:programs} for more details), we proof by induction on the structure of programs. For simplicity, we denote $pre_C([\varPhi])$ by $[\varPhi(C)]$, which represent the evaluation of $[\varPhi]$ after the execution of $C$. 
    \begin{itemize}
        \item Case $C \equiv \textsf{skip}$.
        \begin{eqnarray*}
            & &\overline{\varPhi}_f(\overline{\varPsi}_h(h)) \\
            &=& [\neg \varphi] \cdot f + [\varphi] \cdot pre_C(\overline{\varPsi}_h(h)) \\
            &=& [\neg \varphi] \cdot f + [\varphi] \cdot \overline{\varPsi}_h(h) \\
            &=& [\neg \varphi] \cdot f + [\varphi] \cdot \min \{\overline{\varPhi}_f(h), h\} \\
            &=& [\neg \varphi] \cdot f + [\varphi] \cdot \min \{[\neg \varphi] \cdot f + [\varphi] \cdot h, h\} \\
            &=& [\neg \varphi] \cdot f + \min \{[\varphi] \cdot h, [\varphi] \cdot h\} \\
            &=& [\neg \varphi] \cdot f + [\varphi] \cdot h \\
            &=& \overline{\varPhi}_f(h)
        \end{eqnarray*}
        It corresponds to pre-expectation of the loop-free program unfolded with twice (only one program).
        \item Case $C \equiv x := e$.
        \begin{eqnarray*}
            & &\overline{\varPhi}_f(\overline{\varPsi}_h(h)) \\
            &=& [\neg \varphi] \cdot f + [\varphi] \cdot pre_C(\overline{\varPsi}_h(h)) \\
            &=& [\neg \varphi] \cdot f + [\varphi] \cdot \overline{\varPsi}_h(h)([x/e]) \\
            &=& [\neg \varphi] \cdot f + [\varphi] \cdot \min \{[\neg \varphi] \cdot f + [\varphi] \cdot h([x/e]), h\}([x/e]) \\
            &=& [\neg \varphi] \cdot f + [\varphi] \cdot  \min \{[\neg \varphi([x/e])]\cdot f([x/e]) + \\
            & & [\varphi([x/e])] \cdot h([x/e])([x/e]), h([x/e])\} \\
            &=& \min\{[\neg \varphi] \cdot f + [\varphi \wedge \neg \varphi([x/e])] \cdot f([x/e]) + \\
            & & [\varphi \wedge \varphi([x/e])] \cdot h([x/e])([x/e]),
            [\neg \varphi] \cdot f + [\varphi] \cdot h([x/e]) \}\\
            &=& \min\{[\neg \varphi] \cdot f + [\varphi \wedge \neg \varphi([x/e])] \cdot f([x/e]) + \\
            & & [\varphi \wedge \varphi([x/e])] \cdot pre_{C;C}(h),
            [\neg \varphi] \cdot f + [\varphi] \cdot h([x/e]) \}
        \end{eqnarray*}
        the expressions in the minimize operator correspond to pre-expectation of the two loop-free programs unfolded within twice (one for once, and another for twice).
        \item Case $C \equiv C_1;C_2$.
        \begin{eqnarray*}
            & &\overline{\varPhi}_f(\overline{\varPsi}_h(h)) \\
            &=& [\neg \varphi] \cdot f + [\varphi] \cdot pre_C(\overline{\varPsi}_h(h)) \\
            &=& [\neg \varphi] \cdot f + [\varphi] pre_{C_1}(pre_{C_2}(\min \{[\neg \varphi] \cdot f + [\varphi] \cdot pre_{C_1}(pre_{C_2}(h)), h\})) \\
            &=& [\neg \varphi] \cdot f + [\varphi] \cdot \min \{[\neg \varphi(C_1;C_2)] \cdot pre_{C_1;C_2}(f) + \\
            & & [\varphi(C_1;C_2)]\cdot pre_{C_1;C_2}(h), pre_{C_1;C_2}(h)\} \\
            &=& \min \{[\neg \varphi] \cdot f + [\varphi \wedge \varphi(C_1;C_2)] \cdot pre_{C_1;C_2}(f) + \\
            & &[\varphi \wedge \neg \varphi(C_1;C_2)] \cdot pre_{C_1;C_2}(h), \\
            & &[\neg \varphi] \cdot f + [\varphi] \cdot pre_{C_1;C_2}(h) \} \\
        \end{eqnarray*}
        the expressions in the minimize operator correspond to pre-expectation of the two loop-free programs unfolded within twice (one for once, and another for twice)
        \item case $C \equiv \{C_1\}[p]\{C_2\}$.
        \begin{eqnarray*}
            & &\overline{\varPhi}_f(\overline{\varPsi}_h(h)) \\
            &=& [\neg \varphi] \cdot f + [\varphi] \cdot pre_C(\overline{\varPsi}_h(h)) \\
            &=& [\neg \varphi] \cdot f + [\varphi] \cdot p \cdot pre_{C_1}(\overline{\varPsi}_h(h)) + [\varphi] \cdot (1-p) \cdot pre_{C_2}(\overline{\varPsi}_h(h))
        \end{eqnarray*}
        wherein
        \begin{eqnarray*}
            pre_{C_1}(\overline{\varPsi}_h(h)) &=& pre_{C_1}(\min \{[\neg \varphi] \cdot f + [\varphi]\cdot (p\cdot pre_{C_1}(h) + (1-p)  \cdot pre_{C_2}(h)), h\} \\ 
            &=& \min \{[\neg \varphi(C_1)] \cdot pre_{C_1}(f) + [\varphi(C_1)]\cdot \\
            & &(p\cdot pre_{C_1;C_1}(h) + (1-p) \cdot pre_{C_1;C_2}(h)), pre_{C_1}(h)\}
        \end{eqnarray*}
        and
        \begin{eqnarray*}
            pre_{C_2}(\overline{\varPsi}_h(h)) &=& pre_{C_2}(\min \{[\neg \varphi] \cdot f + [\varphi]\cdot (p\cdot pre_{C_1}(h) + (1-p)  \cdot pre_{C_2}(h)), h\} \\ 
            &=& \min \{[\neg \varphi(C_2)] \cdot pre_{C_2}(f) + [\varphi(C_2)]\cdot \\
            & &(p\cdot pre_{C_2;C_1}(h) + (1-p) \cdot pre_{C_2;C_2}(h)), pre_{C_2}(h)\}
        \end{eqnarray*}
        Thus we have
        \begin{eqnarray*}
            & &\overline{\varPhi}_f(\overline{\varPsi}_h(h)) \\
            &=& [\neg \varphi] \cdot f + [\varphi] \cdot p \cdot \min \{[\neg \varphi(C_1)] \cdot pre_{C_1}(f) \\
            & &+ [\varphi(C_1)]\cdot (p\cdot pre_{C_1;C_1}(h) + (1-p) \cdot pre_{C_1;C_2}(h)), pre_{C_1}(h)\} +\\
            & &[\varphi] \cdot (1-p) \cdot \min \{[\neg \varphi(C_2)] \cdot pre_{C_2}(f) \\
            & &+[\varphi(C_2)]\cdot (p\cdot pre_{C_2;C_1}(h) + (1-p) \cdot pre_{C_2;C_2}(h)), pre_{C_2}(h)\} \\
            &=& \min \{[\neg \varphi] \cdot f + [\varphi \wedge \neg \varphi(C_1)] \cdot p \cdot pre_{C_1}(f) + [\varphi \wedge \neg \varphi(C_2)] \cdot (1-p) \cdot pre_{C_2}(f) \\
            & &+ [\varphi \wedge \varphi(C_1)] \cdot (p^2 \cdot pre_{C_1;C_1}(h) + p(1-p) \cdot pre_{C_1;C_2}(h)) \\
            & &+ [\varphi \wedge \varphi(C_2)] \cdot ((1-p)p\cdot pre_{C_2;C_1}(h) + (1-p)^2 \cdot pre_{C_2;C_2}(h)), \\            
            & & [\neg \varphi] \cdot f + [\varphi \wedge \neg \varphi(C_1)] \cdot p \cdot pre_{C_1}(f) + \\
            & & [\varphi \wedge \varphi(C_1)] \cdot (p^2 \cdot pre_{C_1;C_1}(h) + p(1-p) \cdot pre_{C_1;C_2}(h)) +\\
            & & [\varphi]\cdot (1-p) \cdot pre_{C_2}(h), \\            
            & & [\neg \varphi] \cdot f + [\varphi \wedge \neg \varphi(C_2)] \cdot (1-p) \cdot pre_{C_2}(f) + \\
            & & [\varphi \wedge \varphi(C_2)] \cdot ((1-p)p \cdot pre_{C_2;C_1}(h) + (1-p)^2 \cdot pre_{C_2;C_2}(h)) +\\
            & & [\varphi]\cdot p \cdot pre_{C_1}(h), \\            
            & & [\neg \varphi] \cdot f + [\varphi]\cdot p \cdot pre_{C_1}(h) + [\varphi]\cdot (1-p) \cdot pre_{C_2}(h) \\
        \end{eqnarray*}
        The first expression corresponds to the case that we unfold for twice at each state we reach (after the execution of $C_1$ and $C_2$), and the second (resp. third) expression corresponds to the case that we unfold for twice at the state that we reach after the execution of $C_1$ (resp. $C_2$) and unfold for once at the state that we reach after the execution of $C_2$ (resp. $C_1$). The fourth expression corresponds to the case that we unfold for once at both states, i.e., 1-induction principle.
        \item case $C \equiv \textsf{if} ~(\phi)~\{C_1\} ~\textsf{else} ~\{C_2\}$.
        \begin{eqnarray*}
            & &\overline{\varPhi}_f(\overline{\varPsi}_h(h)) \\
            &=& [\neg \varphi] \cdot f + [\varphi] \cdot pre_C(\overline{\varPsi}_h(h)) \\
            &=& [\neg \varphi] \cdot f + [\varphi \wedge \phi] \cdot pre_{C_1}(\overline{\varPsi}_h(h)) + [\varphi \wedge \neg \phi] \cdot pre_{C_2}(\overline{\varPsi}_h(h))
        \end{eqnarray*}
        wherein
        \begin{eqnarray*}
            pre_{C_1}(\overline{\varPsi}_h(h)) &=& pre_{C_1}(\min \{[\neg \varphi] \cdot f + [\varphi]\cdot ([\phi]\cdot pre_{C_1}(h) + [\neg \phi] \cdot pre_{C_2}(h)), h\} \\ 
            &=& \min \{[\neg \varphi(C_1)] \cdot pre_{C_1}(f) + [\varphi(C_1)]\cdot \\
            & &([\phi(C_1)]\cdot pre_{C_1;C_1}(h) + [\neg \phi(C_1)] \cdot pre_{C_1;C_2}(h)), pre_{C_1}(h)\}
        \end{eqnarray*}
        and
        \begin{eqnarray*}
            pre_{C_2}(\overline{\varPsi}_h(h)) &=& pre_{C_2}(\min \{[\neg \varphi] \cdot f + [\varphi]\cdot ([\phi]\cdot pre_{C_1}(h) + [\neg \phi]  \cdot pre_{C_2}(h)), h\} \\ 
            &=& \min \{[\neg \varphi(C_2)] \cdot pre_{C_2}(f) + [\varphi(C_2)]\cdot \\
            & &([\phi(C_2)]\cdot pre_{C_2;C_1}(h) + [\neg \phi(C_2)] \cdot pre_{C_2;C_2}(h)), pre_{C_2}(h)\}
        \end{eqnarray*}
        Thus we have
        \begin{eqnarray*}
            & & \overline{\varPhi}_f(\overline{\varPsi}_h(h)) \\
            &=& [\neg \varphi] \cdot f + [\varphi \wedge \phi]  \cdot \min \{[\neg \varphi(C_1)] \cdot pre_{C_1}(f) \\
            & &+ [\varphi(C_1)]\cdot ([\phi(C_1)]\cdot pre_{C_1;C_1}(h) + [\neg \phi(C_1)] \cdot pre_{C_1;C_2}(h)), pre_{C_1}(h)\} +\\
            & &[\varphi\wedge \neg \phi ] \cdot \min \{[\neg \varphi(C_2)] \cdot pre_{C_2}(f) \\
            & &+[\varphi(C_2)]\cdot ([\phi(C_2)]\cdot pre_{C_2;C_1}(h) + [\neg \phi(C_2)] \cdot pre_{C_2;C_2}(h)), pre_{C_2}(h)\} \\
            &=& \min \{[\neg \varphi] \cdot f + [\varphi \wedge \phi \wedge \neg \varphi(C_1)]  \cdot pre_{C_1}(f) + \\
            & &[\varphi \wedge \neg \phi \wedge \neg \varphi(C_2)]  \cdot pre_{C_2}(f)+ \\
            & &[\varphi \wedge \phi \wedge \varphi(C_1) \wedge \phi(C_1)]  \cdot pre_{C_1;C_1}(h) + [\varphi \wedge \phi \wedge \varphi(C_1) \wedge \neg \phi(C_1)] \cdot pre_{C_1;C_2}(h)) +\\
            & & [\varphi \wedge \neg \phi \wedge \varphi(C_2) \wedge \phi(C_2)] \cdot pre_{C_2;C_1}(h) + [\varphi \wedge \neg \phi \wedge \varphi(C_2) \wedge \neg \phi(C_2)] \cdot pre_{C_2;C_2}(h)), \\            
            & & [\neg \varphi] \cdot f + [\varphi \wedge \phi \wedge\neg \varphi(C_1)]  \cdot pre_{C_1}(f) + \\
            & & [\varphi \wedge \phi \wedge \varphi(C_1) \wedge \phi(C_1)]  \cdot pre_{C_1;C_1}(h) + [\varphi \wedge \phi \wedge \varphi(C_1) \wedge \neg \phi(C_1)] \cdot pre_{C_1;C_2}(h)) + \\
            & & [\varphi \wedge \neg \phi] \cdot pre_{C_2}(h), \\            
            & & [\neg \varphi] \cdot f + [\varphi \wedge \neg \phi \wedge \neg \varphi(C_2)] \cdot pre_{C_2}(f) + \\
            & & [\varphi \wedge \neg \phi \wedge \varphi(C_2) \wedge \phi(C_2)] \cdot pre_{C_2;C_1}(h) + [\varphi \wedge \neg \phi \wedge \varphi(C_2) \wedge \neg \phi(C_2)] \cdot pre_{C_2;C_2}(h)) + \\
            & & [\varphi \wedge \phi] \cdot pre_{C_1}(h), \\            
            & & [\neg \varphi] \cdot f + [\varphi \wedge \phi] \cdot pre_{C_1}(h) + [\varphi \wedge \neg \phi] \cdot pre_{C_2}(h) \\
        \end{eqnarray*}
        The one-to-one relation is the same as that in the former case (probabilistic choice case). 
    \end{itemize} 

    Then we proof the case of $k>2$ by mathematical induction.
    Suppose that the proposition holds when $k = n$, i.e., the upper $n$-induction condition $\overline{\varPhi}_f(\overline{\varPsi}_h^{n-1}(h)) \preceq h$ is equivalent with $\min\{h_1, h_2,\dots, h_m\} \preceq h$ , where each $h_i$ uniquely corresponds to one $C_d \in \{C_1,\dots, C_m\}$ and is equal to $pre_{C_{d}}(h)$, where $\{C_1,\dots, C_m\}$ are all the loop-free programs generated by following the decision process in \textbf{Stage 2} in~\cref{sec:algorithm} within $m$ unfolding.

    Then we proof the case of $n+1$.
    \begin{align*}
        \overline{\varPhi}_f(\overline{\varPsi}_h^{n}(h)) &=\overline{\varPhi}_f(\overline{\varPsi}_h(\overline{\varPsi}_h^{n-1}(h))) \\
        &=\overline{\varPhi}_f(\min\{\overline{\varPhi}_f(\overline{\varPsi}_h^{n-1}(h)), h\})\\
        &=\overline{\varPhi}_f(\min\{ \min\{h_1, h_2,\dots, h_m\}, h\}) \\
        &=\overline{\varPhi}_f(\min\{h_1, h_2,\dots, h_m, h\}) \\
        &= [\neg \varphi] \cdot f + [\varphi] \cdot pre_C(\min\{h_1, h_2,\dots, h_m, h\})
    \end{align*}
    Through the same inference on the structure $C$ as above, we show it is equivalent to $\min\{g_1, g_2,\dots, g_M\}$, where $M \geq m+1$ and each $g_i$ uniquely corresponds to one $C_d \in \{C_1,\dots, C_M\}$ and is equal to $pre_{C_{d}}(h)$, where $\{C_1,\dots, C_M\}$ are all the loop-free programs generated by following the decision process in \textbf{Stage 2} in~\cref{sec:algorithm} within $n+1$ unfolding. Thus the proposition holds when $k = n+1$. Notice that the operators $\overline{\varPhi}_f$ and pointwise $\min$ are noncommutative.
    
    By mathematical induction, the proposition holds for $k\geq 2$.
\end{proof}

\begin{remark}
    In~\cref{prop:relation}, We only propose the case of upper $k$-induction condition, and the case of lower $k$-induction condition is completely dual. \qed 
\end{remark}

\subsection{Supplementary Material for the Pedagogical Explanation in Stage 2}\label{app:explanation}
We now present a detailed mathematical analysis of the program in (\ref{eg:simple_loop}). 

Recall that we denote $f$ as the return function, and denote $\overline{\varPhi}_f$ as the function given by 
\[
\overline{\varPhi}_f(h)(x):=[\neg \varphi(x)] \cdot f(x) + [\varphi(x)] (p\cdot h(a_1 x + b_1) + (1-p)\cdot h(a_2 x + b_2))
\]
for every function $h:\mathbb{R}\rightarrow\mathbb{R}$. We use the $k$-induction operator $\varPsi_h$  from~\cite{DBLP:conf/cav/BatzCKKMS20} ($k$ is dummy here) which is given by
$\varPsi_h(g) := \min \{\overline{\varPhi}_f(g), h \}$. We apply the $k=2$-induction condition 
to upper-bound the expected value of $X_f$ and perform a key simplification for this condition via loop unfolding as follows. 
For the ease of understanding, we let $H_1 = [\neg \varphi(a_1 x+b_1)] \cdot f(a_1 x + b_1) + [\varphi(a_1x+b_1)] \cdot (p \cdot h(a_1(a_1 x+b_1)+b_1) + (1-p) \cdot h(a_2(a_1 x+b_1)+b_2))$, which intuitively represents that we unfold the loop once at the state of $a_1 x + b_1$, and $H_2 = [\neg \varphi(a_2x+b_2)] \cdot f(a_2 x + b_2) + [\varphi(a_2x+b_2)] \cdot (p \cdot h(a_1(a_2 x+b_2)+b_1) + (1-p) \cdot h(a_2(a_2 x+b_2)+b_2))$, which intuitively represents that we unfold the loop once at the state of $a_2 x + b_2$. 

\begin{itemize}
    \item 
    \textbf{Case 1}: In this case, the loop is executed once, reaching two states $a_1 x + b_1$ and $a_2 x + b_2$, and does not continue. In other words, we unfold the loop only once and obtain the loop-free program $C_1$ as in ~\cref{fig:case1}. This amounts to $h_1 = [\neg \varphi(x)] \cdot f(x) + [\varphi(x)] (p\cdot h(a_1 x + b_1) + (1-p)\cdot h(a_2 x + b_2))$, which is the expected value of $h(x)$ after the execution of the program $C_1$.
    \item 
        \textbf{Case 2}: In this case, the loop is first executed once, reaching two states $a_1 x + b_1$ and $a_2 x + b_2$. Then, we clarify two cases below.
        \begin{itemize}
            \item At the state $a_1 x + b_1$, we stop the execution of the loop and have the value $h(a_1 x + b_1)$.  
            \item At the state $a_2 x + b_2$, we continue the execution of the loop and obtain two branches: (i) if $\varphi$ is not satisfied, we directly have the return function $f(a_2 x + b_2)$; (ii) if $\varphi$ is satisfied, we arrive at the states $a_1(a_2 x+b_2)+b_1$ and $a_2(a_2 x+b_2)+b_2 $. 
        \end{itemize}
        The unfolding process above generates a loop-free program $C_2$ (see~\cref{fig:case2}), and $h_2$ is derived from the program $C_2$ in a way similar to $h_1$. We have that $h_2 = [\neg \varphi(x)] \cdot f(x) + [\varphi(x)] \cdot (p \cdot h(a_1 x + b_1) + (1-p)\cdot H_2 )$, which is the expected value of $h(x)$ after the execution of  the program $C_2$. 
    \item \textbf{Case 3}: This case is similar to \textbf{Case 2}, with the only difference that we choose to continue the execution of the loop at the state $a_1 x + b_1$ and do not unfold the loop at $a_2 x + b_2$. Then, we clarify two cases below. 
       \begin{itemize}
           \item At the state $a_1 x + b_1$, we continue the execution of the loop and we will attain two branches: (i) if $\varphi$ is not satisfied, output the return function $f(a_1 x + b_1)$; (ii) if $\varphi$ is satisfied, we will arrive at the states $a_1(a_1 x+b_1)+b_1$ and $a_2(a_1 x+b_1)+b_2 $.
           \item At the state of $a_2 x + b_2$, we stop the execution of the loop and have the value $h(a_2 x + b_2)$.
       \end{itemize}
        This generates a loop-free program $C_3$ (see~\cref{fig:case3}), from which
        $h_3$ is derived similar to $h_1,h_2$. We have that $h_3 = [\neg \varphi(x)] \cdot f(x) + [\varphi(x)] \cdot (p \cdot H_1 + (1-p) \cdot h(a_2 x + b_2))$, which is the expected value of $h(x)$ after the execution of  the program $C_3$.
    \item \textbf{Case 4}: In this case, at both the states $a_1 x + b_1$ and $a_2 x + b_2$, we choose to execute the loop once more. This generates a loop-free program $C_4$ (see~\cref{fig:case4}). $h_4$ is derived from the program $C_4$ similar to the previous cases. We have that $h_4 = [\neg \varphi(x)] \cdot f(x) + [\varphi(x)] \cdot (p \cdot H_1 + (1-p)\cdot H_2)$, which is the expected value of $h(x)$ after the execution of the program $C_4$.
\end{itemize}

\subsection{Supplementary Material for \textbf{Stage 4}}  \label{app:constraint}

Motzkin's Transposition Theorem is a classical theorem that provides a dual characterization for the satisfiability of a system of strict and non-strict inequalities. Below we present the original Motzkin's Transposition Theorem.

\begin{theorem}[Motzkin's Transposition Theorem~\cite{motzkin1936beitrage}]\label{thm:Motzkin}
    Given the set of linear, and strict linear, inequalities over real-valued variables $x_1, x_2,..., x_n$, 
    \begin{align*}
    \begin{aligned}
        S = \left[  
        \begin{aligned}
            &\sum_{i = 1}^n \alpha_{(1,i)} \cdot x_{i} + \beta_1 \leq 0 \\
            &\qquad \qquad \quad \vdots \\
            &\sum_{i = 1}^n \alpha_{(m,i)} \cdot x_{i} + \beta_m \leq 0 \\
        \end{aligned}
        \right] and
    \end{aligned}
    \quad
    \begin{aligned}
        T = \left[  
        \begin{aligned}
            &\sum_{i = 1}^n \alpha_{(m+1,i)} \cdot x_{i} + \beta_{m+1} < 0 \\
            &\qquad \qquad \qquad \vdots \\
            &\sum_{i = 1}^n \alpha_{(m+k,i)} \cdot x_{i} + \beta_{m+k} < 0 \\
        \end{aligned}
        \right]
    \end{aligned}
\end{align*}

    in which $\alpha_{(1,1)}, ..., \alpha_{(m+k, n)}$ and $\beta_1, ..., \beta_{m+k}$ are real-valued, we have that $S$ and $T$ simultaneously are \textit{not} satisfiable (i.e., they have no solution in $x$) if and only if there exist non-negative real numbers $\lambda_0, \lambda_1, ..., \lambda_{m+k}$ such that either the condition $(A_1)$: 
    $$ \textstyle 0 = \sum_{i=1}^{m+k} \lambda_i \alpha_{(i,1)}, ..., 0 = \sum_{i=1}^{m+k} \lambda_i \alpha_{(i,n)}, 1 = (\sum_{i=1}^{m+k} \lambda_i \beta_i) - \lambda_0,  $$
    or condition $(A_2)$: at least one coefficient $\lambda_i$ for $i$ in the range $\lbrace m+1, ..., m+k \rbrace$ is non-zero and 
    $$ \textstyle 0 = \sum_{i=1}^{m+k} \lambda_i \alpha_{(i,1)}, ..., 0 = \sum_{i=1}^{m+k} \lambda_i \alpha_{(i,n)}, 0 = (\sum_{i=1}^{m+k} \lambda_i \beta_i) - \lambda_0. $$
\end{theorem}

In our work, we consider the variant form of Motzkin's Transposition Theorem (see~\cref{corollary_motzkin}). Theorem~\ref{corollary_motzkin} is first proposed in \cite[Theorem 4.5 and Remark 4.6]{DBLP:journals/toplas/ChatterjeeFNH18} without proof. We give a complete proof here.

\smallskip
\noindent\textbf{Theorem}~\ref{corollary_motzkin}. [Corollary of Motzkin's Transposition Theorem]
Let $S$ and $T$ be the same systems of linear inequalities as that in~\cref{thm:Motzkin}. If $S$ is satisfiable, then $S\wedge T$ is unsatisfiable iff there exist non-negative reals $\lambda_0, \lambda_1, ..., \lambda_{m+k}$ and at least one coefficient $\lambda_i$ for $i\in \lbrace m+1, ..., m+k \rbrace$ is non-zero, such that:
\begin{equation*}
\textstyle 0 = \sum_{i=1}^{m+k} \lambda_i \alpha_{(i,1)}, ..., 0 = \sum_{i=1}^{m+k} \lambda_i \alpha_{(i,n)}, 0 = (\sum_{i=1}^{m+k} \lambda_i \beta_i) - \lambda_0. 
\end{equation*}
i.e., the condition $(A_2)$ in~\cref{thm:Motzkin}.

\smallskip
Before we proof the theorem, we introduce the desired theorem: Farkas's Lemma:

\begin{lemma}[Farkas's lemma~\cite{farkas1894fourier}]
    Consider the following system of linear inequalities over real-valued variables $x_1, x_2,..., x_n$,
    \begin{equation*}
        S = \left[
        \begin{aligned}
            &\alpha_{(1,1)}x_{1}&  &+ \cdots + &  &\alpha_{(1, n)} x_{n}&  + &\beta_1&  \leq 0\\
            &\vdots&               &\vdots&       &\vdots&                   &\vdots  \\ 
            &\alpha_{(m,1)}x_{1}&  &+ \cdots + &  &\alpha_{(m, n)} x_{n}&  + &\beta_m&  \leq 0 \\
        \end{aligned}
        \right]
    \end{equation*}
    When $S$ is satisfiable, it entails a given linear inequality 
    $$\phi: c_1 x_1 + ... +c_n x_n +d \leq 0$$
    if and only if there exist non-negative real numbers $\lambda_0, \lambda_1, ..., \lambda_{m}$, such that 
    $$ c_1 = \sum_{i=1}^{m} \lambda_i \alpha_{(i,1)}, ..., c_n = \sum_{i=1}^{m} \lambda_i \alpha_{(i,n)}, d = (\sum_{i=1}^{m} \lambda_i \beta_i) - \lambda_0$$
    Furthermore, $S$ is unsatisfiable if and only if the inequality $1 \leq 0$ can be derived
    as shown above.
\end{lemma}

Now we proof the corollary (\cref{corollary_motzkin}).

\begin{proof}
    Proof by contradiction. According to Motzkin's Transposition Theorem, $S$ and $T$ have no solution in $x$ if and only if there exists non-negative real numbers $\lambda_0, \lambda_1, ..., \lambda_{m+k}$ such that either condition $(A_1)$ or $(A_2)$ is satisfied. We first proof $(\lambda_{m+1} \neq 0) \vee (\lambda_{m+2} \neq 0) \vee ... \vee (\lambda_{m+k} \neq 0)$.

    If it is not satisfied, we assume that $\lambda_{m+1} = ... = \lambda_{m+k} = 0$. Then we know the condition $(A_1)$ must be satisfied and we have (By applying the assumption $\lambda_{m+1} = ... = \lambda_{m+k} = 0$):
    $$ 0 = \sum_{i=1}^{m} \lambda_i \alpha_{(i,1)}, ..., 0 = \sum_{i=1}^{m} \lambda_i \alpha_{(i,n)}, \sum_{i=1}^{m} \lambda_i \beta_i = \lambda_0 + 1 \geq 1,  $$
    By applying Farkas's Lemma, we have:
    $$c_1 = \sum_{i=1}^{m} \lambda_i \alpha_{(i,1)} = 0, ..., c_n = \sum_{i=1}^{m} \lambda_i \alpha_{(i,n)}  = 0, d = (\sum_{i=1}^{m} \lambda_i \beta_i) - \lambda_0 = \lambda_0 + 1 - \lambda_0 = 1,$$
    Thus we have:
    $$\phi = c_1 x_1 + ... +c_n x_n +d = d = 1 \leq 0$$
    if and only if $S$ is not satisfiable, which contradicts the assumption, so the assumption does not hold. We have proved $(\lambda_{m+1} \neq 0) \vee (\lambda_{m+2} \neq 0) \vee ... \vee (\lambda_{m+k} \neq 0)$. 

    If condition $(A_1)$ is satisfied, then exists non-negative real numbers $\lambda_0, \lambda_1, ..., \lambda_{m+k}$ and $(\lambda_{m+1} \neq 0) \vee (\lambda_{m+2} \neq 0) \vee ... \vee (\lambda_{m+k} \neq 0)$(what we just prove) such that 
    $$ 0 = \sum_{i=1}^{m+k} \lambda_i \alpha_{(i,1)}, ..., 0 = \sum_{i=1}^{m+k} \lambda_i \alpha_{(i,n)}, 1 = (\sum_{i=1}^{m+k} \lambda_i \beta_i) - \lambda_0,  $$
    let $\lambda'_0 = \lambda_0 + 1 \geq 0$ and we can find that it also satisfies the condition $(A_2)$, that is $A_1 \implies A_2$. Thus, Motzkin's Transposition Theorem can be simplified as: If $S$ is satisfiable, then $S$ and $T$ have no solution in $x$ if and only if there exists non-negative real numbers $\lambda_0, \lambda_1, ..., \lambda_{m+k}$, such that:
    $$ ((A_1 \vee A_2) \wedge (A_1 \implies A_2)) \iff A_2 $$  
    Thus we prove~\cref{corollary_motzkin}.
\end{proof}

\subsection{Application of Putinar's Positivstellensatz~\cite{putinar}}\label{app:putinar}
We recall Putinar's Positivstellensatz below.
\begin{theorem}[Putinar's Positivstellensatz~\cite{putinar}] \label{thm:putinar} Let $V$ be a finite set of real-valued variables and $g, g_1, \ldots, g_m \in \mathbb{R}[V]$ be polynomials over $V$ with real coefficients. Consider the set $\mathcal{S}:=\{\mathbf{x} \in \mathbb{R}^V\,\mid\, g_i(\mathbf{x}) \geq 0  \mbox{ for all }1\le i\le m \}$ which is the set of all real vectors at which every $g_i$ is non-negative. If (i)~there exists some $g_k$ such that the set $\{ \mathbf{x} \in \mathbb{R}^V ~\mid~ g_k(\mathbf{x}) \geq 0  \}$ is compact and (ii)~$g(\mathbf{x})>0$ for all $\mathbf{x} \in \mathcal{S}$, then we have that 
	\begin{equation} \label{eq:putinar}
	\textstyle g = f_0 + \sum_{i=1}^m f_i \cdot g_i
	\end{equation}
	for some polynomials $f_0,f_1\dots, f_m\in \mathbb{R}[V]$ such that each polynomial $f_i$ is the  a sum of squares (of polynomials in $\mathbb{R}[V]$), i.e.~$f_i = \sum_{j=0}^{k} q_{i,j}^2$ for polynomials $q_{i,j}$'s in $\mathbb{R}[V]$.
\end{theorem}

In our comparison, we utilize the sound form in \eqref{eq:putinar} for witnessing a polynomial $g$ to be non-negative over a semi-algebraic set $P$  for each inductive constraint $\forall x\in P, g(x)\ge 0$. 

In our experiments, the  maximum degree of unknown SOS polynomials is set to the degree of the polynomial template plus 2.

\section{Supplementary Material for Section~\ref{sec:experiment}}\label{app:section6}

\subsection{Continued Fraction}\label{app:continued fraction}
Continued fraction can represent a real number $r$ by an expression as follows:
\begin{equation*}
    r = a_0 + \frac{1}{a_1 + \frac{1}{a_2 + \frac{1}{\ddots}}}
\end{equation*}
and $r$ is abbreviated as $[a_0, a_1, a_2,...]$. In our implementation, we first transform each output float coefficient into its continued fraction form $[a_0, a_1, a_2,...]$. Then we perform the truncation operation that we find the first $a_i (i\geq 1)$ that is greater than a large threshold, for which we choose $100$, and truncate from there (including this number). We keep only the previous parts, as our rational approximation results. 

\subsection{Experimental Results of Piecewise Linear Lower Bounds}
\label{app:linear_lower}

We present the experimental results of piecewise linear lower bounds in this section. For the piecewise linear lower bound experiments, we consider the same benchmarks and return functions $f$ as in~\cref{sec:linear_res}, and use the same invariant from {\em Invariants} in {\bf Experimental Setting} for each benchmark. Moreover, we follow the rest of the experimental setup described in {\bf Experimental Setting} in~\cref{sec:experiment}.

\smallskip
\noindent{\bf Answering RQ1.}
We present the experimental results for the synthesis of piecewise linear lower bounds on the 13 benchmarks in~\cref{table:2}. In this table, we only show the piecewise results with $(k\leq 3)$-induction. We observe that on most of the benchmarks, we can obtain a linear lower bound via the conventional approach, i.e., 1-induction, while we can synthesize better (tighter) piecewise linear lower bounds via $(k>1)$-induction within a few minutes. Only on the benchmark \textsc{Growing Walk-variant}, we require $(k > 1)$-induction to synthesize a lower bound. Our approach derives the exact bound, i.e., the tightest upper bound, on the benchmark~\textsc{Equal-Prob-Grid}. The exactness of these bounds is established by comparison with the piecewise upper bound presented in~\cref{sec:linear_res}.

\begin{table*}

    \renewcommand{\arraystretch}{2.0}
	\caption{Experimental Results for {\bf RQ1} and {\bf RQ2}, Linear Case (Lower Bounds). "$f$" stands for the return function considered in the benchmark, "T(s)" (of our approach) stands for the execution time of our approach (in seconds), including the parsing from the program input, transforming the $k$-induction constraint into the bilinear problems, bilinear solving time and verification time. "Conventional Approach ($k=1$)" stands for the monolithic linear lower bound synthesized via 1-induction, "$k$" stands for the $k$-induction we apply, "Solution" stands for the linear candidate solved by Gurobi, and "Piecewise Linear Lower Bound" stands for our piecewise results. "Result" stands for the synthesized results by other approaches and "T(s)" (of their approaches) stands for the execution time of their tools.}
	\label{table:2}
	\resizebox{\textwidth}{!}{
		\begin{threeparttable}
			\begin{tabular}{|c|c|c|c|c|c|c|c|c|c|c|c|}
				\hline
                \multicolumn{1}{|c|}{\multirow{2}{*}{\textbf{Benchmark}}}  &
				\multicolumn{1}{c|}{\multirow{2}{*}{\textbf{$f$}}}      &
                \multicolumn{2}{c|}{\multirow{1}{*}{\makecell{\textbf{Conventional} \\ \textbf{Approach} ($k=1$)}}}  &
			\multicolumn{4}{c|}{\multirow{1}{*}{\textbf{Our Approach}}}  &
                \multicolumn{2}{c|}{\multirow{1}{*}{\textsc{cegispro2}}}  &
                \multicolumn{2}{c|}{\multirow{1}{*}{\textsc{exist}}}  
                \\ \cline{3-12}
                \multicolumn{1}{|c|}{} & \multicolumn{1}{c|}{} &
                \multicolumn{1}{c|}{\textbf{Result}}  & 
				\multicolumn{1}{c|}{\textbf{T(s)}} &
                    \multicolumn{1}{c|}{\textbf{$k$}} &
				\multicolumn{1}{c|}{\textbf{Solution}}   &    
				\multicolumn{1}{c|}{\textbf{Piecewise Linear Lower Bound}} & 
    		      \multicolumn{1}{c|}{\textbf{T(s)}} &
                    \multicolumn{1}{c|}{\textbf{ Result  }} &
                    \multicolumn{1}{c|}{\textbf{ T(s) }} &
                    \multicolumn{1}{c|}{\textbf{ Result  }} &
                    \multicolumn{1}{c|}{\textbf{ T(s) }}  
                \\ \hline \hline
                {\textsc{Geo}}& $x$  & $x$ & 0.33  &3   &  $x$ & $[c>0]\cdot x + [c\leq 0]\cdot (x+\frac{3}{4})$ & 2.19  & \makecell{$[c>0]\cdot x + $\\$[c\leq 0]\cdot (x+\frac{3}{4})$} & 0.06  & $x+[c=0]$ & 83.01 \\
                \hline
                {\textsc{k-Geo}}& $y$  & $y$ & 100.18 & 3   &{$y$ } & \makecell{$[k>N]\cdot y + [k\leq N] \cdot $ \\ $(0.75x+y+0.25)$} &  133.81   & \makecell{$[k>N]\cdot y + $\\$[k\leq N] \cdot $\\$(-k+N+x+y+1)$} & 0.2 & \makecell{$y + [k\leq n] \cdot$\\$(0.8x  -0.3k $\\$+ 0.3n + 0.5)$} & 239.95 \\
                \hline
                {\textsc{Bin-ran}}& $y$ & \makecell{$-0.5i+$\\$y+5$}  & 1.18 & 2  &  \makecell{$-\frac{21}{29}*i+$\\$y+\frac{210}{29}$}  & \makecell{$[i>10] \cdot y + [1< i\leq 10] \cdot $ \\ $ (-\frac{21}{29}i+y+\frac{9}{20}x+\frac{1068}{145}) $} & 106.59  & \makecell{$[i>10] \cdot y + [ i\leq 10] \cdot $ \\ $ (\frac{9}{19}*x +y -$\\$\frac{53059}{112955}*i +\frac{154900}{22591}) $} & 0.26  &fail &- \\
                \hline
                {\textsc{Coin}}& $i$  & $i$ & 100.51 &2  & $i$ &  $[y\neq x]\cdot i + [y=x] \cdot (i+\frac{13}{8})$ & 5.99  & \makecell{$[y\neq x]\cdot i + $\\$[y=x] \cdot (i+\frac{13}{8})$} & 0.07 & $i + [x = y]\cdot 2.2$ & 116.67\\
                \hline
                {\textsc{Mart}}& $i$  & $i$ & 0.37 & 3   & $i$ &  $[x\leq 0]\cdot i + [x>0]\cdot (i+1.5)$ & 2.44  &\makecell{violation of \\ non-negativity } & -   & $i + [x > 0] *  2 $ & 122.93\\
                \hline
                {\makecell{\textsc{Growing} \\\textsc{Walk}}} & $y$ & $x+y$  & 100.16 & 3   & $x+y$ &  \makecell{$[x<0]\cdot y + [x\geq 0]\cdot (x+y+\frac{5}{4})$} & 101.80 & \makecell{violation of \\ non-negativity } & - & fail & -\\
                \hline
                {\makecell{\textsc{Growing} \\\textsc{Walk} \\ \textsc{-variant}}} & $y$ & \ding{55} & - & 3  & $y-1$  & \makecell{$[x<0] \cdot y + $ \\ $[0\leq x <1] \cdot (y+0.5x-1)$ + \\ $[1\leq x <2]\cdot (y+0.5x-1.5)$ \\$+[2 \leq x]\cdot (y+0.75x-2)$ } & 125.53  & \makecell{violation of \\ non-negativity } & - & fail & -\\
                \hline
                {\makecell{\textsc{Expected} \\ \textsc{Time}}}& $t$ & \makecell{$1.1111x$\\$+t$} & 0.25 & 3  &$1.240x+t$  & \makecell{$[x<0]\cdot t+ $\\$[0\leq x <1]\cdot (0.124x+t+0.9)$\\$[1 \leq x \leq 10] \cdot(1.1284x+t+1.9116)$} & 125.54  & \makecell{violation of \\ non-negativity } & - & fail & - \\
                \hline
                {\makecell{\textsc{Zero-Conf} \\ \textsc{-variant}}} & $\textit{cur}$  &  $\textit{cur}$ &100.32 &3  &$\textit{cur}$ &  
                \makecell{$[\textit{est} >0]\cdot \textit{cur}+$ \\ $[\textit{start}==0 \wedge \textit{est} \leq 0]\cdot (\textit{cur} + 1.9502)$ \\ $+[\textit{start}\geq 1 \wedge \textit{est} \leq 0]\cdot(\textit{cur} + 0.287)$} & 183.63  & \makecell{violation of \\ non-negativity } & - & inner error & - \\
                \hline
                {\makecell{\textsc{Equal-} \\ \textsc{Prob-Grid}}}& $\textit{goal}$  &$\textit{goal}$ & 100.38 & 2  & $\textit{goal}$  & \makecell{$[a>10\vee b>10  \vee \textit{goal} \neq 0]\cdot \textit{goal}$ \\ $[a\leq 10 \wedge b \leq 10 \wedge \textit{goal} = 0]\cdot 1.5$} &139.80  & \makecell{$[a>10\vee b>10 \vee $\\$\textit{goal} \neq 0]\cdot \textit{goal}+$ \\ $[a\leq 10 \wedge b\leq 10 $\\$ \wedge \textit{goal} = 0] \cdot 1.5$} & 0.1  & inner error & - \\
                \hline
                {\textsc{RevBin}}& $z$  & $z+2x-2$ & 100.14 & 3  & \makecell{$z+2x$\\$-2$}   & \makecell{$[x<1] \cdot z +$ \\ $[1 \leq x <2] \cdot (z+x) $ \\ $+[x \geq 2] \cdot (z+2x-2)$} & 129.46  & \makecell{$[x<1] \cdot z +$ \\ $[x \geq 1] \cdot (z+2x-2)$} & 0.11  & \makecell{$[x > 0] \cdot 2x$\\ + z$ $} & 122.85 \\
                \hline
                {\textsc{Fair Coin}}& $i$  & $i$ &100.34 & 3 & $i$  &  \makecell{$[x>0 \vee y>0] \cdot i+$\\$[x=0 \wedge y=0]\cdot(i + \frac{5}{4})$}  &43.84  & \makecell{$[x>0 \vee y>0] \cdot i+$ \\$[x=0 \wedge y=0] $\\$ \cdot(i + \frac{5}{4})$} & 0.06  & \makecell{$[ x+y = 0]\cdot 1.3$\\$+i$} &82.67\\ 
                \hline
                {\makecell{\textsc{St-Petersburg} \\ \textsc{variant}}} & $y$  & $y$ & 0.32 & 3  & $y$ & $[x > 0]\cdot y+ [x \leq 0]\cdot \frac{11}{8} y$ &2.28  & \makecell{$[x > 0]\cdot y+ $\\$[x \leq 0]\cdot \frac{11}{8} y$} & 0.21  & \makecell{$ [x=0] \cdot 0.4y$\\$+y$} & 98.05 \\
                \hline
                
			\end{tabular}
	\end{threeparttable}}

\end{table*}

\smallskip
\noindent{\bf Answering RQ2.}
We answer RQ2 by comparing our approach with the most related approaches~\cite{DBLP:conf/tacas/BatzCJKKM23,DBLP:conf/cav/BaoTPHR22} in~\cref{table:2}. The relevant explanations for {\bf RQ2} in~\cref{table:2} are totally the same to~\cref{table:1}. These two relevant works require a (possibly piecewise) lower bound to be verified as an additional program input and return a sub-invariant that is sufficient to \emph{verify} the input lower bound, which is the most different aspect from our work. \textsc{cegispro2} produce the results by a proof rule derived from the original OST (see Section 6 in~\cite{DBLP:conf/tacas/BatzCJKKM23} and~\cref{app:classical_OST}), while we apply an extended OST (see~\cref{thm:ost-variant}). To have a richer comparison, we also feed our benchmarks paired with the piecewise lower bounds synthesized by our approach to \textsc{cegispro2}. On 5 of our benchmarks (e.g., \textsc{Growing Walk-variant, Zero-Conf-variant}, etc), it reports failure (violation of non-negativity). On 6 of our benchmarks, \textsc{cegispro2} produce the same results with our inputs. Only on two benchmarks (\textsc{k-Geo, Bin-Ran}) \textsc{cegispro2} produce a different result to verify our inputs.

For the comparison with \textsc{exist}, we note that \textsc{exist} synthesizes sub-invariants without the application of OST, which might be unsound for proving the input lower bounds (see also Section 7 in~\cite{DBLP:conf/tacas/BatzCJKKM23}). We compare with their tool on our benchmarks by assuming the soundness of their lower bounds, and feed them our piecewise linear lower bounds as an additional program input. 
On benchmarks~\textsc{Geo, k-Geo, Coin, RevBin, Mart, Fair Coin, St-Petersburg variant}, their tool can generate a tighter sub-invariant to \emph{verify} our piecewise lower bound. On these benchmarks, due to the existence of exact invariants, they are usually able to find a tighter sub-invariant by a heuristic search based on sampling and machine learning at the cost of the long time (usually about or even more than 100s). For the remaining benchmarks, either they cannot generate sub-invariants or there are internal errors within their tool. 

In conclusion,  our approaches can handle many benchmarks that these two works~\cite{DBLP:conf/tacas/BatzCJKKM23,DBLP:conf/cav/BaoTPHR22} cannot handle. When feeding our benchmarks with the bounds synthesized through our approach to
\textsc{cegispro2} and \textsc{exist}, they fail on about $40\%$ of our benchmarks. Over most of the benchmarks that
\textsc{cegispro2} and our approach can handle, our bounds are comparable with theirs. Over most of the benchmarks that
\textsc{exist} and our approach can handle, they spend much more time to generate a slightly tighter bound.

\begin{figure}[htbp]
  \centering
  \subfloat[\textsc{Geo}]
  {\includegraphics[width=0.2\textwidth]{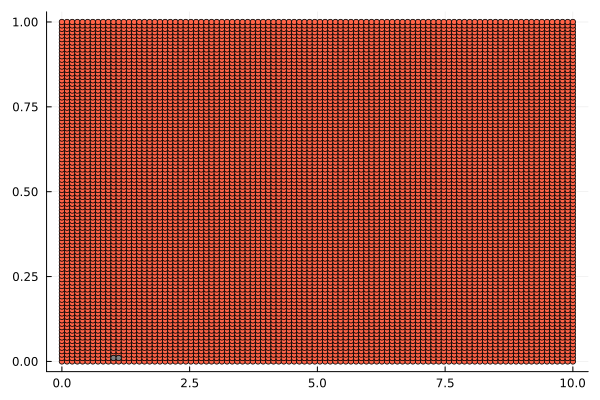}} \hfill   
  \subfloat[\textsc{Mart}]
  {\includegraphics[width=0.2\textwidth]{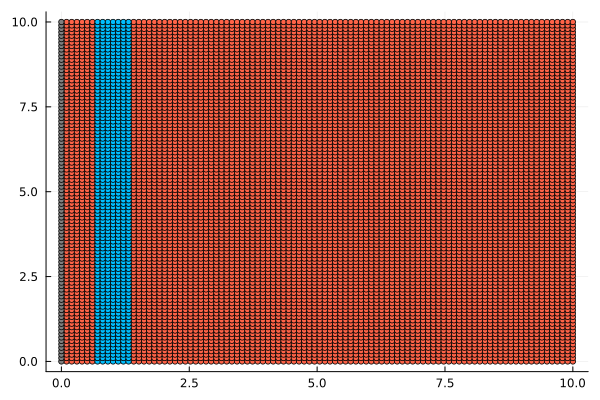}}  \hspace{0.1cm}
  \subfloat[\textsc{Bin-Ran}]
  {\includegraphics[width=0.31\textwidth]{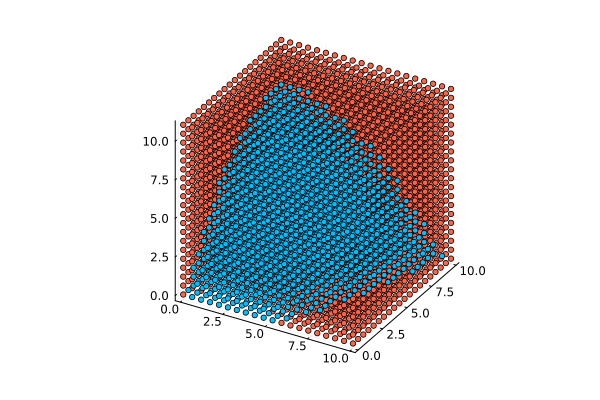}} \hspace{-0.8cm}
  \subfloat[\textsc{Coin}]
  {\includegraphics[width=0.31\textwidth]{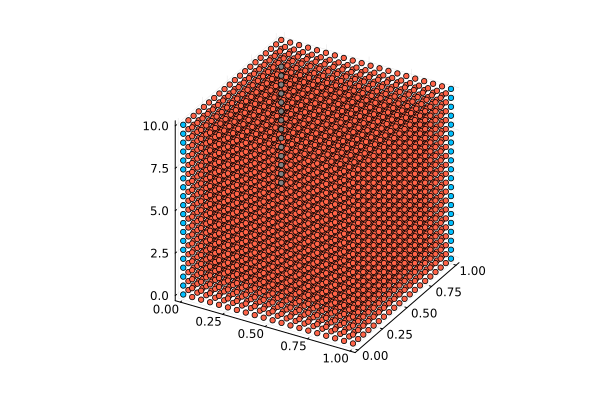}}\hspace{-0.5cm}  \\
  \subfloat[\textsc{Growing Walk}]
  {\includegraphics[width=0.2\textwidth]{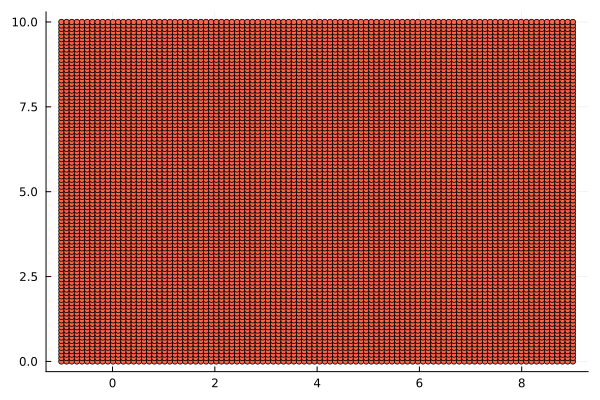}} \hfill   
  \subfloat[\parbox{0.15\textwidth}{\centering \textsc{Growing Walk \\ variant}}]
  {\includegraphics[width=0.2\textwidth]{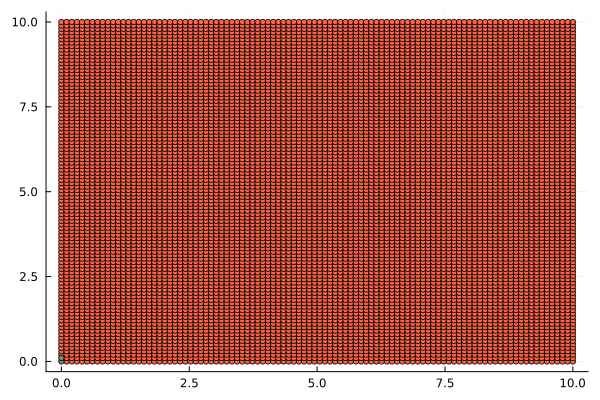}}  \hspace{0.1cm}
  \subfloat[\parbox{0.2\textwidth}{\centering \textsc{Zero Conference\\ variant}}]
  {\includegraphics[width=0.31\textwidth]{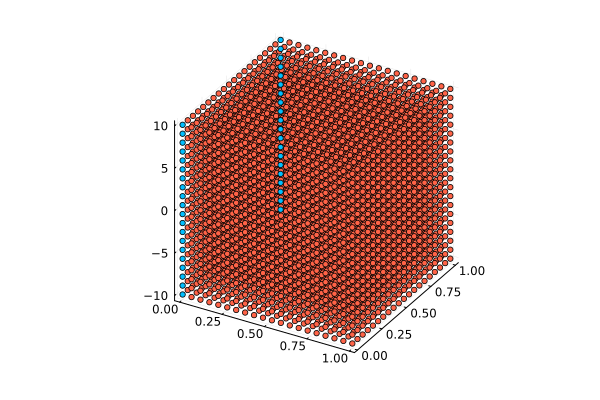}} \hspace{-0.8cm}
  \subfloat[\parbox{0.2\textwidth}{\centering \textsc{Equal Probability \\ Grid Family}}]
  {\includegraphics[width=0.31\textwidth]{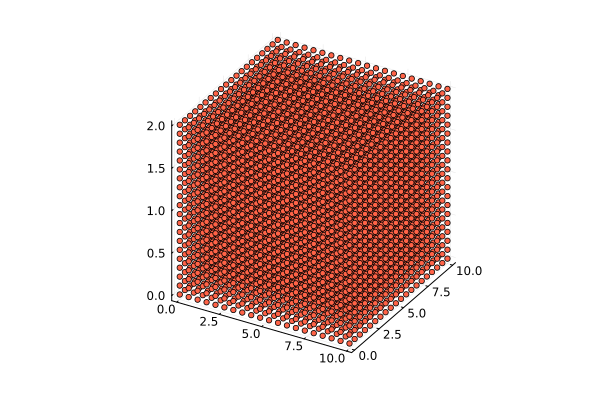}}\hspace{-0.5cm}  \\
  \subfloat[\textsc{Expected Time}]
  {\includegraphics[width=0.2\textwidth]{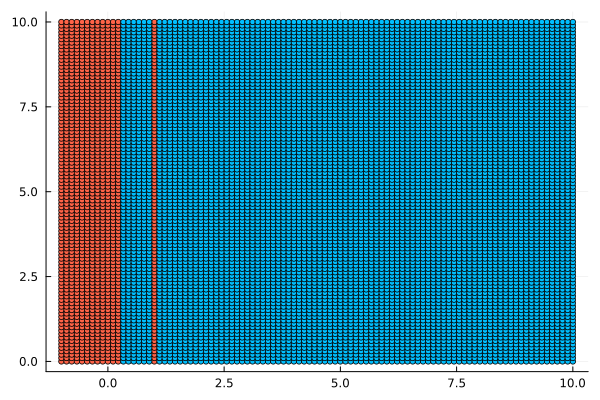}} \hfill   
  \subfloat[\textsc{RevBin}]
  {\includegraphics[width=0.2\textwidth]{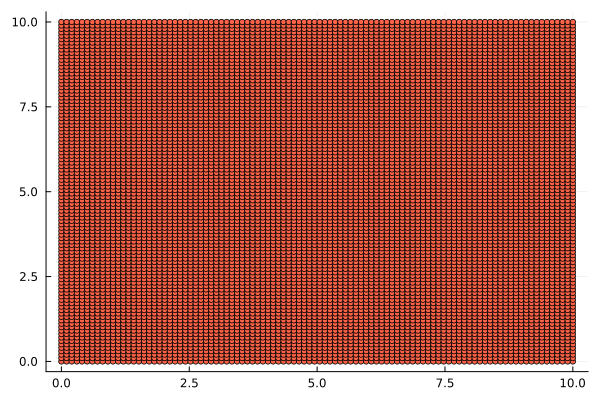}}  \hspace{0.1cm}
  \subfloat[\textsc{Fair Coin}]
  {\includegraphics[width=0.31\textwidth]{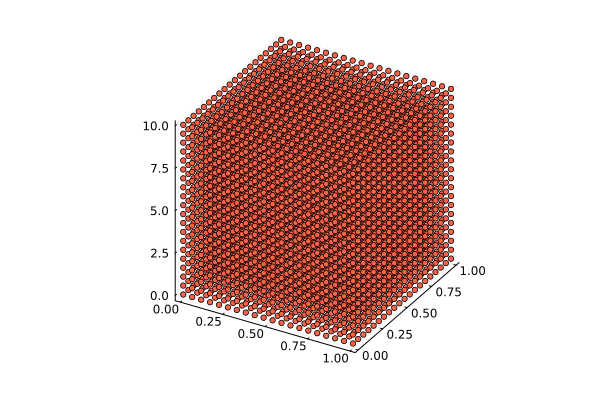}} \hspace{0.1cm}
  \subfloat[\parbox{0.15\textwidth}{\centering \textsc{St-Petersburg \\ variant}}]
  {\includegraphics[width=0.2\textwidth]{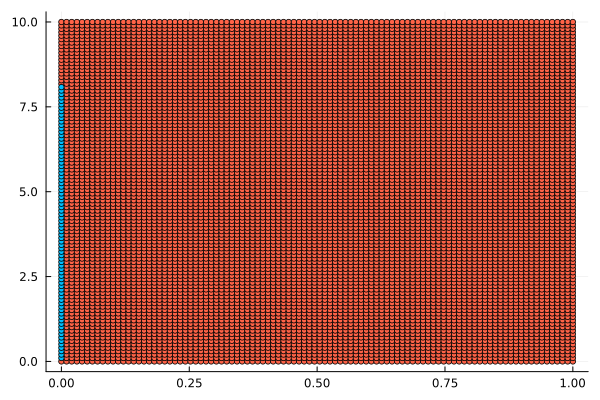}} \hspace{0.6cm}

  \caption{Difference Plots of the Comparison in Piecewise Linear Lower Case}
  \label{fig:diff_plots_lower1}
  
  \begin{tablenotes}
      \item \textcolor{red}{Red} points indicate where our piecewise lower bounds are larger than the monolithic ones \\ by more than $10^{-3}$; 
      \textcolor{blue}{blue} points indicate where the monolithic bounds are larger by more \\ than $10^{-3}$;
      \textcolor{gray}{gray} points denote cases with negligible differences ($\le 10^{-3}$).
  \end{tablenotes}
\end{figure}

\begin{table*}
    \renewcommand{\arraystretch}{1.8}
        \caption{Experimental Results for {\bf RQ3}, Linear Case (Lower Bounds). "$f$" stands for the return function considered in the benchmark, "$k$" stands for the $k$-induction condition we apply in this comparison, "Monolithic Polynomial via 1-Induction" stands for the monolithic polynomial bounds synthesized via 1-induction, and "d" stands for the degree of polynomial template we use. "PCT" stands for the percentage of the points that our piecewise lower bound are lower (i.e., not better) than monolithic polynomial, and "Diff" stands for the unbiased average difference between our piecewise polynomial bound and the monolithic polynomial. A positive value indicates how much tighter our piecewise bounds are on average. }
        \label{table:comparison_lower}
        \resizebox{\textwidth}{!}{
		\begin{threeparttable}
            \begin{tabular}{|c|c|c|c|c|c|c|c|}
				\hline
				\multicolumn{1}{|c|}{\multirow{2}{*}{\textbf{Benchmark}}}  &
				\multicolumn{1}{c|}{\multirow{2}{*}{\textbf{$f$}}}      &
				\multicolumn{2}{c|}{\multirow{1}{*}{\textbf{Our Approach}}}  &
				\multicolumn{2}{c|}{\multirow{1}{*}{\textbf{\textbf{\makecell{Monolithic  Polynomial  via 1-induction }}}}} &
                \multicolumn{1}{c|}{\multirow{2}{*}{\textbf{PCT}}} & 
                \multicolumn{1}{c|}{\multirow{2}{*}{\textbf{Diff}}} \\ 
                \cline{3-6}
                \multicolumn{1}{|c|}{}  & \multicolumn{1}{c|}{} &  
                \multicolumn{1}{c|}{$\; k\;$}  &   
				\multicolumn{1}{c|}{\textbf{Piecewise Linear Lower Bound} } 
                &\multicolumn{1}{c|}{\textbf{d}}
                &\multicolumn{1}{c|}{\textbf{Monolithic Polynomial Lower Bound}}
                &  \multicolumn{1}{c|}{}   &  \multicolumn{1}{c|}{} \\ \hline \hline

                {\textsc{Geo}}& $x$   & 4  & \makecell{$[c>0]\cdot x +$\\$[c\leq 0]\cdot (x+\frac{7}{8})$} &  3 & \makecell{$-0.0313 - 0.1902*c + 1.0478*x - $\\ $0.3980*c^2 + 0.0695*x*c - 0.0019*x^2 - $\\$0.1595*x*c^2 + 0.07227*x^2*c - 0.0147*x^3$} & $0.0\%$  & 2.6872\\
                \hline
                {\textsc{k-Geo}}& $y$  & 3 &     \makecell{$[k>N]\cdot y + $ \\ $[k\leq N] \cdot (0.75x+y+0.25)$} & 2 & \makecell{$44.6223*N -221.2813 - 0.7791*k + 1.0*y $\\$+0.9281*x - 2.1922*N^2 - 0.1043*x^2$} &  $4.19\% $ & 65.432 \\
                \hline
                {\textsc{Bin-ran}}& $y$  &  2  & \makecell{$[i>10] \cdot y+$ \\ $[1< i\leq 10] \cdot (-\frac{21}{29}i+y+\frac{9}{20}x+\frac{1068}{145}) $} & 3 & \makecell{$-22.0746 - 24.4593*i + 33.7063*y +$\\$ 20.7709*x + 1.4945*i^2 + 0.2057*y*i + $\\$ 0.0232*y^2 + 0.4741*x*i + 0.2689*x*y +$\\$ 1.9807*x^2 + 0.0006*i^3 - 0.3133*y*i^2 -$\\$ 0.0111*y^2*i + 0.0049*y^3 - 0.4668*x*i^2 +$\\$ 0.0036*x*y*i + 0.0105*x*y^2 - 0.7437*x^3 $\\$+ 0.04213*x^2*y - 0.7531*x^2*i$} &  33.39\% & 252.02 \\
                \hline
                {\textsc{Coin}}& $i$    &  4 & $[y\neq x]\cdot i + [y=x] \cdot (i+\frac{129}{64})$ & 2 & \makecell{$2.6655 + 1.0002*i - 3622.3830*y - $ \\ $5419.0667*x - 0.0001*i^2 + 0.0007*y*i +$\\$ 3619.71553*y^2 - 0.0008*x*i + $\\$1827.4383*x*y + 3594.2952*x^2$} & 2.0\% & 1587.5 \\
                \hline
                {\textsc{Mart}}& $i$   & 4   & \makecell{$[x\le 0]\cdot i + $\\$[x > 0]\cdot (i+\frac{7}{4})$} & 2 &\makecell{$i + 4.0*x - 2.0*x^2$}  &  1.0\% & 48.736 \\
                \hline
                {\makecell{\textsc{Growing} \\ \textsc{Walk}}}& $y$  & 4  &   \makecell{$[x<0]\cdot y +$ \\ $ [x\geq 0]\cdot (x+y+\frac{13}{8})$} & 3 &\makecell{$-0.0004+1.0003*y+1.3463*x-$ \\ $0.0001*y^2-0.0010*x*y-0.0590*x^2$ \\ $+ 0.0007*x^2*y-0.0022*x^3$} & 0.0\% & 1.909 \\
                \hline
                {\makecell{\textsc{Growing} \\ \textsc{Walk} \\ \textsc{variant}}}& $y$   & 3    &  \makecell{$[x<0]\cdot y +$ \\ $ [0\leq x <1]\cdot (0.5x+y-1) $ \\ $+[1 \leq x<2 ] \cdot (0.5x+y-1.5)$ \\ $+[2 \leq x] \cdot (0.75x+y-2)$} & 3 & \makecell{$-1.0000+1.0000*y-0.3903*x-$ \\ $0.0734*y^2+0.0484*x*y+0.4758*x^2-$ \\ $ 0.0250*x*y^2-0.0484*x^2*y-0.0855*x^3$} &  0.01\% &  23.987 \\
                \hline
                {\makecell{\textsc{Expected} \\ \textsc{Time}}}& $t$  &  3  & \makecell{$[x<0]\cdot t+$\\$[0\leq x <1]\cdot (0.124x+t+0.9)+$\\$[1 \leq x \leq 10] \cdot(1.1284x+t+1.9116)$} & 3 & \makecell{$-0.0784 + 1.0093*t + 3.1426*x - $\\$0.0010*t^2 + 0.0083*x*t - 0.1576*x^2 + $\\$0.0002*x*t^2 + 0.0002*x^2*t + 0.0043*x^3$} &  87.0 \% & -3.2809 \\
                \hline
                {\makecell{\textsc{Zero-Conf} \\ \textsc{-variant}}}& $\textit{cur}$   &3 &  \makecell{$[\textit{est} >0]\cdot \textit{cur}+$ \\ $[\textit{start}==0 \wedge \textit{est} \leq 0]\cdot (\textit{cur} + 1.9502)$ \\ $+[\textit{start}\geq 1 \wedge \textit{est} \leq 0]\cdot(\textit{cur} + 0.287)$} &2
                & \makecell{$140.2458 + 1.0098*cur - 424365.5964*start -$\\$ 587675.0179*est - 0.0066*start*cur + $\\$424267.3602*start^2 -  0.0095*est*cur - $\\$504437.5495*est*start + 587534.7143*est^2$}& 0.64\% & $2.9*10^5$\\
				\hline
                {\makecell{\textsc{Equal-} \\ \textsc{Prob-Grid}}}& $\textit{goal}$  & 2 &    \makecell{$[a>10\vee b>10 \vee \textit{goal} \neq 0]\cdot \textit{goal}$ \\ $+ [a\leq 10 \wedge b\leq 10 \wedge \textit{goal} = 0]\cdot 1.5$} & 2 & \makecell{$0.4950*goal - 0.2020*goal^2 + $\\$0.0053*b*goal - 0.0011*a*goal$} & 0.0\% &  0.8373 \\
				\hline
                {\textsc{RevBin}}& $z$   & 3   &  \makecell{$[x<1] \cdot z +$ \\ $[1 \leq x <2] \cdot (z+x) $ \\ $+[x \geq 2] \cdot (z+2x-2)$} & 2 & \makecell{$-1.5942 + 1.505*z + 2.2516*x - $\\$0.0842*z^2 - 0.3367*x*z - 0.7415*x^2$} & 0.0\% & 32.0514 \\
				\hline
                {\textsc{Fair Coin}}& $i $   &  4  & \makecell{$[x > 0 \vee y>0]\cdot i + $ \\ $[x\leq 0 \wedge y \leq 0]\cdot (i+\frac{21}{16})$} & 2 &\makecell{$1.0000*i - 0.3932*y - 0.39325*x $\\$- 0.3153*i^2 + 0.6305*y*i - 0.7242*y^2 $\\$+ 0.6305*x*i - 0.1796*x*y - 0.7242*x^2$} & 0.0\% & 5.0417\\
                \hline
                {\makecell{\textsc{St-Petersburg} \\ \textsc{variant}}} & $y$ & 3 & $[x > 0]\cdot y+ [x \leq 0]\cdot \frac{11}{8} y$ & 3 &\makecell{$-0.0058 + 1.4976*y - 4576.1318*x - $\\$0.015*y^2 - 40.4398xy + 2420.6062*x^2 +$\\$ 0.0142xy^2 + 39.9432x^2y + 2155.5278*x^3$} & 0.0\% & 964.78\\
                \hline
            \end{tabular}
        \end{threeparttable}}
\end{table*}

\smallskip
\noindent{\bf Answering RQ3.} Consistent with the upper case analysis, we derive monolithic polynomial lower bounds using 1-induction to serve as a baseline. Our proposed piecewise linear lower bounds are then systematically compared against these monolithic counterparts on all benchmarks. We compare two results by uniformly taking the grid points in the invariant and evaluate two results, and we compute the percentage of the points that our piecewise lower bound are lower (i.e., no better) than monolithic polynomial, which is shown in the column "PCT" in~\cref{table:comparison_lower}. We also present difference plots that classify all grid points into three disjoint regions according to the magnitude of difference: red points correspond to cases where our piecewise linear lower bounds are notably larger (diff $ > 10^{-3}$), blue points correspond to cases where the monolithic lower bounds are notably larger, and gray points represent regions where the two bounds are nearly identical (diff $ \le 10^{-3}$). We display the comparison in~\cref{fig:diff_plots_lower1}. We observe that on our benchmarks except \textsc{Expected Time}, our piecewise linear bounds are significantly tighter (i.e., greater) than monolithic polynomial lower bounds. 
In addition, we quantify the difference between the two lower bounds by subtracting the piecewise lower bound from the monolithic one and taking the unbiased average of the resulting values, which is reported in the last column “Diff” of~\cref{table:comparison_lower}. We observe that, except for \textsc{Expected Time}, our piecewise bounds are generally tighter than the monolithic ones, with especially notable improvements on benchmarks such as \textsc{Zero-Conf-Variant}.
We conjecture that the relatively poor performance of our piecewise linear algorithm in  \textsc{Expected Time} is due to the fact that the true expected value is closer to a piecewise polynomial function.

\subsection{Experimental Results of Piecewise Polynomial Lower Bound} \label{app:poly_lower}

In this section, we present the experimental results of piecewise polynomial lower bounds. For the piecewise polynomial lower bounds, we consider the same benchmarks and return functions $f$ as in~\cref{sec:poly_res}, and use the same invariant from {\em Invariants} in {\bf Experimental Setting} for each benchmark. Moreover, we follow the rest of the experimental setup described in {\bf Experimental Setting} in~\cref{sec:experiment}.

\smallskip 
\noindent{\bf Answering RQ1.} We present the experimental results for the synthesis of piecewise polynomial lower bounds on the 20 benchmarks in Table~\ref{table:poly2}. The experimental results show that our approach can compute piecewise polynomial lower bounds for most of the benchmarks within around 10 seconds. Only \textsc{inv-Pend variant} and \textsc{cav-5} require more than five minutes to compute a result. Especially, on the benchmarks~\textsc{Bin0, Bin2, DepRV, Sum0, Prinsys}, the lower bounds we obtain are the same with the upper bounds we obtain in~\cref{sec:poly_res} (see~\cref{table:poly1} for more details), which shows that we obtain the exact expected value of $X_f$ after the execution of the loop, i.e., the tightest lower bounds, on these 5 benchmarks. The exactness of these results is verified by comparison with the exact invariants synthesized in~\cite{DBLP:conf/cav/BaoTPHR22} and with our corresponding upper bounds in~\cref{sec:poly_res}.

\begin{table*}
    \renewcommand{\arraystretch}{1.5}
	\caption{Experimental Results for {\bf RQ1} and {\bf RQ2}, Polynomial Case (Lower Bounds). "$f$" stands for the return function considered in the benchmark, "T(s)" stands for the execution time of our approach (in seconds), including the parsing procedure from the program input, relaxing the $k$-induction constraint into the SDP problems, the SDP solving time and verification time. "d" stands for the degree of polynomial template we use and "Solution $h^*$" is the candidate polynomial solved directly by the solver. "Piecewise Polynomial Lower Bound" stands for the piecewise bound we synthesize. "Sub-invariant" stands for the sub-invariant synthesized by~\textsc{exist} when feeding our lower bounds .}
	\label{table:poly2}
	\resizebox{\textwidth}{!}{
		\begin{threeparttable}
			\begin{tabular}{|c|c|c|c|c|c|c|c|}
				\hline
				\multicolumn{1}{|c|}{\multirow{2}{*}{\textbf{Benchmark}}}  &
				\multicolumn{1}{c|}{\multirow{2}{*}{\textbf{$f$}}}      &
				\multicolumn{4}{c|}{\multirow{1}{*}{\textbf{Our Approach}}}  &
                    \multicolumn{2}{c|}{\multirow{1}{*}{\textbf{\textsc{exist}}}}  
				\\ \cline{3-8}
				\multicolumn{1}{|c|}{} & \multicolumn{1}{c|}{}  &  
                    \multicolumn{1}{c|}{\textbf{d}} &
                    \multicolumn{1}{c|}{\textbf{Solution $h^*$}}  & 
				\multicolumn{1}{c|}{\textbf{ T(s) }} &
				\multicolumn{1}{c|}{\textbf{ Piecewise Polynomial Lower Bound}}   &    
				\multicolumn{1}{c|}{\textbf{ Sub-invariant }} &  
                \multicolumn{1}{c|}{\textbf{T(s)}} 
                \\ \hline \hline
               {\textsc{GeoAr}}& $x$   &  2  &  \makecell{$-0.0467*y^2 + 0.8036*y*z -$\\$ 7.1202*z^2 + x + 0.6668*y $\\$+ 10.2222*z - 2.3795$}  &   3.98  & \makecell{$\max\{[z>0] \cdot (-0.0467y^2 + 0.4018*y*z$\\$ - 3.5601*z^2 + x + 1.0734*y + 5.5129*z $\\$- 1.2594) + [z\le 0] \cdot x, h^*$\}} &  \makecell{inner error} & - \\
                \hline
               {\textsc{Bin0}}& $x$   &  2 &  $x + 0.5*y*n$  &  5.56 & $x + [n>0]\cdot 0.5*y*n$  & fail & - \\
                \hline
               {\textsc{Bin2}}& $x$   &  2  &  $0.25*n + x + 0.25*n^2 + 0.5*y*n$  &  5.44   & \makecell{$x + [n>0]\cdot(0.25*n + x $\\$+ 0.25*n^2 + 0.5*y*n)$} &  fail & -\\
                \hline
               {\textsc{DepRV}}& $x*y$   &  2  &  \makecell{$-0.25*n + 0.25*n^2 + 0.5*y*n$\\$ + 0.5*x*n + x*y$}  &  5.75  & \makecell{$[n >0] \cdot (-0.25*n + 0.25*n^2 +$ \\$ 0.5*y*n + 0.5*x*n + x*y) $\\$+ [n\leq 0] \cdot x*y$} &  \makecell{inner error} & -\\
                \hline
               {\textsc{Prinsys}}& $[x==1]$   &  2  &  $0$  &  2.10   & \makecell{$[x==1]*1 + [x==0]*0.5$} &  \makecell{$[x==1]*1 + $\\$[x==0]*0.5$} & 7.29 \\
                \hline
               {\textsc{Sum0}}& $x$ & 2 & $0.25*i^2+0.25*i+x$ &  1.98 & $[i>0]*(0.25*i^2+0.25*i)+x$ & fail & -\\
                \hline
               {\textsc{Duel}}& $t$  & 2 & \makecell{$21.7319*x^2 - 0.4706*x*t + 1.3703*t^2$\\$ - 21.7099*x - 0.3707*t - 0.0011$} & 6.66 & \makecell{$\max \{[t>0\wedge x \ge 1]\cdot(10.8660x^2 +$\\$ 0.2353*x*t + \dots - 5.5451*x - $\\$ 0.8705*t + 0.2488) + [x<1]\cdot t, h^* \}$} &  inner error& -\\ \hline
               {\textsc{brp}}& \makecell{$[failed$\\$=10]$}  & 2 & \makecell{$-41834.4189*failed^2 - 6.0771*failed*sent $\\$- 0.8349*sent^2 - 1710.0678*failed $\\$+ 655.2652*sent + 2695.5257$} & 9.85 & \makecell{$\max\{[failed<10\wedge sent < 800]\cdot $\\$(-418.3442*failed^2 - 0.0608*failed*sent- $\\$ 0.8349*sent^2 - 853.7891failed + 653.5513sent$\\$ + 2907.9668) + [failed =10], h^* \}$} & inner error & -\\
                \hline
               {\textsc{chain}}& $[y=1]$  & 2 & \makecell{$-0.0001*x*y - 0.0052*y^2 + $\\$0.0032*x + 0.0173*y - 0.0347$} & 4.09 & \makecell{$\max\{[y=0\wedge x < 100]\cdot(-0.0001*x*y $\\$- 0.0051*y^2 + 0.0032*x + $\\$0.0170*y - 0.0314)+[y=1], h^*\}$} & inner error & -\\
                \hline
               {\textsc{grid small}}& \makecell{$[a<10 \wedge $\\$ b \ge 10]$}  & 3 & \makecell{$0.0006*a^3 - 0.0012*a^2*b + 0.0008*a*b^2 $\\$- 0.0071*a^2 + 0.008*a*b - 0.0056*b^2 - $\\$0.046*a + 0.0822*b + 0.4185$} & 6.75 & \makecell{$\max\{[a<10\wedge b < 10]\cdot (0.0006*a^3 - $\\$0.0012*a^2*b + 0.0008*a*b^2 - 0.0068*a^2$\\$ + 0.0076*a*b - 0.0052*b^2 - 0.0478*a + $\\$ 0.0800*b + 0.4306) + [a<10 \wedge b \ge 10], h^*\}$} & inner error & - \\
                \hline
               {\textsc{grid big}}& \makecell{$[a<1000 \wedge $\\$ b \ge 1000]$}  & 2 &   \makecell{$-0.0231*a^2 + 0.0462*a*b - 0.0231*b^2$\\$ - 0.1895*a + 0.2425*b + 0.9503$} & 7.21 &\makecell{$\max\{[a < 1000 \wedge b < 1000]\cdot (-0.0231*a^2 + $\\$ 0.0462*a*b - 0.0231*b^2 - 0.1895*a +$\\$ 0.2425b + 0.9537) + [a < 1000 \wedge b \ge 1000], h^* \}$} & inner error & -\\
                \hline
               {\textsc{CAV-2}}& $[h>1+t]$   & 3 &  \makecell{$0.0001*h^3 - 0.0003*h^2*t$\\$ + 0.0003*h*t^2 - 0.0001*t^3 + 0.0018*h^2 $\\$- 0.0057*h*t + 0.0032*t^2 - 0.002*h $\\$+ 0.054*t - 0.6863$}  & 3.45 &\makecell{$\max\{[t\ge h]\cdot (0.0001*h^3 - 0.0003*h^2*t + $\\$ 0.0003*h*t^2 - 0.0001*t^3 + 0.0023*h^2 $\\$- 0.0066*h*t + 0.0037*t^2 + 0.0076*h +$\\$ 0.0399*t - 0.5852)+[h>1+t], h^*\}$} & \makecell{inner error} & - \\
                \hline
               {\textsc{CAV-4}}& $[x\le 10]$ & 2 & \makecell{$-0.0148*x^2 - 0.0597*x*y + 0.3443*y^2$\\$ + 0.0523*x - 0.3282*y + 0.9537$} & 2.47 & \makecell{$\max\{[y\ge 1] \cdot (-0.0148*x^2 - 0.0072*x + $\\$ 0.9694) + [y<1\wedge x \le 10], h^*\}$} & \makecell{inner error} & - \\
                \hline
               {\textsc{fig-6}}& $[y\le 5]$  & 4 & \makecell{$0.0001*x^4 - 0.0007*x^3*y + 0.0009*x^2*y^2 $\\$- 0.0006*x*y^3 - 0.0011*x^3 + 0.0143*x^2*y $\\$- 0.0035*x*y^2 + 0.0032*y^3 + 0.0556*x^2 - $\\$0.1077*x*y + 0.0085*y^2 - $\\$0.3753*x + 0.1362*y + 0.5438$} & 109.28 & \makecell{$\max\{[x \le 4 ]\cdot (0.0001*x^4 - 0.0007*x^3*y +$\\$0.0009*x^2*y^2 - 0.0006*x*y^3 - 0.0014*x^3 $\\$ + 0.0140*x^2*y - 0.0035*x*y^2 + 0.0026*y^3 $\\$ + 0.0690*x^2 - 0.0960*x*y + 0.0173*y^2 $\\$- 0.3696*x + 0.1229*y + 0.5508)$\\$ + [x > 4 \wedge y \le 5], h^*\}$}  & \makecell{inner error} & - \\
                \hline
               {\textsc{fig-7}}& $[x \le 1000]$  & 2 & \makecell{$-0.0002*x*y - 0.0029*y^2 + 0.0038*y*i $\\$- 0.0009*i^2 + 0.0002*x - 0.0037*y$\\$ + 0.002*i + 0.9978$} & 21.38 & \makecell{$\max\{[y\le 0]\cdot (-0.0009*i^2 + 0.0002*x + $\\$  0.0021*i + 0.997)+[y > 0\wedge x \le 1000], h^* \}$} & \makecell{inner error} & - \\
                \hline
               {\textsc{\makecell{inv-Pend \\ variant}}}& $[pA\le1]$  & 3 & \makecell{$0.0008*pAD^2*pA - 0.0023*pAD^2*cV +$\\$ 0.0991*pAD^2*cP + 0.4931*pAD*pA^2 +$\\$0.1464*pAD*pA*cV \cdots - 5.002*cV*cP $\\$- 44.9405*cP^2 - 5.7109*cV + 1.0$} & 436.04 & \makecell{$\max\{[cP>0.5\vee pA <-0.1 \vee cP<-0.5\vee $\\$pA > 0.1]\cdot (0.0011*pAD^2*pA + $\\$0.0011*pAD^2*cV + \cdots + 0.999* cV $\\$- 0.0688*cP + 1.6061)+[cP\le 0.5 $\\$\wedge pA\le 0.1 \wedge cP \ge -0.5 \wedge cP \le 0.5], h^*\}$} & \makecell{inner error} & - \\
                \hline
               {\textsc{CAV-7}}& $[x\leq 30]$  &  3    &  \makecell{$0.0001*i^3 - 0.0002*i^2*x + 0.0001*i*x^2 - $\\$0.0006*i^2 + 0.001*i*x - 0.0002*x^2 $\\$+ 0.0007*i - 0.0005*x + 0.9981$}   & 5.17 & \makecell{$\max\{[i < 5]\cdot (- 0.0001*i^2*x - 0.0001*i^2 $\\$+ 0.0006*i*x - 0.0001*x^2 + 0.0003*i+ $\\$ 0.0001*x + 0.9985) + [i \ge 5 \wedge x \le 30], h^* \}$} &  \makecell{inner error} & -\\
                \hline
               {\textsc{cav-5}}& $[i\leq 10]$  &  3  & \makecell{$0.0009*i^2*money + 0.0043*i*money^2 $\\$+ 0.0013*money^3 - 0.9614*i^2 - $\\$17.8117*i*money - 66.2212*money^2$\\$ - 29.2611*i + 1.0$} &  897.32   &  \makecell{$\max\{[money \ge 10]\cdot (0.0009*i^2*money + $\\$0.0043*i*money^2 + 0.0013*money^3$\\$ - 0.9624*i^2 - 17.8205*i*money-$\\$ 66.2275*money^2 - 12.8062*i + $\\$118.2861*money - 1379.4033)$\\$ + [money < 10 \wedge i \le 10], h^* \}$} &  \makecell{inner error} & - \\
                \hline
               {\textsc{Add}}& $[x>5]$   &  3  &  \makecell{$0.0316 - 0.1528*y + 0.0164*x$\\$ + 0.0688*y^2 - 0.0271*x*y + 0.0048*x^2$\\$ - 0.0004*y^3 - 0.0083*x*y^2 + $\\$0.002*x^2*y - 0.0002*x^3$} &  3.74 & \makecell{$\max\{[y \le 1]\cdot (0.0315 - 0.1456y + 0.0165*x $\\$+ 0.0620*y^2 - 0.0272*x*y + 0.0047*x^2$\\$ - 0.0004*y^3 - 0.0083*x*y^2 + 0.002*x^2*y$\\$ - 0.0002*x^3) + [y > 1\wedge x > 5], h^*\}$} &  \makecell{inner error} & - \\ \hline
               {\textsc{\makecell{Growing\\Walk\\Variant2}}} & $y$ & 2 & \makecell{$-0.0055*x^2 - 0.0013*x*y - 0.0132*x*r$\\$ - 0.0027*y^2 + 0.0123*y*r - 0.0261*r^2 $\\$+ 0.0288*x + 1.0125*y + 0.0111*r - 0.0454$} & 4.83 & \makecell{$\max\{[r \le 0]\cdot (-0.0075*x^2 - 0.004*x*y $\\$- 0.0027*y^2 + 0.5230*x + 1.0174*y$\\$ - 0.0362) + [r > 0]\cdot y,  h^*\}$} & inner error & -\\
                \hline
			\end{tabular}
	\end{threeparttable}}
\end{table*}

\smallskip
\noindent{\bf Answering RQ2.} We answer RQ2 by comparing our approach with the relevant work~\textsc{Exist} in~\cref{table:poly2}. Since their tool requires a lower bound to be verified as an extra program input, we feed them our lower bounds (the column "Solution $h^*$" in~\cref{table:poly2}) synthesized by our approach. Across all benchmarks, they only successfully synthesize a sub-invariant to {\emph verify} our lower bounds on \textsc{Prinsys} and the sub-invariant they generate is the same as our piecewise lower bound. For the benchmarks \textsc{Bin0, Bin2,Sum0}, they can learn some candidates for sub-invariants but they are not able to verify them so that they fail to generate a sub-invariant. For the other 16 benchmarks, they fail to generate due to some inner errors within their tool. Overall, compared to~\cite{DBLP:conf/cav/BaoTPHR22}, we are able to handle more benchmarks, and for the benchmarks that both approaches can handle, our results coincide with theirs.

\begin{figure}[htbp]
  \centering
  \subfloat[\textsc{Add}]
  {\includegraphics[width=0.2\textwidth]{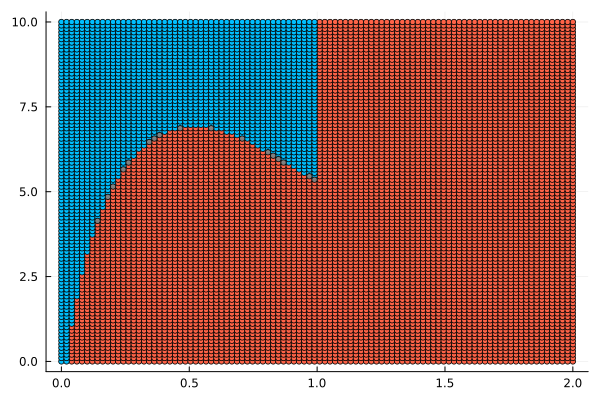}} \hfill   
  \subfloat[\textsc{cav-5}]
  {\includegraphics[width=0.2\textwidth]{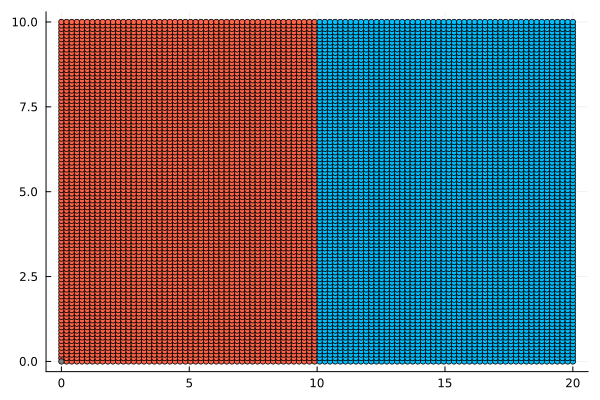}}  \hspace{0.1cm}
  \subfloat[\textsc{fig-7}]
  {\includegraphics[width=0.31\textwidth]{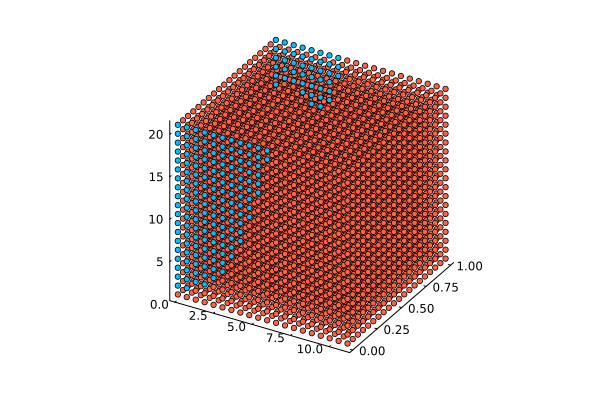}} \hspace{-0.8cm}
  \subfloat[\textsc{Growing Walk variant2}]
  {\includegraphics[width=0.31\textwidth]{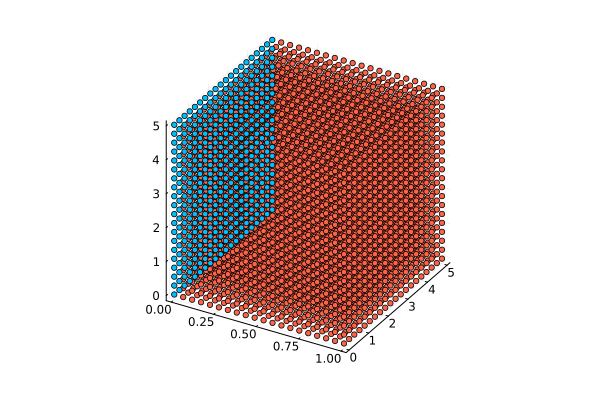}}\hspace{-0.5cm}  \\
  \subfloat[\textsc{Sum0}]
  {\includegraphics[width=0.2\textwidth]{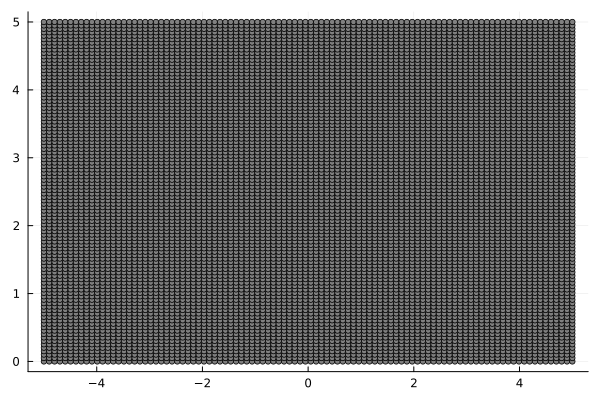}} \hfill   
  \subfloat[\textsc{Duel}]
  {\includegraphics[width=0.2\textwidth]{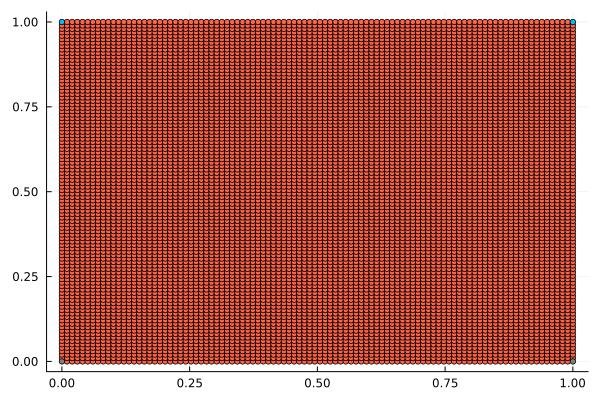}}  \hspace{0.1cm}
  \subfloat[\textsc{GeoAr}]
  {\includegraphics[width=0.31\textwidth]{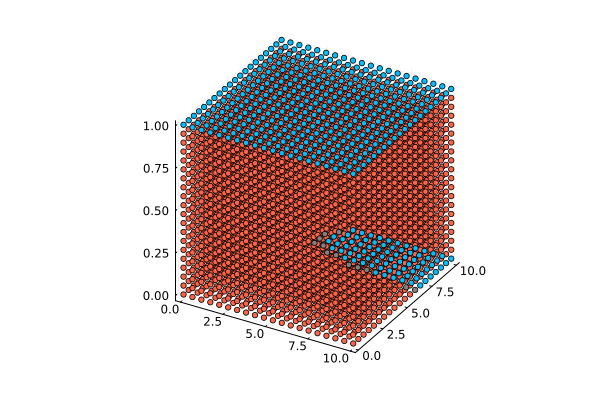}} \hspace{-0.8cm}
  \subfloat[\textsc{Bin0}]
  {\includegraphics[width=0.31\textwidth]{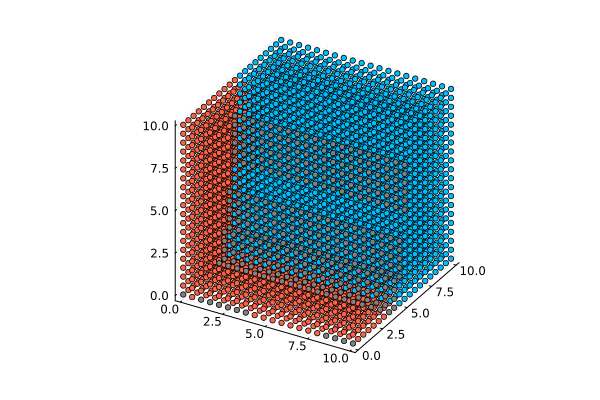}}\hspace{-0.5cm}  \\
  \subfloat[\textsc{brp}]
  {\includegraphics[width=0.2\textwidth]{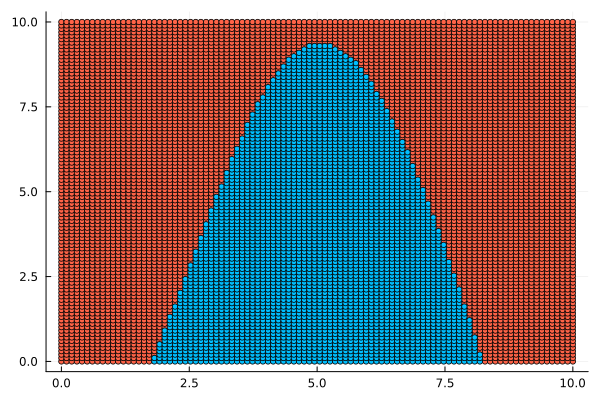}} \hfill
  \subfloat[\textsc{chain}]
  {\includegraphics[width=0.2\textwidth]{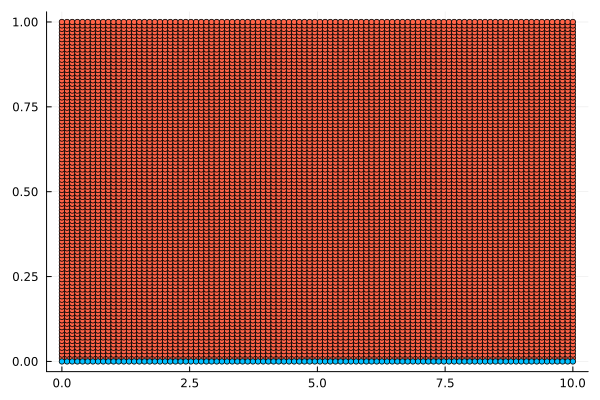}} \hspace{0.1cm}
  \subfloat[\textsc{Bin2}]
  {\includegraphics[width=0.31\textwidth]{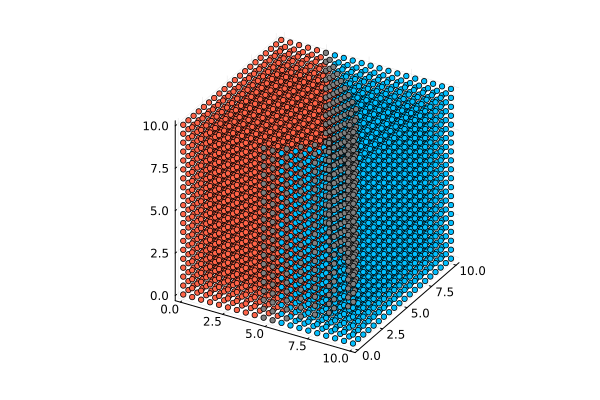}} \hspace{-0.8cm}
  \subfloat[\textsc{DepRV}]
  {\includegraphics[width=0.31\textwidth]{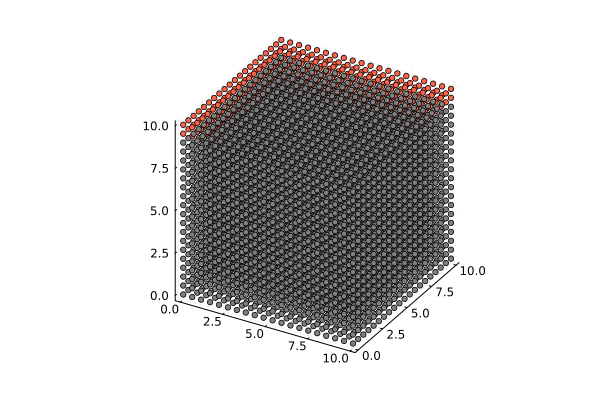}} \hspace{-0.5cm}  \\
  \subfloat[\textsc{grid small}]
  {\includegraphics[width=0.2\textwidth]{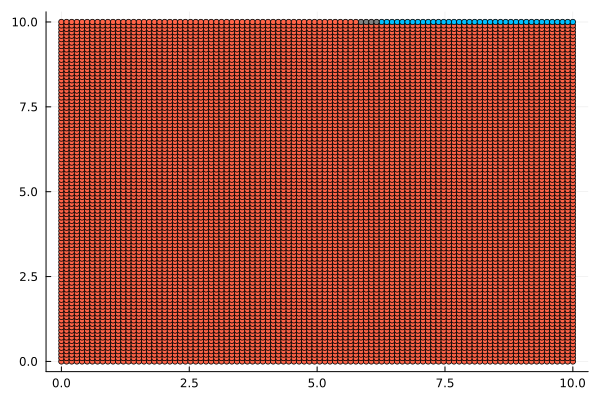}} \hfill   
  \subfloat[\textsc{grid big}]
  {\includegraphics[width=0.2\textwidth]{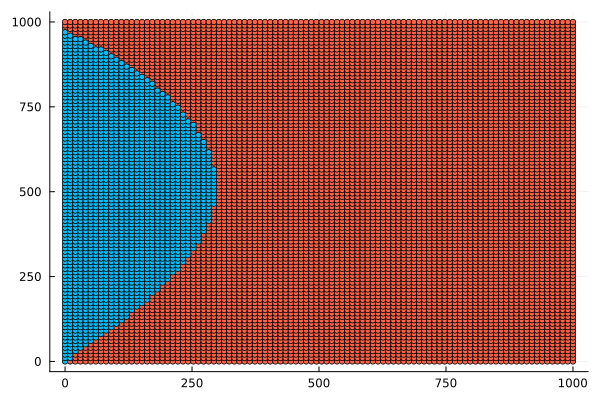}}  \hfill
  \subfloat[\textsc{cav-2}]
  {\includegraphics[width=0.2\textwidth]{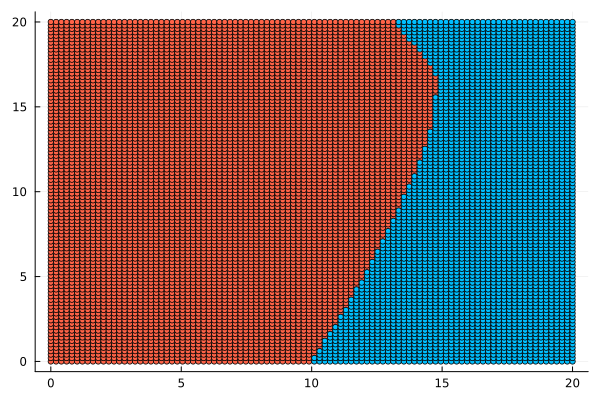}} \hfill   
  \subfloat[\textsc{cav-4}]
  {\includegraphics[width=0.2\textwidth]{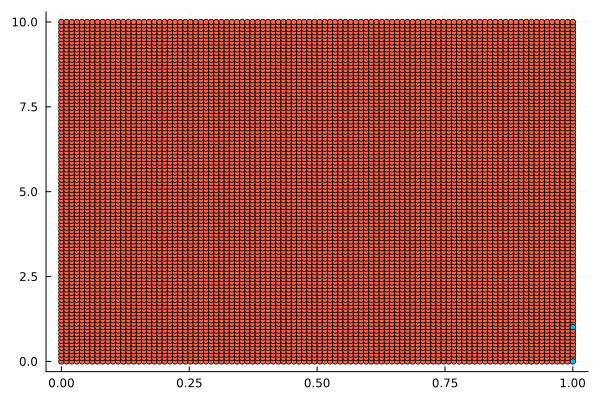}}  \hfill  \\
  \subfloat[\textsc{fig-6}]
  {\includegraphics[width=0.2\textwidth]{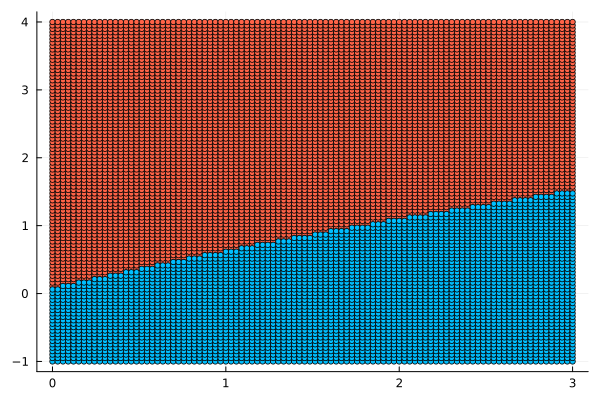}} \hspace{0.75cm}   
  \subfloat[\textsc{CAV-7}]
  {\includegraphics[width=0.2\textwidth]{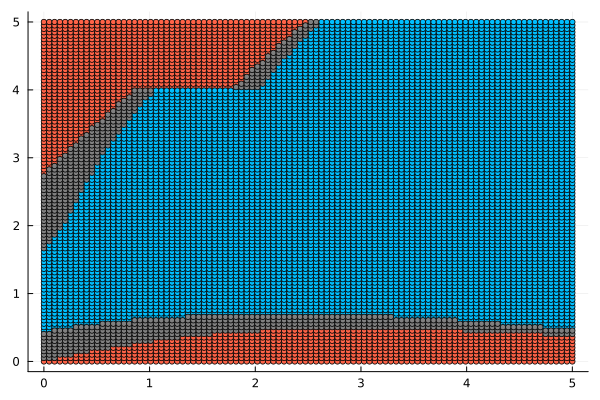}}  \hfill  

  \caption{Difference Plots of the Comparison in Piecewise Polynomial Lower Case}
  \label{fig:diff_plots_lower2}
  
  \begin{tablenotes}
      \item \textcolor{red}{Red} points indicate where our piecewise lower bounds are larger than the monolithic ones \\ by more than $10^{-3}$; 
      \textcolor{blue}{blue} points indicate where the monolithic bounds are larger by more \\ than $10^{-3}$;
      \textcolor{gray}{gray} points denote cases with negligible differences ($\le 10^{-3}$).
  \end{tablenotes}
\end{figure}

\begin{table*}
    \renewcommand{\arraystretch}{1.5}
        \caption{Experimental Results for {\bf RQ3}, Polynomial Case (Lower Bounds). "$f$" stands for the return function considered in the benchmark, "Piecewise Polynomial Lower Bound" stands for the results synthesized by our algorithm. "Monolithic Polynomial via 1-Induction" stands for the monolithic polynomial bounds synthesized via 1-induction, and "T(s)" stands for the total execution time. "PCT" stands for the percentage of the points that our piecewise polynomial lower bound are smaller (i.e., not better) than (higher degree) monolithic polynomial, and "Diff" stands for the unbiased average difference between our piecewise polynomial bound and the monolithic polynomial. A positive value indicates how much tighter our piecewise bounds are on average. }
        \label{table:comparison_lower_poly}
        \resizebox{\textwidth}{!}{
		\begin{threeparttable}
            \begin{tabular}{|c|c|c|c|c|c|c|c|c|c|}
				\hline
				\multicolumn{1}{|c|}{\multirow{2}{*}{\textbf{Benchmark}}}  &
				\multicolumn{1}{c|}{\multirow{2}{*}{\textbf{$f$}}}      &
				\multicolumn{3}{c|}{\multirow{1}{*}{\textbf{Our Approach}}}  &
				\multicolumn{3}{c|}{\multirow{1}{*}{\makecell{\textbf{Monolithic Polynomial  via 1-induction}}}} &
                \multicolumn{1}{c|}{\multirow{2}{*}{\textbf{PCT}}} & \multicolumn{1}{c|}{\multirow{2}{*}{\textbf{Diff}}}\\ 
                \cline{3-8}
                \multicolumn{1}{|c|}{}  & \multicolumn{1}{c|}{} &  
                \multicolumn{1}{c|}{\textbf{d}}  &  \multicolumn{1}{c|}{\textbf{T(s)}} &  
				\multicolumn{1}{c|}{\textbf{Piecewise Polynomial lower Bound}} &\multicolumn{1}{c|}{\textbf{d}} &  \multicolumn{1}{c|}{\textbf{T(s)}} &   
				\multicolumn{1}{c|}{\textbf{Monolithic Polynomial lower Bound}} 
                & \multicolumn{1}{c|}{} & \multicolumn{1}{c|}{}  \\ \hline \hline

                {\textsc{GeoAr}}& $x$ & 2 &  3.98 &  \makecell{$\max\{[z>0] \cdot (-0.0467y^2 + 0.4018*y*z$\\$ - 3.5601*z^2 + x + 1.0734*y + 5.5129*z $\\$- 1.2594) + [z\le 0] \cdot x, h^*$} &  3 & 0.66  & \makecell{$0.0021*x^2*z + 0.0003*x*y^2 + 0.0128xyz $\\$- 0.0745*x*z^2 - 0.0013*y^3 + 0.0006*y^2*z $\\$- 0.9011*y*z^2 - 34359.8787*z^3 - 0.0001*x^2$\\$ - 0.0028*x*y + 0.0294*x*z + 0.0154*y^2 $\\$+ 1.9735*y*z + 68717.1029*z^2 + 1.0025*x $\\$- 0.0476*y - 34355.4581*z - 0.0973$}    &  5.98\% & 2662.9\\
                \hline

                {\textsc{Bin0}}& $x$ & 2    &  5.56 &  \makecell{$x + [n>0]\cdot 0.5*y*n$}  & 3 & 0.61 & \makecell{$-0.001 - 0.0002*n - 0.0017*y + x $\\$+ 0.4994*y*n + 0.0005*y^2 + 0.0001*y^2*n$}  &  48.25\% & -0.009 \\
                \hline
                {\textsc{Bin2}}& $x$    &  2  &  5.44 & \makecell{$x + [n>0]\cdot(0.25*n + $\\$x + 0.25*n^2 + 0.5*y*n)$} &  3 & 0.53 & \makecell{$-0.0001*y^3 + 0.0001*y^2*n + 0.0001*y*n^2$\\$ + 0.0006*y^2 + 0.4992*y*n + 0.25*n^2 + x$\\$ - 0.0021*y + 0.249*n - 0.0028$} & 38.2\% & -0.002\\
                \hline
                {\textsc{DepRV}}& $x*y$ &  2 &  5.75  &  \makecell{$[n >0] \cdot (-0.25*n + 0.25*n^2 + 0.5*y*n $ \\$+ 0.5*x*n + x*y) + [n\leq 0] \cdot x*y$}  &  3 & 0.57  & \makecell{$x*y + 0.5*x*n + 0.5*y*n + $\\$ 0.25*n^2 - 0.2501*n - 0.0001$} &  0.0\% & 0.0006 \\
                \hline
                {\textsc{Prinsys}}& $[x==1]$ &   2  & 2.10 & \makecell{$[x==1]*1 + [x==0]*0.5$}     & 3 & 0.45 & 0.0 & 0.0\% & 0.0\\
                \hline
                {\textsc{Sum0}}& $x$ & 2  & 1.98 & \makecell{$[i>0]*(0.25*i^2+0.25*i)+x$}  & 4 & 0.50 & $0.25*i^2 + 0.25*i + x$ & 0.0\% & 0.0 \\
                \hline
                {\textsc{Duel}}& $t$ &   2  & 7.24 & \makecell{$\max \{[t>0\wedge x \ge 1]\cdot(10.8660x^2 +$\\$ 0.2353*x*t + 1.3703*t^2 - 11.0903*x $\\$- 1.3703*t + 0.4987)+[t\le 0\wedge x\ge 1]\cdot$\\$(5.4330*x^2 + 0.1177*x*t + 1.3703*t^2$\\$ - 5.5451*x - 0.8705*t + 0.2488) $\\$+ [x<1]\cdot t, h^* \}$}& 4 & 0.58 & \makecell{$57.6107*x^4 - 0.3086*x^3*t + 32.5537*x^2*t^2$\\$ - 0.9734*x*t^3 - 8.9958*t^4 + 31.3993*x^3 $\\$- 17.8531*x^2*t + 10.7254*x*t^2 + 26.343*t^3$\\$ - 27.2812*x^2 - 24.7154*x*t $\\$+ 13.3805*t^2 - 61.5859*x - 29.7278*t$} & 0.02\% & 30.219 \\
                \hline
                {\textsc{brp}}& \makecell{$[failed$\\$=10]$} &   2  & 9.85 & \makecell{$\max\{[failed<10\wedge sent < 800]\cdot $\\$(-418.3442*failed^2 - 0.0608*failed*sent- $\\$ 0.8349*sent^2 - 853.7891failed + 653.5513sent$\\$ + 2907.9668) + [failed =10], h^* \}$}     & 4 & 1.24 & \makecell{$-5.1928*failed^4 - 0.992*failed^3*sent -$\\$ 0.0002*failed^2*sent^2 - 1.6946*failed^3 + $\\$ 2.1022*failed^2*sent + 0.0001*failed*sent^2 - $\\$3.3782failed^2 - 1.0916failed*sent - 0.0057sent^2$\\$ - 2.09*failed + 1.1127*sent + 0.7991$} & 37.3\% & 3686.8\\
                \hline
                {\textsc{chain}}& $[y=1]$ &   2  & 4.09 & \makecell{$\max\{[y=0\wedge x < 100]\cdot(-0.0001*x*y $\\$- 0.0051*y^2 + 0.0032*x + $\\$0.0170*y - 0.0314)+[y=1], h^*\}$}     & 3 & 0.74 & \makecell{$531.0717 - 2.0558*10^6*y + 161.9284*x + $\\$1.8583*10^6*y^2 - 116.9387*x*y $\\$- 6.8309*x^2 + 196159.2965*y^3 + $\\$41.3507*x*y^2 + 0.9467*x^2*y + 0.051*x^3$} & 1.00\% & $354890.1$ \\
                \hline
                {\textsc{grid small}}& \makecell{$[a<10 \wedge$\\$ b \ge 10]$} &  3  & 6.75 & \makecell{$\max\{[a<10\wedge b < 10]\cdot (0.0006*a^3 - $\\$0.0012*a^2*b + 0.0008*a*b^2 - 0.0068*a^2$\\$ + 0.0076*a*b - 0.0052*b^2 - 0.0478*a + $\\$ 0.0800*b + 0.4306) + [a<10 \wedge b \ge 10], h^*\}$}  & 4 & 0.90 & \makecell{$-0.0002*a^2*b + 0.0001*a*b^2 + $\\$0.0002*a^2 - 0.0007*a*b + 0.0004*b^2 $\\$-0.0371*a + 0.0364*b + 0.4437$} & 0.62\% & 0.5539 \\
                \hline
                {\textsc{grid big}}& \makecell{$[a<1000 \wedge$\\$ b \ge 1000]$} &  2  & 7.24 &   
                \makecell{$\max\{[a < 1000 \wedge b < 1000]\cdot (-0.0231*a^2 + $\\$ 0.0462*a*b - 0.0231*b^2 - 0.1895*a +$\\$ 0.2425b + 0.9537) + [a < 1000 \wedge b \ge 1000], h^* \}$}
                & 3 & 0.56  & \makecell{$0.001*a^3 - 0.0005*a^2*b - 0.0018*a*b^2 $\\$+ 0.0008*b^3 - 2.9594*a^2 + 5.9103*a*b -$\\$ 2.9631b^2 - 499.9807a + 511.5109b + 253.0223$} & 19.23\% & $2.7*10^8$ \\
                \hline
                {\textsc{cav-2}}& $[h > t + 1]$ &   3  & 3.45 & \makecell{$\max\{[t\ge h]\cdot (0.0001*h^3 - 0.0003*h^2*t + $\\$ 0.0003*h*t^2 - 0.0001*t^3 + 0.0023*h^2 $\\$- 0.0066*h*t + 0.0037*t^2 + 0.0076*h +$\\$ 0.0399*t - 0.5852)+[h>1+t], h^*\}$}  & 4 &  0.47 & \makecell{$-300.5839 + 69.1282*t - 21.1489*h -$\\$ 6.092*t^2 + 3.8054*h*t - 0.413*h^2 + $\\$0.2491*t^3 - 0.269*h*t^2 + 0.0985*h^2*t - $\\$0.0179*h^3 - 0.003*t^4 + 0.0053*h*t^3 - $\\$0.0034*h^2*t^2 + 0.0004*h^3*t + 0.0004*h^4$} & 33.33\% & 104.29 \\
                \hline
                {\textsc{cav-4}}& $[x\le 10]$ &   2  & 2.47 & \makecell{$\max\{[y\ge 1] \cdot (-0.0148*x^2 - 0.0072*x + $\\$ 0.9694) + [y<1\wedge x \le 10], h^*\}$}     & 3 & 0.34 & \makecell{$-0.0017*x^3 - 0.0105*x^2*y - 0.0514*x*y^2 + $\\$ 11.376*y^3 + 0.0085*x^2 + 0.0539*x*y - $\\$ 5.0983*y^2 - 0.0103*x - 6.2928*y + 1.0$} & 0.76\% & $2.303$ \\
                \hline
                {\textsc{fig-6}}& $[y\le 5]$ &  4  & 109.28 &  \makecell{$\max\{[x \le 4 ]\cdot (0.0001*x^4 - 0.0007*x^3*y +$\\$0.0009*x^2*y^2 - 0.0006*x*y^3 - 0.0014*x^3 $\\$ + 0.0140*x^2*y - 0.0035*x*y^2 + 0.0026*y^3 $\\$ + 0.0690*x^2 - 0.0960*x*y + 0.0173*y^2 $\\$- 0.3696*x + 0.1229*y + 0.5508)$\\$ + [x > 4 \wedge y \le 5], h^*\}$}  & 5 & 0.94 & \makecell{$0.0002*x^5 + 0.0001*x^4*y + 0.0001*x^2*y^3 $\\$- 0.0001*x*y^4 - 0.0021*x^4 - 0.0033*x^3*y $\\$+ 0.001*x^2*y^2 - 0.0016*x*y^3 $\\$- 0.0001*y^4 + 0.0072*x^3 + 0.033*x^2*y$\\$ - 0.008*x*y^2 + 0.0017*y^3 + 0.0817*x^2 $\\$- 0.2069*x*y + 0.0347*y^2 - 0.8681*x$\\$ + 0.5271*y + 0.5958$} & 40.77\% & 0.2377 \\
                \hline
                {\textsc{fig-7}}& $[x \le 1000]$ &   2  & 21.38 & \makecell{$\max\{[y\le 0]\cdot (-0.0009*i^2 + 0.0002*x$\\$ + 0.0021*i + 0.997)+[y > 0\wedge x \le 1000], h^* \}$}     & 3 & 2.40 & \makecell{$0.0616*x^2*y - 0.0002*x^2*i - 47.1183*x*y^2 $\\$-0.4059*x*y*i + 0.014*x*i^2 - 3529.0989*y^3 $\\$+ 23.9641*y^2*i + 2.4655*y*i^2 - 0.38*i^3 $\\$- 0.0401*x^2 + 43.7495*x*y + 0.318*x*i $\\$+ 86697.7958*y^2 - 25.0167*y*i - 0.5461*i^2 $\\$+ 3.2993*x - 83167.42*y - 0.0624*i - 5.0013$} & 2.37\% & 11158.9 \\
                \hline
                {\textsc{\makecell{inv-Pend \\ variant}}}& $[pA\le1]$ &   3  & 436.04 & \makecell{$\max\{[cP>0.5\vee pA <-0.1 \vee cP<-0.5\vee $\\$pA > 0.1]\cdot (0.0011*pAD^2*pA + $\\$0.0011*pAD^2*cV + \cdots + 0.999* cV $\\$- 0.0688*cP + 1.6061)+[cP\le 0.5 $\\$\wedge pA\le 0.1 \wedge cP \ge -0.5 \wedge cP \le 0.5], h^*\}$}  & 4 & 6.71 & \makecell{$-0.2235*pAD^4 - 1.1293*pAD^3*pA +$\\$ 0.1015*pAD^3*cV + 0.1091*pAD^3*cP - $\\$5.2183*pAD^2*pA^2 +\cdots - 10.4965*cP^2 +$\\$ 0.0001*pA - 53.2106*cV + 1.0$} & 1.18\% & 0.0\\
                \hline
                
                {\textsc{CAV-7}}& $[x\leq 30]$ &  3 &  5.17  &  \makecell{$\max\{[i < 5]\cdot (- 0.0001*i^2*x - 0.0001*i^2 $\\$+ 0.0006*i*x - 0.0001*x^2 + 0.0003*i+ $\\$ 0.0001*x + 0.9985) + [i \ge 5 \wedge x \le 30], h^* \}$}  &  4 & 0.78 &\makecell{$-0.0007*i^4 + 0.001*i^3*x - 0.0005*i^2*x^2 + $\\$ 0.0001*i*x^3 + 0.0044*i^3 - 0.0052*i^2*x $\\$ + 0.0011*i*x^2 - 0.0134*i^2 + 0.0121*i*x $\\$- 0.0019*x^2 + 0.0128*i - 0.004*x + 0.9966$}   &  55.83\% & 0.0021\\
                \hline
                
            \end{tabular}
        \end{threeparttable}}
\end{table*}

\begin{table*}
    \renewcommand{\arraystretch}{1.5}
        \label{table:comparison_lower_poly2}
        \resizebox{\textwidth}{!}{
		\begin{threeparttable}
            \begin{tabular}{|c|c|c|c|c|c|c|c|c|c|}
				\hline
				\multicolumn{1}{|c|}{\multirow{2}{*}{\textbf{Benchmark}}}  &
				\multicolumn{1}{c|}{\multirow{2}{*}{\textbf{$f$}}}      &
				\multicolumn{3}{c|}{\multirow{1}{*}{\textbf{Our Approach}}}  &
				\multicolumn{3}{c|}{\multirow{1}{*}{\makecell{\textbf{Monolithic Polynomial  via 1-induction}}}} &
                \multicolumn{1}{c|}{\multirow{2}{*}{\textbf{PCT}}} &
                \multicolumn{1}{c|}{\multirow{2}{*}{\textbf{Diff}}}\\ 
                \cline{3-8}
                \multicolumn{1}{|c|}{}  & \multicolumn{1}{c|}{} &  
                \multicolumn{1}{c|}{\textbf{d}}  &  \multicolumn{1}{c|}{\textbf{T(s)}} &  
				\multicolumn{1}{c|}{\textbf{Piecewise Polynomial lower Bound}} &\multicolumn{1}{c|}{\textbf{d}} &  \multicolumn{1}{c|}{\textbf{T(s)}} &   
				\multicolumn{1}{c|}{\textbf{Monolithic Polynomial lower Bound}} 
                & \multicolumn{1}{c|}{} & \multicolumn{1}{c|}{}  \\ \hline \hline
                {\textsc{cav-5}}& $[i \leq 10 ]$ &  3 &  897.32   &  \makecell{$\max\{[money \ge 10]\cdot (0.0009*i^2*money + $\\$0.0043*i*money^2 + 0.0013*money^3$\\$ - 0.9624*i^2 - 17.8205*i*money-$\\$ 66.2275*money^2 - 12.8062*i + $\\$118.2861*money - 1379.4033)$\\$ + [money < 10 \wedge i \le 10], h^* \}$}&  4 & 1.08 &\makecell{$-0.0001*i^2*money^2 - 0.0004*i*money^3 $\\$- 0.0002*money^4 - 0.001*i^3 + 0.0222*i^2 - $\\$0.0257*i^2*money + 0.0526*i*money^2 $\\$+ 0.0298*money^3  - 0.4528*i*money $\\$- 4.1462*money^2 - 3.6304*i + 1.0$}  &  50.0\% & -7178.7 \\ 
                \hline
                {\textsc{Add}}& $[x > 5]$ & 3 &  3.74 & \makecell{$\max\{[y \le 1]\cdot (0.0315 - 0.1456y + 0.0165*x $\\$+ 0.0620*y^2 - 0.0272*x*y + 0.0047*x^2$\\$ - 0.0004*y^3 - 0.0083*x*y^2 + 0.002*x^2*y$\\$ - 0.0002*x^3) + [y > 1\wedge x > 5], h^*\}$} & 4 & 0.57 & \makecell{$0.0938 - 3.1538*y + 0.0629*x $\\$+ 5.9923*y^2 - 0.005*x*y + 0.0091*x^2- $\\$ 3.976*y^3 - 0.1043*x*y^2 - 0.0008*x^2*y$\\$ - 0.0001*x^3 + 0.8737*y^4 + 0.0354*x*y^3 $\\$+ 0.0013*x^2*y^2 - 0.0001*x^3*y$} &  23.59\% & 0.302 \\
                \hline
                {\textsc{\makecell{Growing\\Walk\\Variant2}}} & $y$ & 2 & 4.83 & \makecell{$\max\{[r \le 0]\cdot (-0.0075*x^2 - 0.004*x*y $\\$- 0.0027*y^2 + 0.5230*x + 1.0174*y$\\$ - 0.0362) + [r > 0]\cdot y,  h^*\}$} & 3 & 1.09 & \makecell{$0.9772 - 21942.448*r + 1.0046*y + 0.9983*x $\\$+ 21709.1212*r^2 - 0.0038*y*r - 0.0013*y^2 $\\$- 19.0802*x*r + 0.0003*x*y + 0.002*x^2 + $\\$232.3409*r^3 + 0.0007*y*r^2 + 0.0007*y^2*r + $\\$18.0697xr^2 + 0.0007*x*y*r + 0.0023*x^2*r$\\$ + 0.0001*x^2*y - 0.0006*x^3$}     & 5.03\%    & 3488.2\\
                \hline
            \end{tabular}
        \end{threeparttable}}
\end{table*}

\smallskip
\noindent{\bf Answering RQ3.} Consistent with the upper case analysis, we compare our piecewise polynomial lower bounds with high-degree monolithic polynomial lower bounds derived through 1-induction, as summarized in \cref{table:comparison_lower_poly}. For consistency across benchmarks, we generate both bounds using identical invariants and optimal objective functions. In this setup, the degree of the monolithic polynomials is restricted to a maximum of 5.

We compare two results by uniformly taking the grid points in the invariant and evaluate two results, and we compute the percentage of the points that our piecewise lower bound are lower (i.e., no better) than (high degree) monolithic polynomial, which is shown in the column "PCT" in~\cref{table:comparison_lower_poly}. We also present difference plots that classify all grid points into three disjoint regions according to the magnitude of difference: red points correspond to cases where our piecewise linear lower bounds are notably larger (diff $ > 10^{-3}$), blue points correspond to cases where the monolithic lower bounds are notably larger, and gray points represent regions where the two bounds are nearly identical (diff $ \le 10^{-3}$). We display the comparison in~\cref{fig:diff_plots_lower2}. In addition, we quantify the difference between the two lower bounds by subtracting the piecewise lower bound from the monolithic one and taking the unbiased average of the resulting values, which is reported in the column “Diff” of~\cref{table:comparison_lower_poly}. 
From the comparison results "PCT" and "Diff" in~\cref{table:comparison_lower_poly} and difference plots~\cref{fig:diff_plots_lower2},  We observe that on all the benchmarks except~\textsc{cav-5}, our piecewise polynomial bounds are significantly tighter than monolithic polynomial bounds. On the benchmarks~\textsc{Sum0, DepRV, Prinsys, inv-Pend variant}, our piecewise results are almost identical to the (high degree) monolithic bounds, and in addition, on the benchmarks~\textsc{Bin0, DepRV, Bin2}, our piecewise bounds are almost identical to the (high degree) monolithic bounds, although the difference plots show discrepancies, a closer inspection reveals that the results for these examples are nearly identical as well. On the remaining benchmarks, our piecewise bounds are generally tighter than the monolithic ones, with especially notable improvements on benchmarks such as \textsc{chain, fig-7, grid big}.
Although our running time is also a bit longer than that of monolithic polynomial experiments, our approach allows to synthesize lower-degree polynomials while achieving better precision against higher-degree polynomials. This advantage is critical as the synthesis of higher-degree polynomials suffers from a large amount of numerical errors as stated previously. Thus our approach has a value to use lower-degree piecewise polynomials to surpass the numerical problem of higher-degree polynomials.

\subsection{Full Expressions for Experimental Results}

For readability and conciseness, some of the experimental results in the main text were partially omitted and denoted the symbol $ \cdots$. \cref{table:full} provides the complete expressions corresponding to those abbreviated entries in this appendix.

\begin{table*}
    \renewcommand{\arraystretch}{2.0}
    \caption{\textbf{Full Expressions} for \textbf{Abbreviated} (Main-Text) Results}
    \centering
    \resizebox{\textwidth}{!}{
    \begin{threeparttable}
    \begin{tabular}{|c|c|c|c|}
        \hline
         \textbf{Benchmark} & \textbf{Location}  & \textbf{Abbreviated Expressions} & \textbf{Full Expressions} \\ \hline \hline
         {\textsc{grid small}} & \makecell{\textbf{Piecewise} \\ \textbf{Polynomial}\\ \textbf{Upper} Bound \\ in \textbf{\cref{table:poly1}\&\ref{table:comparison_upper_poly}}} & \makecell{$\min\{[a<10 \wedge b < 10] \cdot (-0.0003*a^3 - $\\$0.0011b^3 -0.0008*a^2*b+ 0.0018ab^2  $\\$+ 0.0109*a^2 \cdots + 0.0277*b$\\$  + 0.5109)+[a<10 \wedge b\ge 10], h^*\}$} & \makecell{$\min\{[a<10 \wedge b < 10] \cdot (-0.0003*a^3 - 0.0011*b^3 -0.0008*a^2*b+ $\\$0.0018*a*b^2   + 0.0109*a^2  - 0.0144*a*b + 0.0129*b^2 - 0.0926*a + $\\$ 0.0277*b + 0.5109)+[a<10 \wedge b\ge 10], h^*\}.$}\\ \hline
         
         {\textsc{GeoAr}}& \makecell{\textbf{Monolithic} \\ \textbf{Polynomial}\\ \textbf{Upper} Bound \\ in \textbf{\cref{table:comparison_upper_poly}}} & \makecell{$-0.0001x^3 + 0.0001*x^2*y - 0.0011*x^2*z-$\\$ 0.0004xy^2 - 0.0112xyz + 0.164*x*z^2+$\\$ 0.0012*y^3+ \cdots - 0.0137*y^2 + 2.7194yz +$\\$0.9993*x + 0.0417*y + 89867.2768z + 0.078$}  & \makecell{$-0.0001*x^3 + 0.0001*x^2*y - 0.0011*x^2*z - 0.0004*x*y^2 - 0.0112*x*y*z + $\\$0.164*x*z^2 + 0.0012*y^3 + 0.0046*y^2*z - 1.8186*y*z^2+ 89866.1344*z^3 + $\\$0.0027*x*y - 0.1236*x*z - 0.0137*y^2 + 2.7194*y*z - 179731.0721*z^2$\\$ + 0.9993*x + 0.0417*y + 89867.2768*z + 0.078.$}\\ \hline
         
         {\textsc{fig-6}}& \makecell{Solution \textbf{$h^*$} of \\ \textbf{Upper} Bound \\ in \textbf{\cref{table:poly1}}} & \makecell{$ 0.0011*x^3*y -0.0001*x^4  - 0.0001*y^4  $\\$+ 0.0008*x*y^3  - 0.001*x^2*y^2 + \cdots $\\$+ 0.5712*x - 0.281*y + 0.6009$} & \makecell{$-0.0001*x^4 + 0.0011*x^3*y - 0.001*x^2*y^2 + 0.0008*x*y^3 - 0.0001*y^4$\\$ + 0.0016*x^3 - 0.0195*x^2*y + 0.006*x*y^2 - 0.003*y^3 - 0.0627*x^2 $\\$+ 0.1018*x*y - 0.0028*y^2 + 0.5712*x - 0.281*y + 0.6009.$}\\ \hline

         {\textsc{fig-6}}& \makecell{\textbf{Piecewise} \\ \textbf{Polynomial}\\ \textbf{Upper} Bound \\ in \textbf{\cref{table:poly1}\&\ref{table:comparison_upper_poly}}} & \makecell{$\min\{[x \le 4 ]\cdot (-0.0001*x^4 + 0.0011x^3y $\\$- 0.001x^2y^2 + 0.0008*x*y^3 - 0.0001y^4+ $\\$ 0.0023*x^3 \cdots - 0.0094y^2+ 0.5530x - $\\$ 0.2782y + 0.6027) + [x > 4 \wedge y \le 5], h^*\}$} & \makecell{$\min\{[x \le 4 ]\cdot (-0.0001*x^4 + 0.0011*x^3*y - 0.001*x^2*y^2 + 0.0008xy^3 - 0.0001y^4 +$\\$ 0.0023*x^3 - 0.0182*x^2*y + 0.0064*x*y^2 - 0.0026*y^3 - 0.0788*x^2 + 0.0913xy $\\$- 0.0094*y^2 + 0.5530*x - 0.2782*y + 0.6027) + [x > 4 \wedge y \le 5], h^*\}.$}\\ \hline

        {\textsc{fig-6}} & \makecell{\textbf{Monolithic} \\ \textbf{Polynomial}\\ \textbf{Upper} Bound \\ in \textbf{\cref{table:comparison_upper_poly}}} & \makecell{$-0.0001x^5 - 0.0002*x^4*y - 0.0003x^2y^3 $\\$+ 0.0001*x*y^4 - 0.0002*y^5 + 0.0011*x^4+$\\$  0.0037x^3y \cdots + 0.1432xy + 0.0064*y^2 $\\$+ 0.9708*x - 0.6526*y + 0.575$}& \makecell{$-0.0001*x^5 - 0.0002*x^4*y - 0.0003*x^2*y^3 + 0.0001*x*y^4 - $\\$0.0002*y^5 + 0.0011*x^4 + 0.0037*x^3*y - 0.0008*x^2*y^2 + 0.0021*x*y^3 + 0.0005*y^4 - $\\$0.0012*x^3 - 0.0361*x^2*y + 0.0088*x*y^2 - 0.0042*y^3 - 0.084*x^2 + 0.1432*x*y + $\\$0.0064*y^2 + 0.9708*x - 0.6526*y + 0.575.$}\\ \hline

        {\textsc{fig-7}} & \makecell{\textbf{Monolithic} \\ \textbf{Polynomial}\\ \textbf{Upper} Bound \\ in \textbf{\cref{table:comparison_upper_poly}}} & \makecell{$0.0003x^2i-0.083x^2y  + 48.5638*x*y^2+$\\$  0.5267xyi - 0.018*x*i^2 + 2600.9691*y^3 $\\$- 36.705y^2i - 2.646*y*i^2 \cdots- 3.3923*x $\\$+ 56310.8279*y - 0.0114*i + 7.2868$} & \makecell{$-0.083*x^2*y + 0.0003*x^2*i + 48.5638*x*y^2 + 0.5267*x*y*i - 0.018*x*i^2 + 2600.9691*y^3 $\\$- 36.705*y^2*i - 2.646*y*i^2 + 0.4053*i^3 + 0.0539*x^2 - 45.1036*x*y - 0.4109*x*i- $\\$58912.9534*y^2 + 37.7582*y*i + 0.6223*i^2 - 3.3923*x + 56310.8279*y - 0.0114*i + 7.2868.$}\\ \hline

        {\textsc{Duel}} & \makecell{\textbf{Piecewise} \\ \textbf{Polynomial}\\ \textbf{Lower} Bound \\ in \textbf{\cref{table:poly2}\&\ref{table:comparison_lower_poly}}}& \makecell{$\max \{[t>0\wedge x \ge 1]\cdot(10.8660x^2 +$\\$ 0.2353*x*t + \dots - 5.5451*x - $\\$ 0.8705*t + 0.2488) + [x<1]\cdot t, h^* \}$} & \makecell{$\max \{[t>0\wedge x \ge 1]\cdot(10.8660x^2 + 0.2353*x*t + 1.3703*t^2 - 11.0903*x$\\$ - 1.3703*t + 0.4987)+[t\le 0\wedge x\ge 1]\cdot(5.4330*x^2 + 0.1177*x*t + 1.3703*t^2 $\\$- 5.5451*x - 0.8705*t + 0.2488) + [x<1]\cdot t, h^* \}$} \\ \hline

        {\textsc{\makecell{inv-Pend \\ variant}}} & \makecell{Solution \textbf{$h^*$} of \\ \textbf{Upper} Bound \\ in \textbf{\cref{table:poly1}}} &  \makecell{$0.0058*pAD^2*pA + 0.0023*pAD^2*cV- $\\$0.1313*pAD^2*cP - 0.6278*pAD*pA^2-$\\$ 0.2352pAD*pA*cV\cdots + 5.9637cV*cP$\\$+ 60.4194*cP^2 + 7.1495*cV + 1.0$} & \makecell{$0.0058*pAD^2*pA + 0.0023*pAD^2*cV - 0.1313*pAD^2*cP - 0.6278*pAD*pA^2 -$\\$ 0.2352*pAD*pA*cV - 4.2984*pAD*pA*cP + 0.0034*pAD*cV^2 - 0.0776*pAD*cV*cP + $\\$0.2901*pAD*cP^2 - 3.3499*pA^3 + 1.2174*pA^2*cV - 18.4697*pA^2*cP + 0.8063*pA*cV^2 +$\\$ 7.4278*pA*cV*cP + 2.1607*pA*cP^2 + 0.1664*cV^3 + 0.0048*cV^2*cP - 0.5863*cV*cP^2 - $\\$101.7368*cP^3 + 0.7678*pAD^2 + 4.7849*pAD*pA - 0.1664*pAD*cV - 3.5565*pAD*cP +$\\$ 28.2784*pA^2 - 2.7311*pA*cV - 20.9853*pA*cP - 1.1597*cV^2 + 5.9637*cV*cP +$\\$ 60.4194*cP^2 - 0.0002*pA + 7.1495*cV + 0.001*cP + 1.0.$}\\ \hline
         
         {\textsc{\makecell{inv-Pend \\ variant}}}& \makecell{\textbf{Piecewise} \\ \textbf{Polynomial}\\ \textbf{Upper} Bound \\ in \textbf{\cref{table:poly1}\&\ref{table:comparison_upper_poly}}} & \makecell{$\min \{[cp>0.5 \vee cp < -0.5 \vee pA > 0.1 \vee $\\$  pA < -0.1] \cdot (0.0058*pAD^2*pA - $\\$0.0011*pAD^2*cV - 0.1313*pAD^2*cP $\\$+ \cdots + 0.0689cP + 0.3238) + $\\$[-0.5 \le cp \le 0.5 \wedge -0.1 \le pA \le 0.1], h^* \}$} & \makecell{$\min \{[cp>0.5 \vee cp < -0.5 \vee pA > 0.1 \vee   pA < -0.1] \cdot (0.0058*pAD^2*pA - 0.0011*pAD^2*cV - $\\$ 0.1313*pAD^2*cP - 0.6279*pAD*pA^2 - 0.2408*pAD*pA*cV - 4.2984*pAD*pA*cP - $\\$0.0124*pAD*cV^2 - 0.0021*pAD*cV*cP + 0.2901*pAD*cP^2 - 3.3498*pA^3 + 0.4776*pA^2*cV $\\$- 18.4697*pA^2*cP + 0.4734pA*cV^2 + 5.5455*pA*cV*cP + 2.1607*pA*cP^2 + 0.1014*cV^3 - $\\$0.0334*cV^2*cP - 3.4879*cV*cP^2 - 101.7368*cP^3 + 0.5916*pAD^2 + 4.0443*pAD*pA + $\\$0.0057pADcV - 3.5023*pAD*cP + 26.6426*pA^2 - 1.1436*pA*cV - 20.7584*pA*cP $\\$- 0.5132*cV^2 + 5.5468*cV*cP + 60.3921*cP^2 - 0.4489*pAD - 1.5038*pA + 5.2348*cV $\\$+ 0.0688*cP + 0.3238) + [-0.5 \le cp \le 0.5 \wedge -0.1 \le pA \le 0.1], h^* \}.$}\\ \hline

         {\textsc{\makecell{inv-Pend \\ variant}}} & \makecell{\textbf{Monolithic} \\ \textbf{Polynomial}\\ \textbf{Upper} Bound \\ in \textbf{\cref{table:comparison_upper_poly}}}  &\makecell{$0.2264*pAD^4 + 1.1448*pAD^3*pA $\\$- 0.1026*pAD^3*cV - 0.1107*pAD^3*cP +$\\$ 5.2869*pAD^2*pA^2 + \cdots + 10.6625*cP^2 $\\$- 0.0001*pA + 53.8573*cV + 1.0$} & \makecell{$0.2264pAD^4 + 1.1448*pAD^3*pA - 0.1026*pAD^3*cV - 0.1107*pAD^3*cP + 5.2869pAD^2*pA^2 $\\$- 0.4937*pAD^2*pA*cV - 0.8938*pAD^2*pA*cP + 0.3036*pAD^2*cV^2 + 0.0478*pAD^2*cV*cP$\\$ + 0.4208*pAD^2*cP^2 + 6.8201*pAD*pA^3 - 3.2518*pAD*pA^2*cV - 2.3942*pAD*pA^2*cP + $\\$1.3927pADpAcV^2 + 0.7868*pAD*pA*cV*cP + 4.5143*pAD*pA*cP^2 - 0.1912pAD*cV^3 -$\\$ 0.1023pADcV^2cP - 0.1906*pAD*cV*cP^2 - 2.8734*pAD*cP^3 + 53.6801*pA^4 + 1.323pA^3cV $\\$ - 6.8123pA^3cP + 5.2663*pA^2*cV^2 + 2.473*pA^2*cV*cP + 47.9517*pA^2*cP^2 - 0.5451pAcV^3 $\\$ - 0.7983pAcV^2cP - 0.9821*pA*cV*cP^2 - 20.6044*pA*cP^3 + 0.0986*cV^4 + 0.0333cV^3cP +$\\$ 0.3483cV^2cP^2 + 4.5559cVcP^3 + 30.7504*cP^4 - 0.3716*pAD^3 + 3.8*pAD^2*pA + 0.4985pAD^2cV -$\\$ 0.0606pAD^2cP - 9.7537pADpA^2 - 6.8904*pAD*pA*cV + 0.4876*pAD*pA*cP - 0.748pADcV^2 $\\$- 0.1874*pAD*cV*cP - 3.858*pAD*cP^2 - 11.7619*pA^3 + 18.5549*pA^2*cV - 0.4732*pA^2*cP $\\$+ 4.6011*pA*cV^2 - 0.8554*pA*cV*cP - 9.3429*pA*cP^2 + 1.9127*cV^3 + 0.0857*cV^2*cP + $\\$5.2916*cV*cP^2 - 3.7534*cP^3 + 4.1573*pAD^2 + 17.8582*pAD*pA + 0.2247*pAD*cV - $\\$2.5151*pAD*cP + 34.1498*pA^2 + 1.7805*pA*cV - 5.4381*pA*cP - 10.9045*cV^2  $\\$+ 1.196*cV*cP + 10.6625*cP^2 - 0.0001*pA + 53.8573*cV + 1.0.$}\\ \hline
         
         {\textsc{\makecell{inv-Pend \\ variant}}} & \makecell{Solution \textbf{$h^*$} of \\ \textbf{Lower} Bound \\ in \textbf{\cref{table:poly2}}} &  \makecell{$0.0008*pAD^2*pA - 0.0023*pAD^2*cV +$\\$ 0.0991*pAD^2*cP + 0.4931*pAD*pA^2 +$\\$0.1464*pAD*pA*cV \cdots - 5.002*cV*cP $\\$- 44.9405*cP^2 - 5.7109*cV + 1.0$} &\makecell{$0.0008pAD^2pA - 0.0023pAD^2cV + 0.0991pAD^2cP + 0.4931pADpA^2 + 0.1464pADpAcV+ 2.4945*pA^3$\\$ + 3.1026pAD*pA*cP - 0.0138*pAD*cV^2 + 0.0529*pAD*cV*cP - 0.2253*pAD*cP^2- $\\$ 1.1719*pA^2*cV + 13.2225*pA^2*cP - 0.66*pA*cV^2 - 5.8813*pA*cV*cP - 1.6276*pA*cP^2  $\\$ -0.0137*cV^2*cP - 0.2948*cV*cP^2 + 70.0898*cP^3 - 0.6053*pAD^2 - 3.6586*pAD*pA$\\$ + 0.2164*pAD*cV + 2.8831*pAD*cP - 20.411*pA^2 + 2.3725*pA*cV + 16.7142*pA*cP + $\\$0.9814*cV^2 - 0.1308*cV^3 -5.002*cV*cP - 44.9405*cP^2 - 5.7109*cV + 1.0.$} \\ \hline
         
         {\textsc{\makecell{inv-Pend \\ variant}}} & \makecell{\textbf{Piecewise} \\ \textbf{Polynomial}\\ \textbf{Lower} Bound \\ in \textbf{\cref{table:poly2}\&\ref{table:comparison_lower_poly}}} & \makecell{$\max\{[cP>0.5\vee pA <-0.1 \vee cP<-0.5\vee $\\$pA > 0.1]\cdot (0.0011*pAD^2*pA + $\\$0.0011*pAD^2*cV + \cdots + 0.999* cV $\\$- 0.0688*cP + 1.6061)+[cP\le 0.5 $\\$\wedge pA\le 0.1 \wedge cP \ge -0.5 \wedge cP \le 0.5], h^*\}$} & \makecell{$\max \{[cp>0.5 \vee cp < -0.5 \vee pA > 0.1 \vee   pA < -0.1] \cdot (0.0011*pAD^2*pA + 0.0011*pAD^2*cV + $\\$0.0983*pAD^2*cP + 0.5052*pAD*pA² + 0.1661*pAD*pA*cV + 3.1298*pAD*pA*cP + $\\$0.0035*pAD*cV^2 - 0.0028*pAD*cV*cP - 0.1801*pAD*cP^2 + 2.5528*pA^3 - 0.5263*pA^2*cV $\\$+ 13.3751*pA^2*cP - 0.3872*pA*cV^2 - 4.3942*pA*cV*cP - 1.4142*pA*cP^2 - 0.0803*cV^3 $\\$+ 0.0045*cV^2*cP + 1.8662*cV*cP^2 + 70.0747*cP^3 - 0.4694*pAD^2 - 3.0860*pAD*pA + $\\$0.0414pADcV + 2.8487*pAD*cP - 19.2280*pA^2 + 0.9998*pA*cV + 16.5897pAcP + 0.45cV^2 -$\\$4.5457*cV*cP - 44.9216*cP^2 + 0.3480*pAD + 1.1902*cV^2 + 1.1903*pA + 0.999*cV $\\$- 0.0688*cP + 1.6061) + [-0.5 \le cp \le 0.5 \wedge -0.1 \le pA \le 0.1], h^* \}.$}\\ \hline
         
         {\textsc{\makecell{inv-Pend \\ variant}}} & \makecell{\textbf{Monolithic} \\ \textbf{polynomial}\\ \textbf{lower} bound \\ in \textbf{\cref{table:comparison_lower_poly}}} & \makecell{$-0.2235*pAD^4 - 1.1293*pAD^3*pA +$\\$ 0.1015pAD^3*cV + 0.1091*pAD^3*cP - $\\$5.2183pAD^2*pA^2 +\cdots - 10.4965*cP^2 $\\$ + 0.0001*pA - 53.2106*cV + 1.0$} & 
         \makecell{$-0.2235*pAD^4 - 1.1293*pAD^3*pA + 0.1015*pAD^3*cV + 0.1091*pAD^3*cP - 5.2183*pAD^2*pA^2 $\\$ + 0.4869*pAD^2*pA*cV + 0.8825*pAD^2*pA*cP - 0.3*pAD^2*cV^2 - 0.0472*pAD^2*cV*cP  - $\\$6.7225*pAD*pA^3 + 3.227*pAD*pA^2*cV + 2.3658*pAD*pA^2*cP - 1.3787*pAD*pA*cV^2 - $\\$ 0.7789*pAD*pA*cV*cP - 4.4659*pAD*pA*cP^2 + 0.189*pAD*cV^3 + 0.1012*pAD*cV^2*cP + $\\$0.189*pAD*cV*cP^2 + 2.8456*pAD*cP^3 - 53.0957*pA^4 - 1.3037*pA^3*cV + 6.7519*pA^3*cP - $\\$5.2076pA^2*cV^2 -  2.4518*pA^2*cV*cP - 47.4808*pA^2*cP^2 + 0.54*pA*cV^3 + 0.7896pA*cV^2*cP $\\$+0.9738*pA*cV*cP^2 + 20.4078*pA*cP^3 - 0.0975*cV^4 - 0.033*cV^3*cP - 0.3448*cV^2*cP^2 - $\\$4.5124*cV*cP^3 - 30.4568*cP^4 + 0.3659*pAD^3 - 3.7578*pAD^2*pA - 0.4937*pAD^2*cV +$\\$ 0.06pAD^2*cP + 9.654pAD*pA^2 + 6.8253pAD*pA*cV - 0.4846*pAD*pA*cP + 0.7351pAD*cV^2$\\$ +0.1847*pAD*cV*cP + 3.8138*pAD*cP^2 + 11.671*pA^3 - 18.344*pA^2*cV + 0.4762*pA^2*cP - $\\$4.5653*pA*cV^2 + 0.8466*pA*cV*cP + 9.2637*pA*cP^2 - 1.8886*cV^3 - 0.0833*cV^2*cP - $\\$5.2318*cV*cP^2 + 3.7165*cP^3 - 4.1093*pAD^2 - 17.6591*pAD*pA - 0.204*pAD*cV +$\\$ 2.4836*pAD*cP - 33.745*pA^2 - 1.6881*pA*cV + 5.3915*pA*cP + 10.7733*cV^2 - $\\$0.4159*pAD^2*cP^2 - 1.1846*cV*cP - 10.4965*cP^2 + 0.0001*pA - 53.2106*cV + 1.0$}\\ \hline
    \end{tabular}
    \end{threeparttable}
    }

    \label{table:full}
\end{table*}

\subsection{Other Figures for Piecewise Linear Upper Experiments}\label{app:linear_upper_fig}

\begin{figure}[htbp]
  \centering
  \subfloat[\textsc{Mart}]
  {\includegraphics[width=0.2\textwidth]{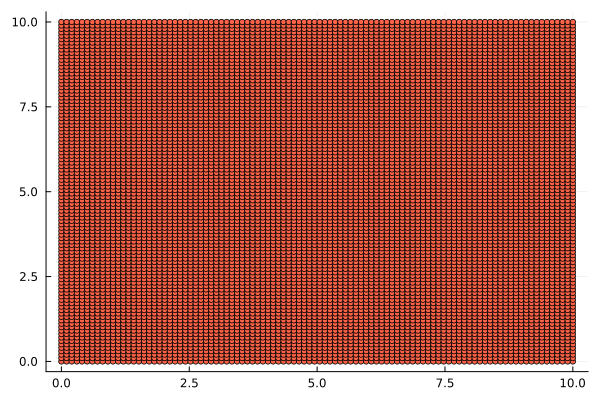}} \hfill   
  \subfloat[\parbox{0.15\textwidth}{\centering \textsc{Growing Walk \\ variant}}]
  {\includegraphics[width=0.2\textwidth]{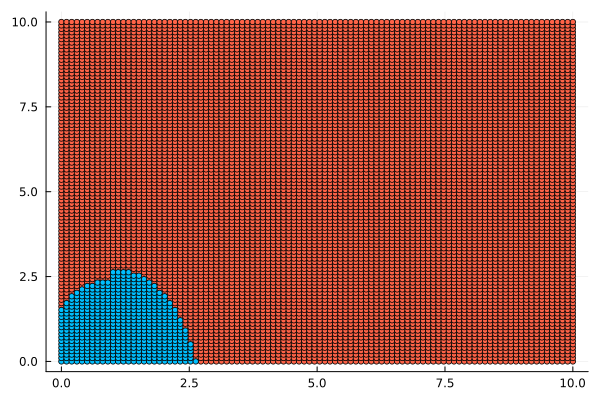}}  \hspace{0.1cm}
  \subfloat[\textsc{Coin}]
  {\includegraphics[width=0.31\textwidth]{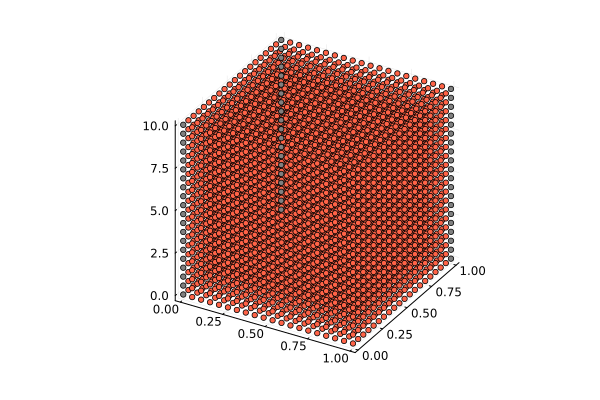}} \hspace{-0.8cm}
  \subfloat[\parbox{0.2\textwidth}{\centering \textsc{Equal Probability \\ Grid Family}}]
  {\includegraphics[width=0.31\textwidth]{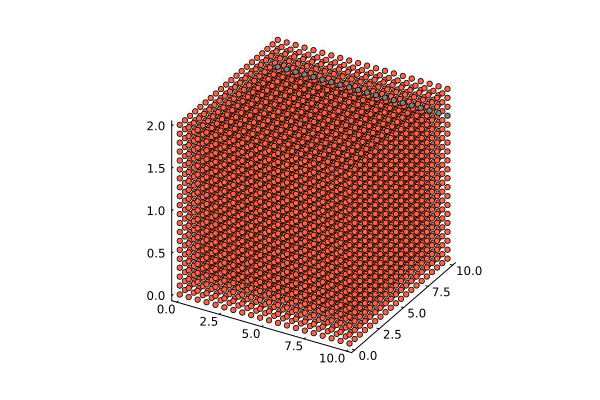}}\hspace{-0.5cm}  \\

  \subfloat[\textsc{Expected Time}]
  {\includegraphics[width=0.2\textwidth]{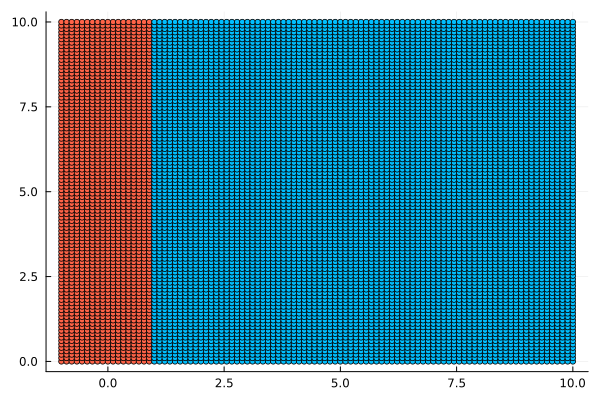}} \hfill   
  \subfloat[\textsc{RevBin}]
  {\includegraphics[width=0.2\textwidth]{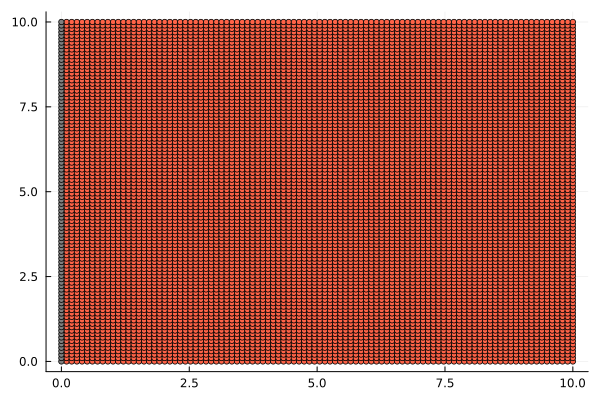}}  \hspace{0.1cm}
  \subfloat[\textsc{Fair Coin}]
  {\includegraphics[width=0.31\textwidth]{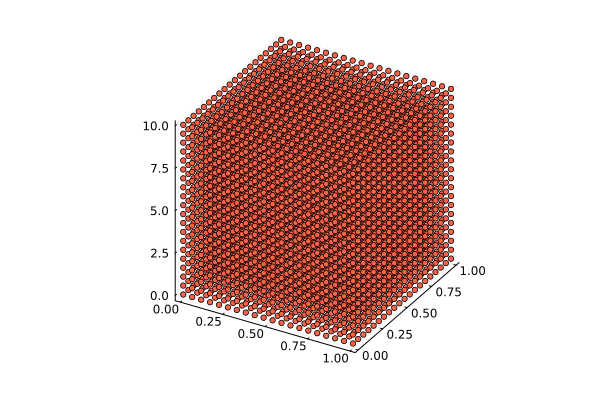}} \hspace{0.1cm}
  \subfloat[\parbox{0.15\textwidth}{\centering \textsc{St-Petersburg \\ variant}}]
  {\includegraphics[width=0.2\textwidth]{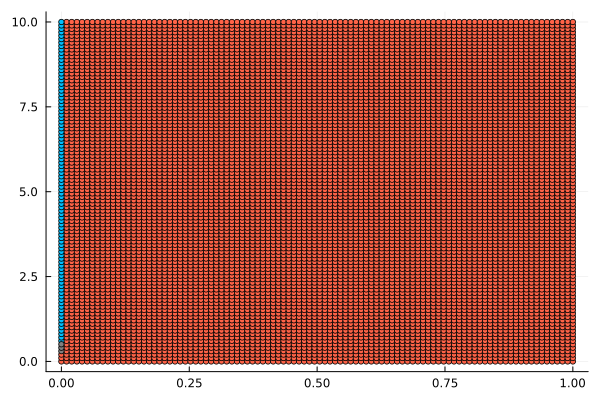}}\hspace{0.6cm}

  \caption{Other Difference Plots of the Comparison in Piecewise Linear Upper Case}
  \label{fig:diff_plots_upper1_other}
  
  \begin{tablenotes}
      \item \textcolor{red}{Red} points indicate where our piecewise upper bounds are smaller than the monolithic ones \\ by more than $10^{-3}$; 
      \textcolor{blue}{blue} points indicate where the monolithic bounds are smaller by more \\ than $10^{-3}$;
      \textcolor{gray}{gray} points denote cases with negligible differences ($\le 10^{-3}$).
  \end{tablenotes}
\end{figure}

\subsection{Other Figures for Piecewise Polynomial Upper Experiments}\label{app:poly_upper_fig}

\begin{figure}[htbp]
  \centering
  \subfloat[\textsc{Sum0}]
  {\includegraphics[width=0.2\textwidth]{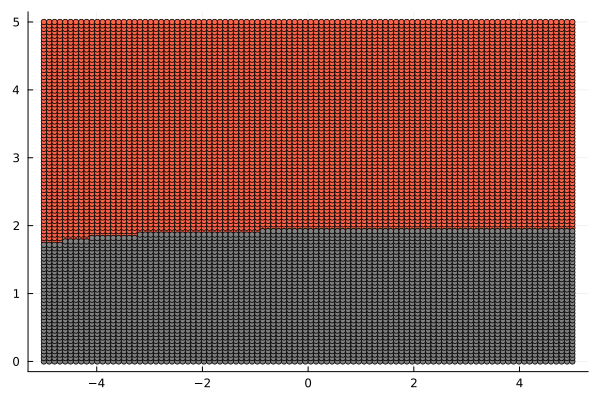}} \hfill   
  \subfloat[\textsc{Duel}]
  {\includegraphics[width=0.2\textwidth]{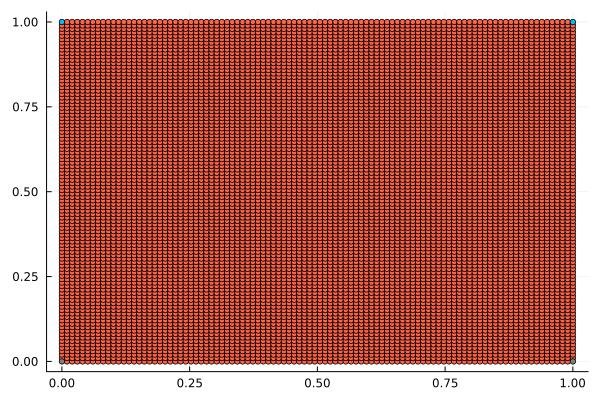}}  \hspace{0.1cm}
  \subfloat[\textsc{GeoAr}]
  {\includegraphics[width=0.31\textwidth]{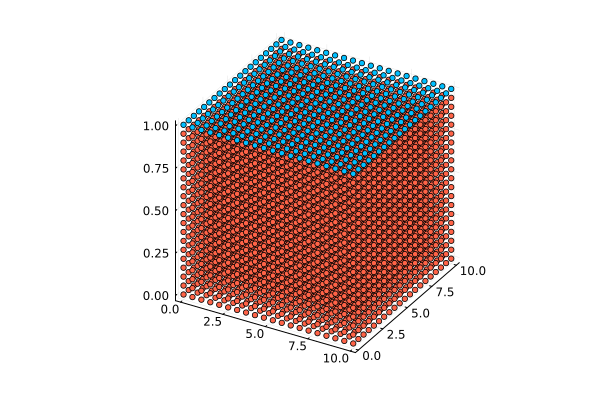}} \hspace{-0.8cm}
  \subfloat[\textsc{Bin0}]
  {\includegraphics[width=0.31\textwidth]{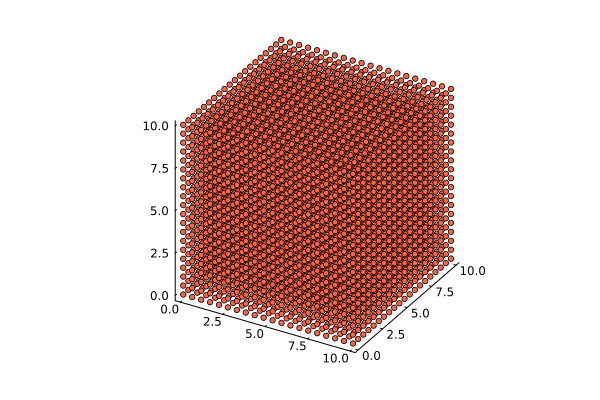}}\hspace{-0.5cm}  \\
  \subfloat[\textsc{brp}]
  {\includegraphics[width=0.2\textwidth]{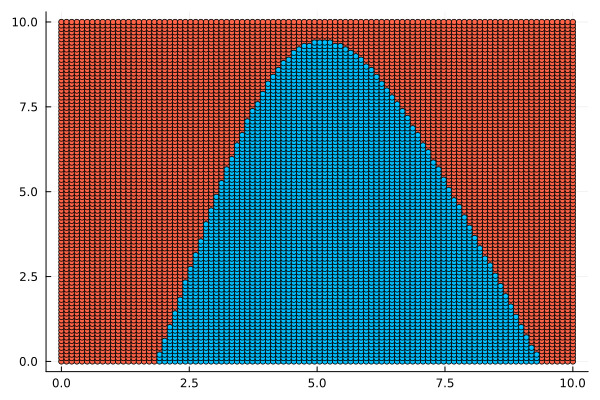}} \hfill
  \subfloat[\textsc{chain}]
  {\includegraphics[width=0.2\textwidth]{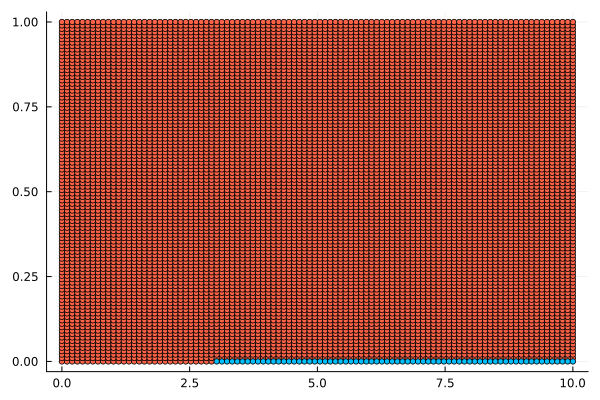}} \hspace{0.1cm}
  \subfloat[\textsc{Bin2}]
  {\includegraphics[width=0.31\textwidth]{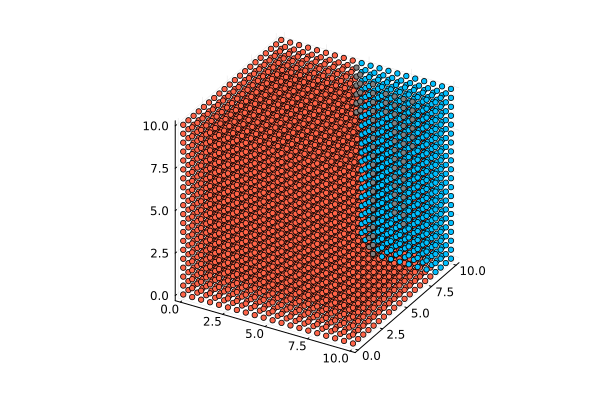}} \hspace{-0.8cm}
  \subfloat[\textsc{DepRV}]
  {\includegraphics[width=0.31\textwidth]{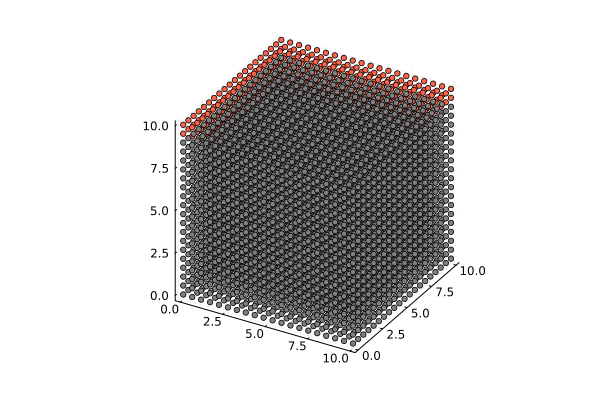}} \hspace{-0.5cm}  \\
  \subfloat[\textsc{grid small}]
  {\includegraphics[width=0.2\textwidth]{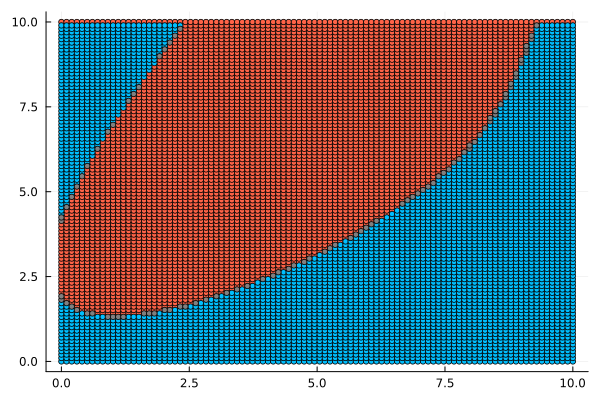}} \hfill   
  \subfloat[\textsc{grid big}]
  {\includegraphics[width=0.2\textwidth]{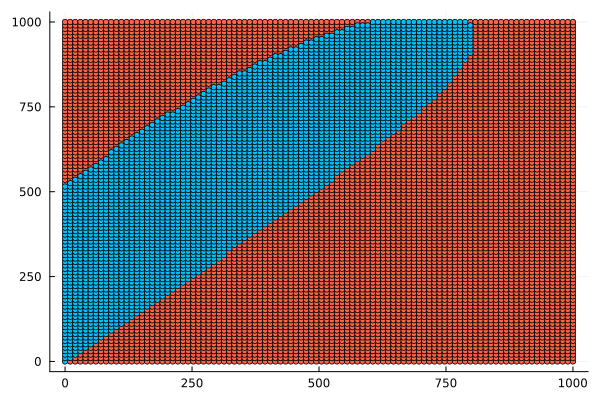}}  \hfill
  \subfloat[\textsc{cav-2}]
  {\includegraphics[width=0.2\textwidth]{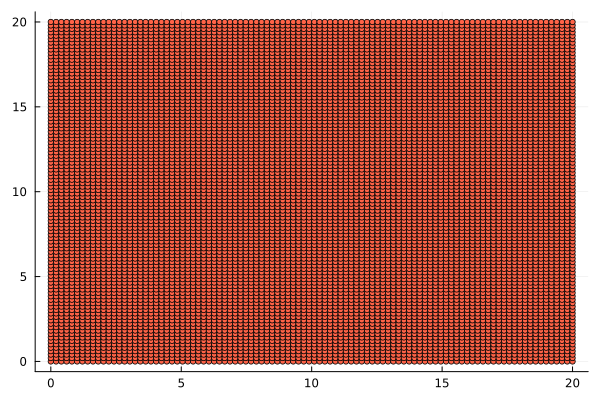}} \hfill   
  \subfloat[\textsc{cav-4}]
  {\includegraphics[width=0.2\textwidth]{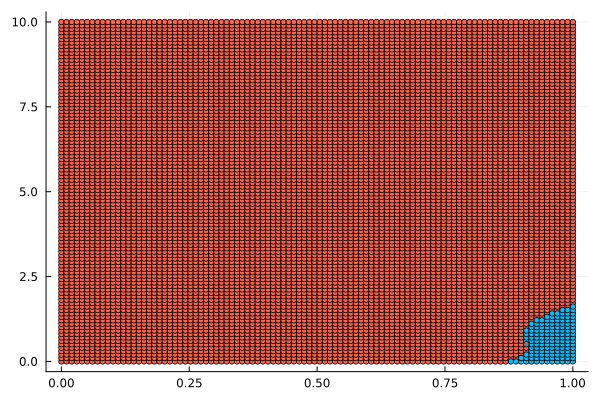}}  \hfill  \\
  \subfloat[\textsc{fig-6}]
  {\includegraphics[width=0.2\textwidth]{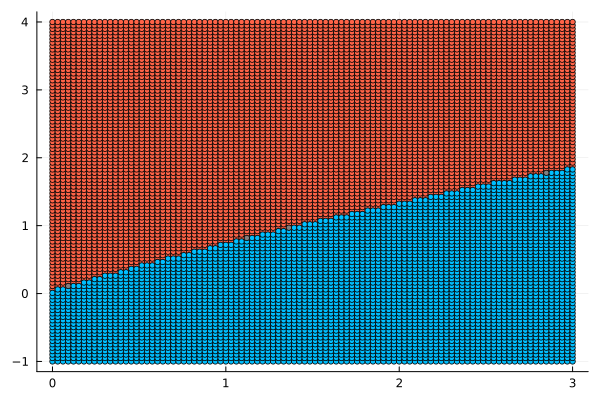}} \hspace{0.75cm}    
  \subfloat[\textsc{CAV-7}]
  {\includegraphics[width=0.2\textwidth]{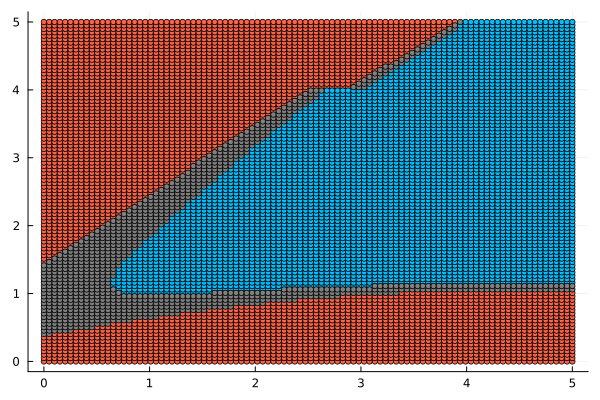}}  \hfill  

  \caption{Other Difference Plots of the Comparison in Piecewise Polynomial Upper Case}
  \label{fig:diff_plots_upper2_other}
  
  \begin{tablenotes}
      \item \textcolor{red}{Red} points indicate where our piecewise upper bounds are smaller than the monolithic ones \\ by more than $10^{-3}$; 
      \textcolor{blue}{blue} points indicate where the monolithic bounds are smaller by more \\ than $10^{-3}$;
      \textcolor{gray}{gray} points denote cases with negligible differences ($\le 10^{-3}$).
  \end{tablenotes}
\end{figure}

\section{Benchmark Programs}\label{app:programs}
This section presents the benchmark programs used in our experiments, along with the invariants employed in our algorithms. In addition, we show the results of checking the prerequisites of~\cref{thm:soundness,thm:soundness_poly}(P2), as discussed in~\cref{sec:experiment}.

\subsection{Programs in Linear Experiments}
This section contains the benchmark programs in our linear experiments, i.e., in~\cref{table:1,table:2}.

\begin{example}[\textnormal{Geo}]
    \begin{align*}
    \cc_{\textnormal{Geo}}\colon
    \quad
    &\WHILESYMBOL\left(\,0\leq c\,\right)\{\\
    &\quad	\{\ASSIGN{c}{1}\}~[0.5]~\{\ASSIGN{x}{x+1}\}\\
    &\}
\end{align*}
\end{example}

In this probabilistic program, we take the \textbf{invariant} $ I = 0\le x$, and every loop iteration terminates directly with probability $p = 0.5$. 

\begin{example}[\textnormal{k-geo}]
    \begin{align*}
        \cc_{\textnormal{k-geo}}\colon
        \quad
        &\WHILESYMBOL\left(\,k\leq N\,\right)\{\\
        &\quad	\{\ASSIGN{k}{k+1};\ASSIGN{y}{y+x};\ASSIGN{x}{0}\}~[0.5]~\{\ASSIGN{x}{x+1}\}\\
        &\}
    \end{align*}
\end{example}

In this probabilistic program, we take the \textbf{invariant} $ I = 0\le x \wedge 0 \le y \wedge k \le N+1$ , and synthesize dbRSM $k-N$.

\begin{example}[\textnormal{Binomial-random}]
    \begin{align*}
        \cc_{\textnormal{Bin-ran}}\colon
        \quad
        &\WHILESYMBOL\left(\,i\leq 10\,\right)\{\\
        &\quad  \{\ASSIGN{x}{x+1}\}~[0.5]~\{\ASSIGN{x}{0}\}\\
        &\quad	\{\ASSIGN{y}{y+x};\ASSIGN{i}{i+1}\}~[0.9]~\{\ASSIGN{y}{y+1};\ASSIGN{i}{0}\}\\
        &\}
    \end{align*}
\end{example}

In this probabilistic program, we take the \textbf{invariant} $ I = 0\le i \le 11 \wedge 0\le x \wedge 0 \le y $, and there is a probability $p \geq 0.9^{10}$ that the program will terminate immediately for every ten loop iterations.

\begin{example}[\textnormal{Coin}]
    \begin{align*}
        \cc_{\textnormal{Coin}}\colon
        \quad
        &\WHILESYMBOL\left(\,x=y\,\right)\{\\
        &\quad  \{\ASSIGN{x}{0}\}~[\nicefrac{3}{4}]~\{\ASSIGN{x}{1}\}\\
        &\quad	\{\ASSIGN{y}{0}\}~[\nicefrac{3}{4}]~\{\ASSIGN{y}{1}\}\\
        &\quad  \ASSIGN{i}{i+1};\\
        &\}
    \end{align*}
\end{example}

In this probabilistic program, we take the \textbf{invariant} $ I = 0\le i \wedge 0\le x \le 1\wedge 0 \le y \le 1$, and every loop iteration terminates directly with probability $p = \frac{5}{8}$.

\begin{example}[\textnormal{Martingale}]
    \begin{align*}
        \cc_{\textnormal{Mart}}\colon
        \quad
        &\WHILESYMBOL\left(\,0< x\,\right)\{\\
        &\quad	\{\ASSIGN{y}{y+x};\ASSIGN{x}{0}\}~[0.5]~\{\ASSIGN{y}{y-x};\ASSIGN{x}{2*x}\}\\
        &\quad  \ASSIGN{i}{i+1};\\
        &\}
    \end{align*}
\end{example}

In this probabilistic program, we take the \textbf{invariant} $ I = 0\le x$, and every loop iteration terminates directly with probability $p = 0.5$.

\begin{example}[\textnormal{Growing Walk}]
    \begin{align*}
        \cc_{\textnormal{Growing Walk}}\colon
        \quad
        &\WHILESYMBOL\left(\,0\leq x\,\right)\{\\
        &\quad	\{\ASSIGN{x}{x+1};\ASSIGN{y}{y+x}\}~[0.5]~\{\ASSIGN{x}{-1}\}\\
        &\}
    \end{align*}
\end{example}

In this probabilistic program, we take the \textbf{invariant} $ I =  -1\le x$, and every loop iteration terminates directly with probability $p = 0.5$.

\begin{example}[\textnormal{Growing Walk variant1}]
    \begin{align*}
        \cc_{\textnormal{Growing Walk1}}\colon
        \quad
        &\WHILESYMBOL\left(\,0\leq x\,\right)\{\\
        &\quad	\{\ASSIGN{x}{x-1};\ASSIGN{y}{y+x}\}~[0.5]~\{\ASSIGN{x}{1}\}\\
        &\}
    \end{align*}
\end{example}

In this probabilistic program, we take the \textbf{invariant} $ I =  -1\le x$, and every loop iteration terminates directly with probability $p = 0.5$.

\begin{example}[\textnormal{Expected Time}]
    \begin{align*}
        \cc_{\textnormal{Expected Time}}\colon
        \quad
        &\WHILESYMBOL\left(\,0\leq x\,\right)\{\\
        &\quad	\{\ASSIGN{x}{x-1};\ASSIGN{t}{t+1}\}~[0.9]~\{\ASSIGN{x}{10};\ASSIGN{t}{t+1}\}\\
        &\}
    \end{align*}
\end{example}

In this probabilistic program, we take the \textbf{invariant} $ I =  -1\le x \le 10$, and there is a probability $p \geq 0.9^{10}$ that the program will terminate immediately for every ten loop iterations.

\begin{example}[\textnormal{Zero Conference variant}]
    \begin{align*}
    \cc_{\textnormal{Zero-Conf-Var}}\colon
    \quad
    &\WHILESYMBOL\left(\,\textit{established}\leq 0  \wedge \textit{start} \leq 1\,\right)\{\\
    &\quad \IFSYMBOL\left(\, \textit{start} \geq 1 \,\right)
    \{\,\\
    &\quad \quad \{\ASSIGN{\textit{start}}{0}\}~[0.3]~\{\ASSIGN{\textit{start}}{0};\ASSIGN{\textit{established}}{1}\,\} \, \}\\
    &\quad \ELSESYMBOL\,\{\,
    \{\ASSIGN{\textit{curprobe}}{\textit{curprobe}+1}\}~[0.99]~\{\ASSIGN{\textit{start}}{1};\ASSIGN{\textit{curprobe}}{\textit{curprobe}-1}\} \, \} \,\\
    &\}
\end{align*}
\end{example}

In this probabilistic program, we take the \textbf{invariant} $ I = 0 \le start \le 1 \wedge 0 \le est \le 1$, and for the prerequisite (P2) checking, when $\textit{start}=1$, the loop iteration terminates directly with probability $p = 0.7$. When $\textit{start}=0$, the value of $\textit{start}$ has the probability of 0.01 to become 1 and turn to the branch of $\textit{start}=1$.

\begin{example}[\textnormal{Equal Probability Grid Family}]
    \begin{align*}
        \cc_{\textnormal{Equal-Prob-Grid-Family}}\colon
        \quad
        &\WHILESYMBOL\left(\, a\leq 10 \wedge b\leq 10 \wedge \textit{goal} = 0 \,\right)\{ \\
        &\quad \IFSYMBOL\left(\, b \geq 10 \,\right)\{\,\\
        &\quad \quad \{\ASSIGN{\textit{goal}}{1} \}~[0.5]~ \{\ASSIGN{\textit{goal}}{2} \} \, \}\\
        &\quad \ELSESYMBOL\,\{\,\\
        &\quad \quad \IFSYMBOL\left(\, a \geq 10 \,\right)\{\,\\
        &\quad \quad \quad \ASSIGN{a}{a-1} \, \}\\
        &\quad \quad \ELSESYMBOL\,\{\,\\
        &\quad \quad \quad \{\ASSIGN{a}{a+1} \}~[0.5]~ \{\ASSIGN{b}{b+1} \} \, \}\\
        &\}
    \end{align*}
\end{example}

In this probabilistic program, we take the \textbf{invariant} $ I =  0 \le a \le 10 \wedge 0 \le b \le 10 \wedge goal \ge 0$, and we synthesize dbRSM $10-b$.

\begin{example}[\textnormal{RevBin}]
    \begin{align*}
        \cc_{\textnormal{RevBin}}\colon
        \quad
        &\WHILESYMBOL\left(\,1\leq x\,\right)\{\\
        &\quad	\{\ASSIGN{x}{x-1};\ASSIGN{z}{z+1}\}~[0.5]~\{\ASSIGN{z}{z+1}\}\\
        &\}
    \end{align*}
\end{example}

In this probabilistic program, we take the \textbf{invariant} $ I =  0 \le x$, and we synthesize dbRSM $x$.

\begin{example}[\textnormal{Fair Coin}]
    \begin{align*}
        \cc_{\textnormal{Fair Coin}}\colon
        \quad
        &\WHILESYMBOL\left(\,x\leq 0 \wedge y\leq 0\,\right)\{\\
        &\quad  \{\ASSIGN{x}{0}\}~[0.5]~\{\ASSIGN{x}{1};\ASSIGN{i}{i+1}\}\\
        &\quad	\{\ASSIGN{y}{0}\}~[0.5]~\{\ASSIGN{y}{1};\ASSIGN{i}{i+1}\}\\
        &\}
    \end{align*}
\end{example}

In this probabilistic program, we take the \textbf{invariant} $ I =  0\le x \le 1 \wedge 0 \le y \le 1$, and every loop iteration terminates directly with probability $p = 0.25$.

\begin{example}[\textnormal{Bernoulli's St. Petersburg Paradox variant}]
    \begin{align*}
    \cc_{\textnormal{St. Petersburg1}}\colon
    \quad
    &\WHILESYMBOL\left(\,x\leq 0\,\right)\{\\
    &\quad	\{\ASSIGN{x}{1}\}~[0.75]~\{\ASSIGN{y}{2*y}\}\\
    &\}
\end{align*}
\end{example}

In this probabilistic program, we take the \textbf{invariant} $ I =  0 \le x \le 1 \wedge y \le 0$, and every loop iteration terminates directly with probability $p = 0.75$.

\subsection{Programs in Polynomial Experiments}
This section contains the benchmarks in our polynomial experiments, i.e., in~\cref{table:poly1,table:poly2}.  

\begin{example}[\textnormal{GeoAr}]
    \begin{align*}
        \cc_{\textnormal{GeoAr}}\colon
        \quad
        &\WHILESYMBOL\left(\,0 < z\,\right)\{\\
        &\quad \ASSIGN{y}{y+1}; \\
        &\quad	\{\ASSIGN{x}{x+y}\}~[0.9]~\{\ASSIGN{z}{0}\}\\
        &\}
    \end{align*}
\end{example}

In this probabilistic program, we take the \textbf{invariant} $ I =  0 \le x \wedge 0 \le y \wedge 0 \le z$, and every loop iteration terminates directly with probability $p = 0.1$.

\begin{example}[\textnormal{Bin0}]
    \begin{align*}
        \cc_{\textnormal{Bin0}}\colon
        \quad
        &\WHILESYMBOL\left(\,n > 0\,\right)\{\\
        &\quad	\{\ASSIGN{x}{x+y};\ASSIGN{n}{n-1}\}~[0.5]~\{\ASSIGN{n}{n-1}\}\\
        &\}
    \end{align*}
\end{example}

In this probabilistic program, we take the \textbf{invariant} $ I =  0 \le x \wedge 0 \le y \wedge 0 \le n$, and synthesize dbRSM $n$.

\begin{example}[\textnormal{Bin2}]
    \begin{align*}
        \cc_{\textnormal{Bin2}}\colon
        \quad
        &\WHILESYMBOL\left(\,n > 0\,\right)\{\\
        &\quad	\{\ASSIGN{x}{x+1};\ASSIGN{n}{n-1}\}~[0.5]~\{\ASSIGN{x}{x+y};\ASSIGN{n}{n-1}\}\\
        &\}
    \end{align*}
\end{example}

In this probabilistic program, we take the \textbf{invariant} $ I =  0 \le x \wedge 0 \le y \wedge 0 \le n$,  and we synthesize dbRSM $n$.

\begin{example}[\textnormal{DepRV}]
    \begin{align*}
        \cc_{\textnormal{DepRV}}\colon
        \quad
        &\WHILESYMBOL\left(\,n > 0\,\right)\{\\
        &\quad	\{\ASSIGN{x}{x+1};\ASSIGN{n}{n-1}\}~[0.5]~\{\ASSIGN{y}{y+1};\ASSIGN{n}{n-1}\}\\
        &\}
    \end{align*}
\end{example}

In this probabilistic program, we take the \textbf{invariant} $ I =  0 \le x \wedge 0 \le y \wedge 0 \le n$, and synthesize dbRSM $n$.

\begin{example}[\textnormal{Prinsys}]
    \begin{align*}
        \cc_{\textnormal{Prinsys}}\colon
        \quad
        &\WHILESYMBOL\left(\,x = 0\,\right)\{\\
        &\quad	\{\ASSIGN{x}{0}\}~[0.5]~\{ \{\ASSIGN{x}{-1}\}~[0.5]~\{\ASSIGN{x}{1}\}\} \\
        &\}
    \end{align*}
\end{example}

In this probabilistic program, we take the \textbf{invariant} $ I =  -1 \le x \le 1$, and every loop iteration terminates directly with probability $p = 0.5$.

\begin{example}[\textnormal{Sum0}]
    \begin{align*}
        \cc_{\textnormal{Sum0}}\colon
        \quad
        &\WHILESYMBOL\left(\,n > 0\,\right)\{\\
        &\quad	\{\ASSIGN{x}{x+n};\ASSIGN{n}{n-1}\}~[0.5]~\{\ASSIGN{n}{n-1}\}\\
        &\}
    \end{align*}
\end{example}

In this probabilistic program, we take the \textbf{invariant} $ I =  i \ge 0$, and synthesize dbRSM $n$.

\begin{example}[\textnormal{Duel Boy}]
    \begin{align*}
    \cc_{\textnormal{Duel}}\colon
    \quad
    &\WHILESYMBOL\left(\,x \geq 1 \right)\{ \\
    &\quad  \IFSYMBOL\left(\, t > 0 \,\right)
    \{\,\\
    &\quad \quad \{\ASSIGN{\textit{x}}{0}\}~[0.5]~\{\ASSIGN{\textit{t}}{1-t} \} \,  \\
    &\quad \} \ELSESYMBOL \{\{\ASSIGN{\textit{x}}{0}\}~[0.75]~\{\ASSIGN{\textit{t}}{1-t}\} \} \\
    &\} 
    \end{align*}
\end{example}

In this probabilistic program, we take the \textbf{invariant} $ I =  0 \le x \le 1 \wedge 0 \le t \le 1$, and every loop iteration terminates directly with probability $p\ge 0.5$.

\begin{example}[\textnormal{brp}]
    \begin{align*}
        \cc_{\textnormal{brp}}\colon
        \quad
        &\WHILESYMBOL\left(\,sent < 800 \wedge failed < 10\,\right)\{\\
        &\quad	\{\ASSIGN{sent}{sent+1};failed=0\}~[0.99]~\{\ASSIGN{failed}{failed+1}\}\\
        &\}
    \end{align*}
\end{example}

In this probabilistic program, we take the \textbf{invariant} $ I =  0 \le failed \le 10 \wedge 0 \le sent \le 800$, and there is a probability $p = 0.01^{10}$ that the program will terminate immediately for every ten loop iterations.

\begin{example}[\textnormal{chain}]
    \begin{align*}
        \cc_{\textnormal{chain}}\colon
        \quad
        &\WHILESYMBOL\left(\,y \le 0 \wedge x < 100\,\right)\{\\
        &\quad	\{\ASSIGN{y}{1}\}~[0.01]~\{\ASSIGN{x}{x+1}\}\\
        &\}
    \end{align*}
\end{example}

In this probabilistic program, we take the \textbf{invariant} $ I =  0 \le x \le 100 \wedge 0 \le y \le 1$, and every loop iteration terminates directly with probability $p=0.01$.

\begin{example}[\textnormal{grid-small}]
    \begin{align*}
        \cc_{\textnormal{grid-small}}\colon
        \quad
        &\WHILESYMBOL\left(\,a< 10 \wedge b <10\,\right)\{\\
        &\quad	\{\ASSIGN{a}{a+1}\}~[0.5]~\{\ASSIGN{b}{b+1}\}\\
        &\}
    \end{align*}
\end{example}

In this probabilistic program, we take the \textbf{invariant} $ I =  0 \le a \le 11 \wedge 0 \le b \le 11$, and synthesize dbRSM $19-(a+b)$.

\begin{example}[\textnormal{grid-big}]
    \begin{align*}
        \cc_{\textnormal{grid-big}}\colon
        \quad
        &\WHILESYMBOL\left(\,a< 1000 \wedge b <1000\,\right)\{\\
        &\quad	\{\ASSIGN{a}{a+1}\}~[0.5]~\{\ASSIGN{b}{b+1}\}\\
        &\}
    \end{align*}
\end{example}

In this probabilistic program, we take the \textbf{invariant} $ I =  0 \le a \le 1001 \wedge 0 \le b \le 1001$, and synthesize dbRSM $1999-(a+b)$.    

\begin{example}[\textnormal{cav-2}]
    \begin{align*}
        \cc_{\textnormal{cav-2}}\colon
        \quad
        &\WHILESYMBOL\left(\,h \le t\,\right)\{\\
        &\quad	\{\ASSIGN{h}{h+10}\}~[0.25]~\{\texttt{skip}\};\\
        &\quad \{\ASSIGN{t}{t+1}\}\\
        &\}
    \end{align*}
\end{example}

In this probabilistic program, we take the \textbf{invariant} $ I =  0 \le t \wedge 0 \le h \wedge h \ge t + 1$, and synthesize dbRSM $t-h$.

\begin{example}[\textnormal{cav-4}]
    \begin{align*}
        \cc_{\textnormal{cav-4}}\colon
        \quad
        &\WHILESYMBOL\left(\,y \ge 1\,\right)\{\\
        &\quad	\{\ASSIGN{y}{1}\}~[0.5]~\{\ASSIGN{y}{0}\};\\
        &\quad \{\ASSIGN{x}{x+1}\}\\
        &\}
    \end{align*}
\end{example}

In this probabilistic program, we take the \textbf{invariant} $I = 0\le y\le 1 \wedge x \ge 0$, and every loop iteration terminates directly with probability $p = 0.5$.

\begin{example}[\textnormal{fig-6}]
    \begin{align*}
        \cc_{\textnormal{fig-6}}\colon
        \quad
        &\WHILESYMBOL\left(\,x \le 4\,\right)\{\\
        &\quad	\{\ASSIGN{x}{x-1}\}~[0.5]~\{\ASSIGN{x}{x+3}\};\\
        &\quad \{\texttt{skip}\}~[0.3333]~\{\{\ASSIGN{y}{y+1}\}~[0.5]~\{\ASSIGN{y}{y+2}\}\};\\
        &\}
    \end{align*}
\end{example}

In this probabilistic program, we take the \textbf{invariant} $y \ge 0 \wedge x \le 7$, and synthesize dbRSM $4-x$.

\begin{example}[\textnormal{fig-7}]
    \begin{align*}
        \cc_{\textnormal{fig-7}}\colon
        \quad
        &\WHILESYMBOL\left(\,y \le 0\,\right)\{\\
        &\quad	\{\ASSIGN{y}{0}\}~[0.5]~\{\ASSIGN{y}{1}\};\\
        &\quad \ASSIGN{x}{2*x};\\
        &\quad \ASSIGN{i}{i+1};\\
        &\}
    \end{align*}
\end{example}

In this probabilistic program, we take the \textbf{invariant} $I = i \ge 0 \wedge x > 0 \wedge 0\le y \le 1$, and every loop iteration terminates directly with probability $p = 0.5$.

\begin{example}[\textnormal{inv-Pend variant}]
    \begin{align*}
        \cc_{\textnormal{inv-Pend variant}}\colon
        \quad
        &\WHILESYMBOL\left(\, exitcond \leq 0\,\right)\{\\
        &\quad \IFSYMBOL \left(-0.5 \le cP \right)\{\\
        &\quad \quad \IFSYMBOL \left(cP \le 0.5 \right)\{\\
        &\quad \quad \quad \IFSYMBOL \left(-0.1 \le pA \right)\{\\
        &\quad \quad \quad \quad \IFSYMBOL \left(pA \le 0.1 \right)\{\\
        &\quad \quad \quad \quad \quad \ASSIGN{exitcond}{1};\\
        &\quad \quad \quad \quad \}\ELSESYMBOL \{\texttt{skip}\}\\
        &\quad \quad \quad \}\ELSESYMBOL \{\texttt{skip}\}\\
        &\quad \quad  \}\ELSESYMBOL \{\texttt{skip}\}\\
        &\quad \}\ELSESYMBOL \{\texttt{skip}\}\\
        \\
        &\quad \ASSIGN{cP}{cP+0.01*cV};\\
        &\quad \{\ASSIGN{cV}{0.02*cP+0.5*cV-0.3*pA-0.06*pAD-1}\}~[0.5]~\\
        &\quad \{\ASSIGN{cV}{0.02*cP+0.5*cV-0.3*pA-0.06*pAD+1}\};\\
        &\quad \ASSIGN{pA}{pA +0.01*pAD};\\
        &\quad \{\ASSIGN{pAD}{0.04*cP+0.07*cV-0.51*pA+0.85*pAD-0.8}\}~[0.5]~\\
        &\quad \{\ASSIGN{pAD}{0.04*cP+0.07*cV-0.51*pA+0.85*pAD+0.8}\};\\
        &\}
    \end{align*}
\end{example}

In this probabilistic program, we take the \textbf{invariant} $I = cV \ge 0$, and synthesize a dbRSM $0.7747 * cP^2 + 0.0004 * cV^2 + 0.0222 * pA^2 + 0.0005 * pAD^2 + 0.0298 * cP * cV - 0.0919 * cP * pA - 0.0168 * cP * pAD - 0.0019 * cV * pA - 0.0003 * cV * pAD + 0.0014 * pA * pAD$ (cut to $10^{-4}$ precision).

\begin{example}[\textnormal{CAV-7}]
    \begin{align*}
        \cc_{\textnormal{CAV-7}}\colon
        \quad
        &\WHILESYMBOL\left(\,i \leq 4\,\right)\{\\
        &\quad	\{\ASSIGN{x}{x+1};\ASSIGN{i}{i+1}\}~[1-0.2*i]~\{\ASSIGN{x}{x+1}\}\\
        &\}
    \end{align*}
\end{example}

In this probabilistic program, we take the \textbf{invariant} $I = 0 \le i \le 5 \wedge 0 \le x$, and synthesize dbRSM $-i$.

\begin{example}[\textnormal{cav-5}]
    \begin{align*}
        \cc_{\textnormal{cav-5}}\colon
        \quad
        &\WHILESYMBOL\left(\,10 \leq \textit{money}\,\right)\{\\
        &\quad \{\ASSIGN{bet}{5}\}~[0.5]~\{\ASSIGN{bet}{10}\};\\
        &\quad \ASSIGN{money}{money-bet};\\
        &\quad \PASSIGN{bank\_guard}{\Uniform{(0.0, 1.0)}}\\
        &\quad \IFSYMBOL \left(bank\_guard \le 0.94737\right) \{\\
        &\quad \quad \PASSIGN{col1\_guard}{\Uniform{(0.0, 1.0)}};\\
        &\quad \quad \IFSYMBOL \left(col1\_guard \le 0.33333\right) \{\\
        &\quad \quad \quad \PASSIGN{flip\_guard1}{\Uniform{(0.0, 1.0)}};\\
        &\quad \quad \quad \IFSYMBOL \left(flip\_guard1 \le 0.5\right) \{\\
        &\quad \quad \quad \quad \ASSIGN{money}{money + 1.5 * bet}; \\
        &\quad \quad \quad \}\ELSESYMBOL \{\ASSIGN{money}{money + 1.1 * bet};\}\\
        &\quad \quad \}\ELSESYMBOL\{\\
        &\quad \quad \quad \PASSIGN{col2\_guard}{\Uniform{(0.0, 1.0)}};\\
        &\quad \quad \quad \IFSYMBOL \left(col2\_guard1 \le 0.5\right) \{\\
        &\quad \quad \quad \quad \PASSIGN{flip\_guard2}{\Uniform{(0.0, 1.0)}};\\
        &\quad \quad \quad \quad \IFSYMBOL \left(flip\_guard2 \le 0.33333\right) \{\\
        &\quad \quad \quad \quad \quad \ASSIGN{money}{money+1.5*bet}; \\
        &\quad \quad \quad \quad \}\ELSESYMBOL \{\ASSIGN{money}{money + 1.1 * bet};\}\\
        &\quad \quad \quad \} \ELSESYMBOL \{\\
        &\quad \quad \quad \quad \PASSIGN{flip\_guard3}{\Uniform{(0.0, 1.0)}};\\
        &\quad \quad \quad \quad \IFSYMBOL \left(flip\_guard3 \le 0.66667\right) \{\\
        &\quad \quad \quad \quad \quad \ASSIGN{money}{money+0.3*bet};\\
        &\quad \quad \quad \quad \} \ELSESYMBOL \{\texttt{skip}\}\\
        &\quad \quad \quad \}\\
        &\quad \quad \}\\
        &\quad \}\\
        &\quad \ASSIGN{i}{i+1}\\
        &\}
    \end{align*}
\end{example}

In this probabilistic program, we take the \textbf{invariant} $I = 0 \le i \wedge -1 \le money$, and synthesize dbRSM $money-10$.

\begin{example}[\textnormal{add}]
    \begin{align*}
        \cc_{\textnormal{add}}\colon
        \quad
        &\WHILESYMBOL\left(\,y \leq 1\,\right)\{\\
        &\quad	\{\ASSIGN{y}{y+1}\}~[0.2]~\{\ASSIGN{x}{x+1}\}\\
        &\}
    \end{align*}
\end{example}

In this probabilistic program, we take the \textbf{invariant} $I = 0 \le x \wedge 0 \le y \le 2$, and synthesize dbRSM $1-y$.

\begin{example}[\textnormal{Growing Walk Variant2}]
    \begin{align*}
        \cc_{\textnormal{Growing Walk Variant2}}\colon
        \quad
        &\WHILESYMBOL\left(\,r \le 0\,\right)\{\\
        &\quad	\{\ASSIGN{r}{0}\}~[0.5]~\{\ASSIGN{r}{1}\};\\
        &\quad  \ASSIGN{y}{y+x*r};\\
        &\quad  \ASSIGN{x}{x+1};\\
        &\}
    \end{align*}
\end{example}

In this probabilistic program, we take the \textbf{invariant} $ I =  0\le x \wedge 0\le y \wedge 0\le r\le 1$, and every loop iteration terminates directly with probability $p = 0.5$.

\end{document}